\documentclass[12pt]{article}
\usepackage[utf8]{inputenc}
\usepackage{csquotes}
\usepackage{bm}
\usepackage{amsmath}
\usepackage{amsthm}
\usepackage{amssymb}
\usepackage{amsfonts}
\usepackage{enumitem}
\usepackage{graphicx}
\usepackage[boxruled]{algorithm2e}
\usepackage{xfrac}
\usepackage[colorlinks=true,linkcolor=blue,citecolor=blue]{hyperref}
\usepackage{caption}
\usepackage{subcaption}
\usepackage{float}
\usepackage{accents}
\usepackage{mathtools}
\usepackage{xcolor}
\usepackage{pifont}% http://ctan.org/pkg/pifont
\newcommand{\cmark}{\ding{51}}%
\newcommand{\xmark}{\ding{55}}
\usepackage[margin=1.0in]{geometry}
\usepackage{graphicx}
\usepackage[authoryear,round]{natbib}
\usepackage{chapterbib}
\usepackage{breakurl}

\def\spacingset#1{\renewcommand{\baselinestretch}%
{#1}\small\normalsize} \spacingset{1}

\allowdisplaybreaks

\newcommand\munderbar[1]{%
  \underaccent{\bar}{#1}}
\newcommand{\F}{\mathcal{F}}
\newcommand{\B}{\mathcal{B}}

\newcommand{\X}{\mathcal{X}}
\newcommand{\dimF}{m_\F}
\newcommand{\dimB}{m_\B}

\newcommand{\XXF}{\X_\F}
\newcommand{\XXB}{\X_\B}

\newcommand{\comment}[1]{\textcolor{magenta}{\textit{JM: #1}}}

\DeclareMathOperator*{\argmin}{argmin}
\DeclareMathOperator*{\argmax}{argmax}

\DeclareMathOperator{\Tr}{trace}

\newcommand{\beginsupplement}{
\setcounter{page}{1}
\setcounter{section}{0}
\setcounter{table}{0}
\setcounter{figure}{0}
\setcounter{equation}{0}
\renewcommand{\theHsection}{SIsection.\arabic{section}}
\renewcommand{\theHtable}{SItable.\arabic{table}}
\renewcommand{\theHfigure}{SIfigure.\arabic{figure}}
\renewcommand{\theHequation}{SIequation.\arabic{equation}}
\renewcommand{\thepage}{S\arabic{page}}  
\renewcommand{\thesection}{S\arabic{section}}   
\renewcommand{\thetable}{S\arabic{table}}   
\renewcommand{\thefigure}{S\arabic{figure}}
\renewcommand{\theequation}{S\arabic{equation}}

}
\newtheorem{theorem}{Theorem}[section]

\newtheorem{proposition}[theorem]{Proposition}
\newtheorem{condition}[theorem]{Condition}
\newtheorem{definition}[theorem]{Definition}

\usepackage{authblk}

\title{Bayesian data selection}

\author[1]{Eli N.\ Weinstein}
\author[2]{Jeffrey W.\ Miller}
\affil[1]{Program in Biophysics, Harvard University, eweinstein@g.harvard.edu}
\affil[2]{Department of Biostatistics, Harvard T.H. Chan School of Public Health, jwmiller@hsph.harvard.edu}

\begin{document}

\maketitle

\begin{abstract}
Insights into complex, high-dimensional data can be obtained by discovering features of the data that match or do not match a model of interest. To formalize this task, we introduce the ``data selection'' problem: finding a lower-dimensional statistic---such as a subset of variables---that is well fit by a given parametric model of interest. A fully Bayesian approach to data selection would be to parametrically model the value of the statistic, nonparametrically model the remaining ``background'' components of the data, and perform standard Bayesian model selection for the choice of statistic. However, fitting a nonparametric model to high-dimensional data tends to be highly inefficient, statistically and computationally. We propose a novel score for performing both data selection and model selection, the ``Stein volume criterion", that takes the form of a generalized marginal likelihood with a kernelized Stein discrepancy in place of the Kullback--Leibler divergence. The Stein volume criterion does not require one to fit or even specify a nonparametric background model, making it straightforward to compute --- in many cases it is as simple as fitting the parametric model of interest with an alternative objective function. We prove that the Stein volume criterion is consistent for both data selection and model selection, and we establish consistency and asymptotic normality (Bernstein--von Mises) of the corresponding generalized posterior on parameters. We validate our method in simulation and apply it to the analysis of single-cell RNA sequencing datasets using probabilistic principal components analysis and a spin glass model of gene regulation.
\end{abstract}

\section{Introduction}

Scientists often seek to understand complex phenomena by developing working models for various special cases and subsets.
Thus, when faced with a large complex dataset, a natural question to ask is where and when a given working model applies.
We formalize this question statistically by saying that given a high-dimensional dataset, we want to identify a lower-dimensional statistic---such as a subset of variables---that follows a parametric model of interest (the working model). We refer to this problem as ``data selection'', in counterpoint to model selection, since it requires selecting the aspect of the data to which a given model applies.

For example, early studies of single-cell RNA expression showed that the expression of individual genes was often bistable, which suggests that the system of cellular gene expression might be described with the theory of interacting bistable systems, or spin glasses, with each gene a separate spin and each cell a separate observation. While it seems implausible that such a model would hold in full generality, it is quite possible that there are subsets of genes for which the spin glass model is a reasonable approximation to reality. Finding such subsets of genes is a data selection problem. In general, a good data selection method would enable one to (a) discover interesting phenomena in complex datasets, (b) identify precisely where naive application of the working model to the full dataset goes wrong, and (c) evaluate the robustness of inferences made with the working model.

Perhaps the most natural Bayesian approach to data selection is to employ a semi-parametric joint model,
using the parametric model of interest for the low-dimensional statistic (the ``foreground'') 
and using a flexible nonparametric model to explain all other aspects of the data (the ``background'').
Then, to infer where the foreground model applies, one would perform standard Bayesian model selection across different choices of the foreground statistic. 
However, this is computationally challenging due to the need to integrate over the nonparametric model for each choice of foreground statistic, making this approach quite difficult in practice.
A natural frequentist approach to data selection would be to perform a goodness-of-fit test for each choice of foreground statistic. However, this still requires specifying an alternative hypothesis, even if the alternative is nonparametric, and ensuring comparability between alternatives used for different choices of foreground statistics is nontrivial. Moreover, developing goodness-of-fit tests for composite hypotheses or hierarchical models is often difficult in practice.

In this article, we propose a new score---for both data selection and model selection---that is similar to the marginal likelihood of a semi-parametric model but does not require one to specify a background model, let alone integrate over it.
The basic idea is to employ a generalized marginal likelihood where we replace the foreground model likelihood by an exponentiated divergence with nice properties,
and replace the background model's marginal likelihood with a simple volume correction factor.
For the choice of divergence, we use a kernelized Stein discrepancy (KSD) since it enables us to provide statistical guarantees
and is easy to estimate compared to other divergences --- for instance, the Kullback--Leibler divergence involves a problematic entropy term that cannot simply be dropped.
The background model volume correction arises roughly as follows: if the background model is well-specified, then asymptotically, its divergence from the empirical distribution converges to zero and all that remains of the background model's contribution is the volume of its effective parameter space.
Consequently, it is not necessary to specify the background model, only its effective dimension.
To facilitate computation further, we develop a Laplace approximation for the foreground model's contribution to our proposed score. 

This article makes a number of novel contributions. We introduce the data selection problem in broad generality, and provide a thorough asymptotic analysis. We propose a novel model/data selection score, which we refer to as the \textit{Stein volume criterion}, that takes the form of a generalized marginal likelihood using a KSD. We provide new theoretical results for this generalized marginal likelihood and its associated posterior, complementing and building upon recent work on the frequentist properties of minimum KSD estimators~\citep{Barp2019-ut}. Finally, we provide first-of-a-kind empirical data selection analyses with two models that are frequently used in single-cell RNA sequencing analysis.

The article is organized as follows. In Section~\ref{sec:method}, we introduce the data selection problem and our proposed method. 
In Section~\ref{sec:method_asymptotics} we study the asymptotic properties of Bayesian data selection methods and compare to model selection.
Section~\ref{sec:related_work} provides a review of related work and Section~\ref{sec:toy_example} illustrates the method on a toy example. In Section~\ref{sec:theory}, we prove (a) consistency results for both data selection and model selection, (b) a Laplace approximation for the proposed score, and (c) a Bernstein--von Mises theorem for the corresponding generalized posterior. 
In Section~\ref{sec:ppca}, we apply our method to probabilistic principal components analysis (pPCA), assess its performance in simulations, and demonstrate it on single-cell RNA sequencing (scRNAseq) data. In Section~\ref{sec:on_off}, we apply our method to a spin glass model of gene expression, also demonstrated on an scRNAseq dataset. Section~\ref{sec:discussion} concludes with a brief discussion.

\section{Method} \label{sec:method}

Suppose the data $X^{(1)}, \ldots, X^{(N)} \in \mathcal{X}$ are independent and identically distributed (i.i.d.), where $\mathcal{X} \subseteq \mathbb{R}^d$. Suppose the true data-generating distribution $P_0$ has density $p_0(x)$ with respect to Lebesgue measure, and let $\{q(x|\theta): \theta \in \Theta \}$ be a parametric model of interest, where $\Theta\subseteq\mathbb{R}^m$.
We are interested in evaluating this model when applied to a projection of the data onto a subspace, $\mathcal{X}_\F \subseteq \X$ (the ``foreground'' space). Specifically, let $X_\F := V^\top X$ be a linear projection of $X\in\mathcal{X}$ onto $\mathcal{X}_\F$, where $V$ is a matrix with orthonormal columns. Let $q(x_\F|\theta)$ denote the distribution of $X_\F$ when $X\sim q(x|\theta)$, and likewise, let $p_0(x_\F)$ be the distribution of $X_\F$ when $X\sim p_0(x)$.
Even when the complete model $q(x|\theta)$ is misspecified with respect to $p_0(x)$, it may be that $q(x_\F|\theta)$ is well-specified with respect to $p_0(x_\F)$; see Figure~\ref{fig:illustrate} for a toy example. In such cases, the parametric model is only partially misspecified --- specifically, it is misspecified on the ``background'' space $\XXB$, defined as the orthogonal complement of $\XXF$. 
Our goal is to find subspaces $\XXF$ for which $q(x_\F|\theta)$ is correctly specified.

\begin{figure}[t!]
    \centering
    \begin{minipage}[t]{0.49\textwidth}\vspace{0pt}%
    \begin{subfigure}[t]{0.9\textwidth}
        \centering
        \includegraphics[height=2.5in]{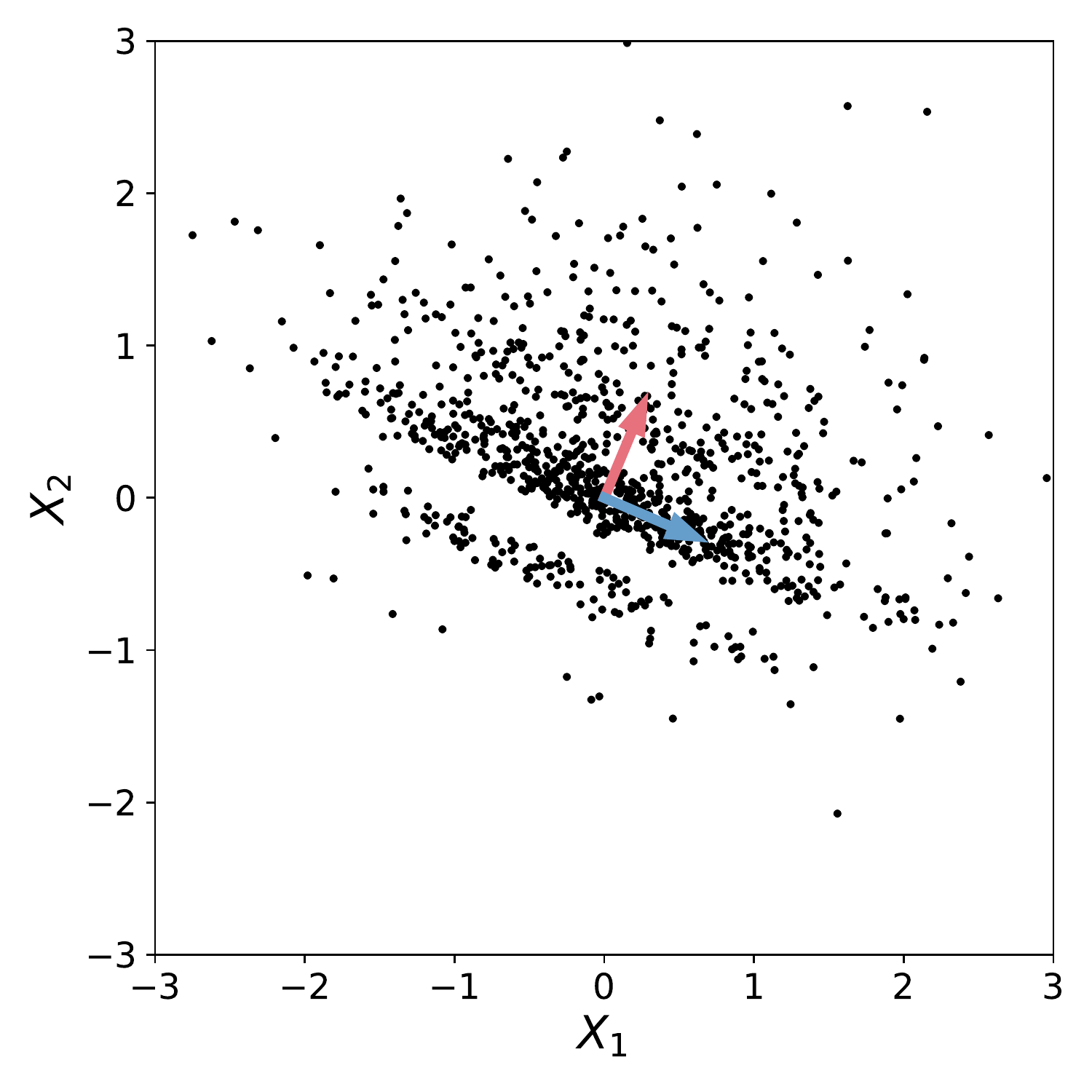}
        \caption{An example for which a bivariate normal model is partially misspecified. Basis vectors for $\X_\F$ (foreground) and $\X_\B$  (background) are blue and red, respectively.}
        \label{fig:illustrate_scat}
    \end{subfigure}
    \end{minipage}%
    \begin{minipage}[t]{0.49\textwidth}\vspace{0pt}%
    \begin{subfigure}[t]{0.9\textwidth}
        \centering
        \includegraphics[height=1.2in]{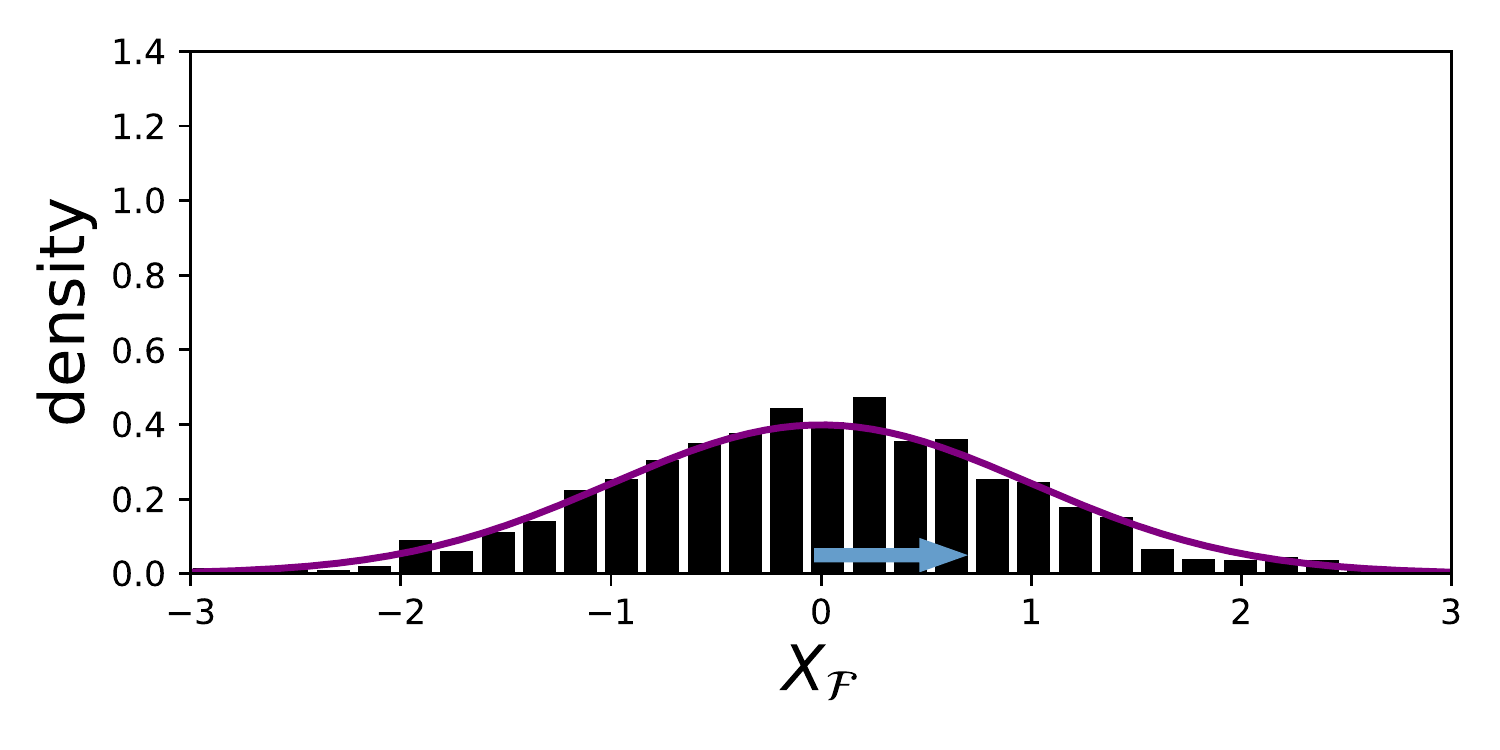}
        \caption{A univariate normal model is well-specified for the data projection onto $\XXF$.}
        \label{fig:illustrate_hist}
    \end{subfigure}
    \\
    \begin{subfigure}[t]{0.9\textwidth}
        \centering
        \includegraphics[height=1.2in]{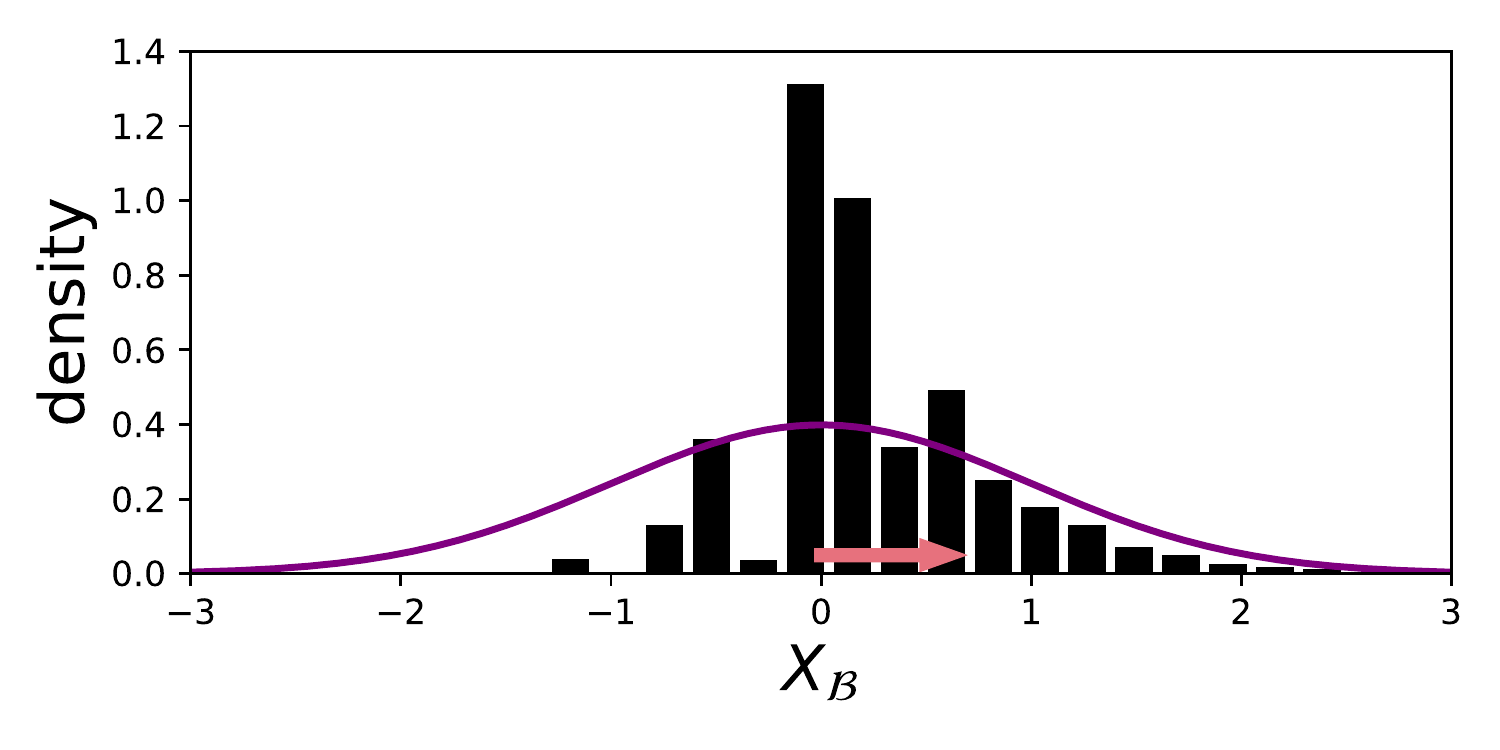}
        \caption{A univariate normal model is misspecified for the data projection onto $\XXB$.}
        \label{fig:illustrate_histB}
    \end{subfigure}
    \end{minipage}%
    \caption{A simple example illustrating the data selection problem.}
\label{fig:illustrate}
\end{figure}

A natural Bayesian solution would be to replace the background component of the assumed model, $q(x_\B | x_\F, \theta)$, with a more flexible component $\tilde{q} (x_\B | x_\F, \phi_\B)$ that is guaranteed to be well-specified with respect to $p_0(x_\B | x_\F)$, such as a nonparametric model. 
The resulting joint model, which we refer to as the ``augmented model", is then
\begin{equation}
\begin{split}
\theta \sim \pi(\theta),\ \ \ & \ \ \ X^{(i)}_{\F} \mid \theta \,\overset{\mathrm{iid}}{\sim}\, q(x_\F \mid \theta),\\
\phi_\B \sim \pi_\mathcal{B}(\phi_\B),\ \ \ & \ \ \  X^{(i)}_{\mathcal{B}} \mid X^{(i)}_{\mathcal{F}}, \phi_\B \,\sim\, \tilde{q}(x_\B \mid X^{(i)}_{\F}, \phi_\B).
\end{split}
\label{eqn:augmented_model}
\end{equation}
The standard Bayesian approach would be to put a prior on the choice of foreground space $\X_\F$, and compute the posterior over the choice of $\X_\F$.
Computing this posterior boils down to computing the Bayes factor $\tilde{q}(X^{(1:N)}|\F)/\tilde{q}(X^{(1:N)}|\F')$
for any given pair of foregrounds $\F$ and $\F'$, where $\tilde{q}(X^{(1:N)}|\F)$ denotes the marginal likelihood of $\F$ under the augmented model,
that is, $\tilde{q}(X^{(1:N)}|\F) = \int \int q(X^{(1:N)}_{\F}|\theta)\, \tilde{q}(X^{(1:N)}_{\B} | X^{(1:N)}_{\F}, \phi_\B) \pi(\theta) \pi_\mathcal{B}(\phi_\B) d\theta d\phi_\B$. 

However, in general, it is difficult to find a background model that (a) is guaranteed to be well-specified with respect to $p_0(x_\B|x_\F)$ and (b) can be integrated over in a computationally tractable way to obtain the posterior on the choice of $\F$.
Our proposed method, which we introduce next, sidesteps these difficulties while still exhibiting similar guarantees.

\subsection{Proposed score for data selection and model selection} \label{sec:proposed_score}

In this section, we propose a model/data selection score that is simpler to compute than the marginal likelihood of the augmented model and has similar theoretical guarantees. This score takes the form of a generalized marginal likelihood with a normalized kernelized Stein discrepancy ($\textsc{nksd}$) estimate taking the place of the log likelihood. Specifically, our proposed model/data selection score, termed the ``Stein volume criterion" (SVC), is 
\begin{equation}
\mathcal{K} := \left(\frac{2\pi}{N}\right)^{\dimB/2} \int \exp\!\Big(-\frac{N}{T} \widehat{\textsc{nksd}}(p_0(x_\F)\| q(x_\F|\theta))\Big) \pi(\theta) d\theta
\label{eqn:est_marg_nksd}
\end{equation}
where the ``temperature'' $T > 0$ is a hyperparameter and $m_\B$ is the effective dimension of the background model parameter space. $\widehat{\textsc{nksd}}(\cdot\|\cdot)$ is an empirical estimate of the $\textsc{nksd}$; see Equations~\ref{eqn:nksd} and \ref{eqn:est_nksd}. 
The integral in Equation~\ref{eqn:est_marg_nksd} can be approximated using techniques discussed in Section~\ref{sec:computation}.
The hyperparameter $T$ can be calibrated by comparing the coverage of the standard Bayesian posterior to the coverage of the \textsc{nksd} generalized posterior (Section~\ref{sec:calibrate_t}).
The $(2\pi/N)^{\dimB/2}$ factor penalizes higher-complexity background models. In general, we allow $m_\B$ to grow with $N$, particularly when the background model is nonparametric. Crucially, the likelihood of the background model does not appear in our proposed score, sidestepping the need to fit or even specify the background model --- indeed, the only place that the background model enters into the SVC is through $m_\B$. 

Thus, rather than specify a background model and then derive $m_\B$, one can simply specify an appropriate value of $m_\B$.
Reasonable choices of $m_\B$ can be derived by considering the asymptotic behavior of a Pitman-Yor process mixture model, a common nonparametric model that is a natural choice for a background model. 
A Pitman-Yor process mixture model with discount parameter $\alpha \in (0,1)$, concentration parameter $\theta > -\alpha$, and $D$-dimensional component parameters will asymptotically have expected effective dimension
\begin{equation} \label{eqn:pymm}
	m_\B \sim D\frac{\Gamma(\theta + 1)}{\alpha\Gamma(\theta + \alpha)} N^{\alpha}
\end{equation}
under the prior, where $a_N \sim b_N$ means that $a_N/b_N \to 1$ as $N \to \infty$ and $\Gamma(\cdot)$ is the gamma function (\citealp{Pitman2002-ma}, \S 3.3).
As a default, we recommend setting $m_\B = c_\B\, r_\B \sqrt{N}$, where $r_\B$ is the dimension of $\X_\B$ and $c_\B$ is a constant chosen to match Equation~\ref{eqn:pymm} with $\alpha = 1/2$. The $\sqrt{N}$ scaling is particularly nice in terms of asymptotic guarantees; see Section~\ref{sec:nested_ds}. 

The SVC uses a novel, normalized version of the $\textsc{ksd}$ between densities $p(x)$ and $q(x)$:
\begin{equation}
    \textsc{nksd}(p(x)\| q(x)) := \frac{\mathbb{E}_{X,Y \sim p} \big[  (s_{q}(X) - s_{p}(X))^\top (s_{q}(Y) - s_{p}(Y)) k(X, Y) \big]}{\mathbb{E}_{X,Y \sim p} [ k(X, Y) ]}
\label{eqn:nksd}
\end{equation}
where $k(x,y)\in\mathbb{R}$ is an integrally strictly positive definite kernel, $s_{q}(x) := \nabla_x \log q(x)$, and $s_{p}(x) := \nabla_x \log p(x)$; see Section~\ref{sec:nksd_properties} for details. The numerator corresponds to the standard $\textsc{ksd}$~\citep{Liu2016-bp}.
The denominator, which is strictly positive and independent of $q(x)$, is a normalization factor that we have introduced to make the divergence comparable across spaces of different dimension. 
See Section~\ref{sec:si_kernel_choice} for kernel recommendations. 
% As a function of $q$, the \textsc{nksd} is uniquely minimized when $p_0(x_\F) = q(x_\F|\theta)$, at which point it becomes zero; see Theorem~\ref{}.
Extending the technique of \citet{Liu2016-bp}, we propose to estimate the normalized KSD using U-statistics:
\begin{equation}
\widehat{\textsc{nksd}}(p(x)\| q(x)) = \frac{\sum_{i \neq j} u(X^{(i)}, X^{(j)})}{\sum_{i \neq j} k(X^{(i)}, X^{(j)})}
\label{eqn:est_nksd}
\end{equation}
where $X^{(i)}\sim p(x)$ i.i.d., the sums are over all $i,j\in\{1,\ldots,N\}$ such that $i\neq j$, and 
\begin{equation*}
	u(x, y) := 
    s_{q}(x)^\top s_{q}(y) k(x, y) + s_{q}(x)^\top \nabla_{y} k(x, y) + s_{q}(y)^\top \nabla_{x} k(x, y) + \Tr(\nabla_{x} \nabla_{y}^\top k(x, y)).
\end{equation*}
Importantly, Equation~\ref{eqn:est_nksd} does not require knowledge of $s_{p}(x)$, which is unknown in practice.

\subsection{Comparison with the standard marginal likelihood} \label{sec:approx_standard}

It is instructive to compare our proposed model/data selection score, the Stein volume criterion, to the standard marginal likelihood $\tilde{q}(X^{(1:N)}|\F)$. In particular, we show that the SVC approximates a generalized version of the marginal likelihood.
To see this, first define $H := -\int p_0(x) \log p_0(x) dx$, the entropy of the complete data distribution, and note that if were $H$ somehow known, then the Kullback-Leibler (\textsc{kl}) divergence between the augmented model and the data distribution could be approximated as
\begin{equation*}
	\widehat{\textsc{kl}}(p_0(x)\|q(x_{\F}|\theta)\, \tilde{q}(x_{\B} | x_{\F}, \phi_\B)) := -\frac{1}{N} \sum_{i=1}^N\log q(X^{(i)}_{\F}|\theta)\, \tilde{q}(X^{(i)}_{\B} | X^{(i)}_{\F}, \phi_\B)-H.
\end{equation*}
Since multiplying the marginal likelihoods by a fixed constant does not affect the Bayes factors, the following expression could be used instead of the marginal likelihood $\tilde{q}(X^{(1:N)}|\F)$ to decide among foreground subspaces:
\begin{equation}
	\frac{\tilde{q}(X^{(1:N)}|\F)}{\exp(-N H)} = \int \int \exp\!\Big( -N\, \widehat{\textsc{kl}}(p_0(x)\|q(x_{\F}|\theta)\, \tilde{q}(x_{\B} | x_{\F}, \phi_\B))\Big) \pi(\theta) \pi_\mathcal{B}(\phi_\B) d\theta d\phi_\B.
\end{equation}
Now, consider a generalized marginal likelihood where the \textsc{nksd} replaces the \textsc{kl}:
\begin{equation}
\tilde{\mathcal{K}} := \int \int \exp\!\Big( -N\frac{1}{T} \widehat{\textsc{nksd}}\big(p_0(x)\| q(x_{\F}|\theta)\, \tilde{q}(x_{\B} | x_{\F}, \phi_\B)\big) \Big) \pi(\theta)\pi_\mathcal{B}(\phi_\B) d\theta d\phi_\B.
\end{equation}
We refer to $\tilde{\mathcal{K}}$ as the ``\textsc{nksd} marginal likelihood" of the augmented model. 
Intuitively, we expect it to behave similarly to the standard marginal likelihood, except that it quantifies the divergence between the model and data distributions using the \textsc{nksd} instead of the \textsc{kl}.

However, a key advantage of the \textsc{nksd} marginal likelihood is that it admits a simple approximation via the SVC when the background model is well-specified, unlike the standard marginal likelihood.  
For instance, if the foreground and background are independent, that is, $p_0(x) = p_0(x_\F) p_0(x_\B)$ and $\tilde{q}(x_\B|x_\F, \phi_\B) = \tilde{q}(x_\B|\phi_\B)$, then the theory in Section~\ref{sec:theory} can be extended to the full augmented model to show that
\begin{equation} \label{eqn:full_svc_approx}
\begin{split}
	\frac{\log \tilde{\mathcal{K}}}{\log \mathcal{K}} \xrightarrow[N \to \infty]{P_0} 1,
\end{split}
\end{equation}
where $\mathcal{K}$ is the SVC (Equation~\ref{eqn:est_marg_nksd}).
Thus, the SVC approximates the \textsc{nksd} marginal likelihood of the augmented model, 
suggesting that the SVC may be a convenient alternative to the standard marginal likelihood.
Formally, Section~\ref{sec:method_asymptotics} shows that the SVC exhibits consistency properties similar to the standard marginal likelihood, even when $p_0(x) \neq p_0(x_\F) p_0(x_\B)$.

\subsection{Computation} \label{sec:computation}

Next, we discuss methods for computing the SVC including exact solutions, Laplace/BIC approximation, variational approximation, and comparing many possible choices of $\F$.

% In this section, we show that the SVC admits approximations similar to those used for the marginal likelihood. We also provide methods for comparing many possible choices of $\F$.

\subsubsection{Exact solution for exponential families}

When the foreground model is an exponential family, the SVC can be computed analytically.
Specifically, in Section~\ref{sec:si_conjugacy}, we show if $q(x_\F|\theta) = \lambda(x_\F) \exp(\theta^\top t(x_\F) - \kappa(\theta))$, then
\begin{equation}
	\widehat{\textsc{nksd}}(p_0(x_\F)\| q(x_\F|\theta)) = \theta^\top A\, \theta + B^\top \theta + C
\end{equation}
where $A$, $B$, and $C$ depend on the data $X^{(1:N)}$ but not on $\theta$. Therefore, we can place a multivariate Gaussian prior on $\theta$ and compute the SVC in closed form; see Section~\ref{sec:si_conjugacy}.

\subsubsection{Laplace and BIC approximations} \label{sec:lapl_bic}

The Laplace approximation is a widely-used technique for computing marginal likelihoods. In Theorem~\ref{thm:marginal_ksd}, we establish regularity conditions under which a Laplace approximation to the SVC is justified by being asymptotically correct. The resulting approximation is
\begin{equation}
\mathcal{K} \approx \frac{\exp\left(-\frac{N}{T} \widehat{\textsc{nksd}}(p_0(x_\F)\| q(x_\F|\theta_N))\right)\pi(\theta_N)}{|\det \frac{1}{T}\nabla^2_\theta\, \widehat{\textsc{nksd}}(p_0(x_\F)\| q(x_\F|\theta_N))|^{1/2}}\left(\frac{2\pi}{N}\right)^{(\dimF + \dimB)/2}
\label{eqn:ksd_laplace}
\end{equation}
where $\theta_N := \argmin_\theta \widehat{\textsc{nksd}}(p_0(x_\F)\| q(x_\F|\theta))$ is the point at which the estimated \textsc{nksd} is minimized, the ``minimum Stein discrepancy estimator" as defined by \citet{Barp2019-ut}.

We can also make a rougher approximation, analogous to the Bayesian information criterion (BIC), which does not require one to compute second derivatives of $\widehat{\textsc{nksd}}$:
\begin{equation} \label{eqn:bic_approx}
\mathcal{K} \approx \exp\!\Big(-\frac{N}{T} \widehat{\textsc{nksd}}(p_0(x_\F)\| q(x_\F|\theta_N))\Big)\left(\frac{2\pi}{N}\right)^{(\dimF + \dimB)/2}.
\end{equation}
This approximation is easy to compute, given a minimum Stein discrepancy estimator $\theta_N$. Like the SVC, it satisfies all of our consistency desiderata (Section~\ref{sec:asymptotics}).
However, we expect it to perform worse than the SVC when there is not yet enough data for the $\textsc{nksd}$ posterior to be highly concentrated,
that is, when a range of $\theta$ values can plausibly explain the data.

\subsubsection{Comparing many foregrounds using approximate optima} \label{sec:approximate_optima}

Often, we would like to evaluate many possible subspaces $\X_\F$ when performing data selection. Even when using the Laplace or BIC approximation to the SVC, this can get computationally prohibitive since we need to re-optimize to find $\theta_N$ for every $\F$ under consideration.  Here, we propose a way to reduce this cost by making a fast linear approximation. 
Define $\ell_j(\theta) := \widehat{\textsc{nksd}}(p_0(x_{\F_j})\| q(x_{\F_j}|\theta))$ for $j \in \{1, 2\}$. For $w \in [0,1]$, we can linearly interpolate
\begin{equation}
	\theta_N(w) := \argmin_\theta \ell_1(\theta) + w (\ell_2(\theta) - \ell_1(\theta)).
\end{equation}
Now, $\theta_N(0)$ and $\theta_N(1)$ are the minimum Stein discrepancy estimators for $\F_1$ and $\F_2$, respectively. Given $\theta_N(0)$, we can approximate $\theta_N(1)$ by applying the implicit function theorem and a first-order Taylor expansion (Section~\ref{sec:si_approx_optima}):
\begin{equation}
	\theta_N(1) \approx \theta_N(0) - \nabla_\theta^2 \ell_1(\theta_N(0))^{-1} \nabla_\theta \ell_2(\theta_N(0)).
\label{eqn:IJ_grad}
\end{equation}
Note that the derivatives of $\ell_j$ are often easy to compute with automatic differentiation~\citep{Baydin2018-qp}. 
Note also that when we are comparing one foreground subspace, such as $\X_{\F_1} = \X$, to many other foreground subspaces $\X_{\F_2}$, the inverse Hessian $\nabla_\theta^2 \ell_1(\theta_N(0))^{-1}$ only needs to be computed once.
Thus, Equation~\ref{eqn:IJ_grad} provides a fast method for computing Laplace or BIC approximations to the SVC for a large number of candidate foregrounds $\F$.

\subsubsection{Variational approximation} \label{sec:vi}

Variational inference is a method for approximating both the posterior distribution and the marginal likelihood of a probabilistic model. Since the SVC takes the form of a generalized marginal likelihood, we can derive a variational approximation to the SVC. Let $r_\zeta(\theta)$ be an approximating distribution parameterized by $\zeta$. By Jensen's inequality, we have
\begin{equation} \label{eqn:svc_vi}
\begin{split}
	\log \int &\exp\!\Big(-\frac{N}{T} \widehat{\textsc{nksd}}(p_0(x_\F)\| q(x_\F|\theta))\Big) \pi(\theta) d\theta\\
	 & = \log \int \frac{\exp\left(-\frac{N}{T} \widehat{\textsc{nksd}}(p_0(x_\F)\| q(x_\F|\theta))\right) \pi(\theta)}{r_\zeta(\theta)} r_\zeta(\theta) d\theta\\
	& \ge \mathbb{E}_{r_\zeta}\left[\log\!\bigg(\frac{\exp\left(-\frac{N}{T} \widehat{\textsc{nksd}}(p_0(x_\F)\| q(x_\F|\theta))\right) \pi(\theta)}{r_\zeta(\theta)} \bigg)\right]\\
	&= -\frac{N}{T} \mathbb{E}_{r_\zeta}\left[\widehat{\textsc{nksd}}(p_0(x_\F)\| q(x_\F|\theta))\right] + \mathbb{E}_{r_\zeta}[\log \pi(\theta)] - \mathbb{E}_{r_\zeta}[\log r_\zeta(\theta)].
\end{split}
\end{equation}
Maximizing this lower bound with respect to the variational parameters $\zeta$ provides an approximation to the SVC,
or more precisely, to $\log \mathcal{K} - (m_\B/2)\log(2\pi/N)$.
Note that this variational approximation falls within the framework of generalized variational inference proposed by \citet{Knoblauch2019-nc}.

\section{Data selection and model selection consistency} \label{sec:method_asymptotics}

This section presents our consistency results when comparing two different foreground subspaces (data selection) or two different foreground models (model selection).
The theory supporting these results is in Sections~\ref{sec:theory} and \ref{sec:asymptotics}.
We consider four distinct properties that a procedure would ideally exhibit:
data selection consistency, nested data selection consistency, model selection consistency, and nested model selection consistency;
see Section~\ref{sec:theory_ds_ms} for precise definitions.
We consider six possible model/data selection scores, and we establish which scores satisfy which properties;
see Table~\ref{tbl:consistency}.
The SVC and the full marginal likelihood are the only two of the six scores that satisfy all four consistency properties.

The intuition behind Bayesian model selection is often explained in terms of Occam's razor: a theory should be as simple as possible but no simpler. Data selection and nested data selection encapsulate a complementary intuition: a theory should explain as much of the data as possible but no more. In other words, when choosing between foreground spaces, a consistent data selection score will asymptotically prefer the highest-dimensional space on which the model is correctly specified.

% The Stein volume criterion is designed to satisfy data selection consistency, model selection consistency, nested data selection consistency and nested model selection consistency. We have described how various modifications to this score would no longer satisfy all of these properties at once. This is summarized in Table~\ref{tbl:consistency}.

As in standard model selection, a practical concern in data selection is robustness. For instance, if the foreground model is even slightly misspecified on $\X_{\F_2}$, then the empty foreground $\X_{\F_1} = \varnothing$ will be asymptotically preferred over $\X_{\F_2}$.
Since the SVC takes the form of a generalized marginal likelihood, techniques for improving robustness with the standard marginal likelihood---such as coarsened posteriors, power posteriors, and BayesBag---could potentially be extended to address this issue~\citep{Miller2019-zj,Huggins2020-bu}. We leave exploration of such approaches to future work.

\begin{table}[]
\centering
\begin{tabular}{lc|c|c|c|}
                                                    & \multicolumn{4}{c}{Consistency property} \\ \cline{2-5}
\multicolumn{1}{l|}{Score}                          & d.s. & nested d.s. & m.s. & nested m.s. \\ \hline
\multicolumn{1}{|l|}{$\tilde{q}(X^{(1:N)}|\F)$ full marginal likelihood} & \cmark & \cmark & \cmark  &   \cmark \\ \hline
\multicolumn{1}{|l|}{$\mathcal{K}^{(\mathrm{a})}$ foreground marg lik, background volume} & \xmark & \xmark & \cmark  &   \cmark \\ \hline
\multicolumn{1}{|l|}{$\mathcal{K}^{(\mathrm{b})}$ foreground marg NKSD} & \cmark & \xmark & \cmark  &   \cmark \\ \hline
\multicolumn{1}{|l|}{$\mathcal{K}^{(\mathrm{c})}$ foreground marg KL, background volume} & \cmark & \xmark & \cmark  &   \cmark \\ \hline
\multicolumn{1}{|l|}{$\mathcal{K}^{(\mathrm{d})}$ foreground NKSD, background volume} & \cmark & \cmark & \cmark  &   \xmark \\ \hline
\multicolumn{1}{|l|}{$\mathcal{K}$ foreground marg NKSD, background volume} & \cmark & \cmark & \cmark  &   \cmark \\ \hline
\end{tabular}
\caption{Consistency properties satisfied by various model/data selection scores.
Only the Stein volume criterion $\mathcal{K}$ and the full marginal likelihood $\tilde{q}(X^{(1:N)}|\F)$ satisfy all four desiderata.
(d.s.\ = data selection, m.s.\ = model selection, marg = marginal, lik = likelihood.)}
\label{tbl:consistency}
\end{table}

\subsection{Data selection consistency} \label{sec:data_select_consistency}

First, consider comparisons between different choices of foreground, $\F_1$ and $\F_2$. 
When the model is correctly specified over $\F_1$ but not $\F_2$, we refer to asymptotic concentration on $\F_1$ as ``data selection consistency"
(and vice versa if $\F_2$ is correct but not $\F_1$).
For the standard marginal likelihood of the augmented model, we have (see Section~\ref{sec:si_ds})
\begin{equation}
	\frac{1}{N}\log\frac{\tilde{q}(X^{(1:N)}|\F_1)}{\tilde{q}(X^{(1:N)}|\F_2)} \xrightarrow[N \to \infty]{P_0} \textsc{kl}(p_0(x_{\F_2}) \| q(x_{\F_2}|\theta^{\textsc{kl}}_{2,*})) - \textsc{kl}(p_0(x_{\F_1})\| q(x_{\F_1}|\theta^{\textsc{kl}}_{1,*}))
\label{eqn:data_select_kl}
\end{equation}
where $\theta^{\textsc{kl}}_{j,*} := \argmin \textsc{kl}(p_0(x_{\F_j}) \| q(x_{\F_j}|\theta))$ for $j\in\{1,2\}$, that is, $\theta^{\textsc{kl}}_{j,*}$ is the parameter value that minimizes the \textsc{kl} divergence between the projected data distribution $p_0(x_{\F_j})$ and the projected model $q(x_{\F_j}|\theta)$.
Thus, $\tilde{q}(X^{(1:N)}|\F_j)$ asymptotically concentrates on the $\F_j$ on which the projected model can most closely match the data distribution in terms of $\textsc{kl}$. 

In Theorem~\ref{thm:selection_consistency}, we show that under mild regularity conditions, the Stein volume criterion behaves precisely the same way but with the $\textsc{nksd}$ in place of the $\textsc{kl}$:
\begin{equation}
	\frac{1}{N}\log\frac{\mathcal{K}_1}{\mathcal{K}_2} \xrightarrow[N \to \infty]{P_0} \frac{1}{T}\textsc{nksd}(p_0(x_{\F_2})\| q(x_{\F_2}|\theta^{\textsc{nksd}}_{2,*})) - \frac{1}{T}\textsc{nksd}(p_0(x_{\F_1})\| q(x_{\F_1}|\theta^{\textsc{nksd}}_{1,*}))
\label{eqn:data_select_nksd}
\end{equation}
where $\theta^{\textsc{nksd}}_{j,*} := \argmin \textsc{nksd}(p_0(x_{\F_j})\| q(x_{\F_j}|\theta))$ for $j\in\{1,2\}$. 
% Asymptotically, then, the difference between the marginal likelihood and the Stein volume criterion is simply that the SVC measures distance between distributions in terms of the $\textsc{nksd}$ divergence instead of $\textsc{kl}$.
Therefore, $\tilde{q}(X^{(1:N)}|\F)$ and $\mathcal{K}$ both yield data selection consistency.
It is important here that the SVC uses a true divergence, rather than a divergence up to a data-dependent constant. If we instead used
\begin{equation}\label{eqn:A1}
\mathcal{K}^{(\mathrm{a})} := \left(\frac{2\pi}{N}\right)^{m_\B/2} q(X_\F^{(1:N)}),
\end{equation}
which employs the foreground marginal likelihood $q(X_\F^{(1:N)}) = \int q(X_\F^{(1:N)}|\theta)\pi(\theta)d\theta$ and a background volume correction, we would get qualitatively different behavior (Section~\ref{sec:si_ds}):
\begin{equation}
	\frac{1}{N}\log\frac{\mathcal{K}^{(\mathrm{a})}_1}{\mathcal{K}^{(\mathrm{a})}_2} \xrightarrow[N \to \infty]{P_0} \textsc{kl}(p_0(x_{\F_2})\|q(x_{\F_2}|\theta^{\textsc{kl}}_{2,*})) - \textsc{kl}(p_0(x_{\F_1}) \| q(x_{\F_1}|\theta^{\textsc{kl}}_{1,*})) + H_{\F_2} - H_{\F_1}
\label{eqn:K1_asymptotic}
\end{equation}
where $H_{\F_j} := - \int p_0(x_{\F_j}) \log p_0(x_{\F_j}) dx_{\F_j}$ is the entropy of $p_0(x_{\F_j})$ for $j \in \{1,2\}$. In short, the naive score $\mathcal{K}^{(\mathrm{a})}$ is a bad choice: it decides between data subspaces based not just on how well the parametric foreground model performs, but also on the entropy of the data distribution in each space. As a result, $\mathcal{K}^{(\mathrm{a})}$ does not exhibit data selection consistency.

\subsection{Nested data selection consistency} \label{sec:nested_ds}

When $\X_{\F_2} \subset \X_{\F_1}$, we refer to the problem of deciding between subspaces $\F_1$ and $\F_2$ as nested data selection, in counterpoint to nested model selection, where one model is a subset of another~\citep{Vuong1989-gj}. 
If the model $q(x|\theta)$ is well-specified over $\X_{\F_1}$, then it is guaranteed to be well-specified over any lower-dimensional sub-subspace $\X_{\F_2} \subset \X_{\F_1}$; in this case, we refer to asymptotic concentration on $\F_1$ as ``nested data selection consistency". 
In this situation, $\textsc{kl}(p_0(x_{\F_j})\|q(x_{\F_j}|\theta^{\textsc{kl}}_{j,*}))$ and $\textsc{nksd}(p_0(x_{\F_j}), q(x_{\F_j}|\theta^{\textsc{nksd}}_{j,*}))$ are both zero for $j\in\{1,2\}$, making it necessary to look at higher-order terms in Equations~\ref{eqn:data_select_kl} and \ref{eqn:data_select_nksd}.
In Section~\ref{sec:si_nds}, we show that 
if $\X_{\F_2} \subset \X_{\F_1}$, $q(x|\theta)$ is well-specified over $\X_{\F_1}$, 
the background models are well-specified, and their dimensions $m_{\B_1}$ and $m_{\B_2}$ are constant with respect to $N$, then
 % to establish the asymptotics of the augmented model marginal likelihood and the Stein volume criterion. 
\begin{equation}
		\frac{1}{\log N}\log \frac{\tilde{q}(X^{(1:N)}|\F_1)}{\tilde{q}(X^{(1:N)}|\F_2)} \xrightarrow[N \to \infty]{P_0} \frac{1}{2}(m_{\F_2} + m_{\B_2} - m_{\F_1} - m_{\B_1})
\label{eqn:nested_data_kl}
\end{equation}
where $m_{\F_j}$ is the effective dimension of the parameter space of $q(x_{\F_j}|\theta)$.
% If $\X_{\F_2} \subset \X_{\F_1}$ then $\X_{\B_1} \subset \X_{\B_2}$ and typically $m_{\B_1} < m_{\B_2}$, since a larger background subspace would typically require a greater number of background model parameters.
% Thus, in this case, $\tilde{q}(X^{(1:N)}|\F_j)$ asymptotically concentrates on the larger foreground $\F_1$ by Equation~\ref{eqn:nested_data_kl}. 
% This makes practical sense: our goal is to identify not just where the parametric model applies, but also where it fails, so we would like to select the largest possible foreground subspace without running into misspecification. 
In Theorem~\ref{thm:selection_consistency}, we show that under mild regularity conditions, the SVC behaves the same way:
\begin{equation}
		\frac{1}{\log N}\log \frac{\mathcal{K}_1}{\mathcal{K}_2} \xrightarrow[N \to \infty]{P_0} \frac{1}{2}(m_{\F_2} + m_{\B_2} - m_{\F_1} - m_{\B_1}).
\label{eqn:nested_data_nksd}
\end{equation}
% We are often interested in models $q(x|\theta)$ for which $q(x_{\F_j}|\theta)$ is only dependent on a subset of the components of $\theta$.
% and the rest are nuisance parameters with respect to $q(x_{\F_j}|\theta)$.
% When $m_{\F_1} \neq m_{\F_2}$, the limits of Equations~\ref{eqn:nested_data_kl} and \ref{eqn:nested_data_nksd} both become $\frac{1}{2}(m_{\B_2} + m_{\F_2} - m_{\B_1} - m_{\F_1})$; see Theorem~\ref{thm:marginal_ksd} for technical conditions. 
Thus, so long as $m_{\F_2} + m_{\B_2} > m_{\F_1} + m_{\B_1}$ whenever $\X_{\F_2} \subset \X_{\F_1}$, the marginal likelihood and the SVC asymptotically concentrate on the larger foreground $\F_1$; hence, they both exhibit nested data selection consistency.  This is a natural assumption since the background model is generally more flexible---on a per dimension basis---than the foreground model.

The volume correction $(2\pi/N)^{m_\B/2}$ in the definition of the SVC is 
important for nested data selection consistency (Equation~\ref{eqn:nested_data_nksd}).
An alternative score without that correction,
\begin{equation}\label{eqn:A2}
	\mathcal{K}^{(\mathrm{b})} := \int \exp\!\Big(-\frac{N}{T} \widehat{\textsc{nksd}}(p_0(x_\F)\| q(x_\F|\theta))\Big) \pi(\theta) d\theta,
\end{equation}
exhibits data selection consistency (Equation~\ref{eqn:data_select_nksd} holds for $\mathcal{K}^{(\mathrm{b})}$), but not nested data selection consistency; see Sections~\ref{sec:si_ds} and \ref{sec:si_nds}.
More subtly, the asymptotics of the SVC in the case of nested data selection also depend on the variance of U-statistics. To illustrate, consider a score that is similar to the SVC but uses $\widehat{\textsc{kl}}$ instead of $\widehat{\textsc{nksd}}$:
\begin{equation}
	\mathcal{K}^{(\mathrm{c})} := \left(\frac{2\pi}{N}\right)^{\dimB/2} \int \exp\!\Big(-N \widehat{\textsc{kl}}(p_0(x_\F)\| q(x_\F|\theta))\Big)\pi(\theta)d\theta
\end{equation}
where $\widehat{\textsc{kl}}(p_0(x_\F)\| q(x_\F|\theta)) := -\frac{1}{N} \sum_{i=1}^N \log q(X_\F^{(i)}|\theta) - H_\F$ and $H_\F$ is required to be known. The score $\mathcal{K}^{(\mathrm{c})}$ exhibits data selection consistency, but not nested data selection consistency.
The reason is that the error in estimating the \textsc{kl} is of order $1/\sqrt{N}$ by the central limit theorem,
and this source of error dominates the $\log N$ term contributed by the volume correction; see Section~\ref{sec:si_nds}.
% \begin{equation} \label{eqn:nksd_OsN}
	% N^{1/2}\big[\widehat{\textsc{kl}}(p_0(x_{\F})\| q(x_\F|\theta_{*}^{\textsc{kl}})) - \textsc{kl}(p_0(x_{\F})\| q(x_{\F}|\theta_{*}^{\textsc{kl}}))\big] \xrightarrow[N\to\infty]{D} 
    % \mathcal{N}\big(0,\,\mathbb{V}_{P_0}(\log q(X_\F|\theta^{\textsc{kl}}_{*}))\big)
% \end{equation}
% by the central limit theorem, 
% where $\theta^{\textsc{kl}}_{*} := \argmin \textsc{kl}(p_0(x_{\F}) \| q(x_{\F}|\theta))$; see Section~\ref{sec:si_nds}.
% In the large $N$ limit, this source of error dominates the $O_{P_0}(\log N)$ term contributed by the volume correction.
Meanwhile, the error in estimating the \textsc{nksd} is of order $1/N$ when the model is well-specified, due to the rapid convergence rate of the U-statistic estimator.
% \begin{equation} \label{eqn:nksd_O1}
	% \widehat{\textsc{nksd}}(p_0(x_{\F})\| q(x_{\F}|\theta_{*}^{\textsc{nksd}})) - \textsc{nksd}(p_0(x_{\F})\| q(x_{\F}|\theta_{*}^{\textsc{nksd}}))  = O_{P_0}(N^{-1})
% \end{equation}
% where $\theta^{\textsc{nksd}}_{*} := \argmin \textsc{nksd}(p_0(x_{\F})\| q(x_{\F}|\theta))$.
Thus, in the SVC, this source of error is dominated by the volume correction; see Theorem~\ref{prop:laplace_scaling}.
% , and also see Theorem 4.1 of \citet{Liu2016-bp}.

The nested data selection results we have described so far assume $m_\B$ does not depend on $N$, or at least 
$m_{\B_2} - m_{\B_1}$ does not depend on $N$ (Theorem~\ref{thm:selection_consistency}).
However, in Section~\ref{sec:proposed_score}, we suggest setting $m_\B = c_\B \, r_\B \sqrt{N}$ where $c_\B$ is a constant and $r_\B$ is the dimension of $\X_\B$. With this choice, the asymptotics of the SVC for nested data selection become (Theorem~\ref{thm:selection_consistency})
\begin{equation}
		\frac{1}{\sqrt{N}\log N}\log \frac{\mathcal{K}_1}{\mathcal{K}_2} \xrightarrow[N \to \infty]{P_0} \frac{1}{2} c_\B \, (r_{\B_2} - r_{\B_1}).
\label{eqn:nested_data_nksd_np}
\end{equation}
Since $r_{\B_1} < r_{\B_2}$ when $\X_{\F_2} \subset \X_{\F_1}$, the SVC concentrates on the larger foreground $\F_1$, yielding nested data selection consistency.
Going beyond the well-specified case, Theorem~\ref{thm:selection_consistency} shows that Equation~\ref{eqn:nested_data_nksd_np} holds when $\textsc{nksd}(p_0(x_{\F_1}) \| q(x_{\F_1}|\theta^{\textsc{nksd}}_{1,*})) = \textsc{nksd}(p_0(x_{\F_2}) \| q(x_{\F_2}|\theta^{\textsc{nksd}}_{2,*})) \neq 0$, that is, when the models are misspecified by the same amount as measured by the $\textsc{nksd}$.
Equation~\ref{eqn:nested_data_nksd_np} holds regardless of whether $m_{\F_1}$ is equal to $m_{\F_2}$.

\subsection{Model selection and nested model selection consistency} \label{sec:model_selection_asymptotics}

Consider comparing different foreground models $q_1(x_\F|\theta_1)$ and $q_2(x_\F|\theta_2)$ over the same subspace $\X_\F$,
while using the same background model.
We say that a score exhibits ``model selection consistency'' if it concentrates on the correct model,
when one of the models is correctly specified and the other is not.
When the two models are nested and both are correct, a score exhibits ``nested model selection consistency'' if it concentrates on the simpler model.

Like the standard marginal likelihood, the SVC exhibits both types of model selection consistency. 
The standard marginal likelihood satisfies (Section~\ref{sec:si_ms})
\begin{equation}
	\frac{1}{N}\log \frac{\tilde{q}_1(X^{(1:N)}|\F)}{\tilde{q}_2(X^{(1:N)}|\F)} \xrightarrow[N \to \infty]{P_0} \textsc{kl}(p_0(x_{\F})\| q_2(x_{\F}|\theta^{\textsc{kl}}_{2,*})) - \textsc{kl}(p_0(x_{\F})\| q_1(x_{\F}|\theta^{\textsc{kl}}_{1,*}))
\label{eqn:model_selection}
\end{equation}
where $\theta^{\textsc{kl}}_{j,*} := \argmin \textsc{kl}(p_0(x_{\F}) \| q_j(x_{\F}|\theta_j))$ for $j\in\{1,2\}$. Analogously, by Theorem~\ref{thm:selection_consistency},
\begin{equation}
	\frac{1}{N}\log \frac{\mathcal{K}_1}{\mathcal{K}_2} \xrightarrow[N \to \infty]{P_0} \frac{1}{T}\textsc{nksd}(p_0(x_{\F})\| q_2(x_{\F}|\theta^{\textsc{nksd}}_{2,*})) - \frac{1}{T}\textsc{nksd}(p_0(x_{\F})\| q_1(x_{\F}|\theta^{\textsc{nksd}}_{1,*}))
\label{eqn:model_selection_nksd}
\end{equation}
where $\theta^{\textsc{nksd}}_{j,*} := \argmin \textsc{nksd}(p_0(x_{\F}) \| q_j(x_{\F}|\theta_j))$ for $j\in\{1,2\}$. 
Thus, for both scores, concentration occurs on the model that comes closer to the data distribution in terms of the corresponding divergence ($\textsc{kl}$ or $\textsc{nksd}$).

For nested model selection, suppose both foreground models are well-specified and $m_{\B_1} = m_{\B_2}$. 
Letting $m_{\F,j}$ be the parameter dimension of $q_j(x_{\F}|\theta_j)$, we have (Section~\ref{sec:si_nms})
\begin{equation}
		\frac{1}{\log N}\log \frac{\tilde{q}_1(X^{(1:N)}|\F)}{\tilde{q}_2(X^{(1:N)}|\F)} \xrightarrow[N \to \infty]{P_0} \frac{1}{2}(m_{\F,2} - m_{\F,1}).
\label{eqn:nested_model_selection}
\end{equation}
In Theorem~\ref{thm:selection_consistency}, we show that the SVC behaves identically:
\begin{equation}
		\frac{1}{\log N}\log \frac{\mathcal{K}_1}{\mathcal{K}_2} \xrightarrow[N \to \infty]{P_0} \frac{1}{2}(m_{\F,2} - m_{\F,1}).
\label{eqn:nested_model_nksd}
\end{equation}
Here, a key role is played by the volume of the foreground parameter space, which quantifies the foreground model complexity.
The SVC accounts for this by integrating over foreground parameter space.
Meanwhile, a naive alternative that ignores the foreground volume, 
\begin{equation}
	\mathcal{K}^{(\mathrm{d})} := \left(\frac{2\pi}{N}\right)^{m_\B/2} \exp\!\Big(-\frac{N}{T} \min_\theta \widehat{\textsc{nksd}}(p_0(x_{\F})\| q(x_{\F}|\theta))\Big),
\end{equation}
exhibits model selection consistency (Equation~\ref{eqn:model_selection_nksd} holds for $\mathcal{K}^{(\mathrm{d})}$) but not nested model selection consistency (Section~\ref{sec:si_nms}).
The Laplace and BIC approximations to the SVC (Equations~\ref{eqn:ksd_laplace} and \ref{eqn:bic_approx}) explicitly correct for the foreground parameter volume without integrating.

\section{Related work} \label{sec:related_work}

Projection pursuit methods are closely related to data selection in that they attempt to identify ``interesting'' subspaces of the data.  However, projection pursuit uses certain pre-specified objective functions to optimize over projections, whereas our method allows one to specify a model of interest~\citep{Huber1985-xl}.

Another related line of research is on Bayesian goodness-of-fit (GOF) tests, which compute the posterior probability that the data comes from a given parametric model versus a flexible alternative such as a nonparametric model. Our setup differs in that it aims to compare among different semiparametric models. Nonetheless, in an effort to address the GOF problem, a number of authors have developed nonparametric models with tractable marginals \citep{Verdinelli1998-dx,Berger2001-li}, and using these models as the background component in an augmented model could in theory solve data selection problems. In practice, however, such models can only be applied to one-dimensional or few-dimensional data spaces. In Section~\ref{sec:ppca}, we show that naively extending the method of \citet{Berger2001-li} to the multi-dimensional setting has fundamental limitations.

There is a sizeable frequentist literature on GOF testing using discrepancies \citep{Gretton2012-do,Barron1989-jg,Gyorfi1991-ki}. Our proposed method builds directly on the KSD-based GOF test proposed by \citet{Liu2016-bp} and~\citet{Chwialkowski2016-tk}. 
However, using these methods to draw comparisons between different foreground subspaces is non-trivial, since the set of alternative models considered by the GOF test, though nonparametric, will be different over data spaces with different dimensionality.
Moreover, the Bayesian aspect of the SVC makes it more straightforward to integrate prior information and employ hierarchical models. 

In composite likelihood methods, instead of the standard likelihood, one uses the product of the conditional likelihoods of selected statistics~\citep{Lindsay1988-vi,Varin2011-sy}. Composite likelihoods have seen widespread use, often for robustness or computational purposes.
However, in composite likelihood methods, the choice of statistics is fixed before performing inference. 
In contrast, in data selection the choice of statistics is a central quantity to be inferred.

Relatedly, our work connects with the literature on robust Bayesian methods. \citet{Doksum1990-pd} propose conditioning on the value of an insufficient statistic, rather than the complete dataset, when performing inference; also see~\citet{Lewis2018-ri}. However, making an appropriate choice of statistic requires one to know which aspects of the model are correct; in contrast, our procedure infers the choice of statistic.
The $\textsc{nksd}$ posterior also falls within the general class of Gibbs posteriors, which have been studied in the context of robustness, randomized estimators, and generalized belief updating ~\citep{Zhang2006-sq,Zhang2006-wo,Jiang2008-zw,Bissiri2016-gk,Jewson2018-mw,Miller2019-zj}.

Our theoretical results also contribute to the emerging literature on Stein discrepancies~\citep{Anastasiou2021-kf}. \citet{Barp2019-ut} recently proposed minimum kernelized Stein discrepancy estimators and established their consistency and asymptotic normality. 
In Section~\ref{sec:theory}, we establish a Bayesian counterpart to these results, showing that the $\textsc{nksd}$ posterior is asymptotically normal (in the sense of Bernstein--von Mises) and admits a Laplace approximation. To prove this result, we rely on the recent work of \citet{Miller2019-ur} on the asymptotics of generalized posteriors. 
Since \citet{Barp2019-ut} show that the kernelized Stein discrepancy is related to the Hyv\"{a}rinen divergence in that both are Stein discrepancies, our work bears an interesting relationship to that of \citet{Shao2018-ef}, who use a Bayesian version of the Hyv\"{a}rinen divergence to perform model selection with improper priors. They derive a consistency result analogous to Equation~\ref{eqn:data_select_nksd},
however, their model selection score takes the form of a prequential score, not a Gibbs marginal likelihood as in the SVC, and cannot be used for data selection.

In independent recent work, \citet{Matsubara2021-kz} propose a Gibbs posterior based on the KSD and derive a Bernstein-von Mises theorem similar to Theorem~\ref{thm:marginal_ksd} using the results of \citet{Miller2019-ur}. 
Their method is not motivated by the Bayesian data selection problem but rather by (1) inference for energy-based models with intractable normalizing constants and (2) robustness to $\epsilon$-contamination. 
Their Bernstein-von Mises theorem differs from ours in that it applies to a V-statistic estimator of the KSD rather than a U-statistic estimator of the NKSD.

Our linear approximation to the minimum Stein discrepancy estimator (Section~\ref{sec:approximate_optima}) is directly inspired by the Swiss Army infinitesimal jackknife of \citet{Giordano2018-xh}, which similarly computes the linear response of an extremum estimator with respect to perturbations of the dataset.

\section{Toy example} \label{sec:toy_example}

The purpose of this toy example is to illustrate the behavior of the Stein volume criterion, and compare it to some of the defective alternatives listed in Table~\ref{tbl:consistency}, in a simple setting where all computations can be done analytically (Section~\ref{sec:si_conjugacy}). In all of the following experiments, we simulated data from a bivariate normal distribution: $X^{(1)}, \ldots, X^{(N)} \text{ i.i.d.} \sim \mathcal{N}((0,0)^\top, \Sigma_0)$.

To set up the Stein volume criterion, we set $T = 5$ and we choose a radial basis function kernel, $k(x, y) = \exp(-\frac{1}{2}\|x - y\|_2^2)$, which factors across dimensions.
We considered both dataset size-independent values of $m_\B$ (in particular, $m_\B = 5\, r_\B$) and dataset size-dependent values of $m_\B$ (in particular, Equation~\ref{eqn:pymm} with $\alpha=0.5$, $\theta=1$, and $D=0.2$, where fractional values of $D$ correspond to shared parameters across components in the Pitman-Yor mixture model), obtaining very similar results in each case (shown in Figures~\ref{fig:toy} and \ref{fig:toy_py}, respectively).
These choices of $m_\B$ ensure that, except for at very small $N$, the background model has more parameters per data dimension than each of the foreground models considered below, which have just one. In particular, $m_\B > 1\, r_\B$ for all $N$ (in the size-independent case) and for $N \ge 5$ (in the size-dependent case).

\subsubsection*{Data selection consistency}
First, we set $\Sigma_0$ to be a diagonal matrix with entries $(1, 1/2)$, that is, $\Sigma_0 = \mathrm{diag}(1, 1/2)$, and for $x\in\mathbb{R}^2$, we consider the model
\begin{equation}
\begin{split}
	q(x|\theta) &= \mathcal{N}(x \mid \theta, I)\\
	\pi(\theta) &= \mathcal{N}(\theta \mid (0, 0)^\top, 10 I)
\end{split}
\label{eqn:toy_model}
\end{equation} 
where $I$ denotes the identity matrix.
This parametric model is misspecified, owing to the incorrect choice of covariance matrix. We consider two choices of foreground subspace: the first dimension (defined by the projection matrix $V_{\F_1} = (1, 0)^\top$) or the second dimension (projection matrix $V_{\F_2} = (0, 1)^\top$). 
The model is only well-specified for $\F_1$ (not $\F_2$), so a successful data selection procedure would asymptotically select $\F_1$. 

In Figure~\ref{fig:data_select}, we see that the SVC correctly concentrates on $\F_1$ as the number of datapoints $N$ increases, with the log SVC ratio growing linearly in $N$, as predicted by Equation~\ref{eqn:data_select_nksd}. Meanwhile, the naive alternative score $\mathcal{K}^{(\mathrm{a})}$ (Equation~\ref{eqn:A1}) fails since it depends on the foreground entropies, while $\mathcal{K}^{(\mathrm{b})}$ (Equation~\ref{eqn:A2}) succeeds since the volume correction is negligible in this case; see Section~\ref{sec:data_select_consistency} and Table~\ref{tbl:consistency}.

% The $\mathcal{K}^{(\mathrm{b})}$ score (Equation~\ref{eqn:A2}) also concentrates on $\F_1$, since the volume correction does not contribute in this example. However, the $\mathcal{K}^{(\mathrm{a})}$ score (Equation~\ref{eqn:A1}) fails, as described in Section~\ref{sec:data_select_consistency}.

\begin{figure}[t!]
    \centering
    \begin{subfigure}[t!]{0.48\textwidth}
        \centering
        \includegraphics[height=2.5in]{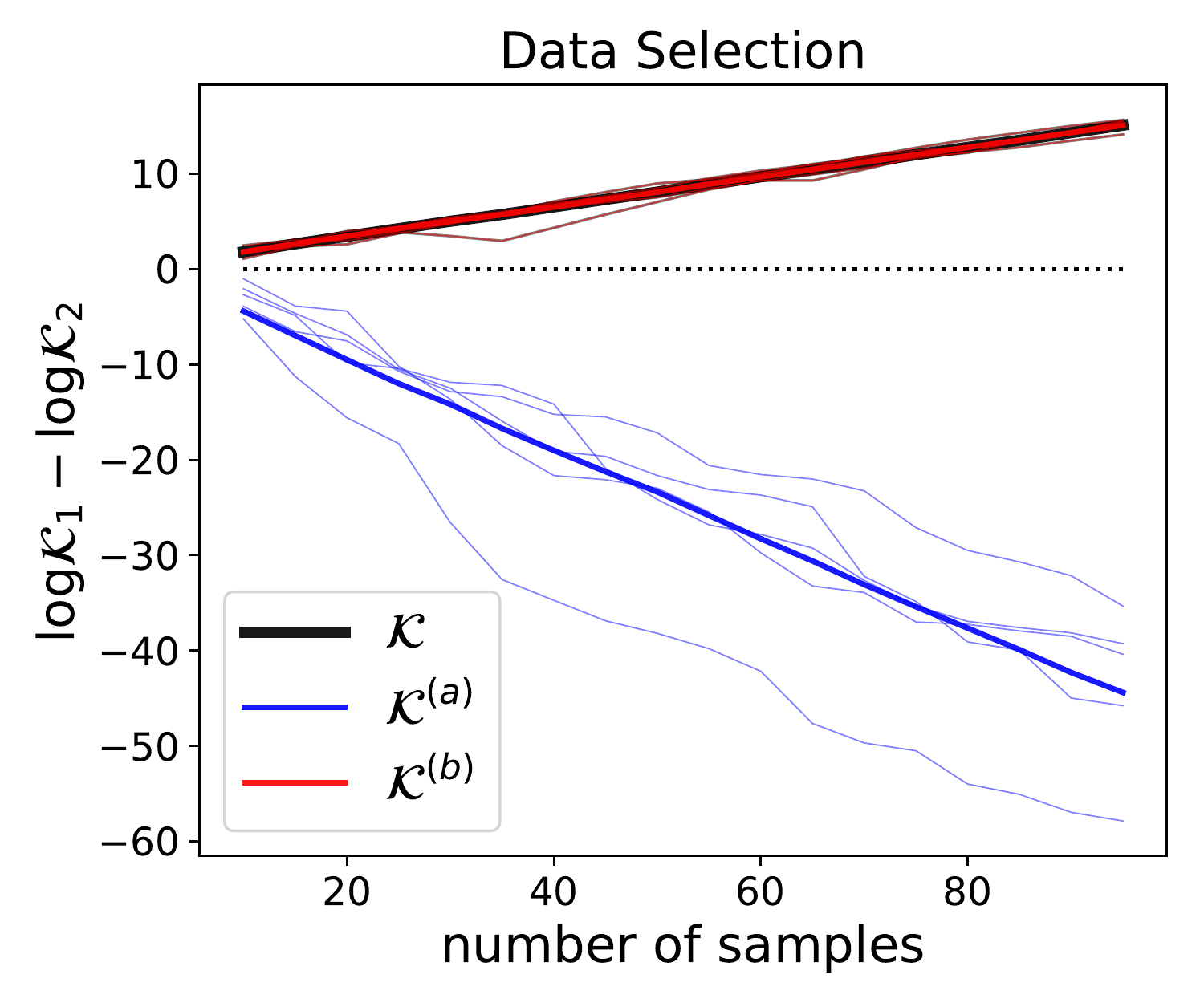}
        \caption{}
        \label{fig:data_select}
    \end{subfigure}%
    ~ 
    \begin{subfigure}[t!]{0.48\textwidth}
        \centering
        \includegraphics[height=2.5in]{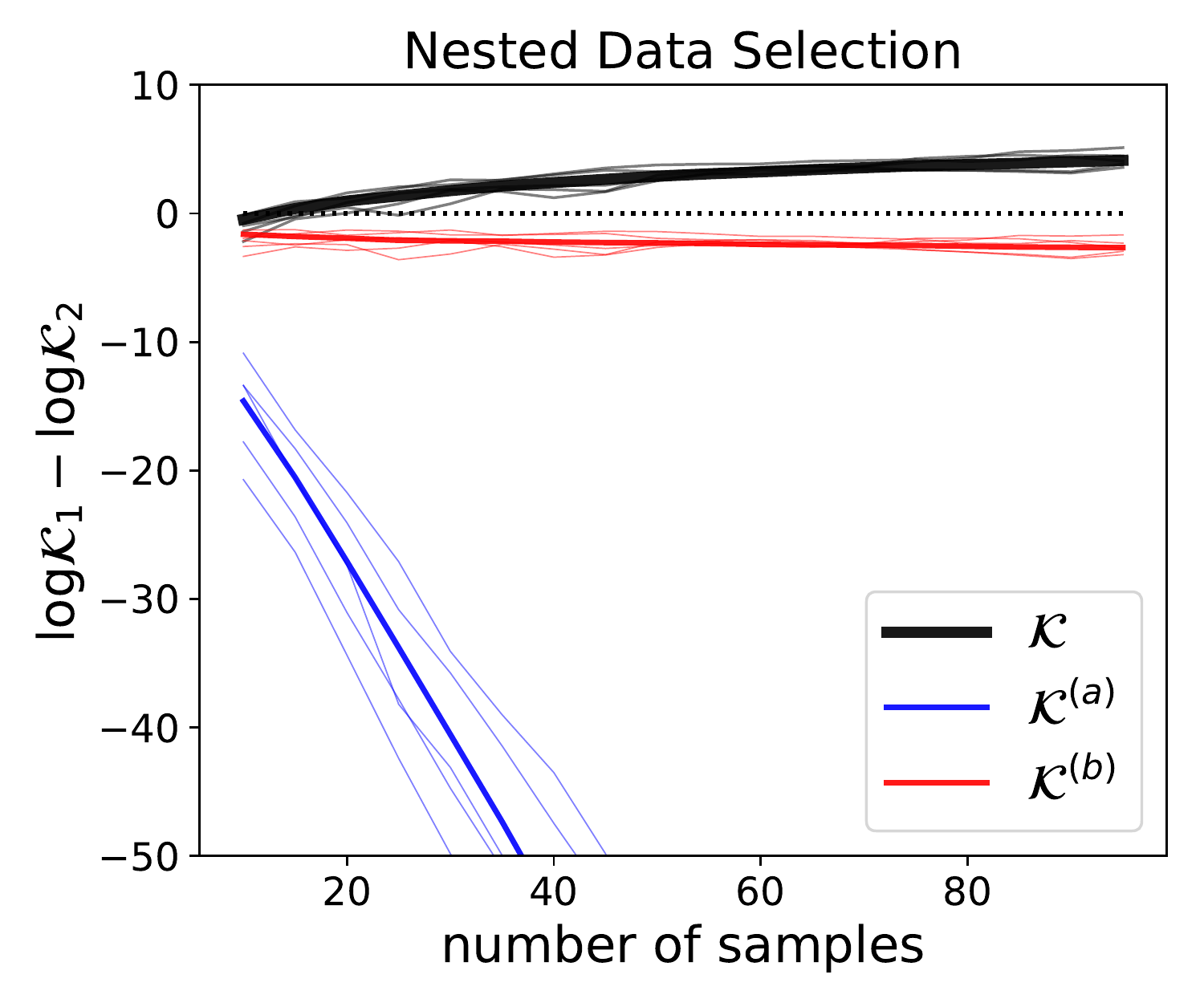}
        \caption{}
        \label{fig:nest_data_select}
    \end{subfigure}
    \\
    \begin{subfigure}[t!]{0.48\textwidth}
        \centering
        \includegraphics[height=2.5in]{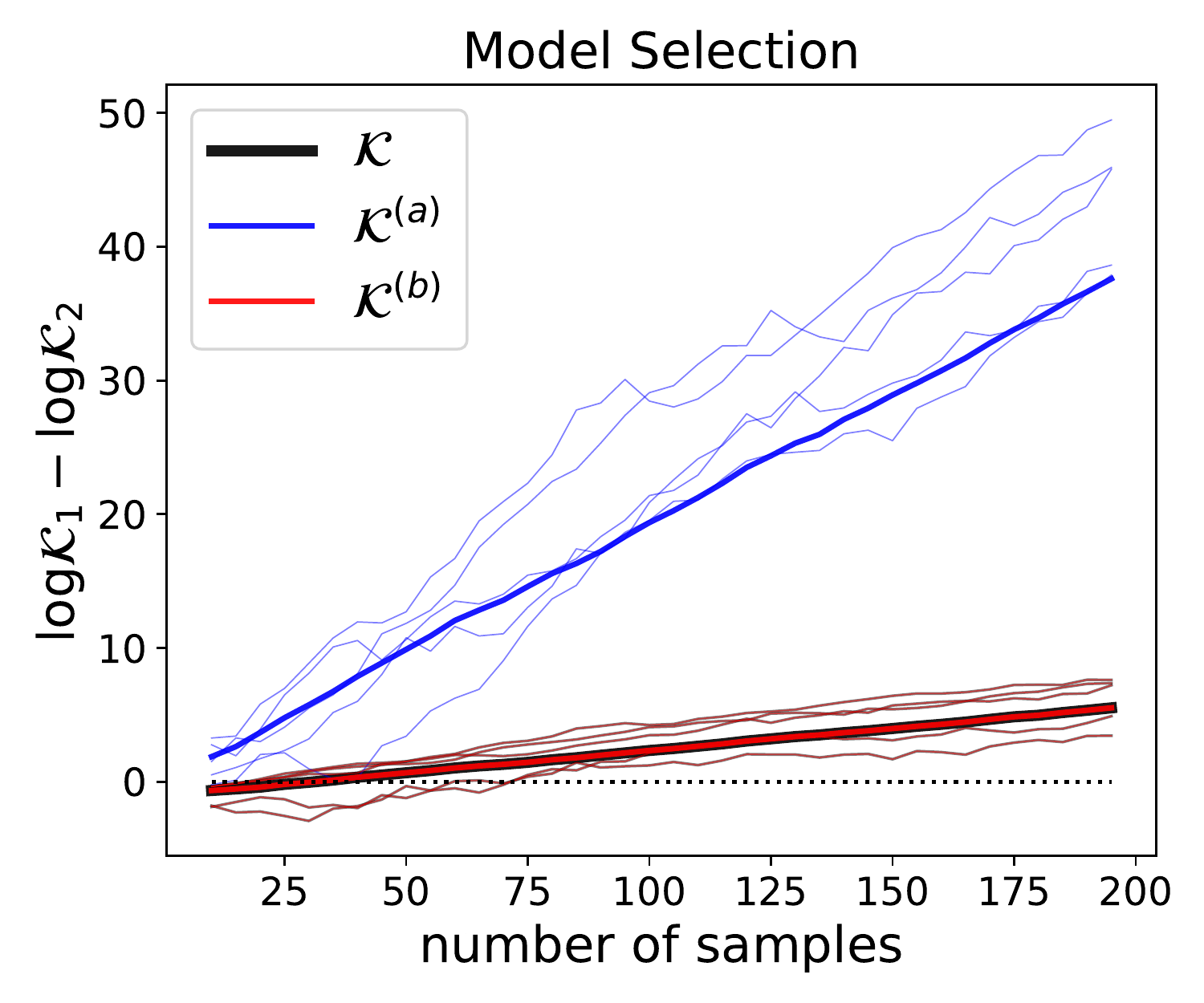}
        \caption{}
        \label{fig:model_select}
    \end{subfigure}
    ~ 
    \begin{subfigure}[t!]{0.48\textwidth}
        \centering
        \includegraphics[height=2.5in]{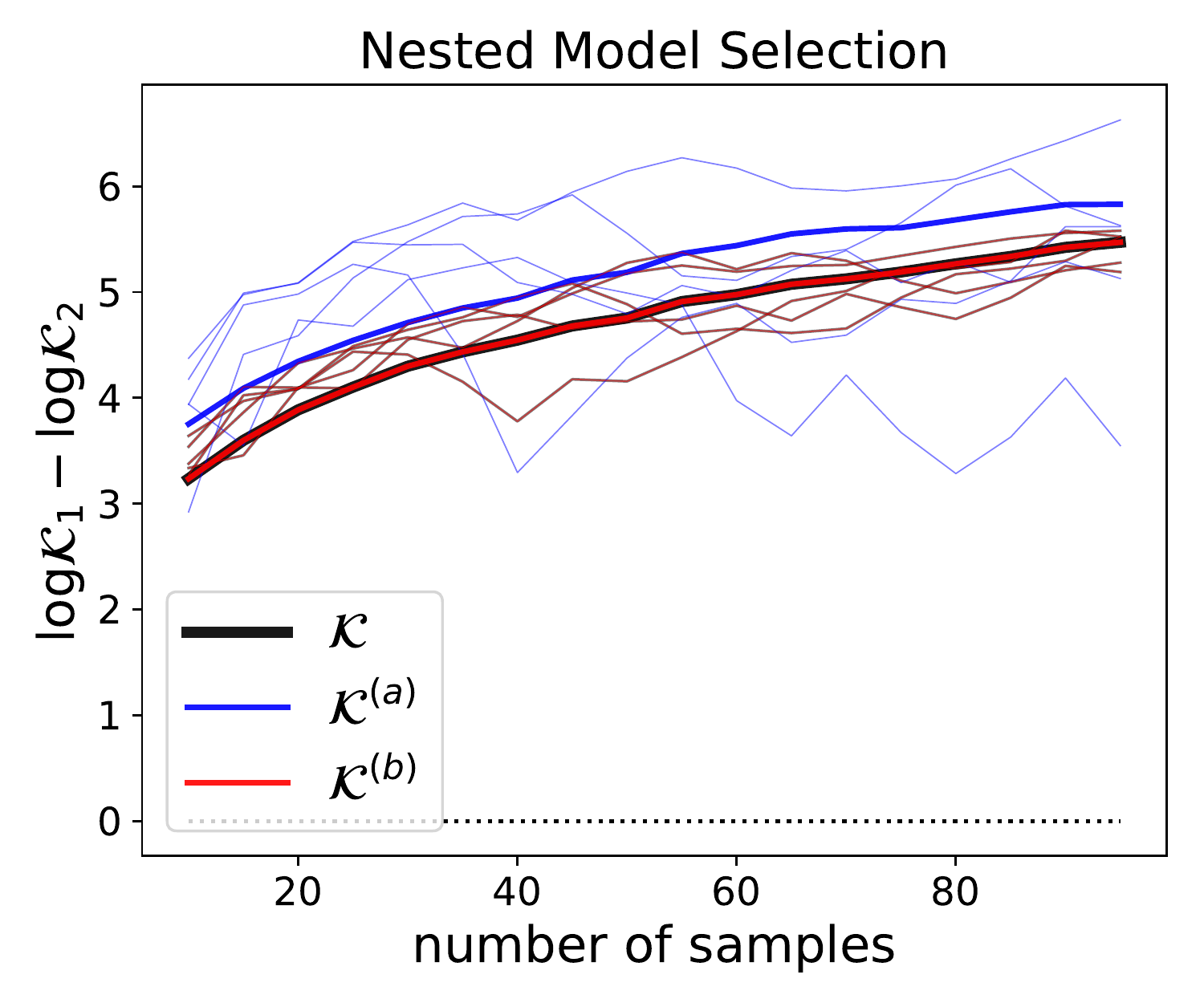}
        \caption{}
        \label{fig:nest_model_select}
    \end{subfigure}
    \caption{Behavior of the Stein volume criterion $\mathcal{K}$, the foreground marginal likelihood with a background volume correction $\mathcal{K}^{(\mathrm{a})}$, and the foreground marginal \textsc{nksd} $\mathcal{K}^{(\mathrm{b})}$ on toy examples. Here, we set $m_\B = 5\, r_\B$. The plots show the results for 5 randomly generated datasets (thin lines) and the average over 100 random datasets (bold lines).}
    \label{fig:toy}
\end{figure}

\subsubsection*{Nested data selection consistency}
Next, we examine the nested data selection case. We use the same model (Equation~\ref{eqn:toy_model}),
but we set $\Sigma_0 = I$ so that the model is well-specified even without being projected.
We compare the complete data space ($\X_{\F_1} = \X$, projection matrix $V_{\F_1} = I$) to the first dimension alone (projection matrix $V_{\F_1} = (1, 0)^\top$). 
Nested data selection consistency demands that the higher-dimensional data space $\X_{\F_1}$ be preferred asymptotically, since the model is well-specified for both $\X_{\F_1}$ and $\X_{\F_2}$. 
Figure~\ref{fig:nest_data_select} shows that this is indeed the case for the Stein volume criterion, with the log SVC ratio growing at a $\log N$ rate when $m_\B$ is independent of $N$, as predicted by Equation~\ref{eqn:nested_data_nksd}. 
When $m_\B$ depends on $N$ via the Pitman-Yor expression, the log SVC ratio grows at a $N^{\alpha} \log N$ rate (Figure~\ref{fig:py_nest_data_select}).
Meanwhile, Figure~\ref{fig:nest_data_select} shows that $\mathcal{K}^{(\mathrm{a})}$ and $\mathcal{K}^{(\mathrm{b})}$ both fail to exhibit nested data selection consistency, in accordance with our theory (Section~\ref{sec:nested_ds} and Table~\ref{tbl:consistency}).

\subsubsection*{Model selection consistency (nested and non-nested)}
Finally, we examine model selection and nested model selection consistency. We again set $\Sigma_0 = I$. 
We first compare the (well-specified) model $q(x|\theta) = \mathcal{N}(x\mid \theta, I)$ to the (misspecified) model $q(x|\theta) = \mathcal{N}(x \mid \theta, 2 I)$, using the prior $\pi(\theta) = \mathcal{N}(\theta\mid (0, 0)^\top, 10 I)$ for both models. As shown in Figure~\ref{fig:model_select}, the SVC correctly concentrates on the first model, with the log SVC ratio growing linearly in $N$, as predicted by Equation~\ref{eqn:model_selection_nksd}. The same asymptotic behavior is exhibited by $\mathcal{K}^{(\mathrm{a})}$, which is equivalent to the standard Bayesian marginal likelihood in this setting (Section~\ref{sec:model_selection_asymptotics}). 
Finally, to check nested model selection consistency, we compare two well-specified nested models: $q(x) = \mathcal{N}(x \mid (0, 0)^\top, I)$ and $q(x|\theta) = \mathcal{N}(x \mid \theta, I)$.  Figure~\ref{fig:nest_model_select} shows that the SVC correctly selects the simpler model (that is, the model with smaller parameter dimension) and the log SVC ratio grows as $\log N$ (Equation~\ref{eqn:nested_model_nksd}). This, too, matches the behavior of the standard Bayesian marginal likelihood, seen in the plot of $\mathcal{K}^{(\mathrm{a})}$.

\section{Theory} \label{sec:theory}

\subsection{Properties of the NKSD} \label{sec:nksd_properties}

% \subsubsection{Divergence} \label{sec:divergence}

Suppose $X^{(1)},\ldots,X^{(N)}$ are i.i.d.\ samples from a probability measure $P$ on $\mathcal{X} \subseteq \mathbb{R}^d$ having density $p(x)$ with respect to the Lebesgue measure. Let $L^1(P)$ denote the set of measurable functions $f$ such that $\int \|f(x)\| p(x) d x < \infty$ where $\|\cdot\|$ is the Euclidean norm.
We impose the following regularity conditions to use the \textsc{nksd} to compare $P$ with another probability measure $Q$ having density $q(x)$ with respect to the Lebesgue measure; these are similar to conditions used for the standard \textsc{ksd} in previous work \citep{Liu2016-bp,Barp2019-ut}.
\begin{condition}[Restrictions on $p$ and $q$] \label{condition:restrict_p_q}
Assume $s_{p}(x) := \nabla_x \log p(x)$ and $s_{q}(x) := \nabla_x \log q(x)$ exist and are continuous for all $x\in\mathcal{X}$,
and assume $\mathcal{X}$ is connected and open.  Further, assume $s_{p},s_{q}\in L^1(P)$.
\end{condition}
We refer to $s_{p}$ as the Stein score function of $p$.
Note that existence of $s_{p}(x)$ implies $p(x) > 0$.
Now, consider a kernel $k:\mathcal{X}\times\mathcal{X}\to\mathbb{R}$.
The kernel $k$ is said to be \textit{integrally strictly positive definite} if for any $g: \mathcal{X} \to \mathbb{R}$ such that $0 < \int_\mathcal{X} |g(x)| dx < \infty$, we have $\int_\mathcal{X} \int_\mathcal{X} g(x) k(x,y) g(y) dx dy > 0$.
The kernel $k$ is said to \textit{belong to the Stein class of $P$} if $\int_\mathcal{X} \nabla_x (k(x,y) p(x)) dx = 0$ for all $y\in\mathcal{X}$.
\begin{condition}[Restrictions on $k$] \label{condition:restrict_k}
Assume the kernel $k$ is symmetric, bounded, integrally strictly positive definite, and belongs to the Stein class of $P$.
\end{condition}
% Let $\mathcal{H}$ be the reproducing kernel Hilbert space associated with $k$, and let $\|\cdot\|_\mathcal{H}$ be the associated norm.
% \begin{condition}[Restrictions on $q$] \label{condition:restrict_q}
% Assume $s_{q}(x) := \nabla_x \log q(x)$ exists and is continuous for all $x\in\mathcal{X}$.  Further, assume $s_{q}\in L^1(P)$.
% Further, assume $x \mapsto \|\frac{1}{q(y)}\nabla_y(q(y) k(x, y))\|_\mathcal{H}$ is in $L^1(P)$ for all $y\in\mathcal{X}$.
% \end{condition} 

% Proposition~\ref{proposition:u_statistic_nksd} shows that under these conditions, the \textsc{nksd} admits a convenient representation.
The following result shows that the \textsc{nksd} can be written in a way that does not involve $s_{p}$; this is particularly useful for estimating the \textsc{nksd} when $P$ is unknown.
\begin{proposition} \label{proposition:u_statistic_nksd}
If Conditions~\ref{condition:restrict_p_q} and \ref{condition:restrict_k} hold, then the \textsc{nksd} is finite and 
\begin{equation} \label{eqn:nksd_u_rep}
\textsc{nksd}(p(x)\| q(x)) := \frac{\mathbb{E}_{X,Y \sim p}[u(X,Y)]}{\mathbb{E}_{X,Y \sim p}[k(X,Y)]}
\end{equation}
where
\begin{equation}
    u(x, y) = s_{q}(x)^\top s_{q}(y) k(x, y) + s_{q}(x)^\top \nabla_y k(x, y) + s_{q}(y)^\top \nabla_x k(x, y) + \Tr(\nabla_x \nabla_y^\top k(x, y)).
\end{equation}
\end{proposition}
\noindent The proof is in Section~\ref{sec:proofs_nksd}.
% The next result shows that the \textsc{nksd} is a statistical divergence.
Next, we show the \textsc{nksd} satisfies the properties of a divergence.
\begin{proposition} \label{thm:nksd_divergence}
If Conditions~\ref{condition:restrict_p_q} and \ref{condition:restrict_k} hold, then
\begin{equation}\textsc{nksd}(p(x)\| q(x)) \geq 0,\end{equation} with equality if and only if $p(x) = q(x)$ almost everywhere.
\end{proposition}
\noindent The proof is in Section~\ref{sec:proofs_nksd}.
Unlike the standard \textsc{ksd}, but like the \textsc{kl} divergence, the \textsc{nksd} exhibits subsystem independence \citep{Caticha2003-gb,Caticha2010-fm,Rezende2018-rp}: if two distributions $P$ and $Q$ have the same independence structure, then the total \textsc{nksd} separates into a sum of individual \textsc{nksd} terms. This is formalized in Proposition~\ref{prop:subsystem_indep}.
\begin{condition}[Shared independence structure] \label{condition:independence}
	Let $x = (x_1^\top, x_2^\top)^\top$ be a decomposition of a vector $x \in \mathbb{R}^d$ into two subvectors, $x_1$ and $x_2$.
	Assume $p(x)$ and $q(x)$ factor as $p(x) = p(x_1)p(x_2)$ and $q(x) = q(x_1)q(x_2)$, and that the kernel $k$ factors as $k(x,y) = k_1(x_1, y_1) k_2(x_2, y_2)$ where $k_1$ and $k_2$ both satisfy Condition~\ref{condition:restrict_k}.
\end{condition}
\begin{proposition}[Subsystem independence] \label{prop:subsystem_indep}
If Conditions~\ref{condition:restrict_p_q}, \ref{condition:restrict_k}, and \ref{condition:independence} hold, then
\begin{equation}
\textsc{nksd}(p(x)\|q(x)) = \textsc{nksd}(p(x_1)\|q(x_1)) + \textsc{nksd}(p(x_2)\|q(x_2))
\label{eqn:nksd_decomp}
\end{equation}
where the first term on the right-hand side uses kernel $k_1$ and the second term uses $k_2$.
\end{proposition}
\noindent See Section~\ref{sec:proofs_nksd} for the proof.
Subsystem independence is powerful since it separates the problem of evaluating the foreground model from that of evaluating the background model.
A modified version applies to the estimator $\widehat{\textsc{nksd}}(p\|q)$ (Equation~\ref{eqn:est_nksd});
see Proposition~\ref{prop:approx_subsystem_indep}.

\subsection{Bernstein--von Mises theorem for the NKSD posterior} \label{sec:nksd_bvm}

In this section, we establish asymptotic properties of the SVC and, more broadly, of its corresponding generalized posterior, which we refer to as the \textsc{nksd} posterior, defined as
\begin{equation}
	\pi_N(\theta) \propto \exp\!\Big(-\frac{N}{T} \widehat{\textsc{nksd}}(p_0(x_\F)\| q(x_\F|\theta))\Big) \pi(\theta).
\end{equation}
In particular, in Theorem~\ref{thm:marginal_ksd}, we show that the \textsc{nksd} posterior concentrates and is asymptotically normal, and we establish that the Laplace approximation to the SVC (Equation~\ref{eqn:ksd_laplace}) is asymptotically correct. These results form a Bayesian counterpart to those of \citet{Barp2019-ut}, who establish the consistency and asymptotic normality of minimum $\textsc{ksd}$ estimators. Thus, in both the frequentist and Bayesian contexts, we can replace the average log likelihood with the negative $\textsc{ksd}$ and obtain similar key properties.
Our results in this section do not depend on whether or not we are working with a foreground subspace, so we suppress the $x_\F$ notation.

Let $\Theta\subseteq\mathbb{R}^m$, and let $\{Q_\theta : \theta \in \Theta\}$ be a family of probability measures on $\mathcal{X} \subseteq\mathbb{R}^d$ having densities $q_\theta(x)$ with respect to Lebesgue measure. For notational convenience, we sometimes write $q(x|\theta)$ instead of $q_\theta(x)$.
Suppose the data $X^{(1)},\ldots,X^{(N)}$ are i.i.d.\ samples from some probability measure $P_0$ on $\mathcal{X}$ having density $p_0(x)$ with respect to Lebesgue measure.
To ensure the \textsc{nksd} satisfies the properties of a divergence for all $q_\theta$, and that convergence of $\widehat{\textsc{nksd}}$ is uniform on compact subsets of $\Theta$ (Proposition~\ref{proposition:uniform_convergence}), we require the following.
\begin{condition} \label{condition:apply_nksd}
	Assume Conditions~\ref{condition:restrict_p_q} and \ref{condition:restrict_k} hold for $p_0$, $k$, and $q_\theta$ for all $\theta \in \Theta$. Further, assume that the kernel $k$ has continuous and bounded partial derivatives up to and including second order, and $k(x, y) > 0$ for all $x, y \in \mathcal{X}$.
\end{condition}

Now we can set up the generalized posterior. First define
% the \textsc{nksd} likelihood function,
\begin{equation}
\label{eqn:f_N}
f_N(\theta) := \frac{1}{T}\widehat{\textsc{nksd}}(p_0(x)\| q(x|\theta)) = \frac{1}{T}\frac{\sum_{i \neq j} u_\theta(X^{(i)}, X^{(j)})}{\sum_{i \neq j} k(X^{(i)}, X^{(j)})},
\end{equation}
where $u_\theta(x,y)$ is the $u(x,y)$ function from Equation~\ref{eqn:est_nksd} with $q_\theta$ in place of $q$.
For the case of $N = 1$, we define $f_1(\theta) = 0$ by convention.
Note that $-N f_N(\theta)$ plays the role of the log likelihood.
Also define
\begin{align} \label{eqn:f}
	f(\theta) &:= \frac{1}{T}\textsc{nksd}(p_0(x)\| q(x|\theta)), \\
	z_N &:= \int_\Theta \exp(-N f_N(\theta))\pi(\theta)d\theta, \notag\\
	\pi_N(\theta) &:= \frac{1}{z_N}\exp(-N f_N(\theta))\pi(\theta), \notag
\end{align}
where $\pi(\theta)$ is a prior density on $\Theta$.  Note that $\pi_N(\theta)d\theta$ is the \textsc{NKSD} posterior and $z_N$ is the corresponding generalized marginal likelihood employed in the SVC.
Denote the gradient and Hessian of $f$ by $f'(\theta) = \nabla_\theta f(\theta)$ and $f''(\theta) = \nabla^2_\theta f(\theta)$, respectively. To ensure that the \textsc{nksd} posterior is well defined and has an isolated maximum, we assume the following condition.
\begin{condition} \label{condition:posterior_restrictions}
Suppose $\Theta \subseteq \mathbb{R}^m$ is a convex set and (a) $\Theta$ is compact or (b) $\Theta$ is open and $f_N$ is convex on $\Theta$ with probability 1 for all $N$.
Assume $z_N < \infty$ a.s.\ for all $N$.
Assume $f$ has a unique minimizer $\theta_{*}\in\Theta$,
$f''(\theta_{*})$ is invertible, $\pi$ is continuous at $\theta_{*}$, and $\pi(\theta_{*}) > 0$.
\end{condition}
\noindent By Proposition~\ref{thm:nksd_divergence}, $f$ has a unique minimizer whenever $\{Q_\theta : \theta\in\Theta\}$ is well-specified and identifiable, that is, when $Q_\theta = P_0$ for some $\theta$ and $\theta \mapsto Q_\theta$ is injective.

In Theorem~\ref{thm:marginal_ksd} below, we establish the following results: (1) the minimum $\widehat{\textsc{nksd}}$ converges to the minimum $\textsc{nksd}$; (2) $\pi_N$ concentrates around the minimizer of the $\textsc{nksd}$; (3) the Laplace approximation to $z_N$ is asymptotically correct; and (4) $\pi_N$ is asymptotically normal in the sense of Bernstein--von Mises. The primary regularity conditions we need for this theorem are restraints on the derivatives of $s_{q_\theta}$ with respect to $\theta$ (Condition~\ref{condition:marginal_ksd}). 
Our proof of Theorem~\ref{thm:marginal_ksd} relies on the theory of generalized posteriors developed by \citet{Miller2019-ur}. 
We use $\|\cdot \|$ for the Euclidean--Frobenius norms: for vectors $A \in \mathbb{R}^D$, $\|A\| = (\sum_i A_i^2)^{1/2}$; for matrices $A \in \mathbb{R}^{D \times D}$, $\|A\| = (\sum_{i,j} A^2_{i,j})^{1/2}$; for tensors $A \in \mathbb{R}^{D \times D \times D}$, $\|A\| = (\sum_{i,j,k} A^2_{i,j,k})^{1/2}$; and so on.

\begin{theorem}
\label{thm:marginal_ksd}

If Conditions~\ref{condition:apply_nksd}, \ref{condition:posterior_restrictions}, and \ref{condition:marginal_ksd} hold, then there is a sequence $\theta_N \to \theta_{*}$ a.s.\ such that:
\begin{enumerate}
    \item $f_N(\theta_N) \to f(\theta_{*})$, $f'_N(\theta_N) = 0$ for all $N$ sufficiently large, and $f_N''(\theta_N) \to f''(\theta_{*})$ a.s.,
    \item letting $B_\epsilon(\theta_{*}) := \{\theta \in \mathbb{R}^m : \|\theta - \theta_{*}\| < \epsilon\}$, we have
    \begin{equation}
        \int_{B_\epsilon(\theta_{*})} \pi_N(\theta) d\theta \xrightarrow[N \to \infty]{\textup{a.s.}} 1 \text{ for all } \epsilon > 0,
    \end{equation}
    
    \item 
    \begin{equation}
        z_N \sim \frac{\exp(-N f_N(\theta_N)) \pi(\theta_{*})}{|\det f''(\theta_{*}) |^{1/2}} \left(\frac{2\pi}{N}\right)^{m/2}
    \end{equation}
    almost surely, where $a_N \sim b_N$ means that $a_N/b_N \to 1$ as $N \to \infty$, and 
    
    \item letting $h_N$ denote the density of $\sqrt{N}(\theta - \theta_N)$ when $\theta$ is sampled from $\pi_N$, we have that $h_N$ converges to $\mathcal{N}(0, f''(\theta_{*})^{-1})$ in total variation, that is,
    \begin{equation}
    \int_{\mathbb{R}^m} \Big| h_N(\tilde{\theta}) - \mathcal{N}(\tilde{\theta} \mid 0, f''(\theta_{*})^{-1}) \Big| d\tilde{\theta} \xrightarrow[N \to \infty]{\textup{a.s.}} 0.
    \end{equation}
\end{enumerate}
\end{theorem}
\noindent The proof is in Section~\ref{sec:si_bvm_proof}. We write $\nabla_\theta^2 s_{q_\theta}$ to denote the tensor in $\mathbb{R}^{d\times m\times m}$ in which entry $(i,j,k)$ is $\partial^2 s_{q_\theta}(x)_i / \partial \theta_j \partial \theta_k$.
Likewise, $\nabla_\theta^3 s_{q_\theta}$ denotes the tensor in $\mathbb{R}^{d\times m\times m\times m}$ in which entry $(i,j,k,\ell)$ is $\partial^3 s_{q_\theta}(x)_i / \partial \theta_j \partial \theta_k \partial \theta_\ell$.
We write $\mathbb{N}$ to denote the set of natural numbers.

\begin{condition}[Stein score regularity]
\label{condition:marginal_ksd}
Assume $s_{q_\theta}(x)$ has continuous third-order partial derivatives with respect to the entries of $\theta$ on $\Theta$. 
Suppose that for any compact, convex subset $C \subseteq \Theta$, there exist continuous functions $g_{0,C}, g_{1,C} \in L^1(P_0)$ such that for all $\theta \in C$, $x\in\mathcal{X}$,
\begin{equation}
\label{eqn:score_bounds}
\begin{split}
    \|s_{q_\theta}(x)\| &\le g_{0,C}(x),\\
    \|\nabla_\theta s_{q_\theta}(x)\| &\le g_{1,C}(x).
\end{split}
\end{equation}
Further, assume there is an open, convex, bounded set $E\subseteq\Theta$ such that $\theta_{*}\in E$, $\bar{E} \subseteq \Theta$, and the sets
\begin{equation}
\label{eqn:score_d2_bound}
\Big\{\frac{1}{N} \sum_{i=1}^N \|\nabla_\theta^2 s_{q_\theta}(X^{(i)})\| : N \in \mathbb{N}, \theta \in E\Big\},
\end{equation}
\begin{equation}
\label{eqn:score_d3_bound}
\Big\{\frac{1}{N} \sum_{i=1}^N \|\nabla_\theta^3 s_{q_\theta}(X^{(i)})\| : N \in \mathbb{N}, \theta \in E\Big\}
\end{equation}
are bounded with probability 1.
\end{condition}
Next, Theorem~\ref{thm:expo_marginal_ksd} shows that in the special case where $q_\theta(x)$ is an exponential family, many of the conditions of Theorem~\ref{thm:marginal_ksd} are automatically satisfied.

\begin{theorem}
\label{thm:expo_marginal_ksd}
Suppose $\{Q_\theta : \theta\in\Theta\}$ is an exponential family with densities of the form $q_\theta(x) = \lambda(x) \exp(\theta^\top t(x) - \kappa(\theta))$ for $x\in\mathcal{X}\subseteq\mathbb{R}^d$. Assume $\Theta = \{\theta \in \mathbb{R}^m : |\kappa(\theta)| < \infty\}$, and assume $\Theta$ is convex, open, and nonempty. 
Assume $\log\lambda(x)$ and $t(x)$ are continuously differentiable on $\mathcal{X}$, 
$\|\nabla_x \log \lambda(x)\|$ and $\|\nabla_x t(x)\|$ are in $L^1(P_0)$, 
and the rows of the Jacobian matrix $\nabla_x t(x) \in \mathbb{R}^{m \times d}$ are linearly independent with positive probability under $P_0$.
Suppose Condition~\ref{condition:apply_nksd} holds, $f$ has a unique minimizer $\theta_{*}\in\Theta$, 
the prior $\pi$ is continuous at $\theta_{*}$, and $\pi(\theta_{*}) > 0$.
Then the assumptions of Theorem~\ref{thm:marginal_ksd} are satisfied for all $N$ sufficiently large.
\end{theorem}
\noindent The proof is in Section~\ref{sec:si_bvm_proof}.

% If Condition~\ref{condition:apply_nksd} holds and $f$ has a unique minimizer $\theta_{*}\in\Theta$, then the Hessian $f''(\theta_*)$ is invertible, $f_N$ is convex on $\Theta$, and Condition \ref{condition:marginal_ksd} holds. 
% \comment{If, further, $z_N < \infty$ a.s.\ for all $N$,
% the prior $\pi$ is continuous at $\theta_{*}$, and $\pi(\theta_{*}) > 0$,
% then the assumptions of Theorem~\ref{thm:marginal_ksd} are satisfied.}

\subsection{Asymptotics of the Stein volume criterion} \label{sec:theory_detailed_asym}

The Laplace approximation to the SVC uses the estimate $\widehat{\textsc{nksd}}$ and its minimizer $\theta_N$, rather than the true \textsc{nksd} and its minimizer $\theta_*$. 
To establish the consistency properties of the SVC, we need to understand the relationship between the two. 
To do so, we adapt a standard approach to performing such an analysis of the marginal likelihood, for instance, as in Theorem 1 of \citet{Dawid2011-kb}.
% In particular, by conclusion 3 of Theorem~\ref{thm:marginal_ksd},
% \begin{align}
	% \log z_{N} &+ N[f_{N}(\theta_{N}) - f_{N}(\theta_{*})] + N[f_{N}(\theta_{*}) - f(\theta_{*})]+ N f(\theta_{*})\notag\\
	% & - \log \pi(\theta_{*}) + \frac{1}{2} \log|\det f''(\theta_{*})|
	% + \frac{1}{2}m \log(N/2\pi) \xrightarrow[N \to \infty]{\textup{a.s.}} 0.
% \label{eqn:mc_decomposition}
% \end{align}
% Here the difference between the true minimum \textsc{nksd}, $-f(\theta_*)$, and the estimated minimum \textsc{nksd}, $-f_N(\theta_N)$, has been broken up into two error terms, $f_{N}(\theta_{N}) - f_{N}(\theta_{*})$ and $f_{N}(\theta_{*}) - f(\theta_{*})$. We will analyze the scaling of each with $N$.
\begin{theorem} \label{prop:laplace_scaling} 
Assume the conditions of Theorem~\ref{thm:marginal_ksd} hold, and assume $s_{q_{\theta_{*}}}$ and $\nabla_\theta \big\vert_{\theta=\theta_{*}} s_{q_\theta}$ are in $L^2(P_0)$. Then as $N\to\infty$,
\begin{equation} \label{eqn:fN_O1}
	f_N(\theta_N) - f_N(\theta_{*}) = O_{P_0}(N^{-1}).
\end{equation}
Further, if $\textsc{nksd}(p_0(x)\| q(x|\theta_{*})) > 0$ then
\begin{equation} \label{eqn:fN_OsN}
	f_N(\theta_{*}) - f(\theta_{*}) = O_{P_0}(N^{-1/2}),
\end{equation}
%\begin{equation} \label{eqn:f_ON}
%	N f(\theta_{*}) = O(N),
%\end{equation}
whereas if $\textsc{nksd}(p_0(x)\| q(x|\theta_{*})) = 0$ then
\begin{equation} \label{eqn:fN_O1_thetastar}
	f_N(\theta_{*}) - f(\theta_{*}) = O_{P_0}(N^{-1}).
\end{equation}
%\begin{equation} \label{eqn:f_zero}
%	N f(\theta_{*}) = 0.
%\end{equation}
%Finally,
%\begin{equation} \label{eqn:pi_O1}
%	- \log \pi(\theta_{*}) + \frac{1}{2} \log|\det f_1''(\theta_{*})| = O(1).
%\end{equation}
\end{theorem}
\noindent The proof is in Section~\ref{sec:proof_laplace_scaling}. 
Remarkably, Equation~\ref{eqn:fN_O1_thetastar} shows that $f_N(\theta_{*})$ converges to $f(\theta_{*})$ more rapidly when the model is well-specified, specifically, at a $1/N$ rate instead of $1/\sqrt{N}$. This is unusual and is crucial for our results in Section~\ref{sec:theory_ds_ms}.  The standard log likelihood does not exhibit this rapid convergence; see Section~\ref{sec:asymptotics_setup}. 
This property of the \textsc{nksd} derives from similar properties exhibited by the standard \textsc{ksd} \citep[Theorem 4.1]{Liu2016-bp}.
% Since $f$ and $\pi$ are independent of $N$, 
Combined with Theorem~\ref{thm:marginal_ksd} (part 3),
Theorem~\ref{prop:laplace_scaling} implies that when the model is misspecified, the leading order term of $\log z_N$ is $-N f(\theta_*)$, whereas when the model is well-specified, the leading order term is $-\frac{1}{2}\, m \log N$.

\subsection{Data and model selection consistency of the SVC} \label{sec:theory_ds_ms}

In this section, we establish the asymptotic consistency of the Stein volume criterion (SVC) when used for data selection, nested data selection, model selection, and nested model selection; see Theorem~\ref{thm:selection_consistency}.  This provides rigorous justification for the claims in Section~\ref{sec:method_asymptotics}. These results are all in the context of pairwise comparisons between two models or two model projections, $M_1$ and $M_2$.
Before proving the results, we formally define the consistency properties discussed in Section~\ref{sec:method_asymptotics}. 
Each property is defined in terms of a pairwise score $\rho(M_1,M_2)$, such as $\rho(M_1,M_2) = \log(\mathcal{K}_1 / \mathcal{K}_2)$.  For simplicity, we assume $\rho(M_1,M_2) = -\rho(M_2,M_1)$; this is satisfied for all of the cases we consider.
Let $\dim(\cdot)$ denote the dimension of a real space.
\begin{definition}[Data selection consistency] \label{def:ds}
Consider foreground model projections $M_j := \{q(x_{\F_j}|\theta): \theta \in \Theta\}$ for $j \in \{1,2\}$. We say that $\rho$ satisfies ``data selection consistency" if $\rho(M_1, M_2) \to \infty$ as $N\to\infty$ when $M_1$ is well-specified with respect to $p_0(x_{\F_1})$ and $M_2$ is misspecified with respect to $p_0(x_{\F_{2}})$.
\end{definition}
\begin{definition}[Nested data selection consistency] \label{def:nds}
Consider foreground model projections $M_j := \{q(x_{\F_j}|\theta): \theta \in \Theta\}$ for $j \in \{1,2\}$. We say that $\rho$ satisfies ``nested data selection consistency" if $\rho(M_1, M_2) \to \infty$ as $N\to\infty$ when $M_1$ is well-specified with respect to $p_0(x_{\F_1})$, $\X_{\F_2} \subset \X_{\F_{1}}$, and $\dim(\X_{\F_2}) < \dim(\X_{\F_1})$.
\end{definition}
\begin{definition}[Model selection consistency] \label{def:ms}
Consider foreground models $M_j := \{q_j(x_\F|\theta_j): \theta_j \in \Theta_j\}$ for $j \in \{1, 2\}$. We say that $\rho$ satisfies ``model selection consistency" if $\rho(M_1, M_2) \rightarrow \infty$ as $N\to\infty$ when $M_1$ is well-specified with respect to $p_0(x_\F)$ and $M_2$ is misspecified.
\end{definition}
\begin{definition}[Nested model selection consistency] \label{def:nms}
Consider foreground models $M_j := \{q_j(x_\F|\theta_j): \theta_j \in \Theta_j\}$ for $j \in \{1, 2\}$. We say that $\rho$ satisfies ``nested model selection consistency" if $\rho(M_1, M_2) \rightarrow \infty$ as $N\to\infty$ when $M_1$ is well-specified with respect to $p_0(x_\F)$, $M_1 \subset M_2$, and $\dim(\Theta_1) < \dim(\Theta_2)$.
\end{definition}

In each case, $\rho$ may diverge almost surely (``strong consistency'') or in probability (``weak consistency''). 
Note that in Definitions~\ref{def:ds}--\ref{def:nds}, the difference between $M_1$ and $M_2$ is the choice of foreground data space $\F$, whereas in Definitions~\ref{def:ms}--\ref{def:nms}, $M_1$ and $M_2$ are over the same foreground space but employ different model spaces.

In Theorem~\ref{thm:selection_consistency}, we show that the SVC has the asymptotic properties outlined in Section~\ref{sec:method_asymptotics}.
In combination with the subsystem independence properties of the NKSD (Propositions~\ref{prop:subsystem_indep} and \ref{prop:approx_subsystem_indep}), Theorem~\ref{thm:selection_consistency} also leads to the conclusion that the SVC approximates the NKSD marginal likelihood of the augmented model (Equation~\ref{eqn:full_svc_approx}). % , as described in Section~\ref{sec:approx_standard} 
Our proof is similar in spirit to previous results for model selection with the standard marginal likelihood, notably those of~\citet{Hong2005-kf} and~\citet{Huggins2020-bu}, but relies on the special properties of the \textsc{nksd} marginal likelihood in Theorem~\ref{prop:laplace_scaling}. 

\begin{theorem}
\label{thm:selection_consistency}
For $j\in\{1,2\}$, assume the conditions of Theorem~\ref{prop:laplace_scaling} hold for model $M_j$ defined on $\X_{\F_j}$, with density $q_j(x_{\F_j}|\theta_j)$ for $\theta_j\in\Theta_j\subseteq \mathbb{R}^{m_{\F_j,j}}$.
Let $\mathcal{K}_{j,N}$ be the Stein volume criterion for $M_j$, with background model penalty $m_{\B_j}=m_{\B_j}(N)$, and let $\theta_{j,*} := \argmin_{\theta_j} \textsc{nksd}(p_0(x_{\F_j})\| q_j(x_{\F_j}|\theta_j))$.
Then:
\begin{enumerate}
	\item If $m_{\B_j} = o(N/\log N)$ for $j\in\{1,2\}$, then
	\begin{equation*}
	\frac{1}{N}\log \frac{\mathcal{K}_{1,N}}{\mathcal{K}_{2,N}} \xrightarrow[N \to \infty]{P_0} \frac{1}{T}\textsc{nksd}(p_0(x_{\F_2})\| q_2(x_{\F_2}|\theta_{2,*})) - \frac{1}{T}\textsc{nksd}(p_0(x_{\F_1})\| q_1(x_{\F_1}|\theta_{1,*})).
	\end{equation*}
	
	\item If $\textsc{nksd}(p_0(x_{\F_1})\| q_1(x_{\F_1}|\theta_{1,*})) = \textsc{nksd}(p_0(x_{\F_2})\| q_2(x_{\F_2}|\theta_{2,*})) = 0$ and $m_{\B_2} - m_{\B_1}$ does not depend on $N$, then
	\begin{equation*}
		\frac{1}{\log N}\log \frac{\mathcal{K}_{1,N}}{\mathcal{K}_{2,N}} \xrightarrow[N \to \infty]{P_0} \frac{1}{2}(m_{\F_2,2} + m_{\B_2} - m_{\F_1,1} - m_{\B_1}).
	\end{equation*}
	
	\item If $\textsc{nksd}(p_0(x_{\F_1})\| q_1(x_{\F_1}|\theta_{1,*})) = \textsc{nksd}(p_0(x_{\F_2})\| q_2(x_{\F_2}|\theta_{2,*}))$, $m_{\B_1} = c_{\B_1} \sqrt{N}$, and $m_{\B_2} = c_{\B_2} \sqrt{N}$, where $c_{\B_1}$ and $c_{\B_2}$ are positive and constant in $N$, then
	\begin{equation*}
		\frac{1}{\sqrt{N}\log N}\log \frac{\mathcal{K}_{1,N}}{\mathcal{K}_{2,N}} \xrightarrow[N \to \infty]{P_0} \frac{1}{2}(c_{\B_2} - c_{\B_1}).
	\end{equation*}
\end{enumerate}
\end{theorem}

The proof is in Section~\ref{sec:si_proof_selection}.
In particular, assuming the conditions of Theorem~\ref{prop:laplace_scaling}, we obtain the following consistency results in terms of convergence in probability.
Let $D_j := \textsc{nksd}(p_0(x_{\F_j})\| q_j(x_{\F_j}|\theta_{j,*}))$ for $j\in\{1,2\}$.
\begin{itemize}
\item If $m_{\B_j} = o(N/\log N)$ then the SVC exhibits data selection consistency and model selection consistency.
This holds by Theorem~\ref{thm:selection_consistency} (part 1) since $D_2 > D_1 = 0$.
\item If $m_{\B_1} = m_{\B_2}$ then the SVC exhibits nested model selection consistency.
This holds by Theorem~\ref{thm:selection_consistency} (part 2) since $D_1 = D_2 = 0$, $m_{\B_2} - m_{\B_1} = 0$, and $m_{\F_2,2} > m_{\F_1,1}$.
\item Consider a nested data selection problem with $\X_{\F_2} \subset \X_{\F_1}$.
If (A) $m_{\B_2} - m_{\B_1}$ does not depend on $N$ and $m_{\F_2,2} + m_{\B_2}  > m_{\F_1,1} + m_{\B_1}$ 
or (B) $m_{\B_j} = c_{\B_j} \sqrt{N}$ and $c_{\B_2} > c_{\B_1} > 0$, 
then the SVC exhibits nested data selection consistency.
Cases A and B hold by Theorem~\ref{thm:selection_consistency} (parts 2 and 3, respectively) since $D_1 = D_2 = 0$.
\end{itemize}

\section{Application: Probabilistic PCA} \label{sec:ppca}

Probabilistic principal components analysis (pPCA) is a commonly used tool for modeling and visualization. The basic idea is to model the data as linear combinations of $k$ latent factors plus Gaussian noise. The inferred weights on the factors are frequently used to provide low-dimensional summaries of the data, while the factors themselves describe major axes of variation in the data. 
In practice, pPCA is often applied in settings where it is likely to be misspecified -- for instance, the weights are often clearly non-Gaussian.
In this section, we show how data selection can be used to uncover sources of misspecification and to analyze how this misspecification affects downstream inferences.

The generative model used in pPCA is
\begin{equation}
\begin{split}
	Z^{(i)} &\sim \mathcal{N}(0, I_k),\\
	X^{(i)}|Z^{(i)} &\sim \mathcal{N}(H Z^{(i)}, v I_d),
\end{split}
\end{equation}
independently for $i=1,\ldots,N$, where $I_k$ is the $k$-dimensional identity matrix, $Z^{(i)}\in\mathbb{R}^k$ is the weight vector for datapoint $i$, $H \in \mathbb{R}^{d \times k}$ is the unknown matrix of latent factors, and $v > 0$ is the variance of the noise. To form a Laplace approximation for the Stein volume criterion, we follow the approach developed by \citet{Minka2000-uw} for the standard marginal likelihood. Specifically, we parameterize $H$ as
\begin{equation}
    H = U (L - v I_k)^{1/2}
\end{equation}
where $U$ is a $d \times k$ matrix with orthonormal columns (that is, it lies on the Stiefel manifold) and $L$ is a $k \times k$ diagonal matrix. 
We use the priors suggested by \citet{Minka2000-uw},
\begin{equation}
\begin{split}
U &\sim \mathrm{Uniform}(\mathcal{U}),\\
L_{i i} &\sim \mathrm{InverseGamma}(\alpha/2,\, \alpha/2),\\
v &\sim \mathrm{InverseGamma}\big((\alpha/2 + 1)(d-k) - 1,\; (\alpha/2)(d - k)\big),
\end{split}
\end{equation}
where $\mathcal{U}$ is the set of $d\times k$ matrices with orthonormal columns and $L_{i i}$ is the $i$th diagonal entry of $L$.
We set $\alpha = 0.1$ in the following experiments, and we use pymanopt~\citep{Townsend2016-qk} to optimize $U$ over the Stiefel manifold (Section~\ref{sec:si_ppca_details}).

\subsection{Simulations}

In simulations, we evaluate the ability of the SVC to detect partial misspecification. We set $d = 6$, draw the first four dimensions from a pPCA model with $k = 2$ and
\begin{equation}
\begin{split}
	H = \left(\begin{matrix}
		\phantom{-}1 & \phantom{-}0\\
		-1 & \phantom{-}1\\
		\phantom{-}0 & \phantom{-}1\\
		-1 & -1
	\end{matrix}\right),
\end{split}
\end{equation}
and generate dimensions 5 and 6 in such a way that pPCA is misspecified. 
We consider two misspecified scenarios: scenario A (Figure~\ref{fig:spike_viz}) is that
\begin{equation}
\begin{split}
	W^{(i)} &\sim \mathrm{Bernoulli}(0.5),\\
	X^{(i)}_{5:6}\mid W^{(i)} & \sim \mathcal{N}\left(0, \Sigma_{W^{(i)}} \right),
\end{split}
\label{eqn:spike_corrupt}
\end{equation}
where $\Sigma_{W^{(i)}} = (0.05)^{W^{(i)}} I_2$.
Scenario B (Figure~\ref{fig:ptfooler_viz}) is the same but with 
\begin{equation}
\Sigma_{W^{(i)}} = \begin{pmatrix} 1 & (-1)^{W^{(i)}} 0.99 \\ (-1)^{W^{(i)}} 0.99 & 1 \end{pmatrix}.
\label{eqn:ptfooler_corrupt}
\end{equation}
Scenario B is more challenging because the marginals of the misspecified dimensions are still Gaussian,
and thus, misspecification only comes from the dependence between $X_{5}$ and $X_{6}$. As illustrated in Figures~\ref{fig:spike_scat} and \ref{fig:ptfooler_scat}, both kinds of misspecification are very hard to see in the lower-dimensional latent representation of the data.

\begin{figure}[t!]
    \centering
    ~
    \begin{subfigure}[t]{0.31\textwidth}
        \centering
        \includegraphics[height=2in]{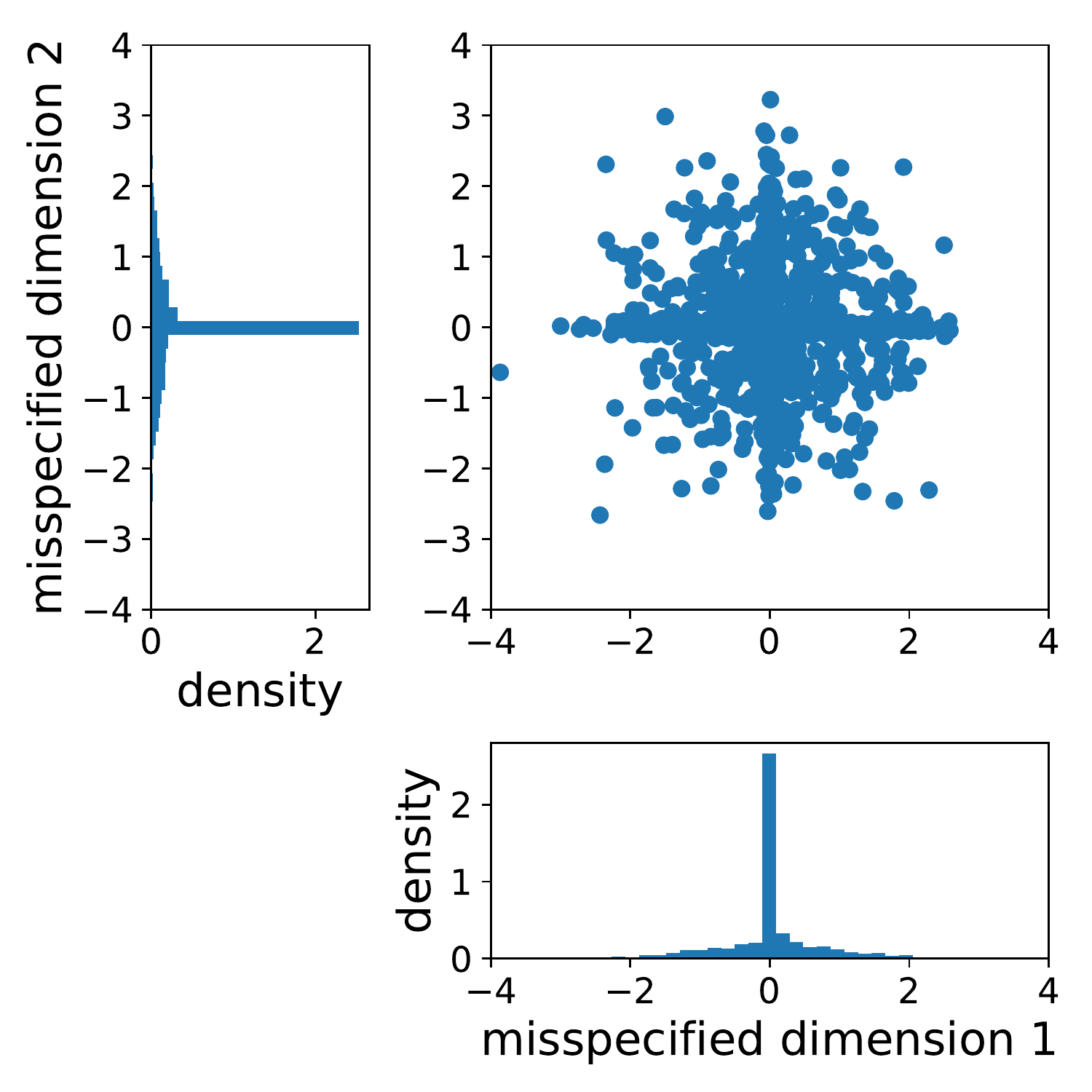}
        \caption{\footnotesize Scenario A, misspecified dimensions.}
        \label{fig:spike_viz}
    \end{subfigure}
    ~ 
    \begin{subfigure}[t]{0.31\textwidth}
        \centering
        \includegraphics[height=2in]{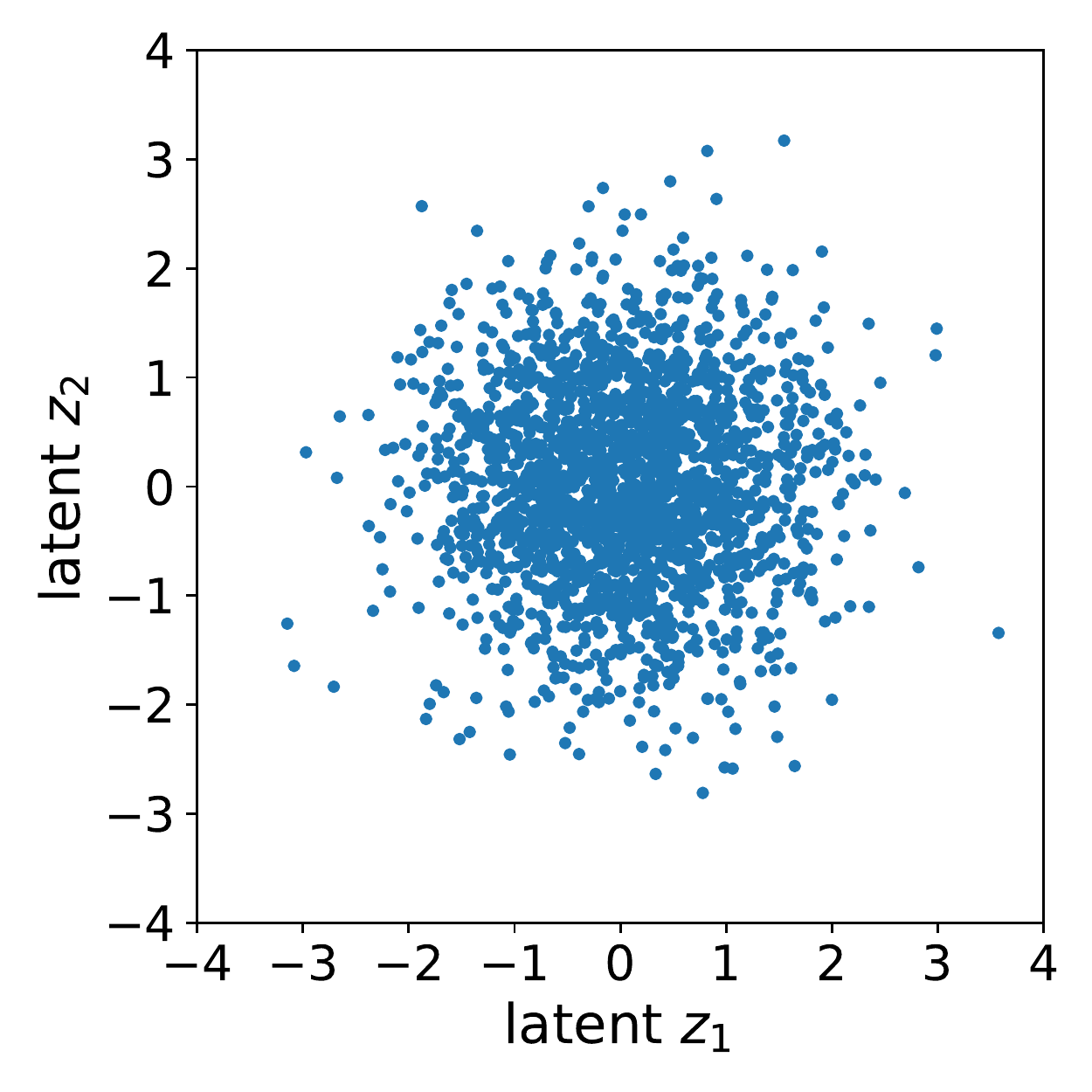}
        \caption{\footnotesize Scenario A, pPCA latent space.}
        \label{fig:spike_scat}
    \end{subfigure}
    ~ 
    \begin{subfigure}[t]{0.31\textwidth}
        \centering
        \includegraphics[height=2in]{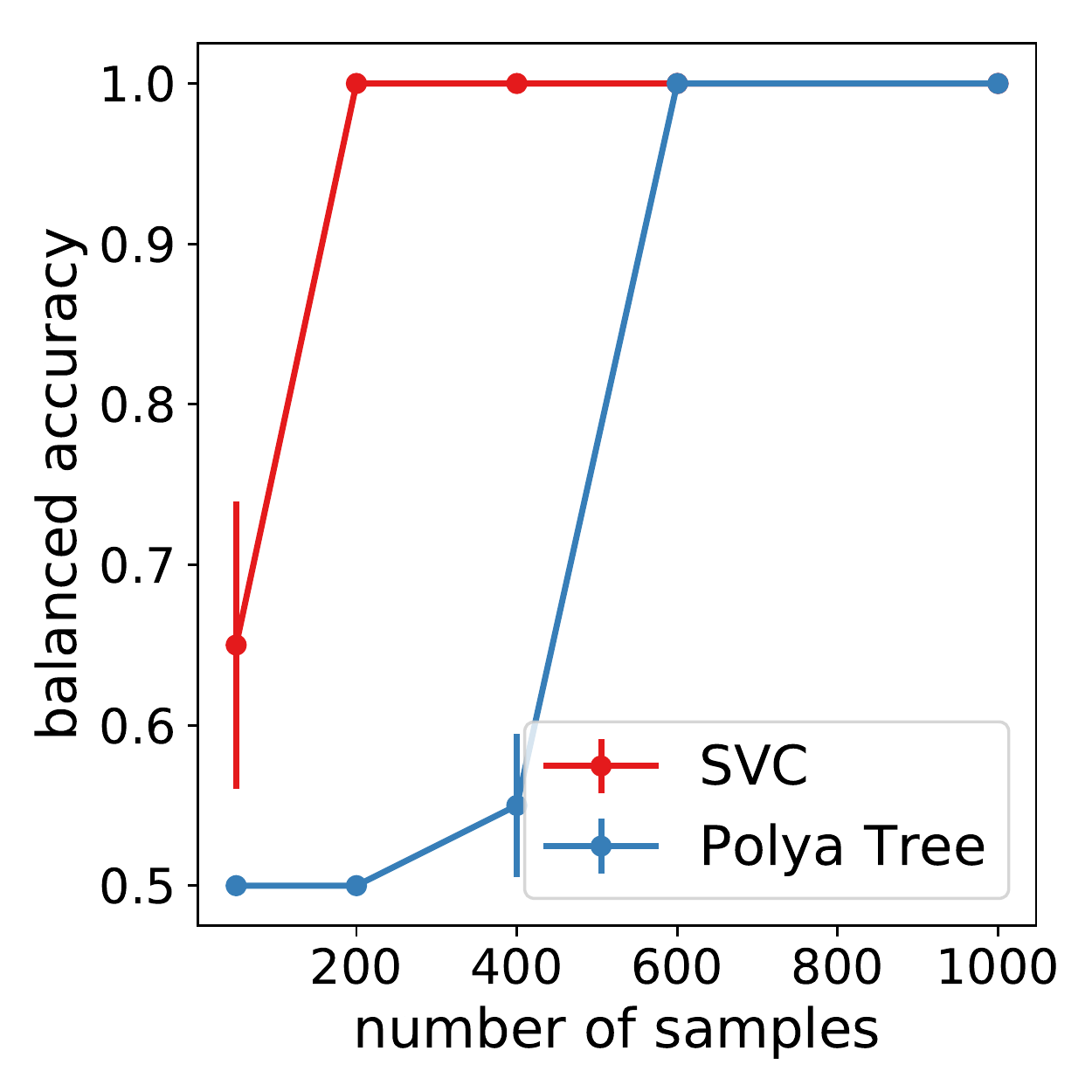}
        \caption{\footnotesize Scenario A, accuracy in detecting misspecified dimensions.}
        \label{fig:spike_acc}
    \end{subfigure}
    \\
    ~
    \begin{subfigure}[t]{0.31\textwidth}
        \centering
        \includegraphics[height=2in]{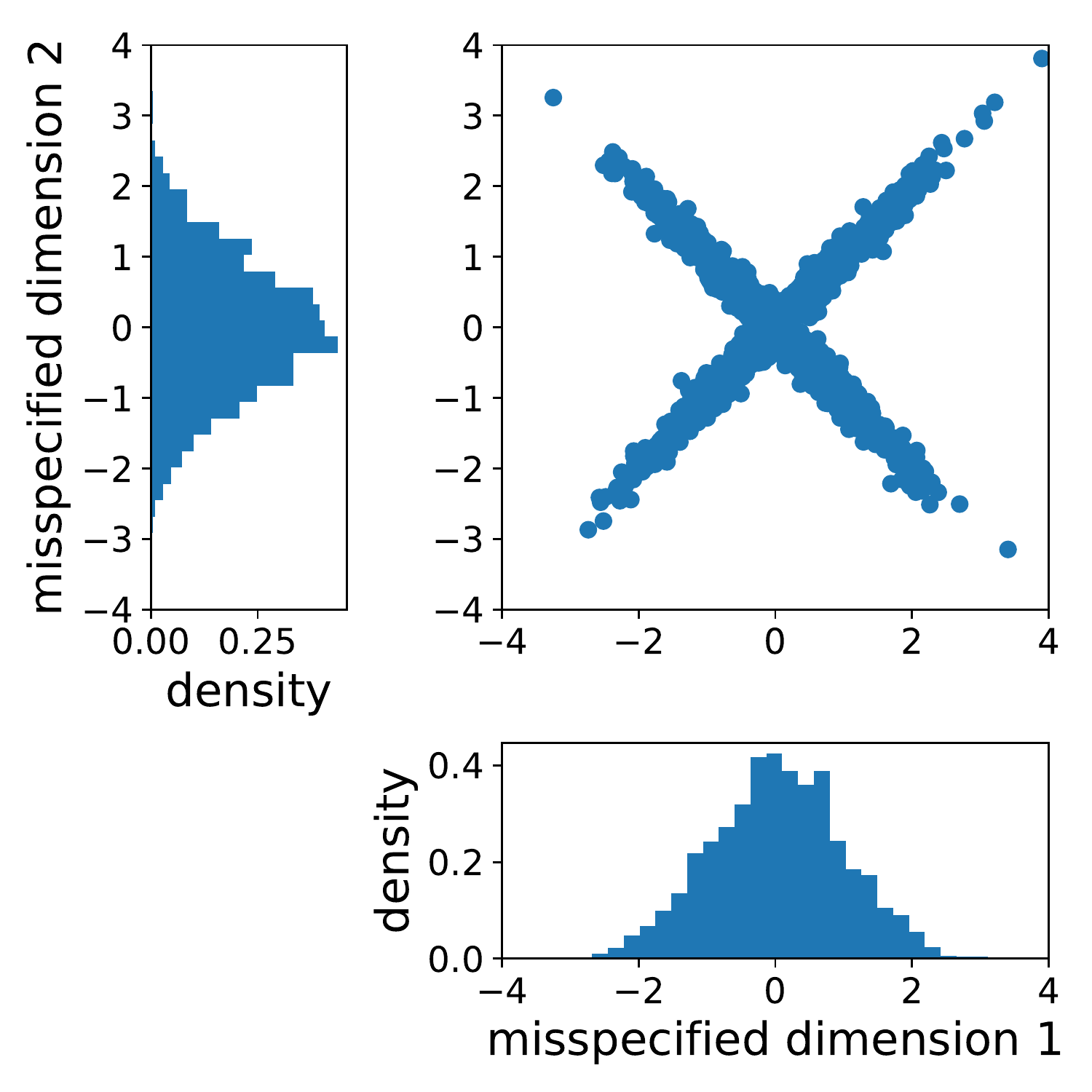}
        \caption{\footnotesize Scenario B, misspecified dimensions.}
        \label{fig:ptfooler_viz}
    \end{subfigure}
    ~
    \begin{subfigure}[t]{0.31\textwidth}
        \centering
        \includegraphics[height=2in]{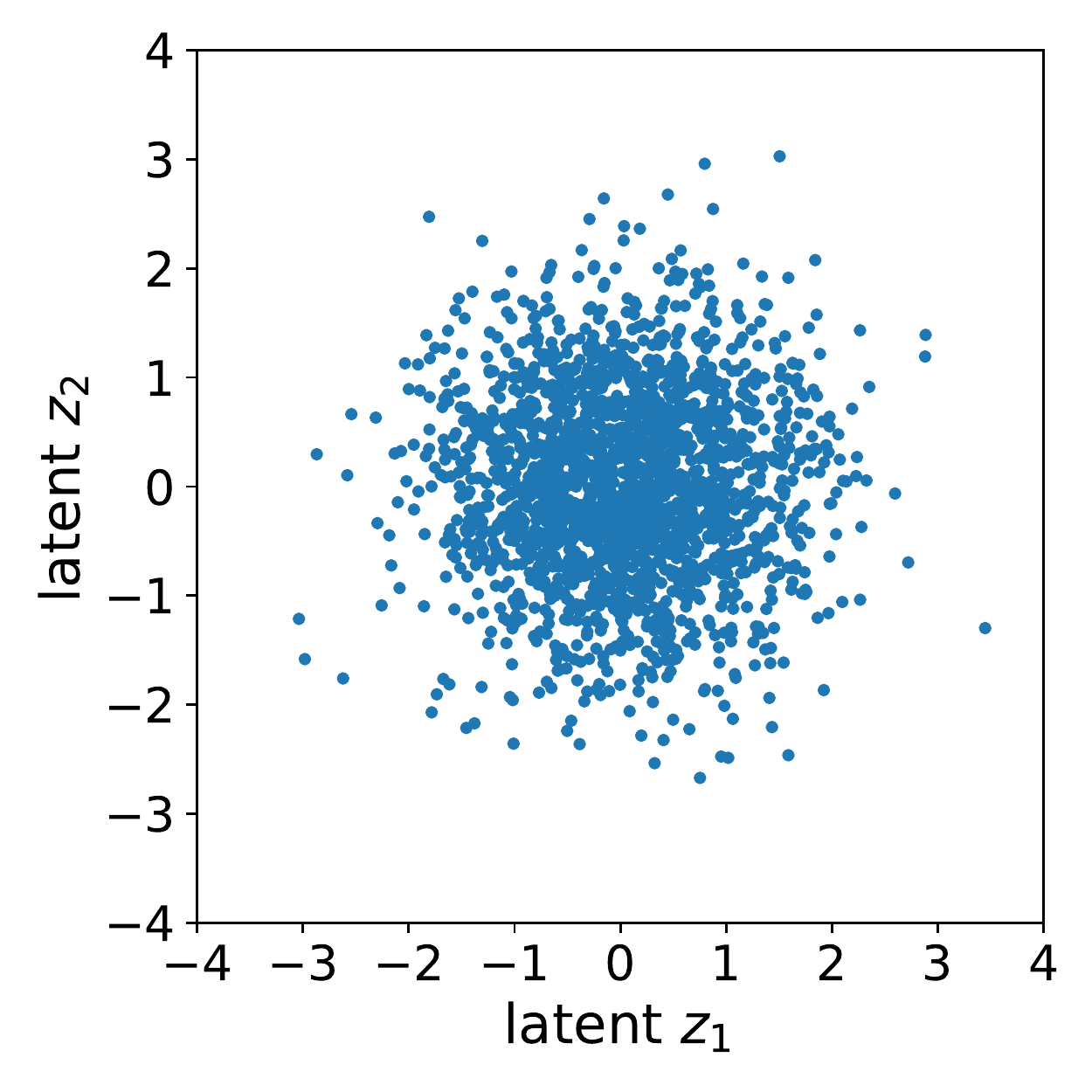}
        \caption{\footnotesize Scenario B, pPCA latent space.}
        \label{fig:ptfooler_scat}
    \end{subfigure}%
    ~ 
    \begin{subfigure}[t]{0.31\textwidth}
        \centering
        \includegraphics[height=2in]{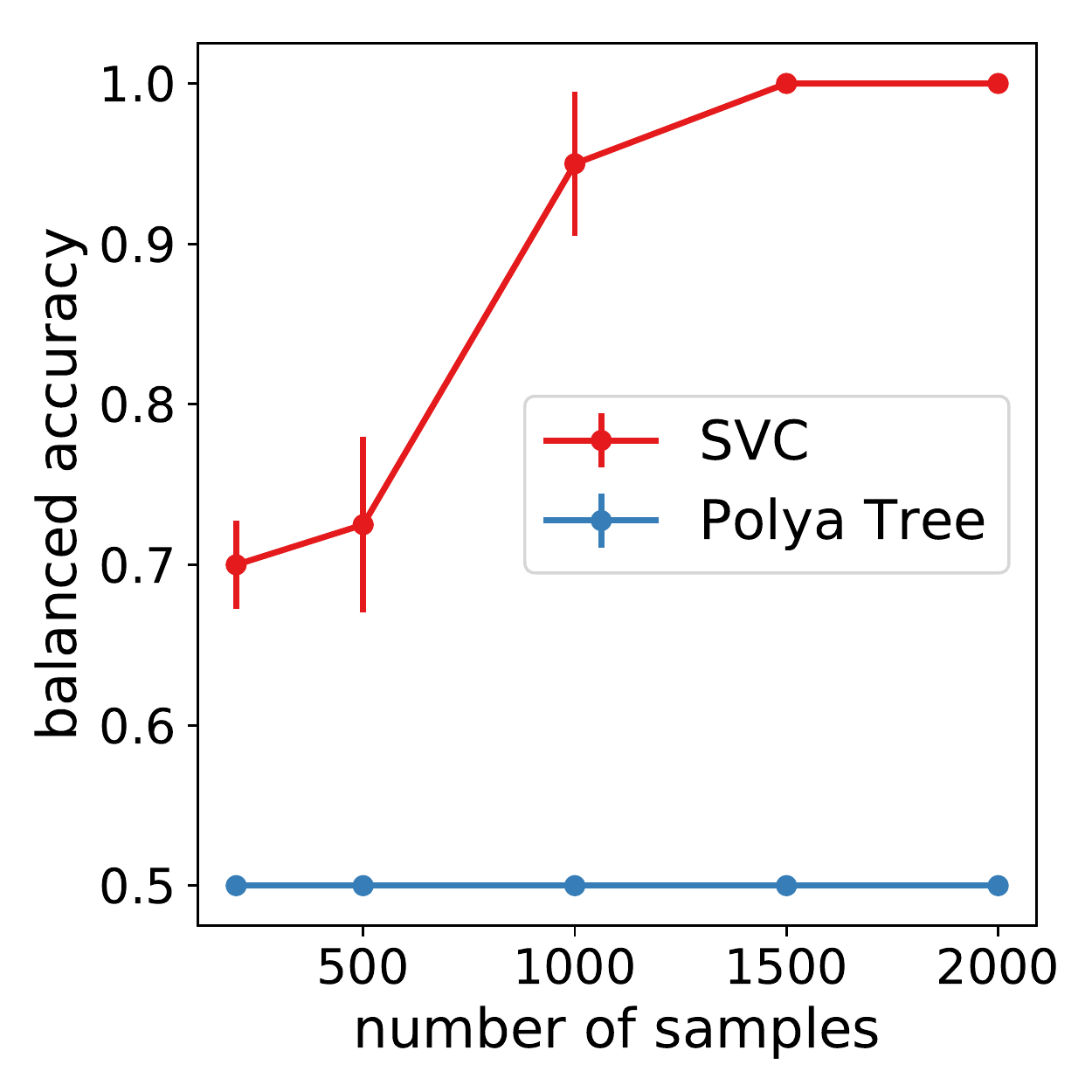}
        \caption{\footnotesize Scenario B, accuracy in detecting misspecified dimensions.}
        \label{fig:ptfooler_acc}
    \end{subfigure}
    \\
    \begin{subfigure}[t]{0.31\textwidth}
        \centering
        \includegraphics[height=2in]{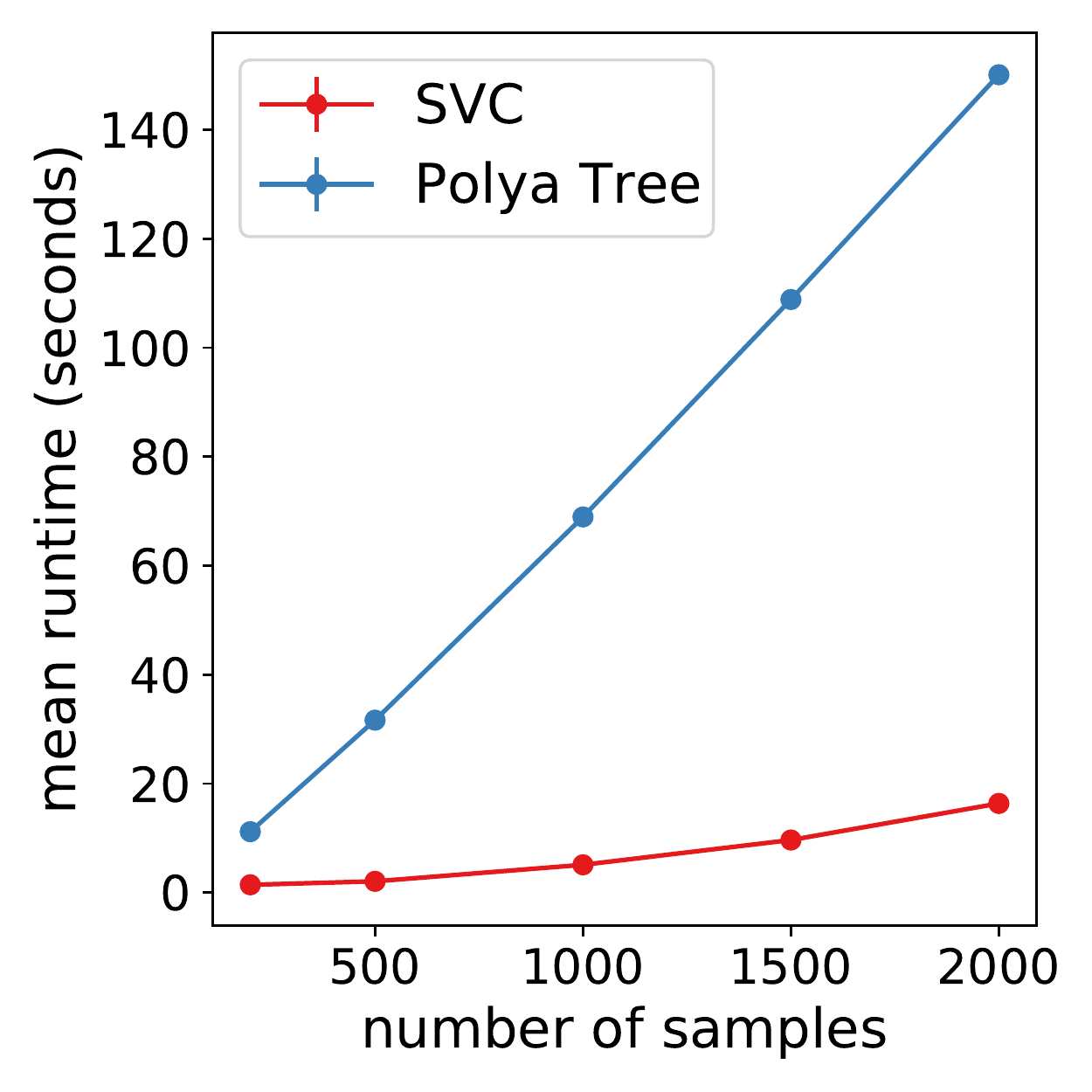}
        \caption{\footnotesize Mean runtime over 5 repeats.}
        \label{fig:spike_timing}
    \end{subfigure}
    \caption{Data selection in the probabilistic PCA model.}
\end{figure}

Our method can be used to both (i) detect misspecified subsets of dimensions, and (ii) conversely, find a maximal subset of dimensions for which the pPCA model provides a reasonable fit to the data.
We set $T = 0.05$ in the SVC, based on the calibration procedure in Section~\ref{sec:calibrate_t} (Section~\ref{sec:si_ppca_calibration}).
We use the Pitman-Yor mixture model expression for the background model dimension (Equation~\ref{eqn:pymm}), with $\alpha=0.5$, $\theta=1$, and $D=0.2$. 
This value of $D$ ensures that the number of background model parameters per data dimension is greater than the number of foreground model parameters per data dimension except for at very small $N$, since there are two foreground parameters for each additional data dimension in the pPCA model, and $m_\B > 2\, r_\B$ for $N \ge 20$. 
%This value of $D$ ensures that the number of background model parameters per data dimension is higher than the number of foreground model parameters per data dimension for moderately small $N$, since there are two foreground parameters for each additional data dimension in the pPCA model, and $m_\B > 2\, r_\B$ for $N \ge 20$. 
% to $\dimB = 5\, r_\B$, where $r_\B$ is the dimension of the background subspace $\XXB$; this is somewhat more than the two parameters added for each data dimension in the pPCA model (Section~\ref{sec:ppca_data_selection_details}). 
We performed leave-one-out data selection, comparing the foreground space $\X_{\F_0} = \X$ to foreground spaces $\X_{\F_j}$ for $j \in \{1, \ldots, d\}$, which exclude the $j$th dimension of the data. We computed the log SVC ratio $\log(\mathcal{K}_j / \mathcal{K}_0) = \log \mathcal{K}_j - \log \mathcal{K}_0$ using the BIC approximation to the SVC (Section~\ref{sec:lapl_bic}) and the approximate optima technique (Section~\ref{sec:approximate_optima}).
We quantify the performance of the method in detecting misspecified dimensions in terms of the balanced accuracy, defined as $(TN/N + TP/P)/2$ where $TN$ is the number of true negatives (dimension by dimension), $N$ is the number of negatives, $TP$ is the number of true positives, and $P$ is the number of positives. Experiments were repeated independently five times. Figures~\ref{fig:spike_acc} and \ref{fig:ptfooler_acc} show that as the sample size increases, the SVC correctly infers that dimensions 1 through 4 should be included and dimensions 5 and 6 should be excluded.

\subsection{Comparison with a nonparametric background model}

To benchmark our method, we compare with an alternative approach that uses an explicit augmented model. The P\'{o}lya tree is a nonparametric model with a closed-form marginal likelihood that is tractable for one-dimensional data~\citep{Lavine1992-fu}. 
We define a flexible background model by sampling each dimension $j$ of the background space independently as
\begin{equation}
X_{j} \sim \mathrm{PolyaTree}(F, \tilde{F}, \eta),
\end{equation}
with the P\'{o}lya tree constructed as by \citet{Berger2001-li} (Section~\ref{sec:si_pt}).
We set $F=\mathcal{N}(0, 10)$, $\tilde{F}=\mathcal{N}(0, 10)$, and $\eta = 1000$ so that the model is weighted only very weakly towards the base distribution. 

We performed data selection using the marginal likelihood of the P\'{o}lya tree augmented model, computing the marginal of the pPCA foreground model using the approximation of~\citet{Minka2000-uw}. The accuracy results for data selection are in Figures~\ref{fig:spike_acc} and~\ref{fig:ptfooler_acc}.
On scenario A (Equation~\ref{eqn:spike_corrupt}), the P\'{o}lya tree augmented model requires significantly more data to detect which dimensions are misspecified.
On scenario B (Equation~\ref{eqn:ptfooler_corrupt}) the P\'{o}lya tree augmented model fails entirely, preferring the full data space $\X_{\F_0} = \X$ which includes all dimensions (Figure~\ref{fig:ptfooler_acc}). The reason is that the background model is misspecified due to the assumption of independent dimensions, and thus, the asymptotic data selection results (Equations~\ref{eqn:data_select_kl} and \ref{eqn:nested_data_kl}) do not hold.
This could be resolved by using a richer background model that allows for dependence between dimensions, however, computing the marginal likelihood under such a model would be computationally challenging.
Even with the independence assumption, the P\'{o}lya tree approach is already substantially slower than the SVC (Figure~\ref{fig:spike_timing}).

\subsection{Application to pPCA for single-cell RNA sequencing}~\label{sec:scrna_ppca}

Single-cell RNA sequencing (scRNAseq) has emerged as a powerful technology for high-throughput characterization of individual cells. It provides a snapshot of the transcriptional state of each cell by measuring the number of RNA transcripts from each gene. PCA is widely used to study scRNAseq datasets, both as a method for visualizing different cell types in the dataset and as a pre-processing technique, where the latent embedding is used for downstream tasks like clustering and lineage reconstruction~\citep{qiu2017reversed,Van_Dijk2018-jw}.
We applied data selection to answer two practical questions in the application of probabilistic PCA to scRNAseq data: (1) Where is the pPCA model misspecified? (2) How does partial misspecification of the pPCA model affect downstream inferences?

\subsubsection*{Model criticism}

Our first goal was to verify that the SVC provides reasonable inferences of partial model misspecification in practice. 
We examined two different scRNAseq datasets, focusing for illustration on a dataset from human peripheral blood mononuclear cells taken from a healthy donor, and pre-processed the data following standard procedures in the field (Section~\ref{sec:si_ppca_datasets}).
We subsampled each dataset to 200 genes (selected randomly from among the 2000 most highly expressed) and 2000 cells (selected randomly) for computational tractability, then mean-subtracted and standardized the variance of each gene, again following standard practice in the field. The number of latent components $k$ was set to 3, based on the procedure of~\citet{Minka2000-rx}.
We performed leave-one-out data selection, comparing the foreground space $\X_{\F_0} := \X$ to foreground spaces $\X_{\F_j}$ that exclude the $j$th gene. We computed the log SVC ratio $\log \mathcal{K}_j - \log \mathcal{K}_0$ using the BIC approximation to the SVC (Section~\ref{sec:lapl_bic}) and the approximate optima technique (Section~\ref{sec:approximate_optima}). 
We used the same setting of $T$ and of $m_\B$ as was used in simulation, resulting in a background model complexity of $\dimB = 20\, r_\B$ for datasets of this size. 
Based on the SVC criterion, 162 out of 200 genes should be excluded from the foreground pPCA model, suggesting widespread partial misspecification. 
Figure~\ref{fig:ppca_exclude_ex} compares the histogram of individual genes to their estimated density under the pPCA model inferred for $\X_{\F_0} = \X$. Those genes most favored to be excluded (namely, UBE2V2 and IRF8) show extreme violations of normality, in stark contrast to those genes most favored to be included (MT-CO1 and RPL6).

\begin{figure}[t!]
    \centering
    \begin{subfigure}[t]{0.4\textwidth}
        \centering
        \includegraphics[height=2in]{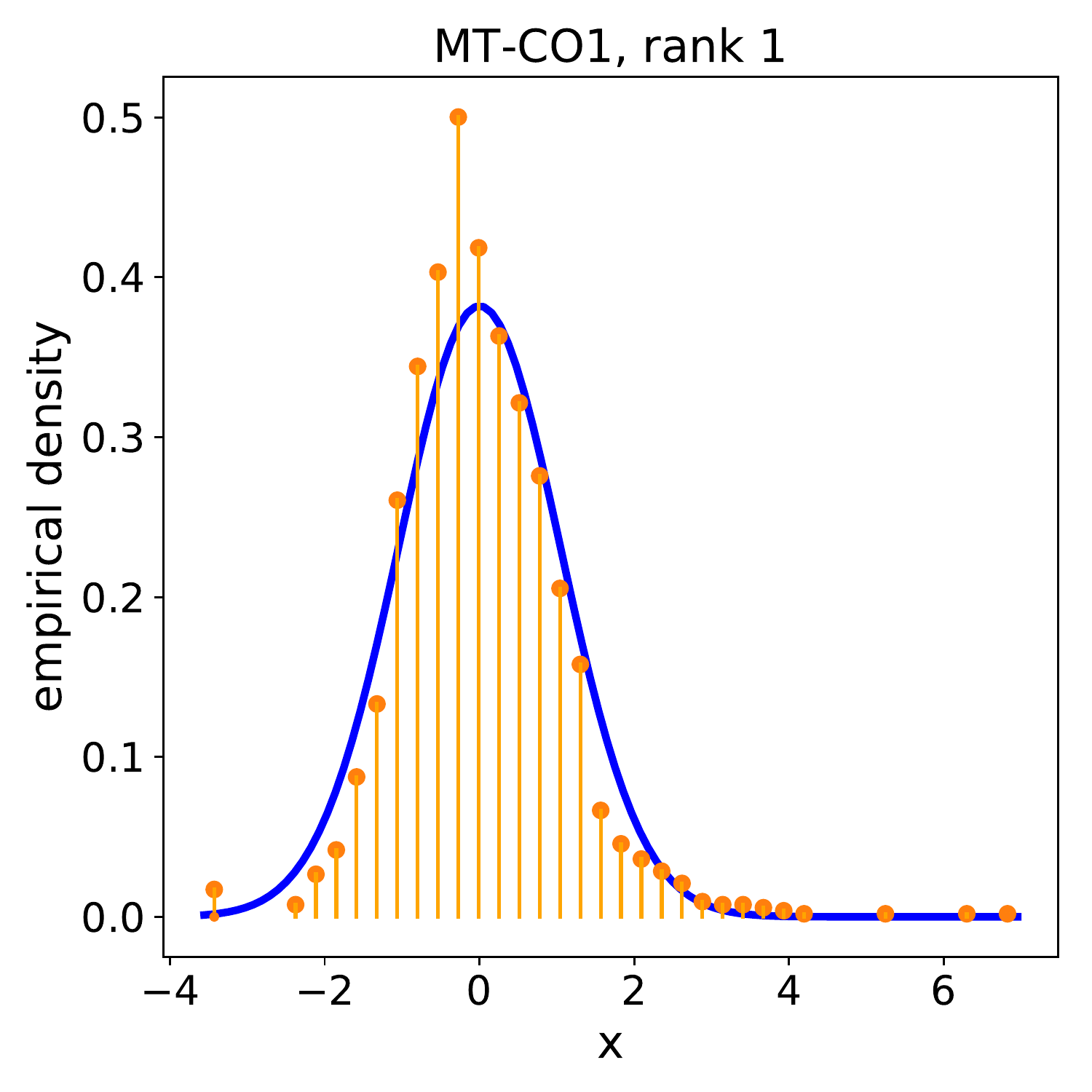}
        \caption{}
    \end{subfigure}
    ~
    \begin{subfigure}[t]{0.4\textwidth}
        \centering
        \includegraphics[height=2in]{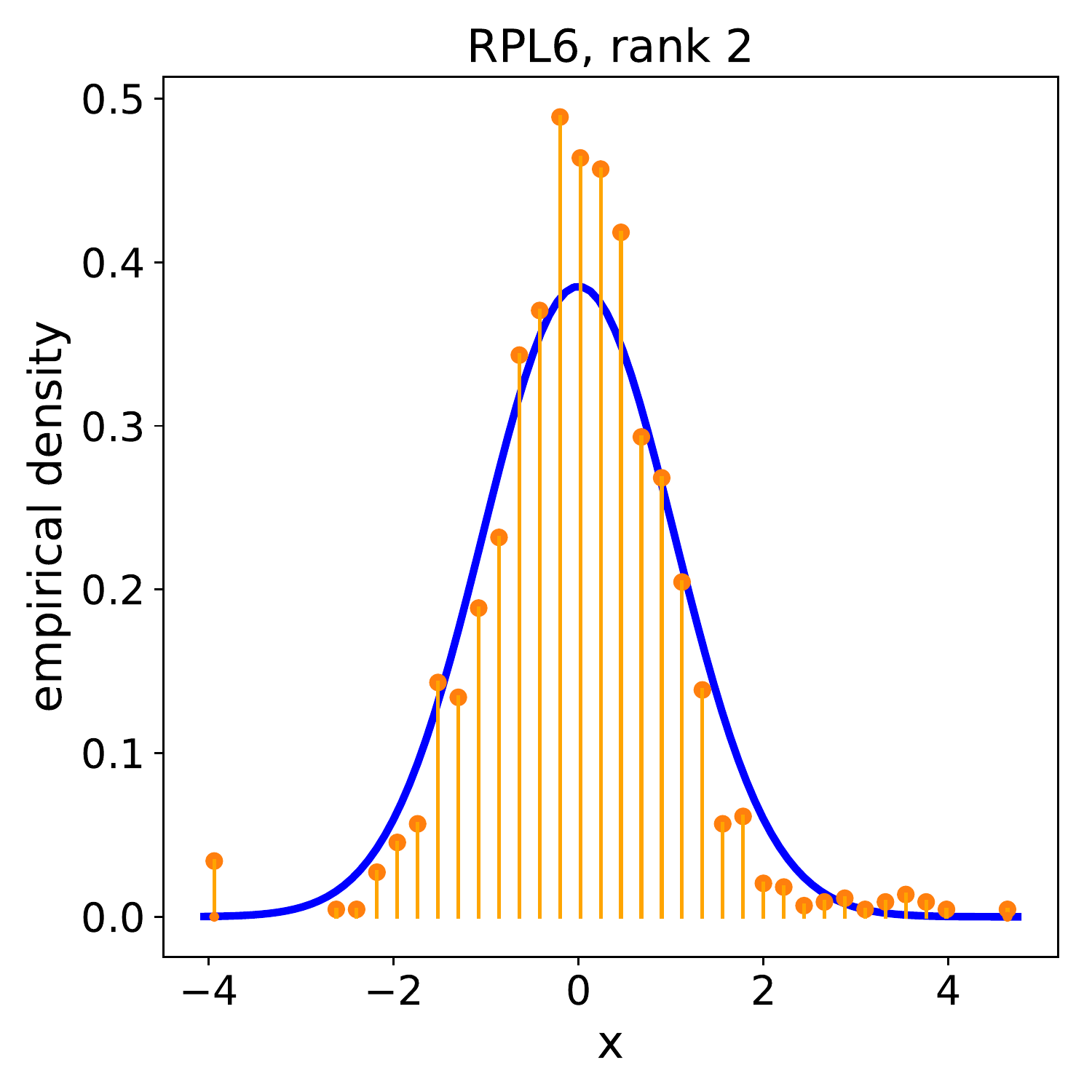}
        \caption{}
    \end{subfigure}
    ~\\
    \begin{subfigure}[t]{0.4\textwidth}
        \centering
        \includegraphics[height=2in]{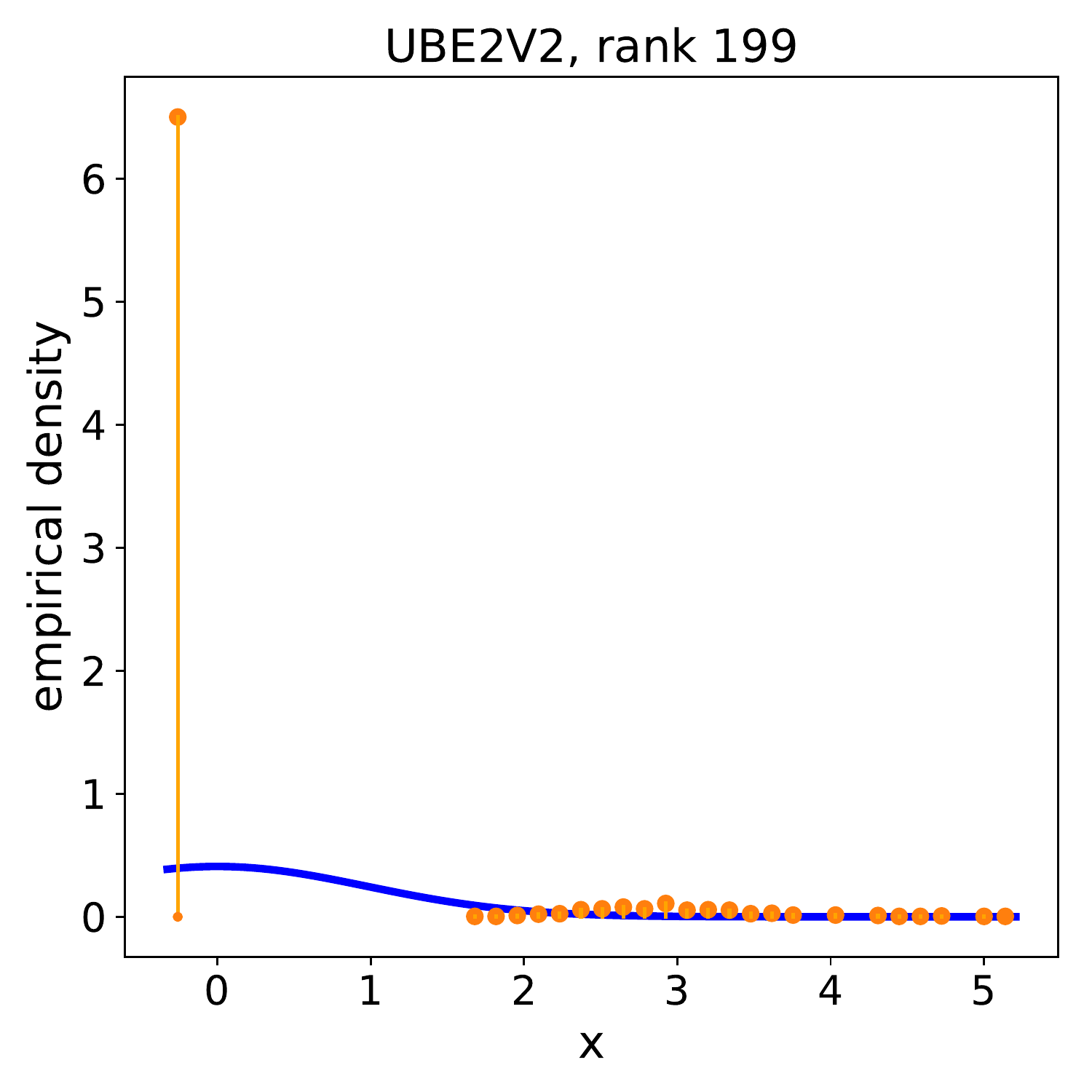}
        \caption{}
    \end{subfigure}
    ~
    \begin{subfigure}[t]{0.4\textwidth}
        \centering
        \includegraphics[height=2in]{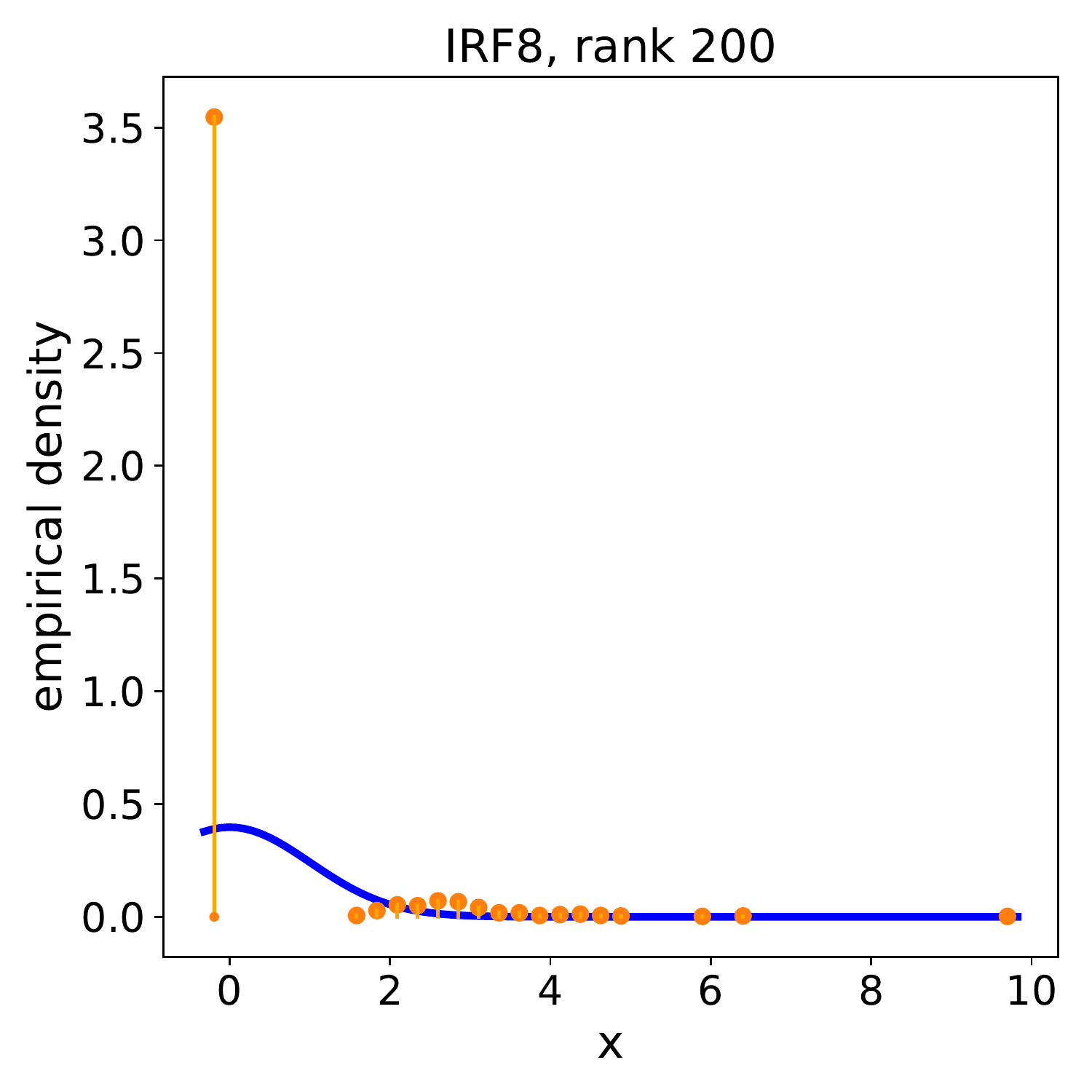}
        \caption{}
    \end{subfigure}
    \caption{(a,b) Histograms of example genes (after pre-processing) selected to be included in the foreground space based on the log SVC ratio, $\log \mathcal{K}_j - \log \mathcal{K}_0$. The estimated density under the pPCA model is shown in blue. (c,d) Histograms of example genes selected to be excluded. Higher ranks (in each title) correspond to larger log SVC ratios.}
    \label{fig:ppca_exclude_ex}
\end{figure}

Next, we compared the results of our data selection approach to a more conventional strategy for model criticism.
Criticism of partially misspecified models can be challenging in practice because misspecification of the model over some dimensions of the data can lead to substantial model-data mismatch in dimensions for which the model is indeed well-specified~\citep{Jacob2017-hu}. The standard approach to model criticism---first fit a model, then identify aspects of the data that the model poorly explains---can therefore be misleading if our aim is to determine how the model might be improved (e.g., in the context of ``Box's loop", \citealp{Blei2014-vn}). 
In particular, standard approaches such as posterior predictive checks will be expected to overstate problems with components of the model that are well-specified and understate problems with components of the model that are misspecified. 
Bayesian data selection circumvents this issue by evaluating augmented models, which replace potentially misspecified components of the model by well-specified components.
To illustrate the difference between these approaches in practice, we compared the SVC to a closely analogous measurement of error for the full foreground model (inferred from $\X_{\F_0} = \X$),
\begin{equation}
	\log \mathcal{E}_j - \log \mathcal{E}_0 := -\frac{N}{T} \widehat{\textsc{nksd}}(p_0(x_{\F_j})\| q(x_{\F_j}|\theta_{0,N})) + \frac{N}{T} \widehat{\textsc{nksd}}(p_0(x)\| q(x|\theta_{0,N}))
\label{eqn:alternative_criticism}
\end{equation}
where $\theta_{0,N} := \argmin \widehat{\textsc{nksd}}(p_0(x)\| q(x|\theta))$ is the minimum \textsc{nksd} estimator for the foreground model when including all dimensions. This model criticism score evaluates the amount of model-data mismatch contributed by the subspace $\X_{\B_j}$ when modeling all data dimensions with the foreground model. For comparison, the BIC approximation to the log SVC ratio is
\begin{equation}
\begin{split}
	\log \mathcal{K}_j - \log \mathcal{K}_0 \approx -&\frac{N}{T} \widehat{\textsc{nksd}}(p_0(x_{\F_j})\| q(x_{\F_j}|\theta_{j,N}) + \frac{N}{T} \widehat{\textsc{nksd}}(p_0(x)\| q(x|\theta_{0,N}))\\
	&+ \frac{m_{\B_j} + m_{\F_j} - m_{\F_0}}{2} \log\bigg(\frac{2\pi}{N}\bigg)
\end{split}
\label{eqn:delta_svc}
\end{equation}
where $\theta_{j,N} := \argmin \widehat{\textsc{nksd}}(p_0(x_{\F_j})\| q(x_{\F_j}|\theta))$ is the minimum \textsc{nksd} estimator for the projected foreground model applied to the restricted dataset, which we approximate as $\theta_{0,N}$ plus the implicit function correction derived in Section~\ref{sec:approximate_optima}. 
Figure~\ref{fig:svc_v_naive} illustrates the differences between the conventional criticism approach ($\log \mathcal{E}_j - \log \mathcal{E}_0$) and the log SVC ratio on an scRNAseq dataset. To enable direct comparison of the two methods, we focus on the lower order terms of Equation~\ref{eqn:delta_svc}, that is, we set $m_{\B_j} = m_{\F_0} - m_{\F_j}$. 
We see that the amount of error contributed by $\X_{\B_j}$, as judged by the SVC, is often substantially higher than the amount indicated by the conventional criticism approach, implying that the conventional criticism approach understates the problems caused by individual genes and, conversely, overstates the problems with the rest of the model. 
 
\begin{figure}[t!]
\centering
\includegraphics[height=3in]{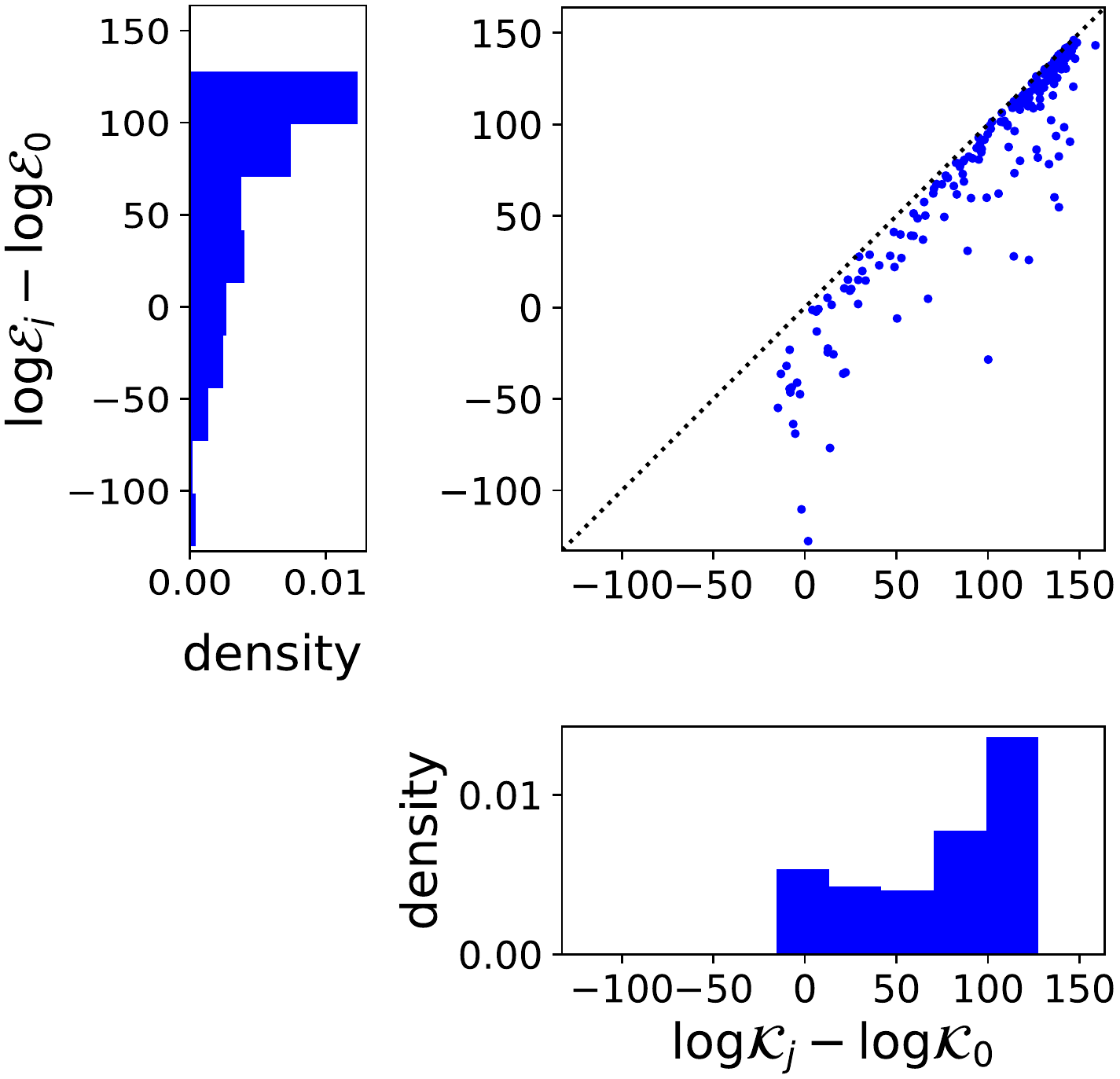}
\caption{Scatterplot comparison and projected marginals of the leave-one-out log SVC ratio, $\log \mathcal{K}_j - \log \mathcal{K}_0$ (with $m_{\B_j} = m_{\F_0} - m_{\F_j}$), and the conventional full model criticism score, $\log \mathcal{E}_j - \log \mathcal{E}_0$, for each gene.
%\comment{(It would be nice to (a) make the scatterplot square and (b) draw a thin gray diagonal line to show where $x = y$.)}
}
\label{fig:svc_v_naive}
\end{figure}

Using the SVC instead of a standard criticism approach can also help clarify trends in where the proposed model fails. 
A prominent concern in scRNAseq data analysis is the common occurrence of cells that show exactly zero expression of a certain gene~\citep{Pierson2015-fv,Hicks2018-zk}. 
We found a Spearman correlation of $\rho=0.89$ between the conventional criticism $\log \mathcal{E}_j - \log \mathcal{E}_0$ for a gene $j$ and the fraction of cells with zero expression of that gene $j$, suggesting that this is an important source of model-data mismatch in this scRNAseq dataset, but not necessarily the only source (Figure~\ref{fig:pbmc_nzeros_v_naive}). 
However, the log SVC ratio yields a Spearman correlation of $\rho=0.98$, suggesting instead that the amount of model-data mismatch can be entirely explained by the fraction of cells with zero expression (Figure~\ref{fig:pbmc_nzeros_v_ij}). 
These observations are repeatable across different scRNAseq datasets (Figure~\ref{fig:malt_nzeros_v_naive}, \ref{fig:malt_nzeros_v_ij}).
% The lower correlation with $\log \mathcal{E}_j - \log \mathcal{E}_0$ is therefore explainable as the result of follow-on errors from partial misspecification.
%We can conclude that, as a rule-of-thumb in the analysis of scRNAseq data, the pPCA model fails to capture the distribution of transcripts for genes that often show no expression at all. This observation raises concerns that downstream analysis based on the pPCA model may obscure the important biological systems that these genes participate in, and suggests that alternative models which account for this phenomenon may be worth considering, particularly in datasets that show a high percentage of such genes.

 \begin{figure}[t!]
\centering
\begin{subfigure}[t]{0.45\textwidth}
    \centering
    \includegraphics[height=2in]{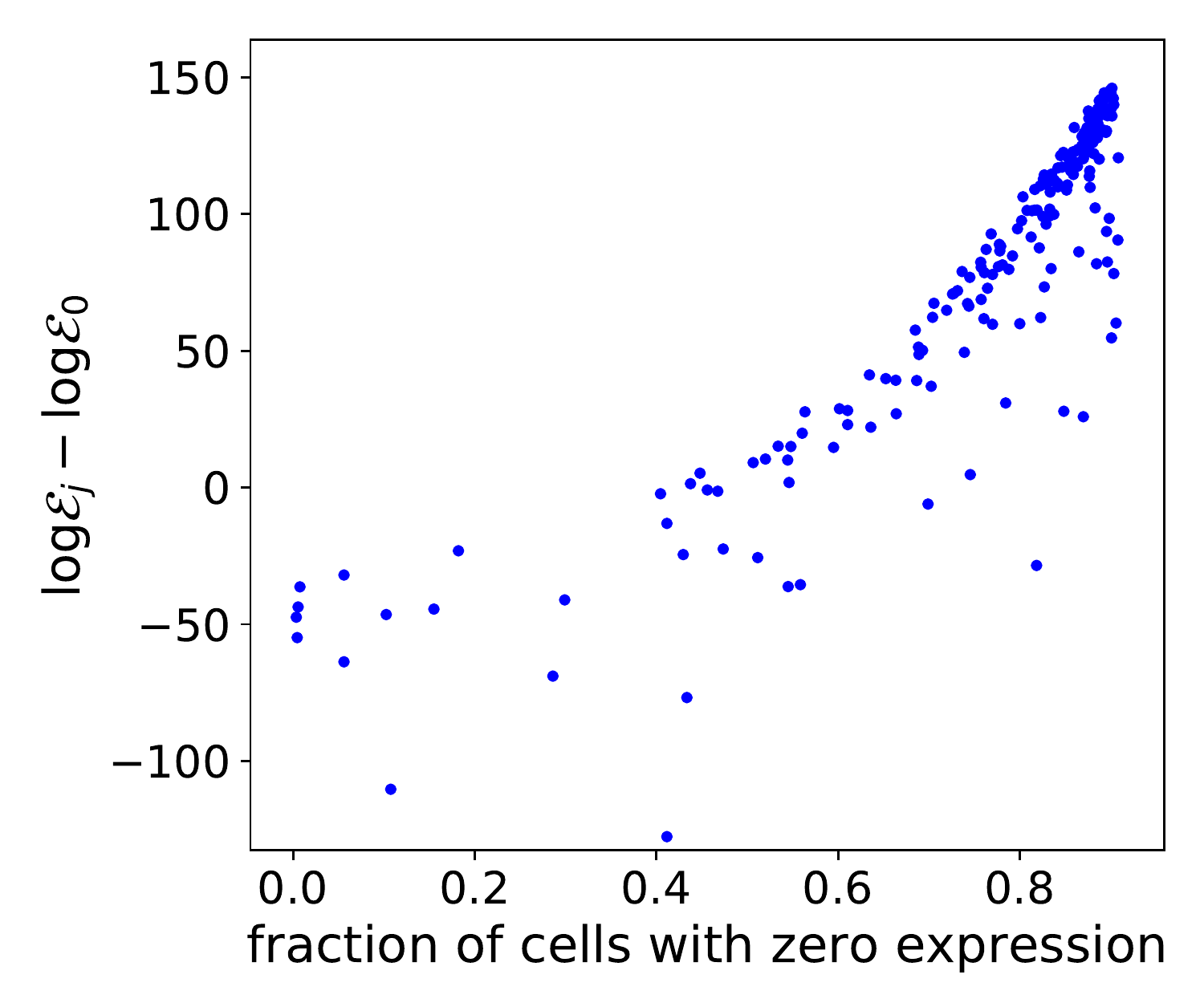}
    \caption{}
    \label{fig:pbmc_nzeros_v_naive}
\end{subfigure}
~
\begin{subfigure}[t]{0.45\textwidth}
    \centering
    \includegraphics[height=2in]{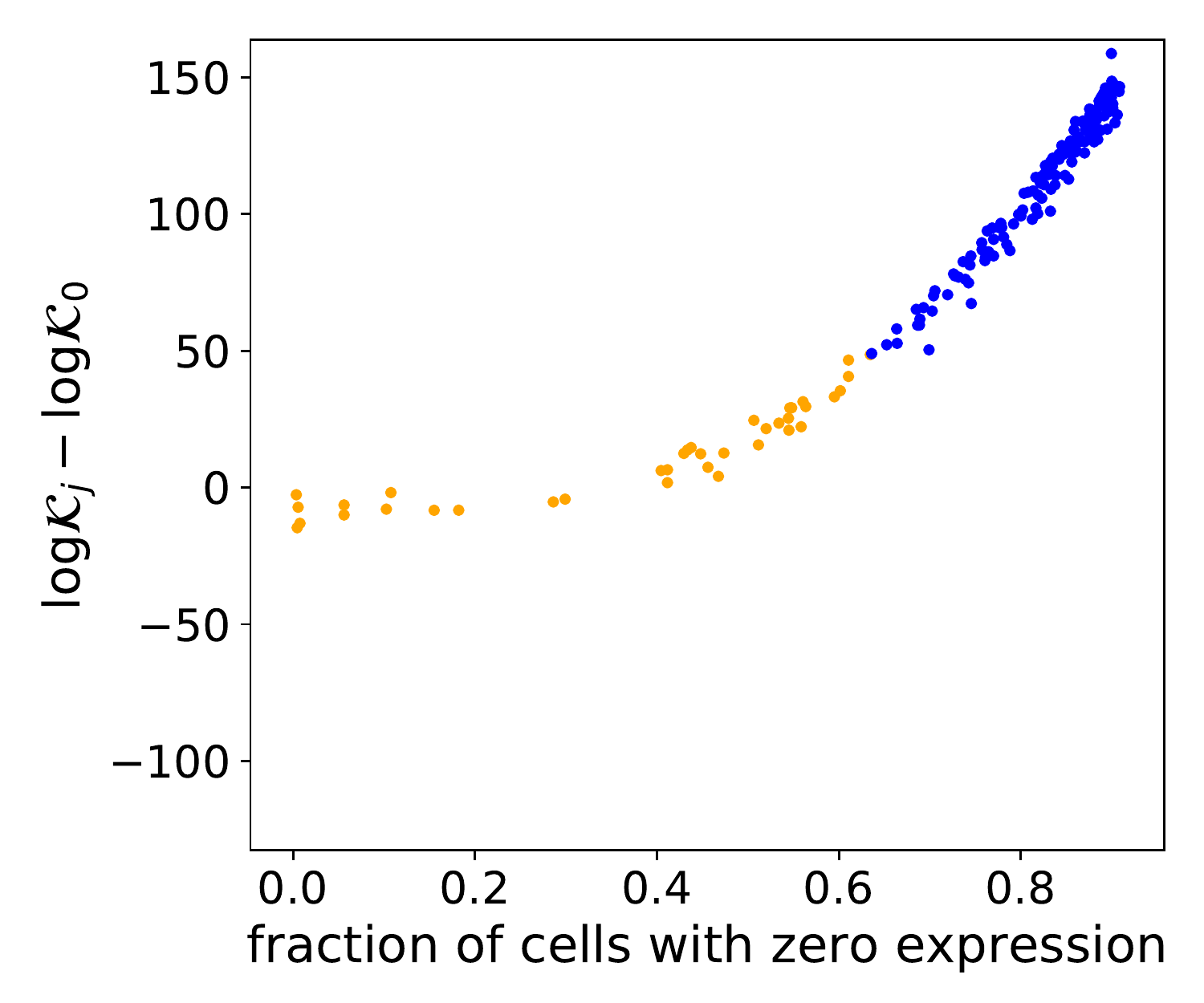}
    \caption{}
    \label{fig:pbmc_nzeros_v_ij}
\end{subfigure}
\begin{subfigure}[t]{0.45\textwidth}
    \centering
    \includegraphics[height=2in]{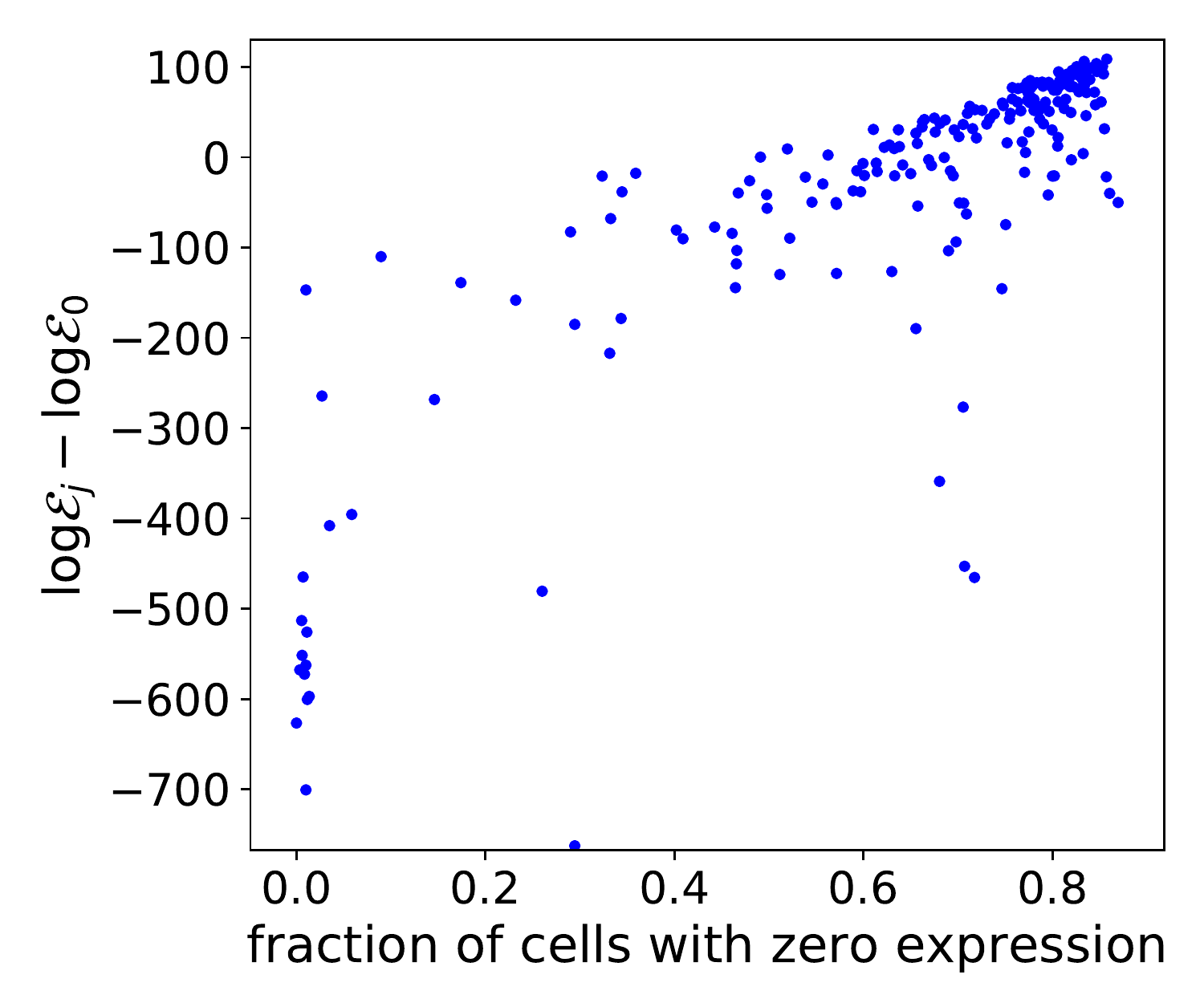}
    \caption{}
    \label{fig:malt_nzeros_v_naive}
\end{subfigure}
~
\begin{subfigure}[t]{0.45\textwidth}
    \centering
    \includegraphics[height=2in]{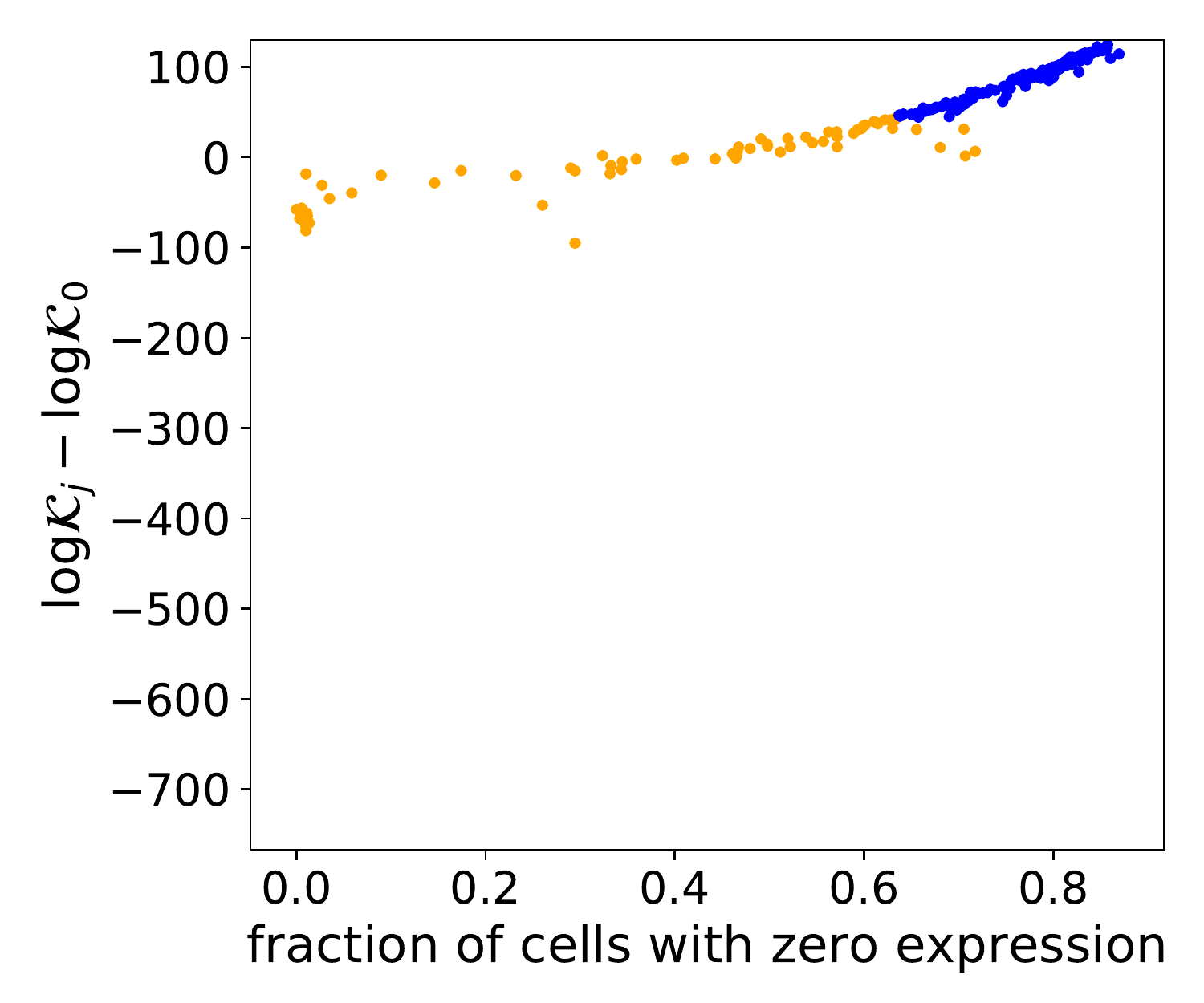}
    \caption{}
    \label{fig:malt_nzeros_v_ij}
\end{subfigure}
\caption{(a) Comparison of the conventional criticism score, for each gene $j$, and the fraction of cells that show zero expression of that gene $j$ in the raw data. Spearman $\rho = 0.89$, $p < 0.01$. (b) Same as (a) but with the log SVC ratio. Spearman $\rho = 0.98$, $p < 0.01$. In orange are genes that would be included when using a background model with $c_\mathcal{B} = 20$ and in blue are genes that would be excluded. 
(c) Same as (a) for a dataset taken from a MALT lymphoma (Section~\ref{sec:si_ppca_datasets}). Spearman $\rho = 0.81$, $p < 0.01$. (d) Same as (b) for the MALT lymphoma dataset. Spearman $\rho = 0.99$, $p < 0.01$.}
\label{fig:ppca_error_explore}
\end{figure}
    
\subsubsection*{Evaluating robustness}

Data selection can also be used to evaluate the robustness of the foreground model to partial model misspecification. This is particularly relevant for pPCA on scRNAseq data, since the inferred latent embeddings of each cell are often used for downstream tasks such as clustering, lineage reconstruction, and so on.
Misspecification may produce spurious conclusions, or alternatively, misspecification may be due to structure in the data that is scientifically interesting.
To understand how partial misspecification of the pPCA model affects the latent representation of cells (and thus, downstream inferences), we performed data selection with a sequence of background model complexities $c_\B$, where $\dimB = c_\B\, r_\B$ (Figure~\ref{fig:PCA_IJ_mean}). We inferred the pPCA parameters based only on genes that the SVC selects to include in the foreground subspace. 
Figures~\ref{fig:PCA_color_0_cut_3}-\ref{fig:PCA_color_0_cut_0} visualize how the latent representation changes as $c_\B$ grows and fewer genes are selected. We can observe the representation morphing into a standard normal distribution, as we would expect in the case where the pPCA model is well-specified. However, the relative spatial organization of cells in the latent space remains fairly stable, suggesting that this aspect of the latent embedding is robust to partial misspecification.
We can conclude that, at least in this example, misspecification strongly contributes to the non-Gaussian shape of the latent representation of the dataset, but not to the distinction between subpopulations.

\begin{figure}[t!]
\centering
\begin{subfigure}[t]{0.4\textwidth}
\centering
\includegraphics[height=2.5in]{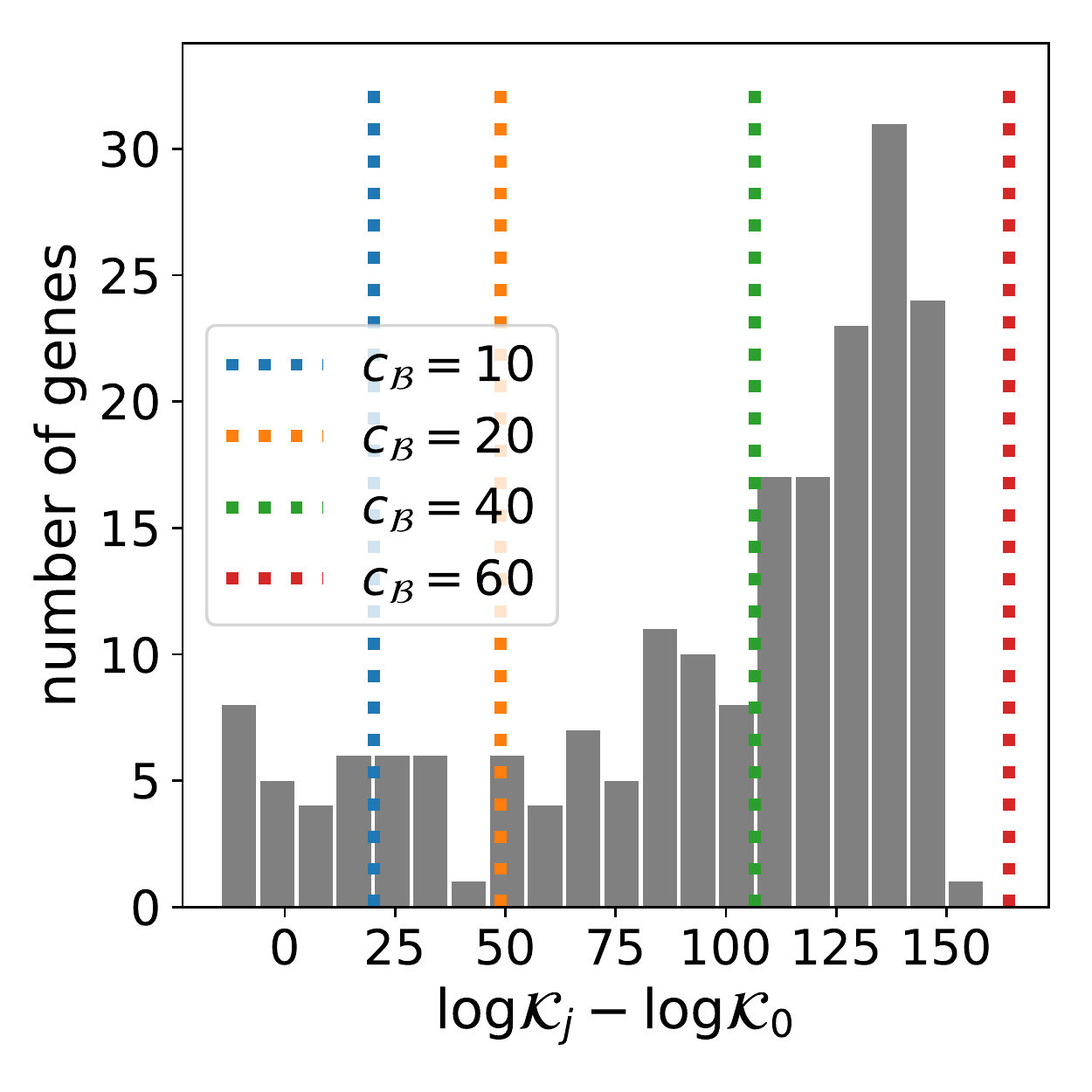}
    \caption{}
\label{fig:PCA_IJ_mean}
\end{subfigure}
~\\
\begin{subfigure}[t]{0.23\textwidth}
    \centering
    \includegraphics[height=1.6in]{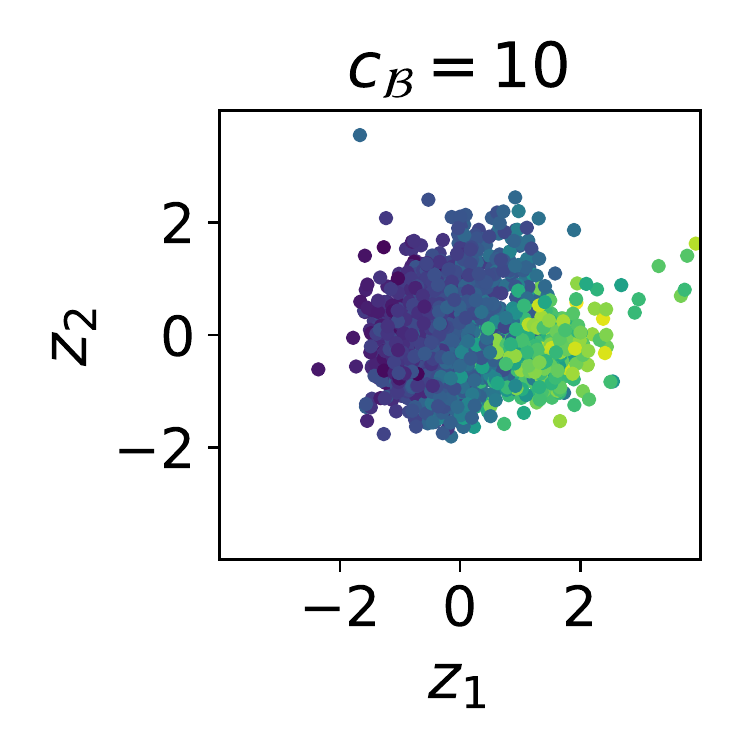}
    \caption{}
\label{fig:PCA_color_0_cut_0}
\end{subfigure}
~
\begin{subfigure}[t]{0.23\textwidth}
    \centering
    \includegraphics[height=1.6in]{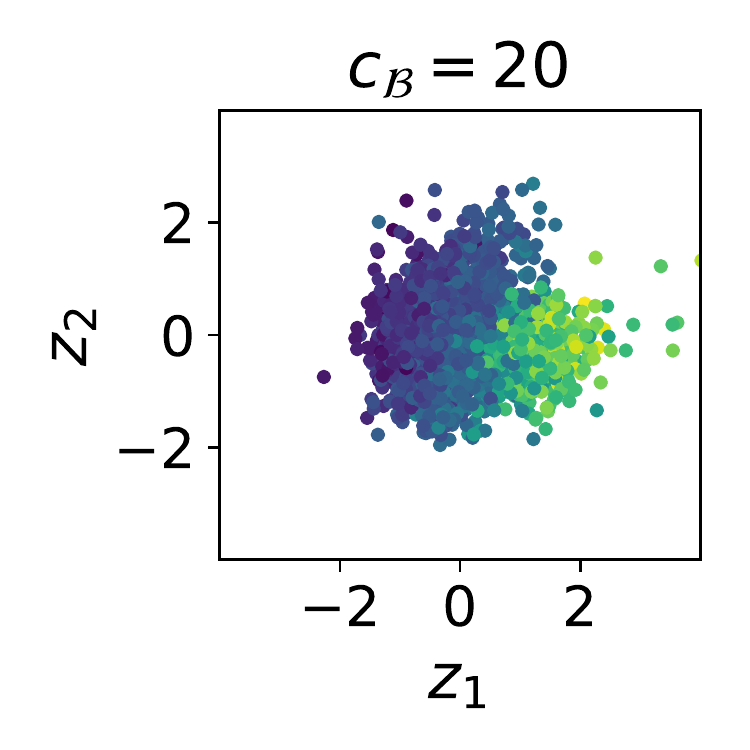}
    \caption{}
\label{fig:PCA_color_0_cut_1}
\end{subfigure}
~
\begin{subfigure}[t]{0.23\textwidth}
    \centering
    \includegraphics[height=1.6in]{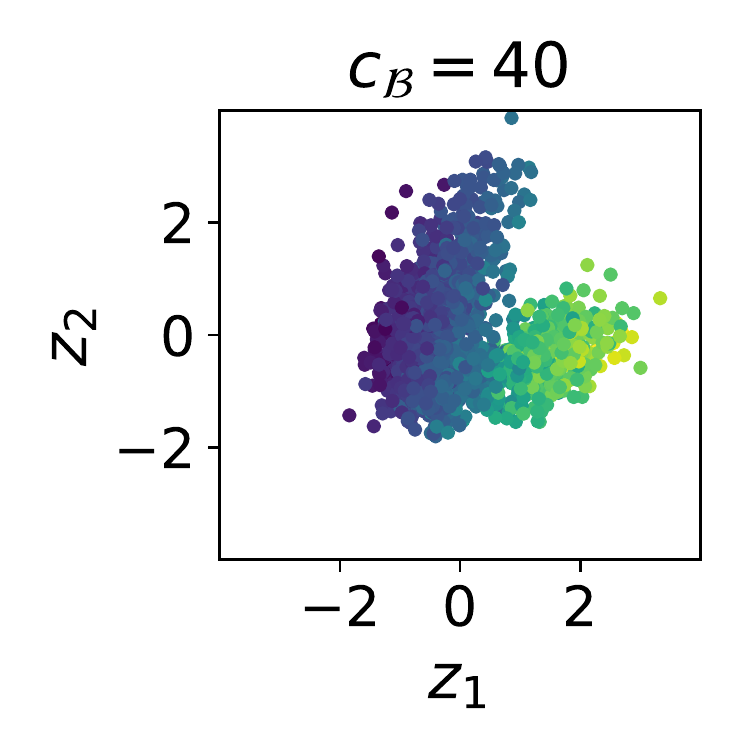}
    \caption{}
\label{fig:PCA_color_0_cut_2}
\end{subfigure}
~
\begin{subfigure}[t]{0.23\textwidth}
    \centering
    \includegraphics[height=1.6in]{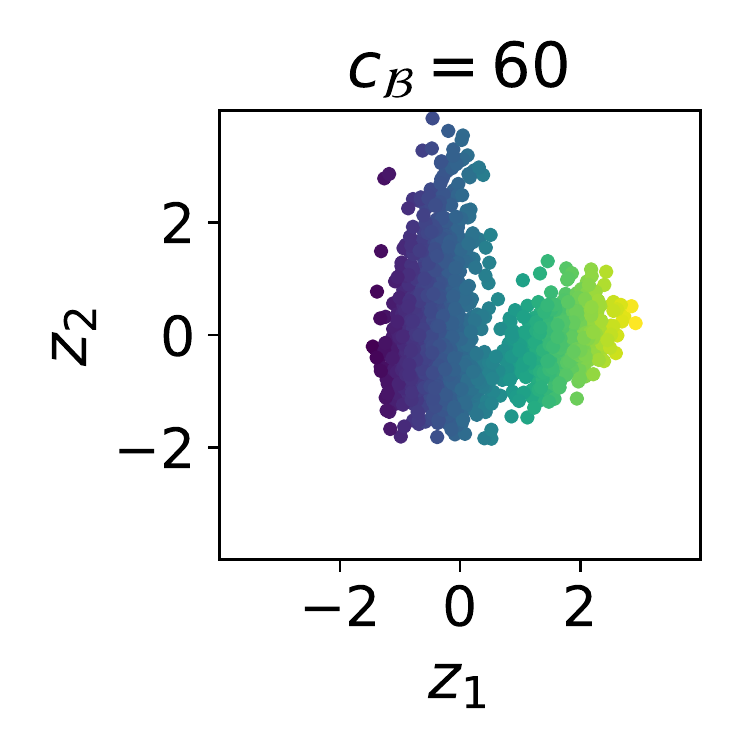}
    \caption{}
\label{fig:PCA_color_0_cut_3}
\end{subfigure}
\caption{(a) Histogram of log SVC ratios $\log \mathcal{K}_j - \log \mathcal{K}_0$ for all 200 genes in the dataset (with $m_{\B_j} = m_{\F_0} - m_{\F_j}$). Dotted lines show the value of the volume correction term in the SVC for different choices of background model complexity $c_\B$; for each choice, genes with $\log \mathcal{K}_j - \log \mathcal{K}_0$ values above the dotted line would be excluded from the foreground subspace based on the SVC.  (b) Posterior mean of the first two latent variables ($z_1$ and $z_2$), with the pPCA model applied to the genes selected with a background model complexity of $c_\B = 10$ (keeping 23 genes in the foreground). (c-e) Same as (b), but with $c_\B = 20$ (keeping 38 genes), $c_\B = 40$ (keeping 87 genes) and $c_\B = 60$ (keeping all 200 genes). In (a)-(d), the points are colored using the $z_1$ value when $c_\B = 60$.
}
\label{fig:ppca_robust}
\end{figure}

\section{Application: Glass model of gene regulation} \label{sec:on_off}
 
A central goal in the study of gene expression is to discover how individual genes regulate one another other's expression. Early studies of single cell gene expression noted the prevalence of genes that were bistable in their expression level~\citep{Shalek2013-gb,Singer2014-lg}. This suggests a simple physical analogy: if individual gene expression is a two-state system, we might study gene regulation with the theory of interacting two-state systems, namely spin glasses. 
We can consider for instance a standard model of this type in which each cell $i$ is described by a vector of spins $z_i = (z_{i 1},\ldots,z_{i d})^\top$ drawn from an Ising model, specifying whether each gene $j \in \{1, \ldots, d\}$ is ``on'' or ``off''.
In reality, gene expression lies on a continuum, so we use a continuous relaxation of the Ising model and parameterize each spin using a logistic function, setting $z_{i j 1}(x_{i j}, \mu, \tau) = 1/(1+\exp(-\tau (x_{i j} - \mu)))$ and $z_{i j 2}(x_{i j}, \mu, \tau) = 1 - z_{i j 1}(x_{i j}, \mu, \tau)$. Here, $x_{i j}$ is the observed expression level of gene $j$ in cell $i$, 
the unknown parameter $\mu$ controls the threshold for whether the expression of a gene is ``on'' (such that $z_{i j} \approx (1, 0)^\top$) or ``off'' (such that $z_{i j} \approx (0, 1)^\top$), and the unknown parameter $\tau > 0$ controls the sharpness of the threshold. 
The complete model is then given by
\begin{equation*}
X^{(i)} \sim p(x_i|H, J, \mu, \tau) := \frac{1}{\mathcal{Z}_{H,J,\mu,\tau}} \exp\!\Big(\sum_j H_j^\top z_{i j}(x_{i j}, \tau, \mu) + \sum_{j'>j} z^\top_{i j}(x_{i j}, \tau, \mu) J_{j j'} z_{i j'}(x_{i j'}, \tau, \mu)\Big)
\end{equation*}
where $\mathcal{Z}_{H,J,\mu,\tau}$ is the unknown normalizing constant of the model, and the vectors $H_j \in \mathbb{R}^2$ and matrices $J_{jj'} \in \mathbb{R}^{2 \times 2}$ are unknown parameters.
This model is motivated by experimental observations and is closely related to RNAseq analysis methods that have been successfully applied in the past \citep{Friedman2000-fd,Friedman2004-vp,Ding2005-bw,Chen2015-ig,Banerjee2008-nh,Duvenaud2008-ir,Liu2009-kx,Huynh-Thu2010-ee,Moignard2015-lc,Matsumoto2017-xi}. 
However, from a biological perspective we can expect that serious problems may occur when applying the model naively to an scRNAseq dataset.
Genes need not exhibit bistable expression: it is straightforward in theory to write down models of gene regulation that do not have just one or two steady states---gene expression may fall on a continuum, or oscillate, or have three stable states---and many alternative patterns have been well-documented empirically~\citep{Alon2019-ni}.
Interactions between genes may also be more complex than the model assumes, involving for instance three-way dependencies between genes.
All of these biological concerns can potentially produce severe violations of the proposed two-state glass model's assumptions.
Data selection provides a method for discovering where the proposed model applies. 

Applying standard Bayesian inference to the glass model is intractable, since the normalizing constant is unknown (it is an energy-based model). However, the normalizing constant does not affect the SVC, so we can still perform data selection.
We used a variational approximation to the SVC (Section~\ref{sec:vi}).  
We placed a Gaussian prior on $H$ and a Laplace prior on each entry of $J$ to encourage sparsity in the pairwise gene interactions; we also used Gaussian priors for $\mu$ and $\tau$ after applying an appropriate transform to remove constraints (Section~\ref{sec:si_glass_inference}).
Following the logic of stochastic variational inference, we optimized the variational approximation using minibatches of the data and a reparameterization gradient estimator~\citep{Hoffman2013-cq,Kingma2014-sp,Kucukelbir2017-ir}. 
We also simultaneously stochastically optimized the set of genes included in the foreground subspace, using the Leave-One-Out REINFORCE estimator~\citep{Kool2019-hl,Dimitriev2021-to}.
We implemented the model and inference strategy within the probabilistic programming language Pyro by defining a new distribution with log probability given by the negative NKSD~\citep{Bingham2019-aa}. 
Pyro provides automated, GPU-accelerated stochastic variational inference, requiring less than an hour for inference on datasets with thousands of cells.  

We examined three scRNAseq datasets, taken from (i) peripheral blood monocytes (PBMCs) from a healthy donor (2,428 cells), (ii) a MALT lymphoma (7,570 cells), and (iii) mouse neurons (10,658 cells) (Section~\ref{sec:si_glass_datasets}).
We preprocessed the data following standard protocols and focused on 200 high expression, high variability genes in each dataset, based on the metric of~\citet{Gigante2020-sc}.
We set $T=0.05$ as in Section~\ref{sec:ppca}, and used the Pitman-Yor expression for $m_\B$ (Equation~\ref{eqn:pymm}) with $\alpha=0.5$, $\theta=1$ and $D = 100$. 
This value of $D$ ensures that the number of background model parameters per data dimension is larger than the number of foreground model parameters per data dimension except for at very small $N$; in particular, there are 798 foreground model parameter dimensions associated with each data dimension (from the 199 interactions $J_{jj'}$ that each gene has with each other gene, plus the contribution of $H_j$), and $m_\B > 798\, r_\B$ for $N \ge 13$.
Our data selection procedure selects 65 genes (32.5\%) in the PBMC dataset, 0 genes in the neuron dataset, and 187 genes (93.5\%) in the MALT dataset; note that for a lower value of $m_\B$, in particular using $D = 10$, no genes are selected in the MALT dataset.
These results suggest substantial partial misspecification in the PBMC and neuron datasets, and more moderate partial misspecification in the MALT dataset.

\begin{figure}[t!]
    \centering
    \includegraphics[height=6.2in]{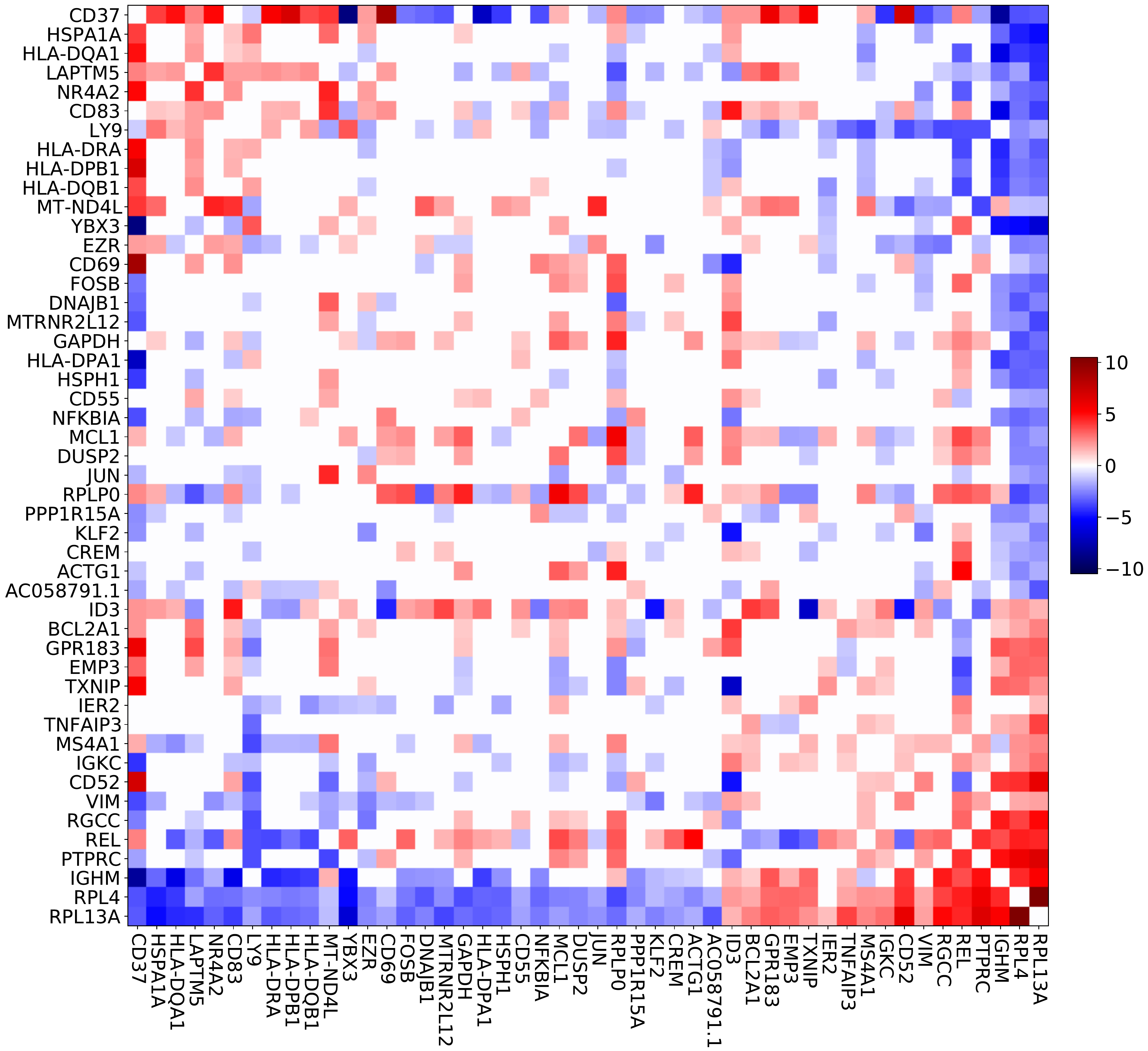}
    \caption{Posterior mean interaction energies $\Delta E_{j j'} := J_{j j' 2 1} + J_{j j' 1 2} - J_{j j' 2 2} - J_{j j' 1 1}$ for a subset of the selected genes. For visualization purposes, weak interactions ($|\Delta E_{j j'}| \le 1$) are set to zero, and genes with less than 10 total strong connections are not shown.
    Genes are sorted based on their (signed) projection onto the top principal component of the $\Delta E$ matrix.}
    \label{fig:interaction_map_selected}
\end{figure}

We investigated the biological information captured by the foreground model on the MALT dataset.
In particular, we looked at the approximate NKSD posterior for the selected 187 genes, and compared it to the approximate NKSD posterior for the model when applied to all 200 genes.
(Note that, since the glass model lacks a tractable normalizing constant, we cannot compare standard Bayesian posteriors.)
Figure~\ref{fig:interaction_map_selected} shows, for a subset of selected genes, the posterior mean of the interaction energy $\Delta E_{j j'} := J_{j j' 2 1} + J_{j j' 1 2} - J_{j j' 2 2} - J_{j j' 1 1}$, that is, the total difference in energy between two genes being in the same state versus in opposite states.
We focused on strong interactions with $|\Delta E_{j j'}| > 1$, corresponding to just 5\% of all possible gene-gene interactions (Figure~\ref{fig:DeltaEjjp_ranks_selected}).

One foreground gene with especially large loading onto the top principal component of the $\Delta E$ matrix is CD37 (Figure \ref{fig:interaction_map_selected}).
In B-cell lymphomas, of which MALT lymphoma is an example, CD37 loss is known to be associated with decreased patient survival~\citep{Xu-Monette2016-so}. 
Further, previous studies have observed that CD37 loss leads to high NF-$\kappa$B pathway activation~\citep{Xu-Monette2016-so}. Consistent with this observation, the estimated interaction energies in our model suggest that decreasing CD37 will lead to higher expression of REL, an NF-$\kappa$B transcription factor ($\Delta E_\textsc{CD37,REL} = 2.5$), decreased expression of NKFBIA, an NF-$\kappa$B inhibitor ($\Delta E_\textsc{CD37,NKFBIA} = -3.6$), and higher expression of BCL2A1, a downstream target of the NF-$\kappa$B pathway ($\Delta E_\textsc{CD37,BCL2A1} = 2.1$).
 Separately, a knockout study of Cd37 in B-cell lymphoma in mice does not show IgM expression~\citep{De_Winde2016-pk}, consistent with our model ($\Delta E_\textsc{CD37,IGHM} = -8.2$). The same study does show MHC-II expression, and our model predicts the same result, for HLA-DQ in particular ($\Delta E_\textsc{CD37,HLA-DQA1} = 5.0$, $\Delta E_\textsc{CD37,HLA-DQB1} = 3.7$).
These results suggest that the data selection procedure can successfully find systems of interacting genes that can plausibly be modeled as a spin glass, and which, in this case, are relevant for cancer.
 
\begin{figure}[t!]
    \centering
    \includegraphics[height=3in]{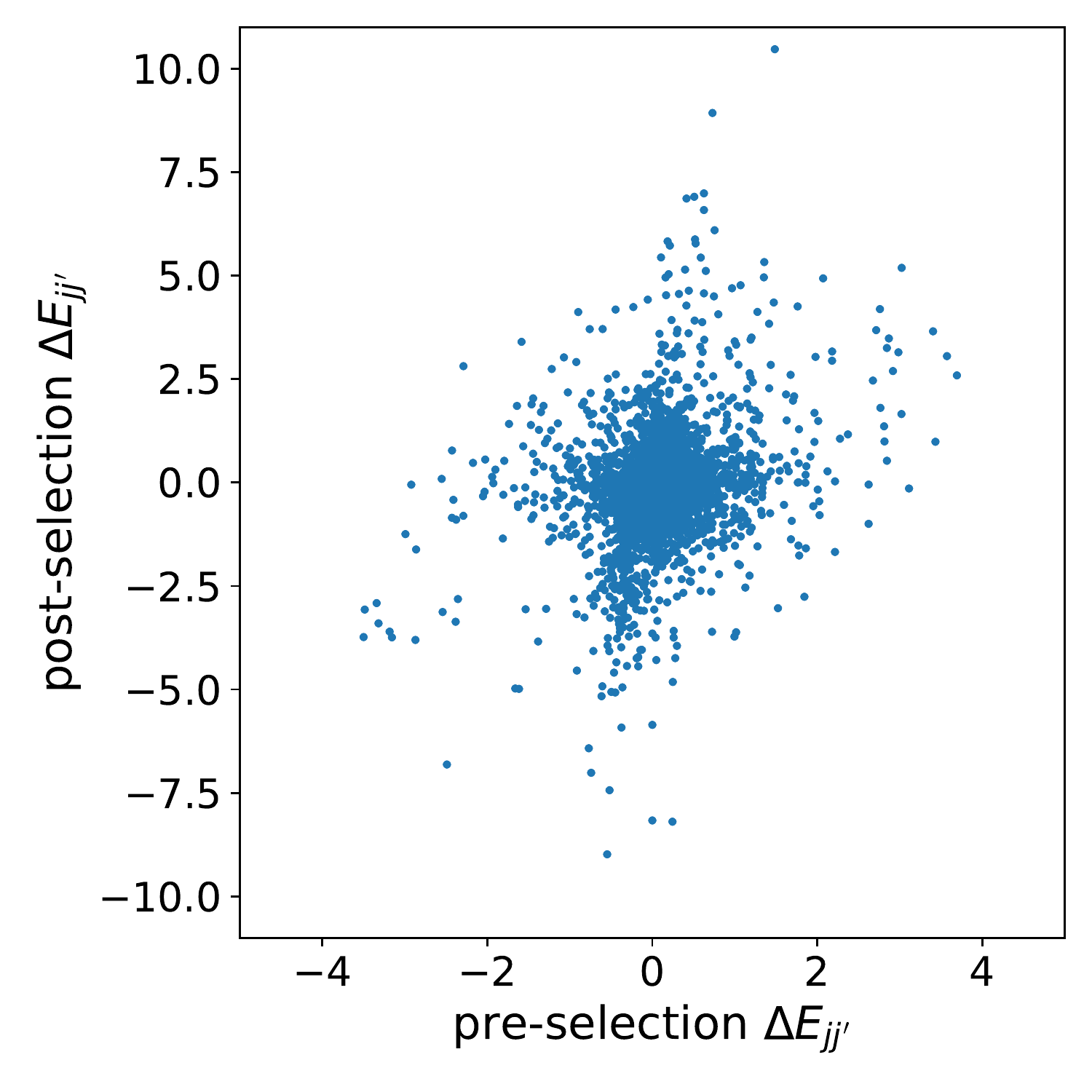}
    \caption{
    Comparison of posterior mean interaction energies $\Delta E_{j j'}$ for a model applied to all 200 genes (pre-data selection) to those learned from a model applied to the selected foreground subspace (post-data selection). 
    Each point corresponds to a pairwise interaction between two of the selected 187 genes.}
    \label{fig:interaction_correlation}
\end{figure}

To investigate whether data selection provided a benefit in this analysis, we compare with the results obtained by applying the foreground model to the full dataset of all 200 genes.
All but one of the interactions listed above have $|\Delta E| < 1$ in the full foreground model, and three have opposite signs ($\Delta E_\textsc{CD37,NFKBIA} = +0.7$, $\Delta E_\textsc{CD37,IGHM} = +0.0$, $\Delta E_\textsc{CD37,HLA-DQB1} = -0.6$); see Figure~\ref{fig:interaction_map_all}. 
%Indeed, the total number of strong interactions with CD37 is 4 in the full model, as compared to 74 in the selected foreground model.
Across all 187 selected genes, we find only a moderate correlation between the interaction energies estimated when using the full foreground model compared with the data selection-based model (Spearman's rho $= 0.30$, $p < 0.01$; Figure~\ref{fig:interaction_correlation}).
These results show that using data selection can lead to substantially different, and arguably more biologically plausible, downstream conclusions as compared to naive application of the foreground model to the full dataset.

As a simple alternative, one might wonder whether genes that are poorly fit by the model could be identified simply by looking their posterior uncertainty under the full foreground model.
This simple approach does not work well, however, since it is possible for parameters to have low uncertainty even when the model poorly describes the data.
Indeed, we found that examining uncertainty in the glass model does not lead to the same conclusions as performing data selection: the genes excluded by our data selection procedure are not the ones with the highest uncertainty in their interactions (as measured by the mean posterior standard deviation of $\Delta E_{j j'}$ under the NKSD posterior), though they do have above average uncertainty (Figure~\ref{fig:interaction_uncertainty}).
Instead, the genes excluded by our data selection procedure are the ones with the highest fraction of cells with zero expression, violating the assumptions of the foreground model (Figure~\ref{fig:fract_zero_selected}).
These results show how data selection provides a sound, computationally tractable approach to criticizing and evaluating complex Bayesian models.

\section{Discussion} \label{sec:discussion} 

Statistical modeling is often described as an iterative process, where we design models, infer hidden parameters, critique model performance, and then use what we have learned from the critique to design new models and repeat the process~\citep{Gelman2013-al}. This process has been called ``Box's loop''~\citep{Blei2014-vn}. From one perspective, data selection offers a new criticism approach. It goes beyond posterior predictive checks and related methods by changing the model itself, replacing potentially misspecified components with a flexible background model. This has important practical consequences: since misspecification can distort estimates of model parameters in unpredictable ways, predictive checks are likely to indicate mismatch between the model and the data across the entire space $\X$ even when the proposed parametric model is only partially misspecified. Our method, by contrast, reveals precisely those subspaces of $\X$ where model-data mismatch occurs.

From another perspective, data selection is outside the design-infer-critique loop. An underlying assumption of Box's loop is that scientists want to model the entire dataset. As datasets get larger, and measurements get more extensive, this desire has led to more and more complex (and often difficult to interpret) models. In experimental science, however, scientists have often followed the opposite trajectory: faced with a complicated natural phenomenon, they attempt to isolate a simpler example of the phenomenon for close study. 
Data selection offers one approach to formalizing this intuitive idea in the context of statistical analysis: we can propose a simple parametric model and then isolate a piece of the whole dataset---a subspace $\XXF$---to which this model applies. When working with large, complicated datasets, this provides a method of searching for simpler phenomena that are hypothesized to exist.

\section{Acknowledgments}

The authors wish to thank Jonathan Huggins, Pierre Jacob, Andre Nguyen, and Elizabeth Wood for helpful discussions and suggestions.
We would like to thank Debora S.\ Marks in particular for suggesting the use of a Potts model in RNAseq analysis.
E.N.W.\ is supported by the Fannie and John Hertz Fellowship.
J.W.M.\ is supported by the National Institutes of Health grant 5R01CA240299-02.

% \spacingset{1}
\bibliographystyle{abbrvnat}
\small
\bibliography{references}
\normalsize
% \spacingset{1.5}

\newpage 
\beginsupplement

\bigskip
\bigskip
\begin{center}
\textbf{\LARGE Supplementary material for\\``Bayesian data selection''}
\end{center}
\medskip

\section{Methods details}

\subsection{Calibrating $T$} \label{sec:calibrate_t}

The SVC contains a hyperparameter $T > 0$.
To choose an appropriate value of $T$, we aim, roughly, to match the coverage of the generalized posterior
\begin{equation*}
	\pi_N^{\textsc{svc}}(\theta)d\theta = \frac{1}{z_N}\exp\!\Big(-\frac{N}{T} \widehat{\textsc{nksd}}(p_0(x))\| q(x|\theta))\Big) \pi(\theta)d\theta
\end{equation*}
to the coverage of the standard Bayesian posterior
\begin{equation*}
	\pi_N^{\textsc{kl}}(\theta)d\theta = \frac{1}{q(X^{(1:N)})} \exp\!\Big(\sum_{i=1}^N \log q(X^{(i)}|\theta)\Big) \pi(\theta)d\theta
\end{equation*}
when the model is well-specified.
 
Let $\theta_{*}$ be the true parameter value, such that $p_0(x) = q(x|\theta_{*})$ almost everywhere. Let $G^{\textsc{kl}}(\theta) := \nabla^2_\theta \mathbb{E}_{X \sim p_0}[-\log q(X|\theta)]$ and let $\theta_N^{\textsc{kl}} := \argmax \sum_{i=1}^N \log q(X^{(i)}|\theta)$ be the maximum likelihood estimator. Let $h_N^{\textsc{kl}}$ be the density of $\sqrt{N}(\theta - \theta_N^{\textsc{kl}})$ when $\theta \sim \pi_N^{\textsc{kl}}$. Under regularity conditions \citep{Miller2019-ur}, according to the Bernstein--von Mises theorem, $h_N^{\textsc{kl}}$ converges to a normal distribution in total variation,
\begin{equation*}
\int_{\mathbb{R}^m} \Big| h_N^{\textsc{kl}}(x) - \mathcal{N}\big(x \mid 0, G^{\textsc{kl}}(\theta_{*})^{-1}\big) \Big| dx \xrightarrow[N \to \infty]{\textup{a.s.}} 0.
\end{equation*}
According to Theorem~\ref{thm:marginal_ksd}, the generalized posterior associated with the SVC has analogous behavior. Let $G^{\textsc{svc}}(\theta) := \nabla^2_\theta \frac{1}{T}\textsc{nksd}(p_0(x)\| q(x|\theta))$ and let $\theta_N^{\textsc{svc}} := \argmin \widehat{\textsc{nksd}}(p_0(x)\| q(x|\theta))$.
Let $h_N^{\textsc{svc}}$ be the density of $\sqrt{N}(\theta - \theta_N^{\textsc{svc}})$ when $\theta \sim \pi_N^{\textsc{svc}}$. Then by Theorem~\ref{thm:marginal_ksd}, $h_N^{\textsc{svc}}$ converges to a normal distribution in total variation,
\begin{equation*}
\int_{\mathbb{R}^m} \Big| h_N^{\textsc{svc}}(x) - \mathcal{N}\big(x \mid 0, G^{\textsc{svc}}(\theta_{*})^{-1}\big) \Big| dx \xrightarrow[N \to \infty]{\textup{a.s.}} 0.
\end{equation*}
For the uncertainty in each posterior to be roughly the same order of magnitude, we want
\begin{equation*}
	\det G^{\textsc{kl}}(\theta_{*}) \approx \det G^{\textsc{svc}}(\theta_{*}),
\end{equation*}
or equivalently,
\begin{equation*}
	T \approx \left(\frac{ \det \big[\nabla^2_\theta\big\vert_{\theta=\theta_{*}} \textsc{nksd}(p_0(x)\| q(x|\theta))\big]}{\det \big[\nabla^2_\theta\big\vert_{\theta=\theta_{*}} \mathbb{E}_{X \sim p_0}[-\log q(X|\theta)]\big]}\right)^{1/m}.
\end{equation*}

To choose a single $T$ value, we simulate true parameters from the prior, generate data from each simulated true parameter, and take the median of the estimated $T$ values.
That is, we use the median $\hat{T}$ across samples drawn as
\begin{equation}
\begin{split}
	\theta_{*} &\sim \pi(\theta)\\
	X^{(i)} &\overset{\mathrm{iid}}{\sim} q(x|\theta_{*})\\
	\hat{T} & = \left(\frac{ |\det \big[\nabla^2_\theta\big\vert_{\theta=\theta_{*}} \widehat{\textsc{nksd}}(p_0(x)\| q(x|\theta))\big]|}{|\det \big[\nabla^2_\theta\big\vert_{\theta=\theta_{*}} \frac{1}{N} \sum_{i=1}^N -\log q(X^{(i)}|\theta)\big]|}\right)^{1/m}.
\end{split}
\label{eqn:T_estimator}
\end{equation}
In practice, we find that the order of magnitude of $\hat{T}$ is stable across samples $\theta_{*}$ from $\pi(\theta)$. See Section~\ref{sec:si_ppca_calibration} for an example.

\subsection{Kernel recommendations} \label{sec:si_kernel_choice}

To obtain subsystem independence (Proposition~\ref{prop:subsystem_indep}), we suggest using a kernel that factors across subspaces, such that $k(X, Y) = k_\F(X_\F, Y_\F) k_\B(X_\B, Y_\B)$ where $k_\F$ and $k_\B$ are integrally strictly positive definite kernels.
In the applications in Sections~\ref{sec:ppca} and \ref{sec:on_off}, we use the following kernel.
\begin{definition}
The \emph{factored inverse multiquadric (IMQ) kernel} is defined as
	\begin{equation*}
	k(x, y) = \prod_{i=1}^d \big(c^2 + (x_i - y_i)^2\big)^{\beta/d}
    \end{equation*}
for $x,y \in \mathbb{R}^d$, where $\beta \in [-1/2, 0)$ and $c > 0$.
\end{definition}
\noindent Note that this kernel factors across any subset of dimensions, that is, if $S \subseteq \{1, \ldots, d\}$ and $S^c = \{1, \ldots, d\} \setminus S$, then we can write $k(x, y) = k_S(x_S, y_S) k_{S^c}(x_{S^c}, y_{S^c})$. 
Thus, if the foreground subspace $\X_\F$ is the span of a subset of the standard basis, such that $x_\F = V^\top x = x_S$ for some $S \subseteq \{1, \ldots, d\}$, then it follows that $k$ factors as $k(x,y) = k_\F(x_\F, y_\F) k_\B(x_\B, y_\B)$. 
Along with this observation,
the next result shows that the factored IMQ satisfies the conditions of Theorem~\ref{thm:marginal_ksd} that pertain to $k$ alone.
\begin{proposition}
The factored IMQ kernel is symmetric, positive, bounded, integrally strictly positive definite, and has continuous and bounded partial derivatives up to order~2.
\end{proposition}
\begin{proof}
It is clear that $k(x,y) = k(y,x)$ and $k(x,y) > 0$. 
Next, we show that $k$ has continuous and bounded partial derivatives up to order 2. Note that we can write $k(x,y) = \prod_{i=1}^d \psi(x_i - y_i)$ where $\psi(r) = (c^2 + r^2)^{\beta/d}$ for $r \in \mathbb{R}$.
Differentiating, we have 
\begin{align*}
\psi'(r) &= \frac{\beta}{d} \frac{2 r}{c^2 + r^2} \psi(r) \\
\psi''(r) &= \Big(\frac{\beta^2}{d^2} - \frac{\beta}{d}\Big) \Big(\frac{2 r}{c^2 + r^2}\Big)^2 \psi(r) + \frac{\beta}{d}\frac{2}{c^2 + r^2} \psi(r).
\end{align*}
Since $r^2 \geq 0$ and $\beta < 0$, $|\psi(r)| \leq c^{2 \beta/d}$ for all $r \in \mathbb{R}$.
Further, it is straightforward to verify that $|\psi'(r)|$ and $|\psi''(r)|$ are bounded on $\mathbb{R}$
by using the fact that $|r|/(c^2 + r^2) \leq 1/(2 c)$.
By the chain rule, it follows that for all $i,j$, the functions $k(x,y)$, $|\partial k / \partial x_i|$, 
and $|\partial^2 k / \partial x_i \partial y_j|$ are bounded.
Thus, we conclude that $k$, $\|\nabla k\|$, and $\|\nabla^2k\|$ are bounded.

Finally, we show that $k$ is integrally strictly positive definite. First, for any $d$, for $x,y\in\mathbb{R}^d$, the function $(x,y)\mapsto (c^2 + \|x - y\|_2^2)^{\beta/d}$ is an integrally strictly positive definite kernel (see, for example, Section 3.1 of \citealp{Sriperumbudur2010-ew}); we refer to this as the standard IMQ kernel.
Since the factored IMQ is a product of one-dimensional standard IMQ kernels, it defines a kernel on $\mathbb{R}^d$ (Lemma 4.6 of~\citealp{Steinwart2008-sz}) and is positive definite (Theorem 4.16 of~\citealp{Steinwart2008-sz}).
By Bochner's theorem (Theorem 3 of~\citealp{Sriperumbudur2010-ew}), a continuous positive definite kernel can be expressed in terms of the Fourier transform of a finite nonnegative Borel measure. In particular, applying Bochner's theorem to $\psi(r)$, we have 
\begin{equation*}
\begin{split}
	k(x,y) &= \Psi(x - y) := \prod_{i=1}^d \psi(x_i - y_i) = \prod_{i=1}^d \int_{\mathbb{R}} \exp\!\big(-\sqrt{-1} (x_i - y_i) \omega_i\big) d\Lambda^0(\omega_i) \\
	&= \int_{\mathbb{R}^d} \exp\!\big(-\sqrt{-1} (x - y)^\top \omega\big) d\Lambda(\omega)
\end{split}
\end{equation*}
by Fubini's theorem, where $\Lambda^0$ is the finite nonnegative Borel measure on $\mathbb{R}$ associated with $\psi(r)$
and $\Lambda = \Lambda^0 \times \cdots \times \Lambda^0$ is the resulting product measure on $\mathbb{R}^d$.
Applying Bochner's theorem in the other direction, we see that $\Psi$ is a positive definite function. 
Moreover, since the standard IMQ kernel is characteristic (Theorem 7 of~\citealp{Sriperumbudur2010-ew}), it follows that the support of $\Lambda^0$ is $\mathbb{R}$ (Theorem 9 of~\citealp{Sriperumbudur2010-ew}), and thus the support of $\Lambda$ is $\mathbb{R}^d$.
This implies that the factored IMQ kernel $k$ is characteristic (Theorem 9 of~\citealp{Sriperumbudur2010-ew}) and, since $k$ is also translation invariant, $k$ must be integrally strictly positive definite (Section 3.4 of~\citealp{Sriperumbudur2011-lv}).
\end{proof}

Our choice of the factored IMQ kernel is motivated by the analysis of \citet{Gorham2017-sd}, which suggests that the standard IMQ is a good default choice for the kernelized Stein discrepancy, particularly when working with distributions that are (roughly speaking) very spread out. 
In particular, it is straightforward to show that the factored IMQ kernel, like the standard IMQ kernel, meets the conditions of Theorem 3.2 of~\citet{Huggins2018-lr}.
However, we do not pursue further the question of whether the \textsc{nksd} with the factored IMQ detects convergence and non-convergence since our statistical setting is different from that of \citet{Gorham2017-sd}, and we are assuming the data consists of i.i.d.\ samples from some underlying distribution rather than correlated samples from an MCMC chain which may or may not converge.

\subsection{Exact solution for exponential families} \label{sec:si_conjugacy}

Here, we show that when $q(x|\theta)$ is an exponential family, the estimated \textsc{nksd} has the form
\begin{equation}
\label{eqn:nksd_expfam}
	\widehat{\textsc{nksd}}(p_0(x)\|q(x|\theta)) = \theta^\top A\, \theta + B^\top \theta + C
\end{equation}
where $A$, $B$, and $C$ depend on the data but not on $\theta$.
Since $q_\theta(x) = q(x|\theta) = \lambda(x) \exp(\theta^\top t(x) - \kappa(\theta))$, we have $s_{q_\theta}(x) = \nabla_x \log \lambda(x) + (\nabla_x t(x))^\top \theta$ where $(\nabla_x t(x))_{i j} = \partial t_i / \partial x_j$. Thus, we can write
\begin{align}
\label{eqn:u_quadratic_form}
u_\theta (x, y) &:= s_{q_\theta}(x)^\top s_{q_\theta}(y) k(x, y) + s_{q_\theta}(x)^\top \nabla_y k(x, y) + s_{q_\theta}(y)^\top \nabla_x k(x, y) + \Tr(\nabla_x \nabla_y^\top k(x, y))\notag\\
&= \theta^\top [(\nabla_x t(x)) (\nabla_y t(y))^\top k(x, y)] \theta\notag\\
&~~~ + [(\nabla_x \log \lambda(x))^\top (\nabla_y t(y))^\top k(x, y) + (\nabla_y \log \lambda(y))^\top (\nabla_x t(x))^\top k(x, y)\notag\\
&~~~~~~ + (\nabla_x k(x, y))^\top (\nabla_y t(y))^\top + (\nabla_y k(x, y))^\top (\nabla_x t(x))^\top] \theta\notag\\
&~~~ + [(\nabla_x \log \lambda(x))^\top (\nabla_y \log \lambda(y)) k(x, y) + (\nabla_y \log \lambda(y))^\top (\nabla_x k(x, y))\notag\\
&~~~~~~ + (\nabla_x \log \lambda(x))^\top (\nabla_y k(x, y)) + \Tr(\nabla_x \nabla_y^\top k(x, y))].
\end{align}
Then the estimated \textsc{nksd} takes the form in Equation~\ref{eqn:nksd_expfam} if we choose
\begin{equation*}
\begin{split}
	A := \frac{1}{\sum_{i\neq j} k(X^{(i)}, X^{(j)})} \sum_{i\neq j} & \nabla_x t(X^{(i)}) \nabla_x t(X^{(j)})^\top k(X^{(i)}, X^{(j)})\\
B^\top := \frac{1}{\sum_{i\neq j} k(X^{(i)}, X^{(j)})} \sum_{i\neq j} \big[& (\nabla_x \log \lambda(X^{(i)}))^\top \nabla_x t(X^{(j)})^\top k(X^{(i)}, X^{(j)})\\ & + (\nabla_x \log \lambda(X^{(j)}))^\top \nabla_x t(X^{(i)})^\top k(X^{(i)}, X^{(j)})\\ & + (\nabla_x k(X^{(i)}, X^{(j)}))^\top \nabla_x t(X^{(j)})^\top + (\nabla_y k(X^{(i)}, X^{(j)}))^\top \nabla_x t(X^{(i)})^\top\big]\\
C := \frac{1}{\sum_{i\neq j} k(X^{(i)}, X^{(j)})} \sum_{i\neq j} \big[ & (\nabla_x \log \lambda(X^{(i)}))^\top (\nabla_x \log \lambda(X^{(j)})) k(X^{(i)}, X^{(j)})\\ & + (\nabla_x \log \lambda(X^{(j)}))^\top \nabla_x k(X^{(i)}, X^{(j)}) \\
&+ (\nabla_x \log \lambda(X^{(i)}))^\top \nabla_y k(X^{(i)}, X^{(j)}) + \Tr(\nabla_x \nabla_y^\top k(X^{(i)}, X^{(j)})) \big].
\end{split}
\end{equation*}
If the prior on $\theta$ is $\mathcal{N}(\mu,\Sigma_0)$, then the SVC is
\begin{equation*}
\begin{split}
	\mathcal{K} = & \left(\frac{2\pi}{N}\right)^{m_\B/2} (2\pi)^{-m_\F/2} (\det \Sigma_0)^{-1/2} \\  & \times \int \exp\!\Big(-\frac{N}{T} [\theta^\top A\, \theta + B^\top \theta + C]\Big) \exp\!\Big(-\frac{1}{2} (\theta - \mu)^\top \Sigma_0^{-1} (\theta - \mu)\Big) d\theta\\
	= & \left(\frac{2\pi}{N}\right)^{m_\B/2} (2\pi)^{-m_\F/2} (\det \Sigma_0)^{-1/2} \\  
	& \times \int \exp\!\bigg(-\frac{1}{2} \theta^\top \Big(\frac{2N}{T} A + \Sigma_0^{-1}\Big) \theta + \Big(-\frac{N}{T} B^\top + \mu^\top \Sigma_0^{-1}\Big) \theta - \frac{N}{T} C - \frac{1}{2} \mu^\top \Sigma_0^{-1} \mu \bigg) d\theta\\
	=& \left(\frac{2\pi}{N}\right)^{m_\B/2} (\det \Sigma_0)^{-1/2} \bigg(\det \Big(\frac{2 N}{T} A + \Sigma_0^{-1}\Big) \bigg)^{-1/2} \\ 
	&\times\exp\!\bigg(\frac{1}{2} \Big(-\frac{N}{T} B^\top + \mu^\top \Sigma_0^{-1}\Big)^\top \Big(\frac{2N}{T} A + \Sigma_0^{-1}\Big)^{-1} \Big(-\frac{N}{T} B^\top + \mu^\top \Sigma_0^{-1}\Big) -\frac{N}{T} C - \frac{1}{2} \mu^\top \Sigma_0^{-1} \mu \bigg).
\end{split}
\end{equation*}
Meanwhile, if $q(x|\theta) = \mathcal{N}(\theta,\Sigma)$ where $\Sigma$ is a fixed covariance matrix, then we have $\nabla_x \log \lambda(x) = -\Sigma^{-1} x$ and $\nabla_x t(x) = \Sigma^{-1}$.

\subsection{Comparing many foregrounds using approximate optima} \label{sec:si_approx_optima}

Here, we justify the technique described in Section~\ref{sec:approximate_optima}.
As in Section~\ref{sec:approximate_optima}, define
$\ell_j(\theta) = \widehat{\textsc{nksd}}(p_0(x_{\F_j})\| q(x_{\F_j}|\theta))$ for $j \in \{1, 2\}$,
and let $\theta_N(w) = \argmin_\theta \mathcal{L}(w, \theta)$ where
\begin{equation*}
	\mathcal{L}(w, \theta) := \ell_1(\theta) + w (\ell_2(\theta) - \ell_1(\theta))
\end{equation*}
for $w\in[0,1]$. 
We assume that the conditions of Theorem~\ref{thm:marginal_ksd} are met, over both $\X_{\F_1}$ and $\X_{\F_2}$.
Since $(\partial\mathcal{L}/\partial\theta_i)(w, \theta_N(w)) = 0$, we have
\begin{equation*}
	0 = \frac{\partial}{\partial w} \Big(\frac{\partial\mathcal{L}}{\partial \theta_i}(w, \theta_N(w))\Big)
	= \frac{\partial^2\mathcal{L}}{\partial w \partial \theta_i}(w,\theta_N(w)) + \sum_{j} \frac{\partial^2\mathcal{L}}{\partial\theta_i \partial\theta_j}(w,\theta_N(w)) \Big(\frac{\partial}{\partial w}\theta_{N,j}(w)\Big),
\end{equation*}
or equivalently, in matrix/vector notation,
\begin{equation*}
	0 = \nabla_w(\nabla_\theta \mathcal{L}(w, \theta_N(w)))
	= \nabla_\theta \nabla_w \mathcal{L}(w, \theta_N) + \nabla_\theta^2 \mathcal{L}(w, \theta_N) \nabla_w(\theta_N(w)).
\end{equation*}
Rearranging, we have
\begin{equation*}
	\nabla_w \theta_N(w) = -\big(\nabla_\theta^2 \mathcal{L}(w, \theta_N)\big)^{-1} \nabla_\theta \nabla_w \mathcal{L}(w, \theta_N).
\end{equation*}
At $w = 0$ we find, plugging back in the definition of $\mathcal{L}$,
\begin{equation*}
\begin{split}
	\nabla_w \theta_N(0) &= -\nabla_\theta^2 \ell_1(\theta_N(0))^{-1} (\nabla_\theta \ell_2(\theta_N(0)) - \nabla_\theta\ell_1(\theta_N(0)))\\
	&= -\nabla_\theta^2 \ell_1(\theta_N(0))^{-1} \nabla_\theta \ell_2(\theta_N(0)).
\end{split}
\end{equation*}
Applying a first-order Taylor series expansion gives us $\theta_N(1) \approx \theta_N(0) + \nabla_w \theta_N(0)$, which yields Equation~\ref{eqn:IJ_grad}.

\section{Asymptotics of the alternative selection criteria} \label{sec:asymptotics}

Theorem~\ref{thm:selection_consistency} shows that the SVC exhibits all four types of consistency: data selection, nested data selection, model selection, and nested model selection.  In this section, we establish the consistency properties of the alternative criteria considered in Section~\ref{sec:method_asymptotics}.

\begin{figure}[t!]
    \centering
    \begin{subfigure}[t!]{0.48\textwidth}
        \centering
        \includegraphics[height=2.5in]{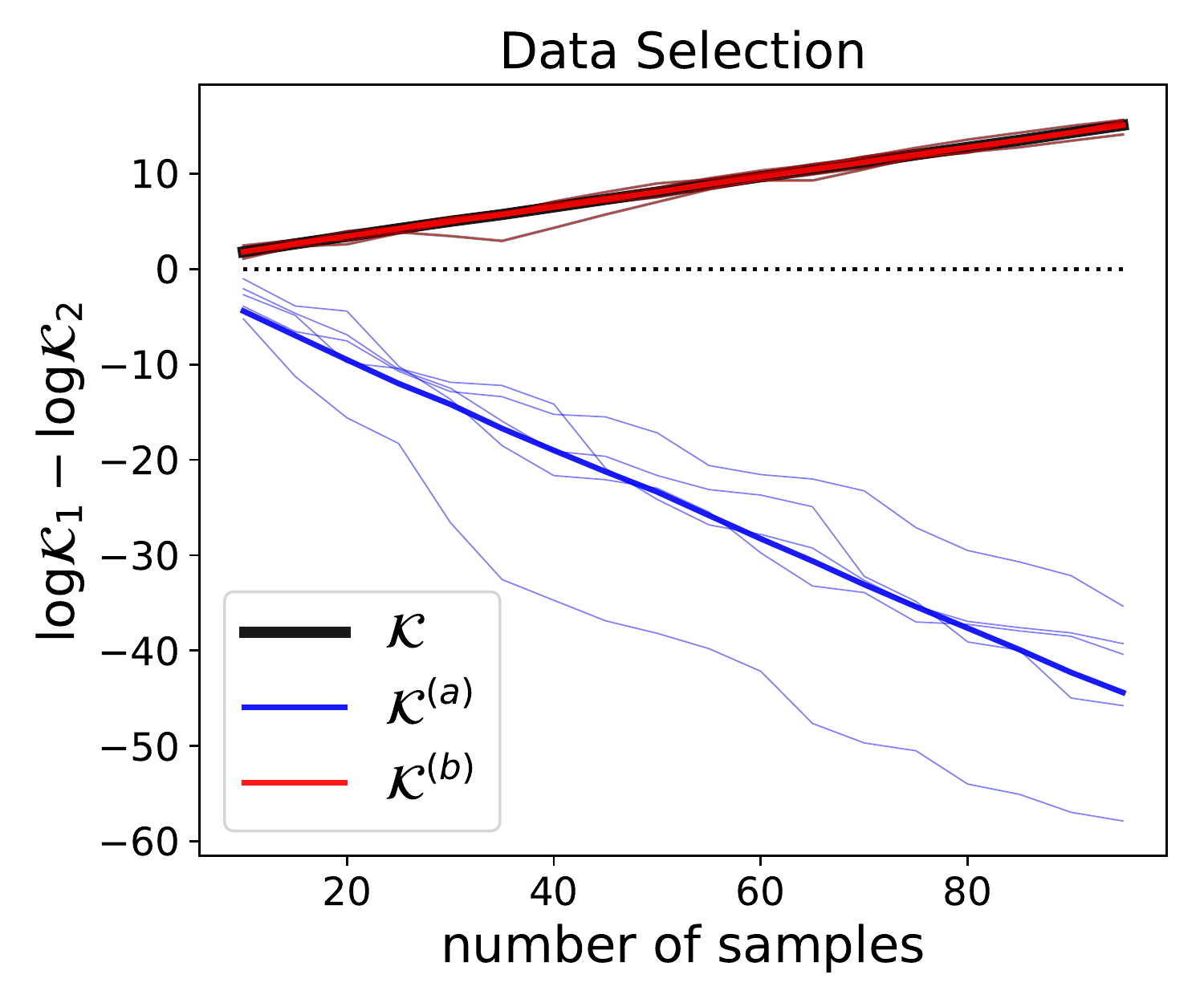}
        \caption{}
        \label{fig:py_data_select}
    \end{subfigure}%
    ~ 
    \begin{subfigure}[t!]{0.48\textwidth}
        \centering
        \includegraphics[height=2.5in]{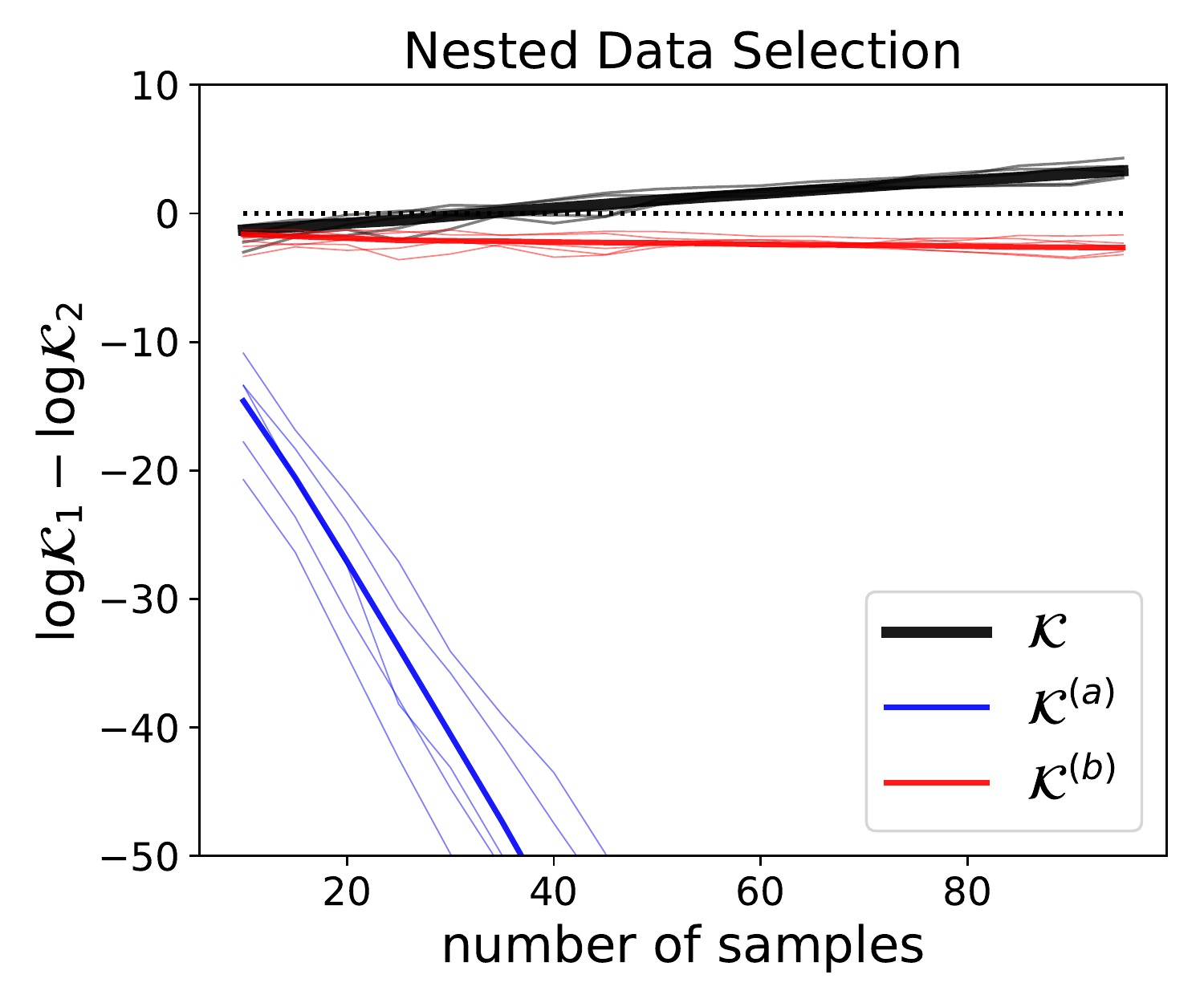}
        \caption{}
        \label{fig:py_nest_data_select}
    \end{subfigure}
    \\
    \begin{subfigure}[t!]{0.48\textwidth}
        \centering
        \includegraphics[height=2.5in]{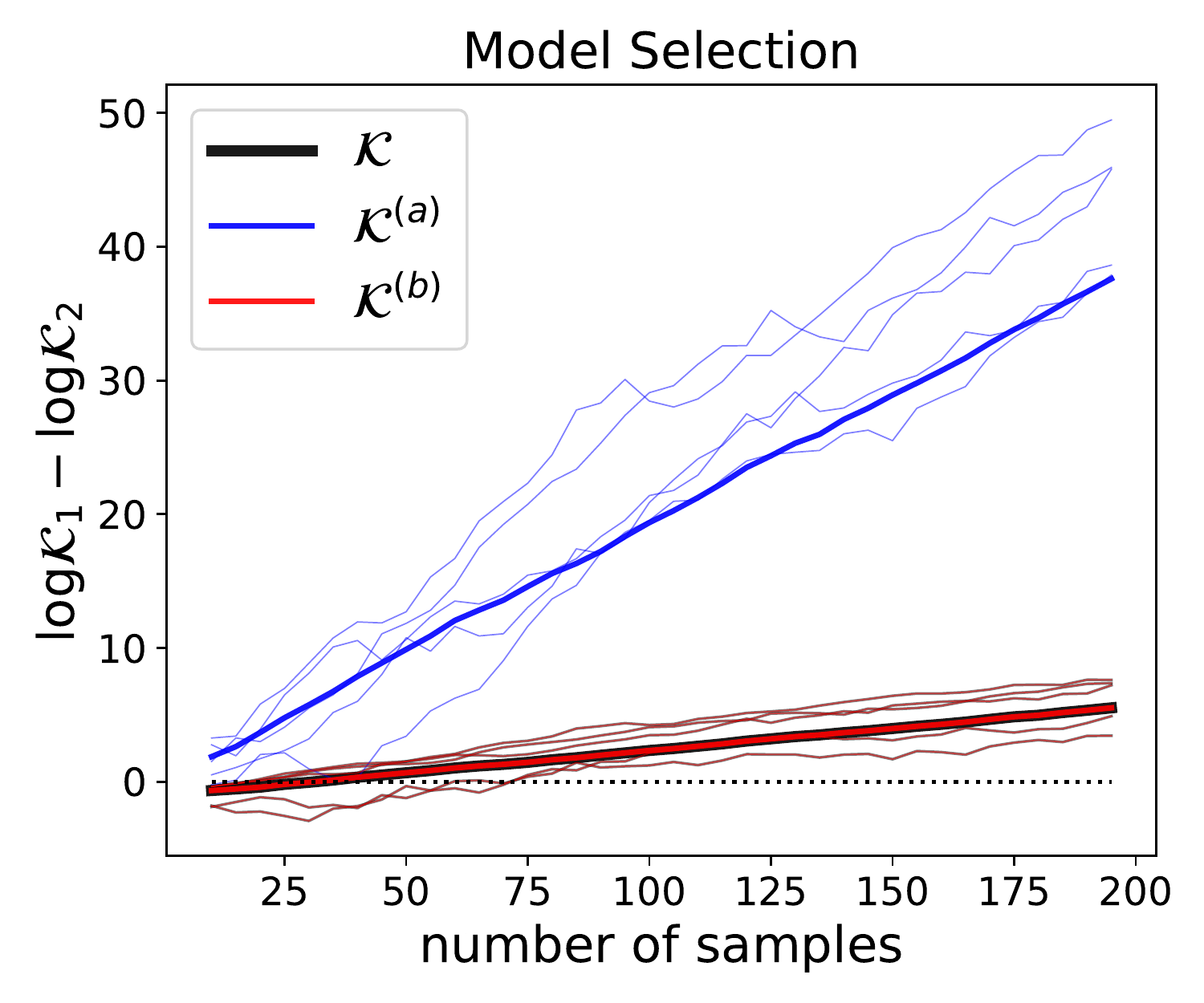}
        \caption{}
        \label{fig:py_model_select}
    \end{subfigure}
    ~ 
    \begin{subfigure}[t!]{0.48\textwidth}
        \centering
        \includegraphics[height=2.5in]{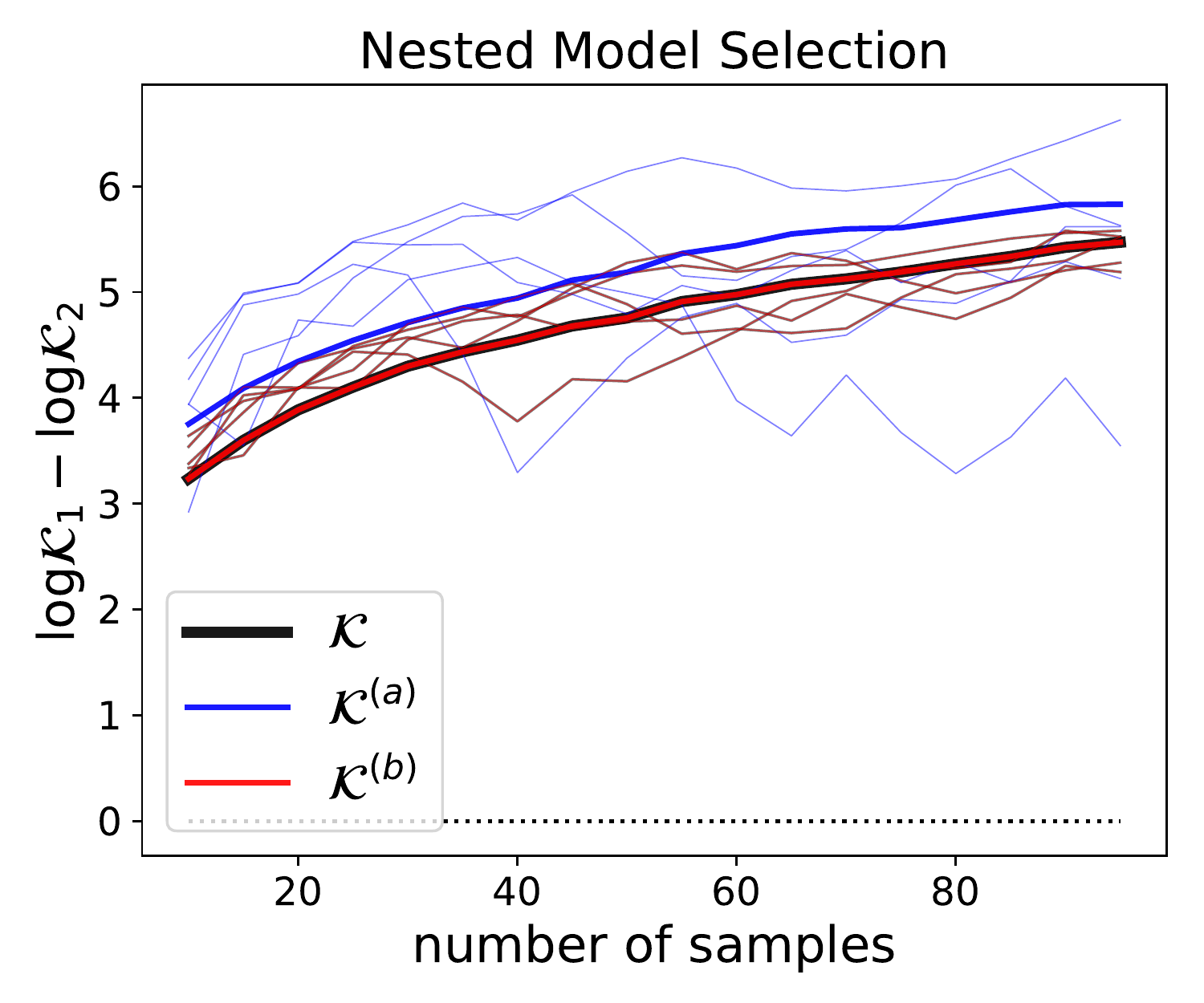}
        \caption{}
        \label{fig:py_nest_model_select}
    \end{subfigure}
    \caption{Behavior of the Stein volume criterion $\mathcal{K}$, the foreground marginal likelihood with a background volume correction $\mathcal{K}^{(\mathrm{a})}$, and the foreground marginal \textsc{nksd} $\mathcal{K}^{(\mathrm{b})}$ on toy examples. The plots show the results for 5 randomly generated datasets (thin lines) and the average over 100 random datasets (bold lines). Here, unlike Figure~\ref{fig:toy}, the Pitman-Yor expression for $m_\B$ is used, with $\alpha=0.5$, $\theta=1$, and $D=0.2$.}
    \label{fig:toy_py}
\end{figure}

\subsection{Setup} \label{sec:asymptotics_setup}

We first review the asymptotics of the standard marginal likelihood, discussed in depth by \citet{Dawid2011-kb} and \citet{Hong2005-kf}, for example. 
Define
\begin{align*}
f_N^{\textsc{kl}}(\theta) &:= -\frac{1}{N}\sum_{i=1}^N \log q(X^{(i)}|\theta), \qquad & \theta_N^{\textsc{kl}} &:= \argmin_\theta f_N^{\textsc{kl}}(\theta),\\
f^{\textsc{kl}}(\theta) &:= -\mathbb{E}_{X \sim p_0} [\log q(X|\theta)],  \qquad & \theta_{*}^{\textsc{kl}} &:= \argmin_\theta f^{\textsc{kl}}(\theta).
\end{align*}
Let $m$ be the dimension of the parameter space.
Under suitable regularity conditions \citep{Miller2019-ur}, the Laplace approximation to the marginal likelihood is
\begin{equation}
\label{eqn:standard_marginal_laplace}
	q(X^{(1:N)}) = \int q(X^{(1:N)}|\theta)\pi(\theta)d\theta \sim \frac{\exp\big(-N f_N^{\textsc{kl}}(\theta_N^{\textsc{kl}})\big) \pi(\theta_{*}^{\textsc{kl}})}{\big|\det \nabla_\theta^2\, f^{\textsc{kl}}(\theta_{*}^{\textsc{kl}})\big|^{1/2}} \left(\frac{2\pi}{N}\right)^{m/2} 
\end{equation}
almost surely, where $a_N \sim b_N$ indicates that $a_N/b_N \to 1$ as $N\to\infty$.
We can rewrite this as
\begin{equation}
\begin{split}
	\log q(X^{(1:N)}) &+  N(f_N^{\textsc{kl}}(\theta_N^{\textsc{kl}}) - f_N^{\textsc{kl}}(\theta_{*}^{\textsc{kl}}))\\ & + N (f_N^{\textsc{kl}}(\theta_{*}^{\textsc{kl}}) - f^{\textsc{kl}}(\theta_{*}^{\textsc{kl}})) + N f^{\textsc{kl}}(\theta_{*}^{\textsc{kl}})\\ &
	+ \frac{m}{2} \log N - \log\left(\frac{\pi(\theta_{*}^{\textsc{kl}}) (2\pi)^{m/2}}{\big|\det \nabla_\theta^2 f^{\textsc{kl}}(\theta_{*}^{\textsc{kl}})\big|^{1/2}} \right) \xrightarrow[N \to \infty]{\textup{a.s.}} 0.
\end{split}
\label{eqn:standard_marginal_asymptotic_decomposition}
\end{equation}
As shown by \citet{Dawid2011-kb} and \citet{Hong2005-kf}, under regularity conditions,
\begin{equation}
\begin{split}
N(f_N^{\textsc{kl}}(\theta_N^{\textsc{kl}}) - f_N^{\textsc{kl}}(\theta_{*}^{\textsc{kl}})) &= O_{P_0}(1)\\
N (f_N^{\textsc{kl}}(\theta_{*}^{\textsc{kl}}) - f^{\textsc{kl}}(\theta_{*}^{\textsc{kl}})) &= O_{P_0}(\sqrt{N})\\
N f^{\textsc{kl}}(\theta_{*}^{\textsc{kl}}) &= O_{P_0}(N)\\
\log\left(\frac{\pi(\theta_{*}^{\textsc{kl}}) (2\pi)^{m/2}}{\big|\det \nabla_\theta^2 f^{\textsc{kl}}(\theta_{*}^{\textsc{kl}})\big|^{1/2}} \right) &= O_{P_0}(1).
\end{split}
\label{eqn:standard_marginal_asymptotic_scaling}
\end{equation}
The \textsc{nksd} marginal likelihood has a similar decomposition. Following Section~\ref{sec:theory}, define
\begin{align*}
	f_N^{\textsc{nksd}}(\theta) &:= \frac{1}{T} \widehat{\textsc{nksd}}(p_0(x)\|q(x|\theta)), \qquad & \theta_N^{\textsc{nksd}} &:= \argmin_\theta f^{\textsc{nksd}}_N(\theta),\\
	f^{\textsc{nksd}}(\theta) &:= \frac{1}{T} \textsc{nksd}(p_0(x)\|q(x|\theta)), \qquad & \theta_{*}^{\textsc{nksd}} &:= \argmin_\theta f^{\textsc{nksd}}(\theta).
\end{align*}
As shown in Theorem~\ref{thm:marginal_ksd},
\begin{equation*}
	z_N :=  \int \exp(-N f_N^{\textsc{nksd}}(\theta))\pi(\theta)d\theta \sim  \frac{\exp(-N f_N^{\textsc{nksd}}(\theta_N^{\textsc{nksd}})) \pi(\theta_{*}^{\textsc{nksd}})}{\big|\det \nabla_\theta^2 f^{\textsc{nksd}}(\theta_{*}^{\textsc{nksd}}) \big|^{1/2}} \left(\frac{2\pi}{N}\right)^{m/2}
\end{equation*}
almost surely as $N\to\infty$. As above, we can rewrite this as
\begin{equation}
\begin{split}
	\log z_N &+ N(f_N^{\textsc{nksd}}(\theta_N^{\textsc{nksd}}) - f_N^{\textsc{nksd}}(\theta_{*}^{\textsc{nksd}}))\\ & + N (f_N^{\textsc{nksd}}(\theta_{*}^{\textsc{nksd}}) - f^{\textsc{nksd}}(\theta_{*}^{\textsc{nksd}})) + N f^{\textsc{nksd}}(\theta_{*}^{\textsc{nksd}})\\ &
	+ \frac{m}{2} \log N - \log\left(\frac{\pi(\theta_{*}^{\textsc{nksd}}) (2\pi)^{m/2}}{\big|\det \nabla_\theta^2 f^{\textsc{nksd}}(\theta_{*}^{\textsc{nksd}})\big|^{1/2}} \right) \xrightarrow[N \to \infty]{\textup{a.s.}} 0.
\end{split}
\label{eqn:svc_asymptotic_decomposition}
\end{equation}
By Theorem~\ref{prop:laplace_scaling}, we have
\begin{equation}
\begin{split}
N(f_N^{\textsc{nksd}}(\theta_N^{\textsc{nksd}}) - f_N^{\textsc{nksd}}(\theta_{*}^{\textsc{nksd}})) &= O_{P_0}(1),\\
N (f_N^{\textsc{nksd}}(\theta_{*}^{\textsc{nksd}}) - f^{\textsc{nksd}}(\theta_{*}^{\textsc{nksd}})) &= O_{P_0}(\sqrt{N}),\\
N f^{\textsc{nksd}}(\theta_{*}^{\textsc{nksd}}) &= O_{P_0}(N),\\
\log\left(\frac{\pi(\theta_{*}^{\textsc{nksd}}) (2\pi)^{m/2}}{\big|\det \nabla_\theta^2 f^{\textsc{nksd}}(\theta_{*}^{\textsc{nksd}})\big|^{1/2}} \right) &= O_{P_0}(1),
\end{split}
\label{eqn:svc_asymptotic_scaling}
\end{equation}
and further, when the model is well-specified, such that $\textsc{nksd}(p_0(x)\|q(x|\theta_{*}^{\textsc{nksd}})) = 0$, 
\begin{equation}
	N (f_N^{\textsc{nksd}}(\theta_{*}^{\textsc{nksd}}) - f^{\textsc{nksd}}(\theta_{*}^{\textsc{nksd}})) = O_{P_0}(1).
\label{eqn:svc_asymptotic_spec}
\end{equation}

For ease of reference, here are the various scores that we consider for model/data selection.\\ 
\textit{Marginal likelihood of the augmented model (foreground+background):}
\begin{equation*}
\tilde{q}(X^{(1:N)}|\F) = \int \int q(X^{(1:N)}_{\F}|\theta)\, \tilde{q}(X^{(1:N)}_{\B} | X^{(1:N)}_{\F}, \phi_\B) \pi(\theta) \pi_\mathcal{B}(\phi_\B) d\theta d\phi_\B.
\end{equation*}
\textit{Foreground marginal \textsc{nksd}, background volume correction (a.k.a.\ the SVC):}
\begin{equation*}
\mathcal{K} := \left(\frac{2\pi}{N}\right)^{\dimB/2} \int \exp\!\Big(-\frac{N}{T} \widehat{\textsc{nksd}}(p_0(x_\F)\| q(x_\F|\theta))\Big) \pi(\theta) d\theta.
\end{equation*}
\textit{Foreground marginal likelihood, background volume correction:}
\begin{equation*}
\mathcal{K}^{(\mathrm{a})} := \left(\frac{2\pi}{N}\right)^{m_\B/2} q(X_\F^{(1:N)}).
\end{equation*}
\textit{Foreground marginal \textsc{nksd}:}
\begin{equation*}
	\mathcal{K}^{(\mathrm{b})} := \int \exp \left(-\frac{N}{T} \widehat{\textsc{nksd}}(p_0(x_\F)\| q(x_\F|\theta))\right) \pi(\theta) d\theta.
\end{equation*}
\textit{Foreground marginal \textsc{kl}, background volume correction:}
\begin{equation*}
	\mathcal{K}^{(\mathrm{c})} := \left(\frac{2\pi}{N}\right)^{\dimB/2} \int \exp \left(-N \widehat{\textsc{kl}}(p_0(x_\F)\| q(x_\F|\theta))\right)\pi(\theta)d\theta.
\end{equation*}
\textit{Foreground \textsc{nksd}, background volume correction:}
\begin{equation*}
	\mathcal{K}^{(\mathrm{d})} := \left(\frac{2\pi}{N}\right)^{m_\B/2} \exp\!\left(-\frac{N}{T} \min_\theta \widehat{\textsc{nksd}}(p_0(x_{\F})\| q(x_{\F}|\theta))\right).
\end{equation*}
\textit{Foreground \textsc{nksd}, foreground and background volume correction (a.k.a.\ BIC for SVC)}
\begin{equation*}
	\mathcal{K}^{\textsc{BIC}} := \left(\frac{2\pi}{N}\right)^{(\dimF + \dimB)/2}\exp\!\Big(-\frac{N}{T} \min_\theta \widehat{\textsc{nksd}}(p_0(x_\F)\| q(x_\F|\theta))\Big).
\end{equation*}

\subsection{Data selection} \label{sec:si_ds}

Assume $m_{\B_j} = o(N/\log N)$ for $j\in\{1,2\}$.
By Equations~\ref{eqn:standard_marginal_asymptotic_decomposition} and \ref{eqn:standard_marginal_asymptotic_scaling},
\begin{align} \label{eqn:Ka}
	\frac{1}{N} \log \frac{\mathcal{K}^{(\mathrm{a})}_1}{ \mathcal{K}^{(\mathrm{a})}_2} \xrightarrow[N \to \infty]{P_0} &\mathbb{E}_{X \sim p_0} [-\log q(X_{\F_2}|\theta_{2,*}^\textsc{kl})] - \mathbb{E}_{X \sim p_0} [-\log q(X_{\F_1}|\theta_{1,*}^\textsc{kl})]\\
	&= \textsc{kl}(p_0(x_{\F_2})\|q(x_{\F_2}|\theta_{2,*}^\textsc{kl})) + H_{\F_2} - \textsc{kl}(p_0(x_{\F_1})\|q(x_{\F_1}|\theta_{1,*}^\textsc{kl})) - H_{\F_1}, \notag
\end{align}
so $\mathcal{K}^{(\mathrm{a})}$ does not satisfy data selection consistency.
The SVC satisfies data selection consistency by Theorem~\ref{thm:selection_consistency} (part 1).
We show that the other scores also satisfy data selection consistency.
Since $\mathcal{K}^{(\mathrm{b})} = (2\pi/N)^{-m_\B/2} \mathcal{K}$ where $\mathcal{K}$ is the SVC, by Theorem~\ref{thm:selection_consistency} (part 1),
\begin{equation}
\begin{split}
	\frac{1}{N} \log \frac{\mathcal{K}^{(\mathrm{b})}_1}{ \mathcal{K}^{(\mathrm{b})}_2} \xrightarrow[N\to\infty]{P_0} \frac{1}{T}\textsc{nksd}(p_0(x_{\F_2})\| q(x_{\F_2}|\theta_{2,*}^\textsc{nksd})) - \frac{1}{T}\textsc{nksd}(p_0(x_{\F_1})\| q(x_{\F_1}|\theta_{1,*}^\textsc{nksd})).
\end{split}
\end{equation}
By Equation~\ref{eqn:Ka} and the fact that $\mathcal{K}^{(\mathrm{c})} = \exp(N H_\F) \mathcal{K}^{(\mathrm{a})}$, we have
\begin{equation}
	\frac{1}{N} \log \frac{\mathcal{K}^{(\mathrm{c})}_1}{ \mathcal{K}^{(\mathrm{c})}_2} \xrightarrow[N\to\infty]{P_0}
	\textsc{kl}(p_0(x_{\F_2})\|q(x_{\F_2}|\theta_{2,*}^\textsc{kl})) - \textsc{kl}(p_0(x_{\F_1})\|q(x_{\F_1}|\theta_{1,*}^\textsc{kl})).
\end{equation}
Since $\mathcal{K}^{(\mathrm{d})} = (2\pi/N)^{m_\B/2} \exp(-N f_N^\textsc{nksd}(\theta_N^\textsc{nksd}))$, then by Equation~\ref{eqn:svc_asymptotic_scaling},
\begin{equation}
	\frac{1}{N} \log \frac{\mathcal{K}^{(\mathrm{d})}_1}{ \mathcal{K}^{(\mathrm{d})}_2} \xrightarrow[N \to \infty]{P_0}  \frac{1}{T}\textsc{nksd}(p_0(x_{\F_2})\|q(x_{\F_2}|\theta_{2,*}^\textsc{nksd})) - \frac{1}{T}\textsc{nksd}(p_0(x_{\F_1})\|q(x_{\F_1}|\theta_{1,*}^\textsc{nksd})).
\end{equation}
Similarly, since $\mathcal{K}^{\textsc{BIC}} = (2\pi/N)^{m_\F/2} \mathcal{K}^{(\mathrm{d})}$,
\begin{equation}
	\frac{1}{N} \log \frac{\mathcal{K}^{\textsc{BIC}}_1}{\mathcal{K}^{\textsc{BIC}}_2} \xrightarrow[N \to \infty]{P_0}  \frac{1}{T}\textsc{nksd}(p_0(x_{\F_2})\|q(x_{\F_2}|\theta_{2,*}^\textsc{nksd})) - \frac{1}{T}\textsc{nksd}(p_0(x_{\F_1})\|q(x_{\F_1}|\theta_{1,*}^\textsc{nksd})).
\end{equation}
These methods therefore satisfy data selection consistency.
For the marginal likelihood of the augmented model, suppose $m_{\B_1}$ and $m_{\B_2}$ do not depend on $N$. Then by Equation~\ref{eqn:standard_marginal_asymptotic_decomposition},
\begin{align} \label{eqn:svc_asymptotic_difflimit}
	\frac{1}{N}\log \frac{\tilde{q}(X^{(1:N)}|\F_1)}{\tilde{q}(X^{(1:N)}|\F_2) } \xrightarrow[N\to\infty]{P_0}
	\; & \mathbb{E}_{X_{\F_2} \sim p_0}[-\log q(X_{\F_2}|\theta_{2,*}^\textsc{kl})] + \mathbb{E}_{X \sim p_0}[-\log \tilde{q}(X_{\B_2}|X_{\F_2}, \phi_{2,*}^\textsc{kl})] \\
	&- \mathbb{E}_{X_{\F_1} \sim p_0}[-\log q(X_{\F_1}|\theta_{1,*}^\textsc{kl})] - \mathbb{E}_{X \sim p_0}[-\log \tilde{q}(X_{\B_1}|X_{\F_1}, \phi_{1,*}^\textsc{kl})]\big]\notag
\end{align}
We can rewrite this in terms of the KL divergence. First note the decomposition,
\begin{equation*}
	H = -\int p_0(x) \log p_0(x) dx = - \int p_0(x_{\F_j}) \log p_0(x_{\F_j}) dx_{\F_j} - \int p_0(x) \log p_0(x_{\B_j}|x_{\F_j}) dx
\end{equation*}
for $j \in \{1,2\}$.
Adding and subtracting the entropy $H$ in Equation~\ref{eqn:svc_asymptotic_difflimit}, and using the fact that the background model is well-specified,
\begin{align}
	\frac{1}{N}\log \frac{\tilde{q}(X^{(1:N)}|\F_1)}{\tilde{q}(X^{(1:N)}|\F_2) } &\xrightarrow[N\to\infty]{P_0} \textsc{kl}(p_0(x_{\F_2})\|q(x_{\F_2}|\theta_{2,*}^\textsc{kl})) + \textsc{kl}(p_0(x_{\B_2}|x_{\F_2})\| \tilde{q}(x_{\B_2}|x_{\F_2}, \phi_{2,*}^\textsc{kl})) \notag \\
	& ~~~~~~~~~~ - \textsc{kl}(p_0(x_{\F_1})\|q(x_{\F_1}|\theta_{1,*}^\textsc{kl})) - \textsc{kl}(p_0(x_{\B_1}|x_{\F_1})\| \tilde{q}(x_{\B_1}|x_{\F_1}, \phi_{1,*}^\textsc{kl})) \notag \\
	& = \textsc{kl}(p_0(x_{\F_2})\|q(x_{\F_2}|\theta_{2,*}^\textsc{kl})) - \textsc{kl}(p_0(x_{\F_1})\|q(x_{\F_1}|\theta_{1,*}^\textsc{kl})).
\end{align}

\subsection{Nested data selection} \label{sec:si_nds}

In nested data selection, we are concerned with situations in which $\X_{\F_2} \subset \X_{\F_1}$ 
and the model is well-specified over both $\X_{\F_1}$ and $\X_{\F_2}$.  Assume further that $m_{\B_2} - m_{\B_1}$ does not depend on $N$.
First, consider $\mathcal{K}^{(\mathrm{d})}$ and $\mathcal{K}^{\textsc{BIC}}$.
Since $\mathcal{K}^{(\mathrm{d})} = (2\pi/N)^{m_\B/2} \exp(-N f_N^\textsc{nksd}(\theta_N^\textsc{nksd}))$
and by Theorem~\ref{prop:laplace_scaling}, $f_N^\textsc{nksd}(\theta_N^\textsc{nksd}) = O_{P_0}(1/N)$,
we have
\begin{equation}
	\frac{1}{\log N}\log \frac{\mathcal{K}^{(\mathrm{d})}_1}{ \mathcal{K}^{(\mathrm{d})}_2} \xrightarrow[N \to \infty]{P_0} \frac{m_{\B_2} - m_{\B_1}}{2}.
\end{equation}
Likewise, since $\mathcal{K}^{\textsc{BIC}} = (2\pi/N)^{m_\F/2} \mathcal{K}^{(\mathrm{d})}$, it follows that
\begin{equation}
	\frac{1}{\log N}\log \frac{\mathcal{K}^{\textsc{BIC}}_1}{\mathcal{K}^{\textsc{BIC}}_2} \xrightarrow[N \to \infty]{P_0} \frac{m_{\F_2} + m_{\B_2} - m_{\F_1} - m_{\B_1}}{2}.
\end{equation}
As in Section~\ref{sec:theory_ds_ms}, it is natural to assume $m_{\B_2} > m_{\B_1}$ and $m_{\F_2} + m_{\B_2} > m_{\F_1} + m_{\B_1}$,
in which case these criteria satisfy nested data selection consistency.

None of $\mathcal{K}^{(\mathrm{a})}$, $\mathcal{K}^{(\mathrm{b})}$, and $\mathcal{K}^{(\mathrm{c})}$ are guaranteed to satisfy nested data selection consistency, because the contribution of background model complexity is negligible or nonexistent.
To see this, note that assuming $m_{\B_j} = o(N/\log N)$, by Equation~\ref{eqn:Ka} we have
\begin{equation}
	\frac{1}{N}\log \frac{\mathcal{K}^{(\mathrm{a})}_1}{ \mathcal{K}^{(\mathrm{a})}_2} \xrightarrow[N \to \infty]{P_0} H_{\F_2} - H_{\F_1}.
\end{equation}
Meanwhile, since $\mathcal{K}^{(\mathrm{b})} = (2\pi/N)^{-m_\B/2} \mathcal{K}$ then by Theorem~\ref{thm:selection_consistency} (part 2),
\begin{equation}
\label{eqn:nds_Kb}
	\frac{1}{\log N} \log \frac{\mathcal{K}^{(\mathrm{b})}_1}{ \mathcal{K}^{(\mathrm{b})}_2} \xrightarrow[N\to\infty]{P_0} \frac{m_{\F_2} - m_{\F_1}}{2}.
\end{equation}
Since $\X_{\F_2} \subset \X_{\F_1}$, we have $m_{\F_2} \leq m_{\F_1}$ except perhaps in highly contrived scenarios.
If $m_{\F_2} < m_{\F_1}$ then Equation~\ref{eqn:nds_Kb} shows that $\log(\mathcal{K}^{(\mathrm{b})}_1 / \mathcal{K}^{(\mathrm{b})}_2) \xrightarrow[]{P_0} -\infty$.
On the other hand, if $m_{\F_2} = m_{\F_1}$, then by Equations~\ref{eqn:svc_asymptotic_decomposition} and \ref{eqn:svc_asymptotic_scaling},
$\log(\mathcal{K}^{(\mathrm{b})}_1 / \mathcal{K}^{(\mathrm{b})}_2) = O_{P_0}(1)$,
so it is not possible to have $\log(\mathcal{K}^{(\mathrm{b})}_1 / \mathcal{K}^{(\mathrm{b})}_2) \xrightarrow[]{P_0} \infty$.
Therefore, $\mathcal{K}^{(\mathrm{b})}$ does not satisfy nested data selection consistency.

Since $\mathcal{K}^{(\mathrm{c})} = e^{N H_\F}\mathcal{K}^{(\mathrm{a})} = e^{N H_\F}(2\pi/N)^{m_\B/2} q(X_\F^{(1:N)})$, then 
by Equations~\ref{eqn:standard_marginal_asymptotic_decomposition} and \ref{eqn:standard_marginal_asymptotic_scaling},
\begin{equation} \label{eqn:kc_ds}
\begin{split}
\frac{1}{\sqrt{N}}\log \frac{\mathcal{K}^{(\mathrm{c})}_1}{ \mathcal{K}^{(\mathrm{c})}_2}
= \sqrt{N} \bigg(\frac{1}{N}\sum_{i=1}^N \log \frac{p_0(X_{\F_1}^{(i)})}{p_0(X_{\F_2}^{(i)})} - \mathbb{E}\Big(\log\frac{p_0(X_{\F_1})}{p_0(X_{\F_2})}\Big)\bigg) + O_{P_0}(N^{-1/2}\log N).
\end{split}	 
\end{equation}
If $\sigma^2 := \mathbb{V}_{P_0}(\log p_0(X_{\F_1})/p_0(X_{\F_2}))$ is positive and finite, then by the central limit theorem and Slutsky's theorem,
$N^{-1/2}\log(\mathcal{K}^{(\mathrm{c})}_1/\mathcal{K}^{(\mathrm{c})}_2) \xrightarrow[]{D} \mathcal{N}(0,\sigma^2)$.
Thus, $\mathcal{K}^{(\mathrm{c})}$ randomly selects $\F_1$ or $\F_2$ with equal probability, 
and therefore, it does not satisfy nested data selection consistency.

For the marginal likelihood of the augmented model, suppose $m_{\B_1}$ and $m_{\B_2}$ do not depend on $N$.
The marginal likelihood achieves nested data selection consistency because the augmented models are both well-specified and describe the complete data space $\X$; this guarantees that the $O_{P_0}(\sqrt{N})$ terms in the marginal likelihood decomposition cancel.
Specifically, $p_0(x) = q(x \mid\theta_{j,*}^\textsc{kl}, \phi_{j,*}^\textsc{kl}, \F_j)$ for $j\in\{1,2\}$, and thus, by Equations~\ref{eqn:standard_marginal_asymptotic_decomposition} and \ref{eqn:standard_marginal_asymptotic_scaling} applied to the augmented model,
\begin{equation}
	\frac{1}{\log N}\log \frac{\tilde{q}(X^{(1:N)}|\F_1)}{\tilde{q}(X^{(1:N)}|\F_2) } \xrightarrow[N\to\infty]{P_0} \frac{m_{\F_2} + m_{\B_2} - m_{\F_1} - m_{\B_1}}{2}.
\end{equation}
Nested data selection consistency follows assuming $m_{\F_2} + m_{\B_2} > m_{\F_1} + m_{\B_1}$ as before.
This can be contrasted with Equation~\ref{eqn:kc_ds}, where although both foreground models are well-specified, they describe different data ($X^{(1:N)}_{\F_1}$ versus $X^{(1:N)}_{\F_2}$), so the $O_{P_0}(\sqrt{N})$ terms remain.

\subsection{Model selection} \label{sec:si_ms}

All of the criteria we consider satisfy model selection consistency.  To see this, we apply the same asymptotic analysis as used for data selection in Section~\ref{sec:si_ds}, under the same conditions on $m_\B$, obtaining
\begin{equation}
\begin{split}
	\frac{1}{N}\log \frac{\tilde{q}_1(X^{(1:N)}|\F)}{\tilde{q}_2(X^{(1:N)}|\F) } \xrightarrow[N\to\infty]{P_0} & \textsc{kl}(p_0(x_{\F})\|q_2(x_{\F}|\theta_{2,*}^\textsc{kl})) - \textsc{kl}(p_0(x_{\F})\|q_1(x_{\F}|\theta_{1,*}^\textsc{kl})),
\end{split}
\end{equation}
\begin{equation}
\begin{split}
	\frac{1}{N} \log \frac{\mathcal{K}^{(\mathrm{a})}_1}{ \mathcal{K}^{(\mathrm{a})}_2} \xrightarrow[N \to \infty]{P_0} & \textsc{kl}(p_0(x_{\F})\|q_2(x_{\F}|\theta_{2,*}^\textsc{kl})) - \textsc{kl}(p_0(x_{\F})\|q_1(x_{\F}|\theta_{1,*}^\textsc{kl})),
\end{split}
\end{equation}
\begin{equation}
\begin{split}
	\frac{1}{N} \log \frac{\mathcal{K}^{(\mathrm{b})}_1}{ \mathcal{K}^{(\mathrm{b})}_2} \xrightarrow[N\to\infty]{P_0} \frac{1}{T}\textsc{nksd}(p_0(x_{\F})\| q_2(x_{\F}|\theta_{2,*}^\textsc{nksd})) - \frac{1}{T}\textsc{nksd}(p_0(x_{\F})\| q_1(x_{\F}|\theta_{1,*}^\textsc{nksd})),
\end{split}
\end{equation}
\begin{equation}
	\frac{1}{N} \log \frac{\mathcal{K}^{(\mathrm{c})}_1}{ \mathcal{K}^{(\mathrm{c})}_2} \xrightarrow[N\to\infty]{P_0}
	\textsc{kl}(p_0(x_{\F})\|q_2(x_{\F}|\theta_{2,*}^\textsc{kl})) - \textsc{kl}(p_0(x_{\F})\|q_1(x_{\F}|\theta_{1,*}^\textsc{kl})),
\end{equation}
\begin{equation}
	\frac{1}{N} \log \frac{\mathcal{K}^{(\mathrm{d})}_1}{ \mathcal{K}^{(\mathrm{d})}_2} \xrightarrow[N \to \infty]{P_0}  \frac{1}{T}\textsc{nksd}(p_0(x_{\F})\|q_2(x_{\F}|\theta_{2,*}^\textsc{nksd})) - \frac{1}{T}\textsc{nksd}(p_0(x_{\F})\|q_1(x_{\F}|\theta_{1,*}^\textsc{nksd})),
\end{equation}
\begin{equation}
	\frac{1}{N} \log \frac{\mathcal{K}^{\textsc{BIC}}_1}{\mathcal{K}^{\textsc{BIC}}_2} \xrightarrow[N \to \infty]{P_0}  \frac{1}{T}\textsc{nksd}(p_0(x_{\F})\|q_2(x_{\F}|\theta_{2,*}^\textsc{nksd})) - \frac{1}{T}\textsc{nksd}(p_0(x_{\F})\|q_1(x_{\F}|\theta_{1,*}^\textsc{nksd})).
\end{equation}
Note that in contrast to the data selection case, $\mathcal{K}^{(\mathrm{a})}$ satisfies model selection consistency
since the entropy terms $H_{\F_j}$ cancel due to the fact that $\F$ is fixed.
We can think of this as a consequence of the \textsc{kl} divergence's subsystem independence; if we are just interested in modeling a fixed foreground space, there is no problem considering the foreground marginal likelihood alone~\citep{Caticha2003-gb,Caticha2010-fm,Rezende2018-rp}. 

\subsection{Nested model selection} \label{sec:si_nms}

In nested model selection, since both models are well-specified, we have $q_j(x_\F|\theta_{j,*}^\textsc{kl}) = p_0(x_\F) = q_j(x_\F|\theta_{j,*}^\textsc{nksd})$ for $j\in\{1,2\}$. Thus, the estimated divergences cancel:
\begin{equation*}
\begin{split}
	\widehat{\textsc{nksd}}(p_0(x_\F)\| q_1(x_\F|\theta_{1,*}^\textsc{nksd})) &= \widehat{\textsc{nksd}}(p_0(x_\F)\| q_2(x_\F|\theta_{2,*}^\textsc{nksd})),\\
	\sum_{i=1}^N\log q_1(X_\F^{(i)}|\theta_{1,*}^\textsc{kl}) &= \sum_{i=1}^N\log q_2(X_\F^{(i)}|\theta_{2,*}^\textsc{kl}),\\
	\widehat{\textsc{kl}}(p_0(x_\F)\| q_1(x_\F|\theta_{1,*}^\textsc{kl})) &= \widehat{\textsc{kl}}(p_0(x_\F)\| q_2(x_\F|\theta_{2,*}^\textsc{kl})).
\end{split}
\end{equation*}   
Using this along with Equations~\ref{eqn:standard_marginal_asymptotic_decomposition}--\ref{eqn:svc_asymptotic_spec}, under the same conditions on $m_\B$ as in Section~\ref{sec:si_ds},
\begin{equation}
\frac{1}{\log N} \log \frac{\tilde{q}_1 (X^{(1:N)}|\F)}{\tilde{q}_2 (X^{(1:N)}|\F)}\xrightarrow[N\to\infty]{P_0} \frac{m_{\F,2} - m_{\F,1}}{2},
\end{equation}
\begin{equation}
\frac{1}{\log N} \log \frac{\mathcal{K}^{(\mathrm{a})}_1}{\mathcal{K}^{(\mathrm{a})}_2} \xrightarrow[N\to\infty]{P_0} \frac{m_{\F,2} - m_{\F,1}}{2},
\end{equation}
\begin{equation}
	\frac{1}{\log N} \log \frac{\mathcal{K}^{(\mathrm{b})}_1}{\mathcal{K}^{(\mathrm{b})}_2} \xrightarrow[N\to\infty]{P_0} \frac{m_{\F,2} - m_{\F,1}}{2},
\end{equation}
\begin{equation}
	\frac{1}{\log N} \log \frac{\mathcal{K}^{(\mathrm{c})}_1}{\mathcal{K}^{(\mathrm{c})}_2} \xrightarrow[N\to\infty]{P_0} \frac{m_{\F,2} - m_{\F,1}}{2},
\end{equation}
\begin{equation}
	\log \frac{\mathcal{K}^{(\mathrm{d})}_1}{\mathcal{K}^{(\mathrm{d})}_2} = O_{P_0}(1),
\end{equation}
\begin{equation}
	\frac{1}{\log N}\log \frac{\mathcal{K}^{\textsc{BIC}}_1}{\mathcal{K}^{\textsc{BIC}}_2} \xrightarrow[N \to \infty]{P_0} \frac{m_{\F,2} - m_{\F,1}}{2},
\end{equation}
where we are using the assumption that the background model is the same in the two augmented models $\tilde{q}_1$ and $\tilde{q}_2$ and so $m_{\B,1} = m_{\B,2}$. Only $\mathcal{K}^{(\mathrm{d})}$ fails to satisfy nested model selection consistency.

\section{Proofs} \label{sec:proofs}

\subsection{Proofs of NKSD properties} \label{sec:proofs_nksd}

\begin{proof}[\textup{\textbf{Proof of Proposition~\ref{proposition:u_statistic_nksd}}}]
By assumption, the kernel is bounded, say $|k(x, y)| \le B$, and $s_p,s_q \in L^1(P)$. Thus, by the Cauchy--Schwarz inequality,
\begin{align*}
\bigg\vert\int_{\mathcal{X}} \int_{\mathcal{X}} &(s_{q}(x) - s_{p}(x))^\top (s_{q}(y) - s_{p}(y)) k(x, y) p(x) p(y)dx dy\bigg\vert \\
&\le B \Big( \textstyle{\int_{\mathcal{X}}} \| s_{q}(x) - s_{p}(x) \| p(x) dx \Big)^2 < \infty.
\end{align*}
Since the kernel is integrally strictly positive definite and $|k(x,y)|\le B$,
\begin{equation}
	0 < \int_{\mathcal{X}} \int_{\mathcal{X}} k(x, y) p(x) p(y) dx dy \le B < \infty.
\end{equation}
Thus, the \textsc{nksd} is finite.
Equation~\ref{eqn:nksd_u_rep} follows from Theorem 3.6 of \citet{Liu2016-bp}.
\end{proof}

\begin{proof}[\textup{\textbf{Proof of Proposition~\ref{thm:nksd_divergence}}}]
The denominator of the \textsc{nksd} is positive since $k$ is integrally strictly positive definite.
Defining $\delta(x) = s_{q}(x) - s_{p}(x)$, the numerator of the \textsc{nksd} is
\begin{equation}
\int_{\mathcal{X}} \int_{\mathcal{X}} \delta(x)^\top \delta(y) k(x,y) p(x) p(y) dx dy = \sum_{i=1}^d \int_{\mathcal{X}} \int_{\mathcal{X}} \delta_i(x) \delta_i(y) k(x, y) p(x) p(y) dx dy.
\end{equation}
If $\delta_i(x)p(x) = 0$ almost everywhere with respect to Lebesgue measure on $\mathcal{X}$,
then the $i$th term on the right-hand side is zero.
Meanwhile, if $\delta_i(x)p(x)$ is not a.e.\ zero, then $\int_{\mathcal{X}} |\delta_i(x)|p(x) d x > 0$,
and hence, the $i$th term is positive 
since $k$ is integrally strictly positive definite and $\delta_i\in L^1(P)$ by assumption.
Hence, the \textsc{nksd} is nonnegative, and equals zero if and only if $\delta(x) p(x) = 0$ almost everywhere.

Suppose $\delta(x) p(x) = 0$ almost everywhere.  Since $p(x) > 0$ on $\mathcal{X}$ by assumption, this implies $s_{p}(x) = s_{q}(x)$ a.e.,
and in fact, $s_{p}(x) = s_{q}(x)$ for all $x\in\mathcal{X}$ by continuity.
Since $\mathcal{X}$ is open and connected, then by the gradient theorem (that is, the fundamental theorem of calculus for line integrals),
$p(x) \propto q(x)$, and hence, $p(x) = q(x)$ on $\mathcal{X}$.
Conversely, if $p(x) = q(x)$ almost everywhere, then $\delta(x)p(x) = 0$ almost everywhere.
\end{proof}

\begin{proof}[\textup{\textbf{Proof of Proposition~\ref{prop:subsystem_indep}}}]
Define
\begin{equation*}
\begin{split}
	\delta_1(x_1) &:= \nabla_{x_1} \log q(x) - \nabla_{x_1} \log p(x) = \nabla_{x_1} \log q(x_1) - \nabla_{x_1} \log p(x_1)\\
	\delta_2(x_2) &:= \nabla_{x_2} \log q(x) - \nabla_{x_2} \log p(x) = \nabla_{x_2} \log q(x_2) - \nabla_{x_2} \log p(x_2).
\end{split}
\end{equation*}
Let $X,Y\sim p(x)$ independently.  Note that $\mathbb{E}[k_1(X_1, Y_1)] > 0$ and $\mathbb{E}[k_2(X_2, Y_2)] > 0$ since $k_1$ and $k_2$ are integrally strictly positive definite by assumption.  Therefore,
\begin{equation*}
\begin{split}
	\textsc{nksd}&(p(x)\|q(x)) = \frac{\mathbb{E}[(\nabla_x \log q(X) - \nabla_x \log p(X))^\top (\nabla_x \log q(Y) - \nabla_x \log p(Y)) k(X, Y)]}{\mathbb{E}[k(X, Y)]}\\
	&= \frac{\mathbb{E}[\delta_1(X_1)^\top \delta_1(Y_1) k_1(X_1, Y_1)]\mathbb{E}[k_2(X_2, Y_2)]}{\mathbb{E}[k_1(X_1, Y_1)]\mathbb{E}[k_2(X_2, Y_2)]}
	 + \frac{\mathbb{E}[\delta_2(X_2)^\top \delta_2(Y_2) k_2(X_2, Y_2)]\mathbb{E}[k_1(X_1, Y_1)]}{\mathbb{E}[k_1(X_1, Y_1)]\mathbb{E}[k_2(X_2, Y_2)]}\\
	&= \frac{\mathbb{E}[\delta_1(X_1)^\top \delta_1(Y_1) k_1(X_1, Y_1)]}{\mathbb{E}[k_1(X_1, Y_1)]} + \frac{\mathbb{E}[\delta_2(X_2)^\top \delta_2(Y_2) k_2(X_2, Y_2)]}{\mathbb{E}[k_2(X_2, Y_2)]}\\
	&= \textsc{nksd}(p(x_1)\|q(x_1)) + \textsc{nksd}(p(x_2)\|q(x_2)).
\end{split}
\end{equation*}
\end{proof}

The following modified version applies to the estimator $\widehat{\textsc{nksd}}(p\|q)$ (Equation~\ref{eqn:est_nksd}).

\begin{proposition} \label{prop:approx_subsystem_indep}
	\begin{equation} \label{eqn:est_nksd_decomp}
		\widehat{\textsc{nksd}}(p(x)\|q(x)) = \overline{\textsc{nksd}}(p(x_1)\|q(x_1)) + \overline{\textsc{nksd}}(p(x_2)\|q(x_2))
	\end{equation}
	where
	\begin{equation*}
	\begin{split}
	\overline{\textsc{nksd}}(p(x_1)\|q(x_1)) :=& \frac{\sum_{i\neq j} u_1(X_1^{(i)}, X_1^{(j)}) k_2(X^{(i)}_2, X^{(j)}_2)}{\sum_{i\neq j} k_1(X^{(i)}_1, X^{(j)}_1) k_2(X^{(i)}_2, X^{(j)}_2)}\\
	 u_1(x_1, y_1) :=& s_{q}(x_1)^\top s_{q}(y_1) k_1(x_1, y_1) + s_{q}(x_1)^\top \nabla_{y_1} k_1(x_1, y_1) + s_{q}(y_1)^\top \nabla_{x_1} k_1(x_1, y_1)\\ & + \Tr(\nabla_{x_1} \nabla_{y_1}^\top k_1(x_1, y_1))\\
	 s_q(x_1) :=& \nabla_{x_1} \log q(x_1),\\
	\end{split}
	\end{equation*}
	and vice versa for $\overline{\textsc{nksd}}(p(x_2)\|q(x_2))$ with the roles of 1 and 2 swapped.
\end{proposition}
\begin{proof}
    Recall the definition of $\widehat{\textsc{nksd}}(p(x)\|q(x))$ in Equation~\ref{eqn:est_nksd}.
    Note that $\nabla_{x_1} k(x,y) = k_2(x_2, y_2) \nabla_{x_1} k_1(x_1, y_1)$ and $\nabla_{x_1} \log q(x) = \nabla_{x_1} \log q(x_1)$. Examining $u(x,y)$ term-by-term,
\begin{equation*}
\begin{split}
	\nabla_x \log q(x)^\top \nabla_y \log q(y) k(x, y) =& \big[\nabla_{x_1} \log q(x_1)^\top \nabla_{y_1} \log q(y_1) k_1(x_1, y_1)\big] k_2(x_2, y_2)\\
	&+\big[\nabla_{x_2} \log q(x_2)^\top \nabla_{y_2} \log q(y_2) k_2(x_2, y_2)\big] k_1(x_1, y_1), \\
	\nabla_x \log q(x)^\top \nabla_y k(x, y) =& [\nabla_{x_1} \log q(x_1)^\top \nabla_{y_1} k_1(x_1,y_1)] k_2(x_2, y_2)\\
	&+[\nabla_{x_2} \log q(x_2)^\top \nabla_{y_2} k_2(x_2, y_2)] k_1(x_1, y_1),\\
	\nabla_x k(x, y)^\top \nabla_y \log q(y) =& [\nabla_{x_1} k_1(x_1, y_1)^\top \nabla_{y_1} \log q(y_1)] k_2(x_2,y_2),\\
	&+[\nabla_{x_2} k_2(x_2, y_2)^\top \nabla_{y_2} \log q(y_2)] k_1(x_1, y_1)\\
	\Tr(\nabla_x \nabla_y^\top k(x, y)) =& \Tr(\nabla_{x_1} \nabla_{y_1}^\top k_1(x_1, y_1)) k_2(x_2, y_2),\\
	&+\Tr(\nabla_{x_2} \nabla_{y_2}^\top k_2(x_2, y_2)) k_1(x_1, y_1).
\end{split}
\end{equation*}
Thus, defining $u_1$ and $u_2$ as in Proposition~\ref{prop:approx_subsystem_indep}, we have
\begin{equation*}
\begin{split}
	u(x, y) &= u_1(x_1, y_1) k_2(x_2, y_2) + u_2(x_2, y_2) k_1(x_1, y_1),\\
	k(x, y) &= k_1(x_1, y_1) k_2(x_2, y_2).
\end{split}
\end{equation*}
The result follows.
\end{proof}

To interpret Proposition~\ref{prop:approx_subsystem_indep}, note that
\begin{equation*}
\begin{split}
	\frac{\mathbb{E}_{X,Y\sim p}[u_1(X_1, Y_1) k_2(X_2, Y_2)]}{\mathbb{E}_{X,Y\sim p}[k_1(X_1, Y_1) k_2(X_2, Y_2)]} = \frac{\mathbb{E}_{X_1,Y_1\sim p(x_1)}[u_1(X_1, Y_1)]}{\mathbb{E}_{X_1,Y_1\sim p(x_1)}[k_1(X_1, Y_1)]}= \textsc{nksd}(p(x_1) \| q(x_1)),
\end{split}
\end{equation*} 
so $\overline{\textsc{nksd}}(p(x_1) \| q(x_1))$ is an estimator of $\textsc{nksd}(p(x_1) \| q(x_1))$, and likewise for $\overline{\textsc{nksd}}(p(x_2)\|q(x_2))$.

\subsection{Proof of Theorems~\ref{thm:marginal_ksd} and \ref{thm:expo_marginal_ksd}} \label{sec:si_bvm_proof}

Our proofs in this section build on the proof of Theorem 3 of \citet{Barp2019-ut}.

\begin{proposition} \label{proposition:uniform_convergence}
Under the assumptions of Theorem~\ref{thm:marginal_ksd},
for any compact convex $C \subseteq \Theta$,
\begin{equation}
\label{eqn:nksd_uniform_convergence}
	\sup_{\theta \in C} |f_N(\theta) - f(\theta)| \xrightarrow{\textup{a.s.}} 0.
\end{equation}
\end{proposition}
\begin{proof}
First, we establish almost sure convergence for the denominator of $f_N(\theta)$. Since $k$ is assumed to be bounded and to have bounded derivatives up to order two, we can choose $B < \infty$ such that $B \ge |k| + \|\nabla_x k\| + \|\nabla_x \nabla_y^\top k\|$.
In particular, the expected value of the kernel is finite:
\begin{equation}
\label{eqn:kernel_integrable}
    \int_\mathcal{X} \int_\mathcal{X} |k(x, y)| P_0(dx) P_0(dy) \le B < \infty.
\end{equation}
By the strong law of large numbers for U-statistics (Theorem 5.4A of \citealp{Serfling2009-dq}), 
\begin{equation}
\label{eqn:nksd_denom_convergence}
    \frac{1}{N(N-1)} \sum_{i \neq j} k(X^{(i)}, X^{(j)}) \xrightarrow[N \to \infty]{\textup{a.s.}} \int_\mathcal{X} \int_\mathcal{X} k(x, y) P_0(dx) P_0(dy).
\end{equation}
Note that the limit is positive since $k(x,y)>0$ for all $x,y\in\mathcal{X}$. 
For the numerator, we establish bounds on $u_\theta$ and $\nabla_\theta u_\theta$. Let $C\subseteq\Theta$ be compact and convex.
By Equation~\ref{eqn:est_nksd}, for all $\theta\in C$ and all $x,y\in\mathcal{X}$,
\begin{equation}
\label{eqn:u_bound}
\begin{split}
|u_\theta(x,y)| &\le |s_{q_\theta}(x)^\top s_{q_\theta}(y) k(x, y)| + |s_{q_\theta}(x)^\top \nabla_y k(x, y)| \\
   &\phantom{\le ~~} + |s_{q_\theta}(y)^\top \nabla_x k(x, y)| + |\Tr(\nabla_x \nabla_y^\top k(x, y))|\\
   &\le \|s_{q_\theta}(x)\| \|s_{q_\theta}(y)\| B + \|s_{q_\theta}(x)\| B + \|s_{q_\theta}(y)\| B + B d\\
   &\le g_{0,C}(x) g_{0,C}(y) B + g_{0,C}(x) B + g_{0,C}(y) B + B d\\
   &=: h_{0,C}(x, y).
\end{split}
\end{equation}
Similarly, for all $\theta\in C$ and all $x,y\in\mathcal{X}$,
\begin{equation}
\label{eqn:grad_u_bound}
\begin{split}
\|\nabla_\theta u_\theta(x, y) \| &\le \|\nabla_\theta(s_{q_\theta}(x)^\top s_{q_\theta}(y)) k(x, y)\| + \|\nabla_\theta (s_{q_\theta}(x)^\top \nabla_y k(x, y))\| \\
    &\phantom{\le ~~} + \|\nabla_\theta(s_{q_\theta}(y)^\top \nabla_x k(x, y))\| + \|\nabla_\theta\Tr(\nabla_x \nabla_y^\top k(x, y))\|\\
    &\le g_{0,C}(x) g_{1,C}(y) B + g_{0,C}(y) g_{1,C}(x) B + g_{1,C}(x) B + g_{1,C}(y) B \\
&=: h_{1,C}(x, y).
\end{split}
\end{equation}
Note that $h_{0,C}$ and $h_{1,C}$ are continuous and belong to $L^1(P_0\times P_0)$. 

Let $S_1\subseteq S_2 \subseteq \cdots\subseteq\mathcal{X}$ be a sequence of compact sets such that $\cup_{M=1}^\infty S_M = \mathcal{X}$. Note that this implies $\cup_{M=1}^\infty S_M \times S_M = \mathcal{X}\times \mathcal{X}$.
Suppose for the moment that, for each $M$, the following collections of functions are equicontinuous on $C$:
(A) $(\theta \mapsto u_\theta(x,y) : x,y\in S_M)$ and (B) $\big(\theta \mapsto \int u_\theta (x, y) P_0(dy) : x \in S_M\big)$.
Assuming this, Theorem 1 of \citet{Yeo2001-jb} shows that
\begin{equation}
    \sup_{\theta \in C} \bigg\vert\frac{1}{N(N-1)} \sum_{i \neq j} u_\theta(X^{(i)}, X^{(j)}) - \int_\mathcal{X} \int_\mathcal{X} u_\theta (x, y) P_0(dx) P_0(dy)\bigg\vert\xrightarrow[N \to \infty]{\textup{a.s.}} 0,
\label{eqn:ksd_uniform_converge}
\end{equation}
and that $\theta \mapsto \int_\mathcal{X} \int_\mathcal{X} u_\theta (x, y) P_0(dx) P_0(dy)$ is continuous.
(Note that although \citet{Yeo2001-jb} assume $\mathcal{X} = \mathbb{R}$, their proof goes through without further modification for any nonempty $\mathcal{X} \subseteq \mathbb{R}^d$.)
Combining Equations~\ref{eqn:nksd_denom_convergence} and \ref{eqn:ksd_uniform_converge}, we have
\begin{equation*}
\frac{\sup_{\theta \in C} \big\vert\frac{1}{N(N-1)} \sum_{i \neq j} u_\theta(X^{(i)}, X^{(j)}) - \int \int u_\theta (x, y) P_0(dx) P_0(dy)\big\vert}{\frac{1}{N(N-1)} \sum_{i \neq j} k(X^{(i)}, X^{(j)})} \xrightarrow[N \to \infty]{\textup{a.s.}} 0.
\end{equation*}
Thus, it follows that $\sup_{\theta\in C} |f_N(\theta) - f(\theta)| \to 0$ a.s.\ by Equations~\ref{eqn:nksd_denom_convergence} and \ref{eqn:u_bound}.
To complete the proof, we must show that (A) and (B) are equicontinuous on $C$.

(A) Since $\theta \mapsto u_\theta(x, y)$ is differentiable on $C$, then by the mean value theorem, we have that for all $\theta_1, \theta_2 \in C$ and all $x,y\in S_M$,
\begin{align*}
    |u_{\theta_1}(x, y) - u_{\theta_2}(x, y)| &\le \|\nabla_\theta|_{\theta = \tilde{\theta}\,} u_{\theta}(x, y)\| \|\theta_1 - \theta_2\| \\
	&\le h_{1,C}(x, y) \|\theta_1 - \theta_2\| \\
	&\le  \Big(\sup_{x, y \in S_M} h_{1,C}(x, y)\Big) \|\theta_1 - \theta_2\| < \infty
\end{align*}
where $\tilde{\theta} = \gamma \theta_1 + (1-\gamma) \theta_2$ for some $\gamma \in [0, 1]$. 
Here, the second inequality holds since $\tilde{\theta} \in C$ by the convexity of $C$,
and the supremum is finite because a continuous function on a compact set attains its maximum.
Therefore, $(\theta \mapsto u_\theta(x,y) : x,y\in S_M)$ is equicontinuous on $C$.

(B) To see that $\big(\theta \mapsto \int u_\theta (x, y) P_0(dy) : x \in S_M\big)$ is equicontinuous on $C$, first note that 
\begin{equation*}
    \int |u_\theta(x, y)| P_0(dy) \le \int h_{0,C}(x, y) P_0(dy) < \infty.
\end{equation*}
Further, due to Equations~\ref{eqn:u_bound} and \ref{eqn:grad_u_bound}, we can apply the Leibniz integral rule \citep[Theorem 2.27]{folland1999real} and find that $\nabla_\theta \int u_\theta(x, y) P_0(dy)$ exists and is equal to $\int \nabla_\theta u_\theta(x, y) P_0(dy)$.
Now we apply the mean value theorem and the same reasoning as before to find that for all $\theta_1, \theta_2 \in C$ and all $x\in S_M$,
\begin{equation*}
\begin{split}
\Big|{\textstyle\int} u_{\theta_1} (x, y) P_0(dy) - {\textstyle\int} u_{\theta_2} (x, y) P_0(dy)\Big| &\le \big\|\nabla_\theta|_{\theta=\tilde{\theta}}\, {\textstyle\int} u_\theta(x, y) P_0(dy)\big\| \|\theta_1 - \theta_2\|\\
& \le \|\theta_1 - \theta_2\|\int \big\| \nabla_\theta|_{\theta=\tilde{\theta}}\, u_\theta (x, y) \big\| P_0(dy) \\
& \le \|\theta_1 - \theta_2\| \sup_{x \in S_M} \int h_{1,C}(x, y) P_0(dy) < \infty
\end{split}
\end{equation*}
where $\tilde{\theta} = \gamma \theta_1 + (1-\gamma) \theta_2$ for some $\gamma \in [0, 1]$. 
The supremum is finite since $x\mapsto\int h_{1,C}(x, y) P_0(dy)$ is continuous, which can easily be seen by plugging in the definition of $h_{1,C}$.
Therefore, $\big(\theta \mapsto \int u_\theta (x, y) P_0(dy) : x \in S_M\big)$ is equicontinuous on $C$. 
\end{proof}

\begin{proposition} \label{proposition:nksd_uniform_bound}
Under the assumptions of Theorem~\ref{thm:marginal_ksd}, $(f'''_N : N \in \mathbb{N})$ is uniformly bounded on $E$.
\end{proposition}
\begin{proof}
First, for any $x,y\in\mathcal{X}$, if we define $g(\theta) = s_{q_\theta}(x)$ and $h(\theta) = s_{q_\theta}(y)$ then
$u_\theta = (g^\top h) k + g^\top (\nabla_y k) + h^\top (\nabla_x k) + \Tr(\nabla_x\nabla_y^\top k)$.
By differentiating, applying Minkowski's inequality to the resulting sum of tensors, and applying the Cauchy--Schwarz inequality to each term, we have
\begin{align*}
\|\nabla_\theta^3 u_\theta(x,y)\| &\le \|\nabla^3 g\| \|h\| k + 3 \|\nabla^2 g\| \|\nabla h\| k + 3 \|\nabla g\| \|\nabla^2 h\| k 
+ \|g\| \|\nabla^3 h\| k \\ & ~~~ + \|\nabla^3 g\| \|\nabla_y k\| + \|\nabla^3 h\| \|\nabla_x k\|.
\end{align*}
Using the symmetry of the kernel to combine like terms, this yields that
\begin{equation*}
\begin{split}
    \Big\|\sum_{i \neq j} & \nabla_\theta^3 u_\theta(X^{(i)}, X^{(j)}) \Big\| \\
    \le \sum_{i \neq j} & \Big(2 \| \nabla^3_\theta s_{q_\theta} (X^{(i)}) \| \|s_{q_\theta}(X^{(j)})\| B
    + 6 \|\nabla^2_\theta s_{q_\theta}(X^{(i)})\| \|\nabla_\theta s_{q_\theta}(X^{(j)})\| B
    + 2 \|\nabla^3_\theta s_{q_\theta}(X^{(i)})\| B\Big)
\end{split}
\end{equation*}
where $B < \infty$ such that $B \ge |k| + \|\nabla_x k\| + \|\nabla_x \nabla_y^\top k\|$.
Since $f_N(\theta) = 0$ when $N=1$ by definition, 
we can assume without loss of generality that $N \ge 2$, so $\frac{1}{N-1} = \frac{1}{N} (1 + \frac{1}{N-1}) \le 2/N$. Since each term is non-negative, we can add in the $i=j$ terms,
\begin{align} \label{eqn:kppp_bound}
\Big\|&\frac{1}{N(N-1)} \sum_{i \neq j} \nabla_\theta^3 u_\theta(X^{(i)}, X^{(j)}) \Big\| \notag\\
&\le \frac{2 B}{N^2} \sum_{i, j} \Big(2 \| \nabla^3_\theta s_{q_\theta} (X^{(i)}) \| \|s_{q_\theta}(X^{(j)})\|
    + 6 \|\nabla^2_\theta s_{q_\theta}(X^{(i)})\| \|\nabla_\theta s_{q_\theta}(X^{(j)})\| + 2 \|\nabla^3_\theta s_{q_\theta}(X^{(i)})\|\Big) \notag \\
&= 4 B \Big(\frac{1}{N} \sum_i \| \nabla^3_\theta s_{q_\theta}(X^{(i)}) \| \Big) \Big(\frac{1}{N} \sum_j \|s_{q_\theta}(X^{(j)})\| \Big) \\
&~~~~ +12 B \Big(\frac{1}{N} \sum_i \| \nabla_\theta^2 s_{q_\theta}(X^{(i)}) \| \Big) \Big(\frac{1}{N} \sum_j \|\nabla_\theta s_{q_\theta}(X^{(j)})\| \Big)\notag\\
&~~~~ +4 B \Big(\frac{1}{N} \sum_i \|\nabla^3_\theta s_{q_\theta}(X^{(i)})\| \Big).\notag
\end{align}
By assumption, $\big\{\frac{1}{N} \sum_i \|\nabla^2_\theta s_{q_\theta}(X^{(i)})\|: N \in \mathbb{N},\theta \in E\big\}$ is bounded with probability 1, and similarly for $\big\{\frac{1}{N} \sum_i \|\nabla^3_\theta s_{q_\theta}(X^{(i)})\|: N \in \mathbb{N},\theta \in E\big\}$. We show the same for $\frac{1}{N} \sum_i \|s_{q_\theta}(X^{(i)})\|$ and $\frac{1}{N} \sum_i \|\nabla_\theta s_{q_\theta}(X^{(i)})\|$. By Equation~\ref{eqn:score_bounds}, we have
\begin{equation*}
    \int \sup_{\theta \in \bar{E}} \|s_{q_\theta}(x)\| P_0(dx) \le \int g_{0,\bar{E}} (x) P_0(dx) < \infty.
\end{equation*}
Hence, by Theorem 1.3.3 of \citet{Ghosh2003-yv}, $\frac{1}{N} \sum_i \|s_{q_\theta}(X^{(i)})\|$ converges uniformly on $\bar{E}$, almost surely. In particular, $\frac{1}{N} \sum_i \|s_{q_\theta}(X^{(i)})\|$ is uniformly bounded on $E$, almost surely. The same argument holds for $\frac{1}{N} \sum_i \|\nabla_\theta s_{q_\theta}(X^{(i)})\|$ using $g_{1,\bar{E}}(x)$.
Therefore, by Equation~\ref{eqn:kppp_bound}, it follows that $\|\frac{1}{N(N-1)} \sum_{i \neq j} \nabla_\theta^3 u_\theta(X^{(i)}, X^{(j)}) \|$ is uniformly bounded on $E$. Since $k$ is positive by assumption, $\frac{1}{N(N-1)} \sum_{i \neq j} k(X^{(i)}, X^{(j)}) > 0$ for all $N \ge 2$  
and by Equations~\ref{eqn:kernel_integrable} and \ref{eqn:nksd_denom_convergence}, $\frac{1}{N(N-1)} \sum_{i \neq j} k(X^{(i)}, X^{(j)})$ converges a.s.\ to a finite quantity greater than 0. We conclude that almost surely,
\begin{equation*}
\| f'''_N(\theta) \| = \frac{1}{T}\frac{\|\frac{1}{N(N-1)} \sum_{i \neq j} \nabla_\theta^3 u_\theta(X^{(i)}, X^{(j)}) \|}{\frac{1}{N(N-1)} \sum_{i \neq j} k(X^{(i)}, X^{(j)})}
\end{equation*}
is uniformly bounded on $E$, for $N \in \{2, 3, \ldots\}$. 
Recall that for $N = 1$, $f_N(\theta) = 0$ by definition.
Therefore, almost surely, $(f'''_N : N \in \mathbb{N})$ is uniformly bounded on $E$.
\end{proof}

\begin{proof}[\textup{\textbf{Proof of Theorem~\ref{thm:marginal_ksd}}}]

We show that the conditions of Theorem 3.2 of \citet{Miller2019-ur} are met, from which the conclusions of this theorem follow immediately. 

By Condition~\ref{condition:marginal_ksd} and Equation~\ref{eqn:f_N}, $f_N$ has continuous third-order partial derivatives on $\Theta$.
Let $E$ be the set from Condition~\ref{condition:marginal_ksd}.
With probability 1, $f_N \to f$ uniformly on $E$ (by Proposition~\ref{proposition:uniform_convergence} with $C = \bar{E}$)
and $(f'''_N)$ is uniformly bounded on $E$ (by Proposition~\ref{proposition:nksd_uniform_bound}).
Note that $f$ is finite on $\Theta$ by Proposition~\ref{proposition:u_statistic_nksd}.
Thus, by Theorem 3.4 of \citet{Miller2019-ur}, $f'$ and $f''$ exist on $E$ and $f_N''\to f''$ uniformly on $E$ with probability 1.
Since $\theta_{*}$ is a minimizer of $f$ and $\theta_{*}\in E$, we know that $f'(\theta_{*}) = 0$ and $f''(\theta_{*})$ is positive semidefinite; thus, $f''(\theta_{*})$ is positive definite since it is invertible by assumption.

Case (a): Now, consider the case where $\Theta$ is compact. Then almost surely, $f_N \to f$ uniformly on $\Theta$ 
by Proposition~\ref{proposition:uniform_convergence} with $C = \Theta$.
Since $\theta_{*}$ is a unique minimizer of $f$, we have $f(\theta) > f(\theta_{*})$ for all $\theta \in \Theta \setminus \{\theta_{*}\}$.
Let $H\subseteq E$ be an open set such that $\theta_{*}\in H$ and $\bar{H} \subseteq E$.
We show that $\liminf_N \inf_{\theta \in \Theta \setminus \bar{H}} f_N(\theta) > f(\theta_{*})$. Since $\Theta\setminus H$ is compact,
\begin{equation*}
\inf_{\theta \in \Theta \setminus \bar{H}} f(\theta) - f(\theta_{*})  =: \epsilon > 0.
\end{equation*}
By uniform convergence, with probability 1, there exists $N$ such that for all $N' > N$,
$\sup_{\theta\in\Theta} |f_{N'}(\theta) - f(\theta)| \le \epsilon/2$, and thus,
\begin{equation*}
\begin{split}
    \inf_{\theta \in \Theta \setminus \bar{H}} f_{N'}(\theta)
    \ge \inf_{\theta \in \Theta \setminus \bar{H}} f(\theta) - \epsilon/2 
    = f(\theta_{*}) + \epsilon/2.
\end{split}
\end{equation*}
Hence, $\liminf_{N} \inf_{\theta \in \Theta \setminus \bar{H}} f_N(\theta) > f(\theta_{*})$ almost surely. Applying Theorem 3.2 of \citet{Miller2019-ur}, the conclusion of the theorem follows.
Note that $f_N''(\theta_N) \to f''(\theta_{*})$ a.s.\ since $\theta_N \to \theta_{*}$ and $f_N''\to f''$ uniformly on $E$.

Case (b): Alternatively, consider the case where $\Theta$ is open and $f_N$ is convex on $\Theta$. 
For all $\theta \in \Theta$, with probability 1, $f_N(\theta) \to f(\theta)$ (by Proposition~\ref{proposition:uniform_convergence} with $C = \{\theta\}$). However, we need to show that with probability 1, for all $\theta \in \Theta$, $f_N(\theta) \to f(\theta)$. We follow the argument in the proof of Theorem 6.3 of \citet{Miller2019-ur}. Let $W$ be a countable dense subset of $\Theta$. Since $W$ is countable, with probability 1, for all $\theta \in W$, $f_N(\theta) \to f(\theta)$. Since $f_N$ is convex, then with probability 1, for all $\theta \in \Theta$, the limit $\tilde{f}(\theta) := \lim_N f_N(\theta)$ exists and is finite, and $\tilde{f}$ is convex (Theorem 10.8 of \citealp{Rockafellar1970-al}).
Since $f_N$ is convex and $f(\theta)$ is finite, $f(\theta)$ is also convex.
Since $f$ and $\tilde{f}$ are convex, they are also continuous (Theorem 10.1 of \citealp{Rockafellar1970-al}). 
Continuous functions that agree on a dense subset of points must be equal. Thus, with probability 1, for all $\theta \in \Theta$, $f_N(\theta) \to f(\theta)$. Applying Theorem 3.2 of \citet{Miller2019-ur}, the conclusion of the theorem follows.
\end{proof}

\begin{proof}[\textup{\textbf{Proof of Theorem~\ref{thm:expo_marginal_ksd}}}]
Our proof builds on Appendix D.3 of \citet{Barp2019-ut}, which establishes a central limit theorem for the \textsc{ksd} when the model is an exponential family.
The outline of the proof is as follows.
First, we establish bounds on $s_{q_\theta}$ and its derivatives, using the assumed bounds on $\nabla_x t(x)$ and $\nabla_x \log \lambda(x)$.
Second, we establish that $f''(\theta)$ is positive definite and independent of $\theta$, and that $f''_N(\theta)$ converges to it almost surely; from this, we conclude that $f''(\theta_{*})$ is invertible and $f_N(\theta)$ is convex. These results rely on the convergence properties of U-statistics and on Sylvester's criterion.

The assumption that $\log \lambda(x)$ is continuously differentiable on $\mathcal{X}$ implies that $\lambda(x) > 0$ for $x\in\mathcal{X}$.
Since $q_\theta(x) = \lambda(x) \exp(\theta^\top t(x) - \kappa(\theta))$, we have
\begin{equation*}
\begin{split}
    s_{q_\theta}(x) &= \nabla_x \log \lambda(x) + (\nabla_x t(x))^\top \theta \\
    \nabla_\theta s_{q_\theta}(x) &= (\nabla_x t(x))^\top \in \mathbb{R}^{d\times m}\\
    \nabla_\theta^2 s_{q_\theta}(x) &= 0  \in \mathbb{R}^{d\times m\times m}
\end{split}
\end{equation*}
where $(\nabla_x t(x))_{i j} = \partial t_i / \partial x_j$.
Thus, $s_{q_\theta}(x)$ has continuous third-order partial derivatives with respect to $\theta$,
and Equations~\ref{eqn:score_d2_bound} and \ref{eqn:score_d3_bound} are trivially satisfied.
Equation~\ref{eqn:score_bounds} holds for all compact $C\subseteq\Theta$ since 
$\|\nabla_x \log \lambda(x)\|$ and $\|\nabla_x t(x)\|$ are continuous functions in $L^1(P_0)$ and
\begin{align*}
    \|s_{q_\theta}(x)\| &= \|\nabla_x \log \lambda(x) + (\nabla_x t(x))^\top \theta\| \le \|\nabla_x \log \lambda(x)\| + \|\nabla_x t(x)\|\|\theta\|,\\
	\|\nabla_\theta s_{q_\theta}(x)\| &= \|\nabla_x t(x)\|.
\end{align*}
Hence, Condition~\ref{condition:marginal_ksd} holds.
By Equation~\ref{eqn:f} and Proposition~\ref{proposition:u_statistic_nksd}, 
\begin{equation}
f(\theta) = \frac{1}{T}\textsc{nksd}(p_0(x)\| q(x|\theta)) = \frac{1}{T K} \int_{\mathcal{X}} \int_{\mathcal{X}} u_\theta(x,y) P_0(d x) P_0(d y)
\end{equation}
where $K :=\int \int k(x,y) P_0(d x) P_0(d y)$.  By Equation~\ref{eqn:u_quadratic_form},
\begin{equation}
\label{eqn:u_quadratic_form2}
u_\theta(x, y) = \theta^\top B_2(x, y) \theta + B_1(x, y)^\top \theta + B_0(x, y)
\end{equation}
where
\begin{equation*}
\begin{split}
B_2(x,y) &= (\nabla_x t(x)) (\nabla_y t(y))^\top k(x, y), \\
B_1(x,y) &= (\nabla_y t(y))(\nabla_x \log \lambda(x)) k(x, y) + (\nabla_x t(x))(\nabla_y \log \lambda(y)) k(x, y)\\
&~~~ + (\nabla_y t(y))(\nabla_x k(x, y)) + (\nabla_x t(x))(\nabla_y k(x, y)),\\
B_0(x,y) &= (\nabla_x \log \lambda(x))^\top (\nabla_y \log \lambda(y)) k(x, y) + (\nabla_y \log \lambda(y))^\top (\nabla_x k(x, y))\\
&~~~ + (\nabla_x \log \lambda(x))^\top (\nabla_y k(x, y)) + \Tr(\nabla_x \nabla_y^\top k(x, y)).
\end{split}
\end{equation*}
By Condition~\ref{condition:apply_nksd}, $|k(x,y)|$, $\|\nabla_x k(x,y)\|$, and $\|\nabla_x \nabla_y^\top k(x,y)\|$ are bounded by a constant, say, $B < \infty$.
Thus, it is straightforward to check that $B_2$, $B_1$, and $B_0$ belong to $L^1(P_0\times P_0)$ 
since $\|\nabla_x t(x)\|$ and $\|\nabla_x \log \lambda(x)\|$ are in $L^1(P_0)$.
Further, $0 < K < \infty$ since $0 < k(x,y) \le B < \infty$ by assumption. Thus, 
\begin{equation*}
    f(\theta) = \frac{1}{T K}\int \int \big(\theta^\top B_2(x, y) \theta + B_1(x, y)^\top \theta + B_0(x, y)\big) P_0(dx) P_0(dy) \in \mathbb{R}.
\end{equation*}
Since $k$ is symmetric, $B_2(x,y)^\top = B_2(y,x)$. Hence, $\nabla_\theta (\theta^\top B_2(x,y) \theta) = (B_2(x,y) + B_2(y,x)) \theta$,
so by Fubini's theorem,
\begin{align*}
    f'(\theta) &= \frac{1}{T K}\int \int \big(2 B_2(x, y) \theta + B_1(x, y)\big) P_0(dx) P_0(dy) \in \mathbb{R}^m, \\
    f''(\theta) &= \frac{2}{T K}\int \int B_2(x, y)P_0(dx) P_0(dy) \in \mathbb{R}^{m\times m}.
\end{align*}
Here, differentiating under the integral sign is justified simply by linearity of the expectation.
Note that $f''(\theta)$ is a symmetric matrix since $B_2(x,y)^\top = B_2(y,x)$.
Next, to show $f''(\theta)$ is positive definite, let $v \in \mathbb{R}^m\setminus\{0\}$.
By assumption, the rows of $\nabla_x t(x)$ are linearly independent with positive probability under $P_0$.
Thus, there is a set $E \subseteq \mathcal{X}$ such that $P_0(E) > 0$ and $(\nabla_x t(x))^\top v \neq 0$ for all $x\in E$.
Define $g(x) = (\nabla_x t(x))^\top v\, p_0(x) \in \mathbb{R}^d$.
Then $\int_{\mathcal{X}} |g_i(x)| d x > 0$ for at least one $i$, and
$\int_{\mathcal{X}} |g_i(x)| d x \le \|v\|\int_{\mathcal{X}}\|\nabla_x t(x)\| p_0(x) d x < \infty$ for all $i$.
Thus,
\begin{equation*}
v^\top f''(\theta) v = \frac{2}{T K} \int \int g(x)^\top g(y) k(x, y) d x d y = \frac{2}{T K} \sum_{i=1}^d \int \int g_i(x) g_i(y) k(x, y) d x d y > 0
\end{equation*}
since $k$ is integrally strictly positive definite.  Therefore, $f''(\theta)$ is positive definite.
In particular, $f''(\theta_{*})$ is invertible.

Finally, we show that with probability 1, for all $N$ sufficiently large, $f_N(\theta)$ is convex. By Equations~\ref{eqn:f_N} and \ref{eqn:u_quadratic_form2},
\begin{equation*}
f_N(\theta) = \frac{1}{T}\frac{\sum_{i \neq j} \big[\theta^\top B_2(X^{(i)}, X^{(j)}) \theta + B_1(X^{(i)}, X^{(j)})^\top \theta + B_0(X^{(i)}, X^{(j)})\big]}{\sum_{i \neq j} k (X^{(i)}, X^{(j)})}.
\end{equation*}
Thus,
\begin{equation*}
\begin{split}
f''_N(\theta) = \frac{2}{T}\frac{\sum_{i \neq j} B_2(X^{(i)}, X^{(j)})}{{\sum_{i \neq j} k (X^{(i)}, X^{(j)})}}.
\end{split}
\end{equation*}
By the strong law of large numbers for U-statistics (Theorem 5.4A of \citealp{Serfling2009-dq}), we have $f''_N(\theta) \to f''(\theta)$ almost surely, since $\int_\mathcal{X} \int_\mathcal{X} \|B_2(x, y)\| P_0(dx)  P_0(dy) < \infty$ and $0 < K < \infty$.
For a symmetric matrix $A$, let $\lambda_{*}(A)$ denote the smallest eigenvalue.
Since $\lambda_*(A)$ is a continuous function of the entries of $A$, we have $\lambda_{*}(f_N''(\theta)) \to \lambda_{*}(f''(\theta))$ a.s.\ as $N\to\infty$.
Thus, with probability 1, for all $N$ sufficiently large, $f''_N(\theta)$ is positive definite, and hence, $f_N$ is convex.
Further, for such $N$, since $f_N$ is a quadratic function with positive definite Hessian, 
we have $M_N := \inf_{\theta\in\Theta} f_N(\theta) > -\infty$ and 
$z_N = \int_\Theta \exp(-N f_N(\theta))\pi(\theta)d\theta \leq \exp(-N M_N) < \infty$.
\end{proof}

\subsection{Proof of Theorem~\ref{prop:laplace_scaling}} \label{sec:proof_laplace_scaling}

To establish Theorem~\ref{prop:laplace_scaling}, we use the properties of U-statistics described in Chapter~5.5 of \citet{Serfling2009-dq}. 
When the data distribution matches the model distribution, $\widehat{\textsc{nksd}}$ converges more quickly than when it does not match; this same property was used by \citet{Liu2016-bp} to develop a goodness-of-fit test based on the \textsc{ksd}.

\begin{proof}
We first study the asymptotics of $f_N'(\theta_{*})$. 
Denoting $\nabla_\theta\big\vert_{\theta=\theta_{*}} u_\theta$ by $\nabla_\theta u_{\theta_{*}}$ for brevity,
\begin{equation*}
f'_N(\theta_{*}) = \frac{1}{T}\frac{\frac{1}{N(N-1)} \sum_{i \neq j} \nabla_\theta u_{\theta_{*}}(X^{(i)}, X^{(j)})}{\frac{1}{N(N-1)} \sum_{i \neq j} k(X^{(i)}, X^{(j)})}.
\end{equation*}
The denominator converges a.s.\ to a finite positive constant, as in the proof of Proposition~\ref{proposition:uniform_convergence}.
It is straightforward to verify that $\mathbb{E}_{X, Y \sim P_0}[\|\nabla_\theta u_{\theta_{*}}(X, Y)\|^2] < \infty$ since 
$s_{q_{\theta_{*}}}$ and $\nabla_\theta \big\vert_{\theta=\theta_{*}} s_{q_\theta}$ are in $L^2(P_0)$ by assumption.
By Theorems 5.5.1A and 5.5.2 of \citet{Serfling2009-dq},
\begin{equation*}
	\frac{1}{N(N-1)} \sum_{i \neq j} \nabla_\theta u_{\theta_{*}}(X^{(i)}, X^{(j)}) - \mathbb{E}_{X, Y \sim P_0}[\nabla_\theta u_{\theta_{*}}(X, Y)] = O_{P_0}(N^{-1/2}).
\end{equation*}
Further, by the Leibniz integral rule \citep[Theorem 2.27]{folland1999real},
\begin{equation*}
	\mathbb{E}_{X, Y \sim P_0}[\nabla_\theta u_{\theta_{*}}(X, Y)] = \nabla_\theta\big\vert_{\theta=\theta_{*}} \mathbb{E}_{X, Y \sim P_0}[u_\theta(X, Y)] = T\,\mathbb{E}_{X, Y \sim P_0}[k(X,Y)]f'(\theta_{*}) = 0,
\end{equation*}
using the fact that $f'(\theta_{*}) = 0$ since $\theta_{*}$ is a minimizer of $f$. Thus,
\begin{equation}
\label{eqn:nksd_gradient_convergence}
	f'_N(\theta_{*}) = O_{P_0}(N^{-1/2}).
\end{equation}

Next, we examine the convergence of $\theta_N$ to $\theta_{*}$. 
For all $N$ sufficiently large, $f_N'(\theta_N) = 0$ by Theorem~\ref{thm:marginal_ksd} (part 1),
and thus, by Taylor's theorem,
\begin{equation*}
	0 = f_N'(\theta_N) = f_N'(\theta_{*}) + f_N''(\theta_N^{+})(\theta_N - \theta_{*}),
\end{equation*}
where $\theta_N^{+}$ is on the line between $\theta_N$ and $\theta_{*}$.  
As in the proof of Theorem~\ref{thm:marginal_ksd}, $f''_N \to f''$ uniformly on the set $E$ defined in Condition~\ref{condition:marginal_ksd}. Thus, since $f''_N$ is continuous on $E$ and $\theta_N^+ \to \theta_{*}$,
\begin{equation} \label{eqn:fpp_N_scaling}
 	f''_N(\theta_N^+) \xrightarrow[N \to \infty]{\textup{a.s.}}  f''(\theta_{*}).
\end{equation}
In particular, $f''_N(\theta_N^+)$ is invertible for all $N$ sufficiently large, since $f''(\theta_{*})$ is invertible by assumption.
Hence,
\begin{equation} \label{eqn:theta_N_star_taylor}
	\theta_N - \theta_{*} = - f_N''(\theta_N^{+})^{-1} f_N'(\theta_{*}),
\end{equation}
and therefore, by Equation~\ref{eqn:nksd_gradient_convergence},
 \begin{equation} \label{eqn:theta_N_star_taylor_scaling}
 \|\theta_N - \theta_{*}\| \le \|f_N''(\theta_N^{+})^{-1}\| \|f_N'(\theta_{*})\| = O_{P_0}(N^{-1/2}).
 \end{equation}
This result matches Theorem 4 in \citet{Barp2019-ut}.
By Taylor's theorem,
 \begin{align}
 	f_N(\theta_{*}) - f_N(\theta_N) &= f_N'(\theta_N)^\top (\theta_{*} - \theta_N) + \frac{1}{2}(\theta_{*} - \theta_N)^\top f_N''(\theta_N^{++})(\theta_{*} - \theta_N)\notag\\
&= \frac{1}{2}(\theta_{*} - \theta_N)^\top f_N''(\theta_N^{++})(\theta_{*} - \theta_N) \notag
 \end{align}
for all $N$ sufficiently large, where $\theta_N^{++}$ is on the line between $\theta_N$ and $\theta_{*}$.
Therefore, using the same reasoning as for Equations~\ref{eqn:fpp_N_scaling} and \ref{eqn:theta_N_star_taylor_scaling},
\begin{equation}
	|f_N(\theta_{*}) - f_N(\theta_N)| \le \frac{1}{2} \|f_N''(\theta_N^{++})\| \|\theta_{*}-\theta_N\|^2 = O_{P_0}(N^{-1}).
\end{equation}
This proves the first part of the theorem (Equation~\ref{eqn:fN_O1}).
Next, consider $f_N(\theta_{*}) - f(\theta_{*})$. Recall that
\begin{equation*}
	f_N(\theta_{*}) = \frac{1}{T}\frac{\frac{1}{N(N-1)} \sum_{i \neq j} u_{\theta_{*}}(X^{(i)}, X^{(j)})}{\frac{1}{N(N-1)} \sum_{i \neq j} k(X^{(i)}, X^{(j)})}.
\end{equation*}
It is straightforward to verify that $\mathbb{E}_{X, Y \sim P_0}[| u_{\theta_{*}}(X, Y)|^2] < \infty$ 
since $s_{q_{\theta_{*}}}$ is in $L^2(P_0)$.
By Theorems 5.5.1A and 5.5.2 of \citet{Serfling2009-dq},
\begin{equation*}
	\frac{1}{N(N-1)} \sum_{i \neq j} u_{\theta_{*}}(X^{(i)}, X^{(j)}) - \mathbb{E}_{X, Y \sim P_0}[u_{\theta_{*}}(X, Y)] = O_{P_0}(N^{-1/2}).
\end{equation*}
Similarly, since $k$ is bounded, 
\begin{equation*}
	\frac{1}{N(N-1)} \sum_{i \neq j} k(X^{(i)}, X^{(j)}) - \mathbb{E}_{X, Y \sim P_0}[k(X, Y)] = O_{P_0}(N^{-1/2}).
\end{equation*}
It is straightforward to check that the second part of the theorem (Equation~\ref{eqn:fN_OsN}) follows.

For the third part, our argument follows that of the proof of Theorem 4.1 of \citet{Liu2016-bp}. 
Suppose $\textsc{nksd}(p_0(x)\| q(x|\theta_{*})) = 0$, 
and note that $P_0(x) = Q_{\theta_{*}}(x)$ by Proposition~\ref{thm:nksd_divergence}.
Given a differentiable function $g:\mathbb{R}^d\to \mathbb{R}^d$, define $\nabla_x^\top g(x) := \sum_{i=1}^d \partial g_i(x) / \partial x_i$.
Then
\begin{align}
\label{eqn:u_degenerate}
\mathbb{E}_{X \sim P_0}[u_{\theta_{*}}(X, y)] &=
s_{p_0}(y)^\top \int_\mathcal{X} \Big((\nabla_x p_0(x)) k(x, y) + p_0(x)(\nabla_x k(x,y)) \Big) d x \notag \\
& ~~~ + \int_\mathcal{X} \Big((\nabla_x p_0(x))^\top \nabla_y k(x,y) + p_0(x) (\nabla_x^\top \nabla_y k(x,y))\Big) dx \notag \\
&= s_{p_0}(y)^\top \int_\mathcal{X} \nabla_x \big(p_0(x)k(x,y)\big) dx + \int_\X \nabla_x^\top \nabla_y  (p_0(x) k(x, y)) dx.
\end{align}
The first term on the right-hand side of Equation~\ref{eqn:u_degenerate} is zero since, by assumption, $k$ is in the Stein class of $P_0$ (Condition~\ref{condition:restrict_k}).
The second term is also zero since, by the Leibniz integral rule \citep[Theorem 2.27]{folland1999real}, 
$\int \nabla_y^\top \nabla_x  (p_0(x) k(x, y)) d x = \nabla_y^\top \int \nabla_x  (p_0(x) k(x, y)) d x$, 
which again equals zero because $k$ is in the Stein class of $P_0$.
Therefore, $\mathbb{E}_{X \sim P_0}[u_{\theta_{*}}(X, y)] = 0$ for all $y\in\mathcal{X}$, and in particular, 
the variance of this expression is also zero: $\mathbb{V}_{Y \sim P_0}[\mathbb{E}_{X \sim P_0}[u_{\theta_{*}}(X, Y)]] = 0$. 
By Theorem 5.5.2 of \citet{Serfling2009-dq}, it follows that
\begin{equation}
\label{eqn:u_OPN}
	\frac{1}{N(N-1)} \sum_{i \neq j} u_{\theta_{*}}(X^{(i)}, X^{(j)}) = O_{P_0}(N^{-1})
\end{equation}
since $\mathbb{E}_{X, Y \sim P_0}[u_{\theta_{*}}(X, Y)] = 0$.
Although \citet{Serfling2009-dq} requires $\mathbb{V}_{X,Y \sim P_0}[u_{\theta_{*}}(X, Y)] > 0$, 
Equation~\ref{eqn:u_OPN} holds trivially if $\mathbb{V}_{X,Y \sim P_0}[u_{\theta_{*}}(X, Y)] = 0$.
As before, since the denominator of $f_N(\theta_{*})$ converges a.s.\ to a finite positive constant, we have that $f_N(\theta_{*}) = O_{P_0}(N^{-1})$.
Equation~\ref{eqn:fN_O1_thetastar} follows since $f(\theta_{*}) = 0$ when $\textsc{nksd}(p_0(x)\| q(x|\theta_{*})) = 0$.
\end{proof}

\subsection{Proof of Theorem~\ref{thm:selection_consistency}} \label{sec:si_proof_selection}
% \subsubsection*{Proof of Theorem~\ref{thm:selection_consistency}}

\begin{proof}
Applying Theorem~\ref{thm:marginal_ksd} (part 3) to each foreground model $j\in\{1,2\}$, we have
\begin{align*}
\log z_{j,N} + N f_{j,N}(\theta_{j,N}) - \log\pi(\theta_{j,*}) + \log|\det f_j''(\theta_{j,*})|^{1/2} - \frac{1}{2}m_{\F_j,j}\log(2\pi/N)
\xrightarrow[N \to \infty]{\textup{a.s.}} 0.
\end{align*}
Since $\mathcal{K}_{j,N} = (2\pi/N)^{m_{\B_j}/2} z_{j,N}$, this implies
\begin{align*}
\log \mathcal{K}_{j,N} + N f_{j,N}(\theta_{j,N}) - \frac{1}{2}(m_{\F_j,j}+m_{\B_j})\log(2\pi/N) + C_j
\xrightarrow[N \to \infty]{\textup{a.s.}} 0
\end{align*}
where $C_j$ is a constant that does not depend on $N$.  Hence,
\begin{align} \label{eqn:logKratio}
\log \frac{\mathcal{K}_{1,N}}{\mathcal{K}_{2,N}} &+ N (f_{1,N}(\theta_{1,N}) - f_{2,N}(\theta_{2,N})) \notag\\
& - \frac{1}{2}(m_{\F_1,1} + m_{\B_1} - m_{\F_2,2} - m_{\B_2})\log(2\pi/N) + C_1 - C_2
\xrightarrow[N \to \infty]{\textup{a.s.}} 0.
\end{align}
By Theorem~\ref{prop:laplace_scaling}, $f_{j,N}(\theta_{j,N}) \xrightarrow[]{P_0} f_j(\theta_{j,*})$, and therefore,
\begin{align*}
\frac{1}{N}\log \frac{\mathcal{K}_{1,N}}{\mathcal{K}_{2,N}} + f_1(\theta_{1,*}) - f_2(\theta_{2,*})
\xrightarrow[N \to \infty]{P_0} 0.
\end{align*}
Plugging in the definition of $f_j$ (Equation~\ref{eqn:f}), this proves part 1 of the theorem.

For part 2, suppose $f_1(\theta_{1,*}) = f_2(\theta_{2,*}) = 0$ and $m_{\B_2} - m_{\B_1}$ does not depend on $N$.
Then by Theorem~\ref{prop:laplace_scaling}, $f_{j,N}(\theta_{j,N}) = O_{P_0}(N^{-1})$.
Using this in Equation~\ref{eqn:logKratio}, we have
\begin{align}
\frac{1}{\log N}\log \frac{\mathcal{K}_{1,N}}{\mathcal{K}_{2,N}} 
+ \frac{1}{2}(m_{\F_1,1} + m_{\B_1} - m_{F_2,2} - m_{\B_2}) \xrightarrow[N \to \infty]{P_0} 0.
\end{align}

For part 3, suppose $f_1(\theta_{1,*}) = f_2(\theta_{2,*})$ and $m_{\B_j} = c_{\B_j}\sqrt{N}$.
Then by Theorem~\ref{prop:laplace_scaling}, $f_{j,N}(\theta_{j,N}) = f_j(\theta_{j,*}) + O_{P_0}(N^{-1/2})$.
Using this in Equation~\ref{eqn:logKratio}, we have
\begin{align}
\frac{1}{\sqrt{N}\log N}\log \frac{\mathcal{K}_{1,N}}{\mathcal{K}_{2,N}} 
+ \frac{1}{2}(c_{\B_1} - c_{\B_2}) \xrightarrow[N \to \infty]{P_0} 0.
\end{align}
\end{proof}

\section{Additional probabilistic PCA details} \label{sec:si_ppca_details}

\subsection{Optimizing the NKSD}

Computing the Laplace or BIC approximation to the SVC requires finding the minimizer of $\widehat{\textsc{nksd}}(p_0(x)\|q(x|\theta))$ with respect to $\theta$. In this section, we describe how components of the \textsc{nksd} can be pre-computed to speed up this optimization process.
The generative model used for pPCA can be rewritten using the properties of multivariate normal distributions as
\begin{equation}
	X \sim \mathcal{N}(0, H H^\top + v I_d).
\label{eqn:ppca_mvn}
\end{equation}
The Stein score function for the pPCA model is then
\begin{equation*}
s_{q_\theta}(x) = \nabla_x \log q(x| H, v) = -(H H^\top + v I_d)^{-1} x.
\end{equation*}
Define the matrices
\begin{equation*}
\begin{split}
K_{i j} &:= \mathbb{I}(i \neq j)\, k(X^{(i)}, X^{(j)}),\\
\dot{K}_{j b} &:=  \sum_{i=1}^N \mathbb{I}(i \neq j)\, \frac{\partial k}{\partial x_b}(X^{(i)}, X^{(j)}),
\end{split}
\end{equation*}
where $\mathbb{I}(E)$ is the indicator function, which equals 1 when $E$ is true and is 0 otherwise. Define the scalars
\begin{equation*}
\begin{split}
\bar{K} &:= \sum_{i,j=1}^N K_{i j},\\
\ddot{K} &:= \sum_{i,j = 1}^N \sum_{b=1}^d \mathbb{I}(i \neq j)\, \frac{\partial^2 k}{\partial x_b \partial y_b}(X^{(i)}, X^{(j)}).
\end{split}
\end{equation*}
Letting $X \in \mathbb{R}^{N \times d}$ be the data matrix, the NKSD can be written as
\begin{equation*}
\begin{split}
	\widehat{\textsc{nksd}}(p_0(x)\|q(x|H, v)) = \frac{1}{\bar{K}} \big[&\Tr(X^\top K X (H H^\top + v I_d)^{-1} (H H^\top + v I_d)^{-1}) \\
& -2 \,\Tr(X^\top \dot{K} (H H^\top + v I_d)^{-1}) + \ddot{K}\big],
\end{split}
\end{equation*}
where we have used the fact that the kernel is symmetric. The terms $X^\top K X$ and $X^\top \dot{K}$ are the only ones that include sums over the entire dataset; these can be pre-computed, before optimizing the parameters $H$ and $v$.

To compute the matrix inversion $(H H^\top + v I_d)^{-1}$ we follow the strategy of \citet{Minka2000-uw},
\begin{equation*}
\begin{split}
(H H^\top + v I_d)^{-1} - v^{-1} I_d &= (H H^\top + v I_d)^{-1} (I_d - v^{-1}(H H^\top + v I_d))\\
&= - (H H^\top + v I_d)^{-1} H H^\top v^{-1}\\
&= - (U (L - v I_k) U^\top + v I_d)^{-1} U (L - v I_k) U^\top v^{-1}.
\end{split} \notag
\end{equation*}
Thus, applying the Woodbury matrix identity and using $I_d U = U = U I_k I_k = U I_k U^\top U$,
\begin{align*}
(H H^\top + v I_d)^{-1} - v^{-1} I_d &= - \big[v^{-1} I_d - v^{-2} U \big((L - v I_k)^{-1} + v^{-1}\big)^{-1} U^\top\big] U (L - v I_k) U^\top v^{-1}\\
&= - U [v^{-1} I_k - v^{-2} ((L - v)^{-1} + v^{-1})^{-1}] (L - v I_k) U^\top v^{-1}\\
&= - U L^{-1} (L - v I_k) U^\top v^{-1}\\
&= U (L^{-1} - v^{-1} I_k) U^\top.
\end{align*}
Therefore,
\begin{equation*}
    (H H^\top + v I_d)^{-1} = U (L^{-1} - v^{-1} I_k) U^\top + v^{-1} I_d.
\end{equation*}
Computing $L^{-1}$ is trivial since the matrix is diagonal. Returning to the \textsc{nksd} we have
\begin{equation*}
\begin{split}
    \widehat{\textsc{nksd}}(p_0(x)\|q&(x|U, L, v))\\
    = \frac{1}{\bar{K}} \bigg[&\Tr\big(X^\top K X [U(L^{-1} - v^{-1}I_k)^2 U^\top + 2 v^{-1} U (L^{-1} - v^{-1} I_k) U^\top + v^{-2} I_d]\big)\\
    & - 2\,\Tr\big(X^\top \dot{K} [U (L^{-1} - v^{-1} I_k) U^\top + v^{-1} I_d]\big) + \ddot{K}\bigg]\\
    = \frac{1}{\bar{K}} \bigg[& \Tr\big(U^\top X^\top K X U (L^{-1} - v^{-1} I_k)^2\big)\\
    & + \Tr\big(U^\top [ 2 v^{-1} X^\top K X - 2 X^\top \dot{K}] U (L^{-1} - v^{-1} I_k)\big)\\
    & + v^{-1} \Tr\big(v^{-1} X^\top K X - 2 X^\top \dot{K}\big) + \ddot{K}\bigg].
\end{split}
\end{equation*}
We optimized $U$, $L$ and $v$ using the trust region method implemented in pymanopt~\citep{Townsend2016-qk}.

\subsection{Data selection with the SVC} \label{sec:ppca_data_selection_details}

We used the approximate optimum technique in Section~\ref{sec:approximate_optima} to estimate the SVC for different foreground subspaces. Following Section~\ref{sec:si_kernel_choice}, we used the factored IMQ kernel with $\beta = -0.5$ and $c = 1$.

We focused on foreground subspaces that correspond to subsets of the data dimensions. More specifically, recall that $X_\F = V^\top X$; then, we impose the restriction that each column of $V$ is a standard basis vector $e^{(b)} \in \mathbb{R}^d$, where $e^{(b)}_b = 1$ and $e^{(b)}_{b'} = 0$ for $b'\neq b$.
A subspace $\X_\F$ is then characterized by the set of included dimensions $S_\F \subseteq \{1, \ldots, d\}$. 
The marginal distribution of the model $q(x_\F|H, v)$ is now straightforward to compute based on Equation~\ref{eqn:ppca_mvn} and the properties of multivariate normals:
\begin{equation*}
	X_\F \sim \mathcal{N}(0, H_{S_\F} H_{S_\F}^\top + v I_{|S_\F|})
\end{equation*}
where $H_{S_\F}$ is the submatrix consisting of rows of $H$ indexed by $S_\F$, and $|S_\F|$ is the size of the set $S_\F$. 

In the projected model, some of the parameters are nuisance variables with no contribution to the likelihood. Since the dimension of a $d \times k$ matrix on the Stiefel manifold is $d k - k (k+1)/2$, the total dimension of the foreground model (including contributions from parameters $U$, $L$ and $v$) is $m_\F = |S_\F| k - k(k+1)/2 + k + 1$, assuming $|S_\F| \ge k$.

Code is available at \url{https://github.com/EWeinstein/data-selection}.

\subsection{Calibration} \label{sec:si_ppca_calibration}

The $T$ hyperparameter was calibrated as in Section~\ref{sec:calibrate_t}. In detail, we sampled 10 independent true parameter values from the prior, with $\alpha = 1$ and $d = 6$. (We used a slightly less disperse prior than during inference, where we set $\alpha = 0.1$, to avoid numerical instabilities in the $\hat{T}$ estimate.) Then, for each of the true parameter values, we simulated $N = 2000$ datapoints. For each simulated true parameter value, we tracked the trend in the $\hat{T}$ estimator (Equation~\ref{eqn:T_estimator}) with increasing $N$ (Figure~\ref{fig:ppca_calibration}). The median estimated $T$ value at $N = 2000$ was $0.052$ across the 10 runs.

\begin{figure}[t!]
	\centering
	\includegraphics[width=0.5\textwidth]{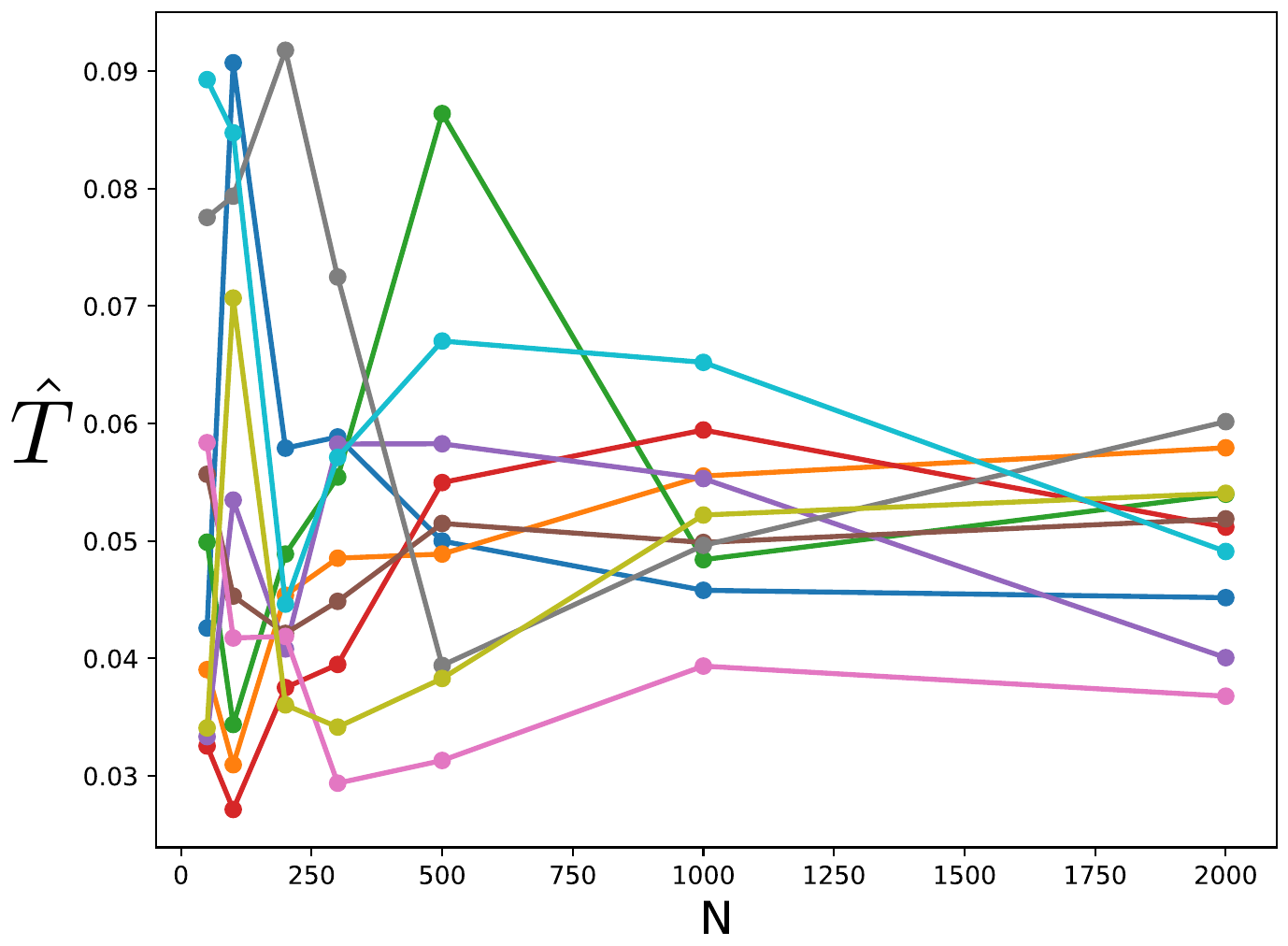}
	\caption{Estimated $T$ for increasing number of data samples, for 10 independent parameter samples from the prior. The median value at $N = 2000$ is $\hat{T} = 0.052$.}
	\label{fig:ppca_calibration}
\end{figure}

\subsection{ P\'{o}lya tree model} \label{sec:si_pt}

In this section, we describe the P\'{o}lya tree model~\citep{Ferguson1974-qn,Mauldin1992-pd,Lavine1992-fu} following the construction of \citet{Berger2001-li}.
Let $\munderbar{\epsilon}_n := (\epsilon_1, \ldots, \epsilon_n)$ denote a vector of length $n$, where each $\epsilon_j \in \{0, 1\}$. Each $\munderbar{\epsilon}_n$ vector indexes an interval in $\mathbb{R}$, given by
\begin{equation*}
   \textstyle B_{\munderbar{\epsilon}_n} := \Big( \tilde{F}^{-1} \big(\sum_{j=1}^n \epsilon_j/2^j\big),\; \tilde{F}^{-1}\big(\sum_{j=1}^n \epsilon_j/2^j + 1/2^n\big) \Big],
\end{equation*}
where $\tilde{F}^{-1}$ is the inverse c.d.f.\ of some probability distribution. 
For all $n\in\{0,1,2,\ldots\}$ and all $\munderbar{\epsilon}_n\in\{0,1\}^n$, let
\begin{equation*}
	Y_{\munderbar{\epsilon}_n} \sim \mathrm{Beta}(\xi_{\munderbar{\epsilon}_n 0}, \xi_{\munderbar{\epsilon}_n 1}),
\end{equation*}
where the $\xi$'s are hyperparameters. We say that a random variable $X\in\mathbb{R}$ is distributed according to a P\'{o}lya tree model if
\begin{equation*}
	P(X \in B_{\munderbar{\epsilon}_n}) = \prod_{j=1}^n (Y_{\munderbar{\epsilon}_{j-1}})^{\mathbb{I}(\epsilon_j = 0)} (1 - Y_{\munderbar{\epsilon}_{j-1}})^{\mathbb{I}(\epsilon_j = 1)},
\end{equation*}
where $\mathbb{I}(E)$ is the indicator function, which equals 1 when $E$ is true and is 0 otherwise. We follow \citet{Berger2001-li} and use
\begin{equation*}
\begin{split}
\mu (B_{\munderbar{\epsilon}_n}) &:= {\textstyle F\big(\tilde{F}^{-1}\big(\sum_{j=1}^n \epsilon_j/2^j + 1/2^n\big)\big) - F\big(\tilde{F}^{-1}\big(\sum_{j=1}^n \epsilon_j/2^j\big)\big)},\\
\rho(\munderbar{\epsilon}_n) &:= \frac{1}{\eta} \bigg(\frac{f(\tilde{F}^{-1}(\sum_{j=1}^n \epsilon_j/2^j + 1/2^{n+1}))}{\mu(B_{\munderbar{\epsilon}_n})}\bigg)^2,\\
\xi_{\munderbar{\epsilon}_n 0} &:= \rho(\munderbar{\epsilon}_n) \sqrt{\frac{\mu (B_{\munderbar{\epsilon}_n 0})}{\mu (B_{\munderbar{\epsilon}_n 1})}},\\
\xi_{\munderbar{\epsilon}_n 1} &:= \rho(\munderbar{\epsilon}_n) \sqrt{\frac{\mu (B_{\munderbar{\epsilon}_n 1})}{\mu (B_{\munderbar{\epsilon}_n 0})}},
\end{split}
\end{equation*}
where $F$ and $f$ are the c.d.f.\ and p.d.f.\ respectively of some probability distribution, and $\eta > 0$ is a scale hyperparameter. We denote this complete model as $X \sim \mathrm{PolyaTree}(F, \tilde{F}, \eta)$. 

\subsection{Datasets and preprocessing} \label{sec:si_ppca_datasets}

We downloaded two publicly available datasets. The first dataset was taken from human peripheral blood mononuclear cells (PBMCs): \url{https://support.10xgenomics.com/single-cell-gene-expression/datasets/1.1.0/pbmc3k}.
This is a standard dataset used in the tutorials for Seurat~\citep{Stuart2019-gc} and Scanpy \citep{Wolf2018-xz}, for example.
The second was taken from a dissociated extranodal marginal zone B-cell tumor, specifically a mucosa-associated lymphoid tissue (MALT) tumor: \url{https://support.10xgenomics.com/single-cell-gene-expression/datasets/3.0.0/malt_10k_protein_v3}.
 
We pre-processed the data using Scprep \citep{Gigante2020-sc}, following its example: we normalized the total expression of each cell to match the median total expression in the dataset, to account for variability in library size, and then square-root transformed the resulting normalized counts.

\section{Additional glass model details} \label{sec:si_glass}

\subsection{Glass model inference} \label{sec:si_glass_inference}

We place a standard normal prior on each entry of $H_j$ and a Laplace prior on each entry of $J_{j j'}$ with scale $0.1$ to encourage sparsity. To enforce that $\mu \geq 0$ (since scRNAseq counts are nonnegative) and $\tau > 0$, we place priors on a transformed version of these parameters, as follows:
\begin{align*}
\begin{split}
	\tilde{\mu} &\sim \mathcal{N}(0, 1)\\
	\mu &= \log(1 + \exp(\tilde{\mu}))\\
	\tilde{\tau} &\sim \mathcal{N}(0, 1)\\
	\tau &= \log(1 + \exp(\tilde{\tau})) + 1.
\end{split}
\end{align*}
For posterior inference, we employ a mean-field variational approximation: independent normal distributions for the entries of $H_j$, normal distributions for $\tilde{\mu}$ and $\tilde{\tau}$, and Laplace distributions for each entry of $J_{j j'}$.
We use the factored IMQ kernel for the NKSD, with $\beta = -0.5$ and $c = 1$.

To optimize the variational approximation (Equation~\ref{eqn:svc_vi}), we construct stochastic estimates of its gradient. 
At each optimization step, the expectation $\mathbb{E}_{r_\zeta}\left[\widehat{\textsc{nksd}}(p_0(x_\F)\| q(x_\F|\theta))\right]$ is estimated using a minibatch of 200 randomly selected datapoints and a single sample from the variational approximation $r_\zeta$.
The rest of the variational inference algorithm follows standard practice in stochastic variational inference, as implemented in Pyro: automatic differentiation to compute gradients, reparameterization estimators for Monte Carlo expectations over the variational distribution, and the Adam optimizer~\citep{Kingma2015-ej,Bingham2019-aa}. 

We also used stochastic optimization to perform data selection, as follows. Let $I = (I_1, \ldots, I_d)^\top$ be an indicator variable that specifies for each gene $j$ whether it is included in the foreground subspace ($I_j = 1$) or not ($I_j = 0$).
We place a distribution on $I$ such that $I_j \sim \mathrm{Bernoulli}(1/(1 + \exp(-\phi_j)))$ for $j=1,\ldots,d$ independently.
Then, to perform data selection over all possible subsets of genes, we optimize
\begin{equation}
	\mathrm{argmax}_{\phi}\, \mathbb{E}(\mathcal{K}(I) \mid \phi)
\end{equation}
where the expectation is taken with respect to $I$, where $\mathcal{K}(I)$ is the (estimated) SVC when genes with $I_j = 1$ are included in the foreground space, and $\phi = (\phi_1, \ldots, \phi_d)^\top\in \mathbb{R}^d$ is a vector of log-odds.
This stochastic approach to discrete optimization has been used extensively in reinforcement learning and related fields.
We use the Leave-One-Out REINFORCE (LOORF) estimator as described in Section 2.1 of \citet{Dimitriev2021-to} to estimate gradients of $\phi$, using 8 samples per step. 

We interleave updates to the variational approximation and to $\phi$, using the Adam optimizer with step size 0.01 for each.
We ran the procedure with 4 random initial seeds, taking the result with the largest final estimated SVC.
We halt optimization using the stopping rule proposed in \citet{Grathwohl2020-pw}, stopping when the estimated mean minus the estimated variance of the SVC begins to decrease, based on the average over 2000 steps.

Code is available at \url{https://github.com/EWeinstein/data-selection}.

\subsection{Datasets and preprocessing} \label{sec:si_glass_datasets}

In addition to the two datasets in~\ref{sec:si_ppca_datasets}, we also explored a dataset of E18 mouse neurons: \url{https://support.10xgenomics.com/single-cell-gene-expression/datasets/3.0.0/neuron_10k_v3}.

We preprocessed each dataset using Scprep \citep{Gigante2020-sc} in the same way as in Section~\ref{sec:si_ppca_datasets}. After preprocessing, we used the top 200 most highly expressed genes from among the top 500 most variable genes, according to the Scprep variability score.
We log transform the counts, that is we define $x_{i j} = \log(1 + c_{i j})$ where $c_{i j}$ is the expression count for gene $j$ in cell $i$.

\begin{figure}[t!]
    \centering
    \includegraphics[height=3in]{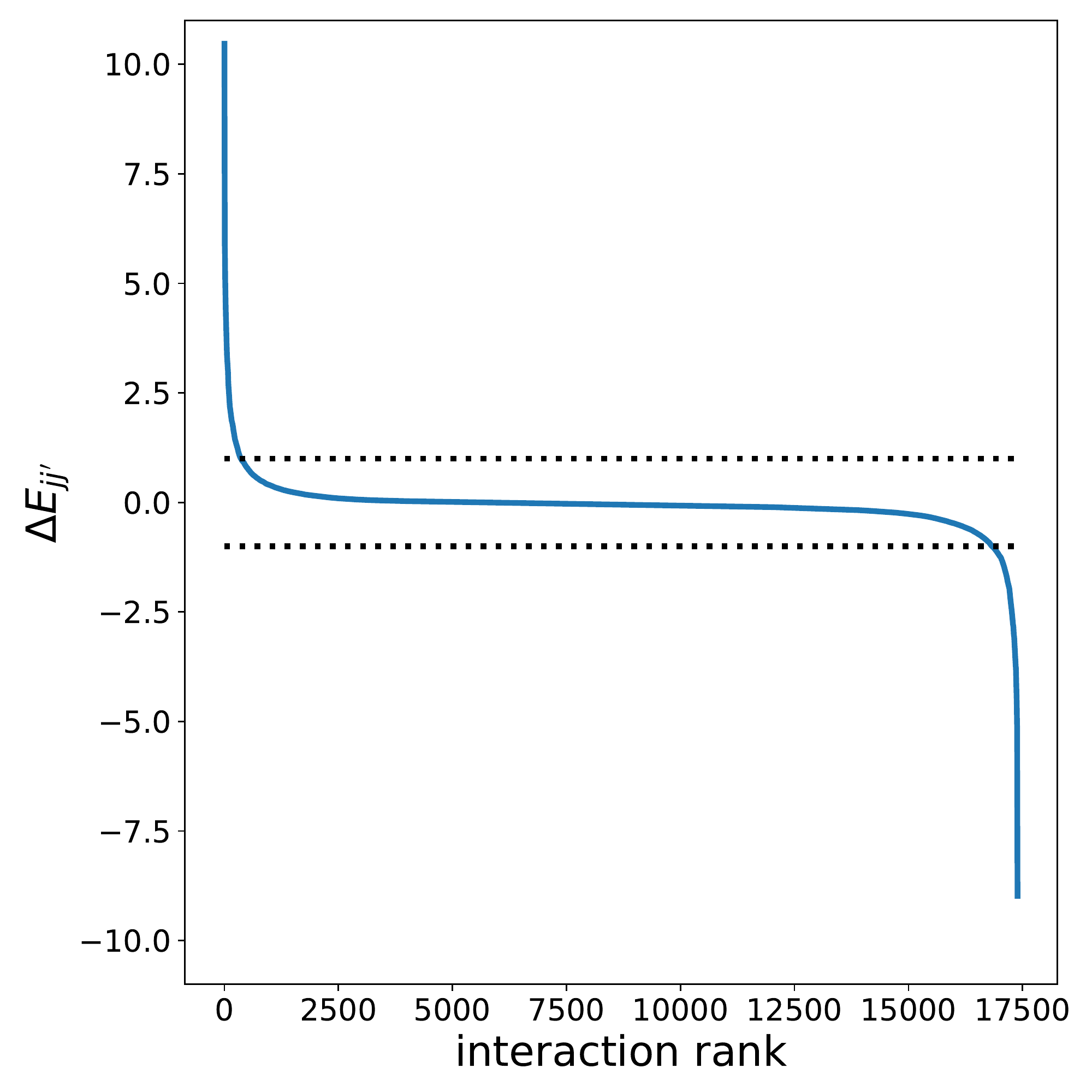}
    \caption{Posterior mean interaction energies $\Delta E_{j j'}$ for all selected genes, sorted. Dotted lines show the thresholds for strong interactions (set by visual inspection).}
    \label{fig:DeltaEjjp_ranks_selected}
\end{figure}

\begin{figure}[t!]
    \centering
    \includegraphics[height=6.2in]{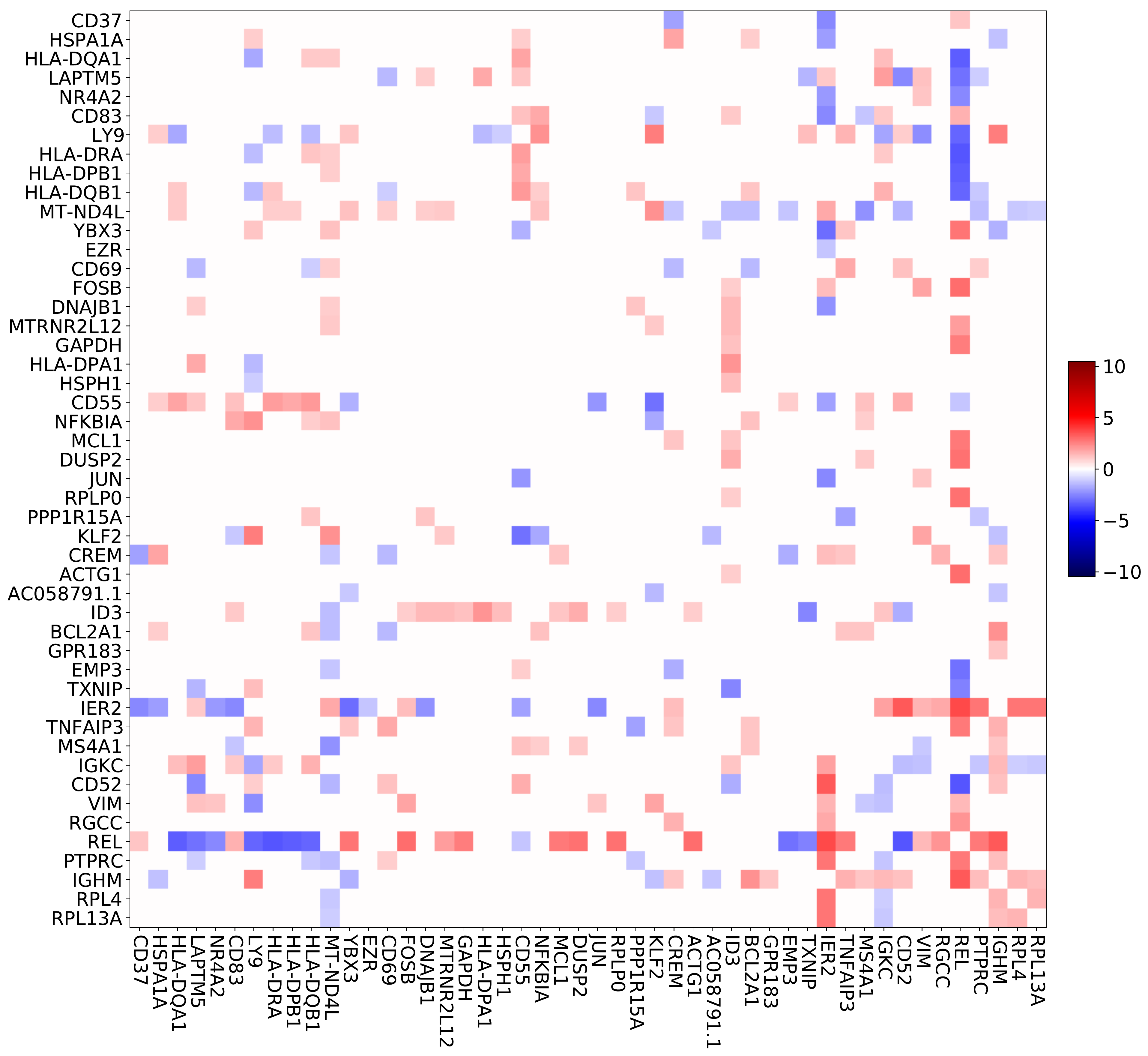}
    \caption{Posterior mean interaction energies $\Delta E_{j j'}$ for the glass model applied to all 200 genes in the MALT dataset (rather than the selected 187). Genes shown are the same as in Figure~\ref{fig:interaction_map_selected}, for visual comparison.}
    \label{fig:interaction_map_all}
\end{figure}

\begin{figure}[t!]
\centering
\begin{subfigure}[t]{0.49\textwidth}
    \centering
    \includegraphics[height=3in]{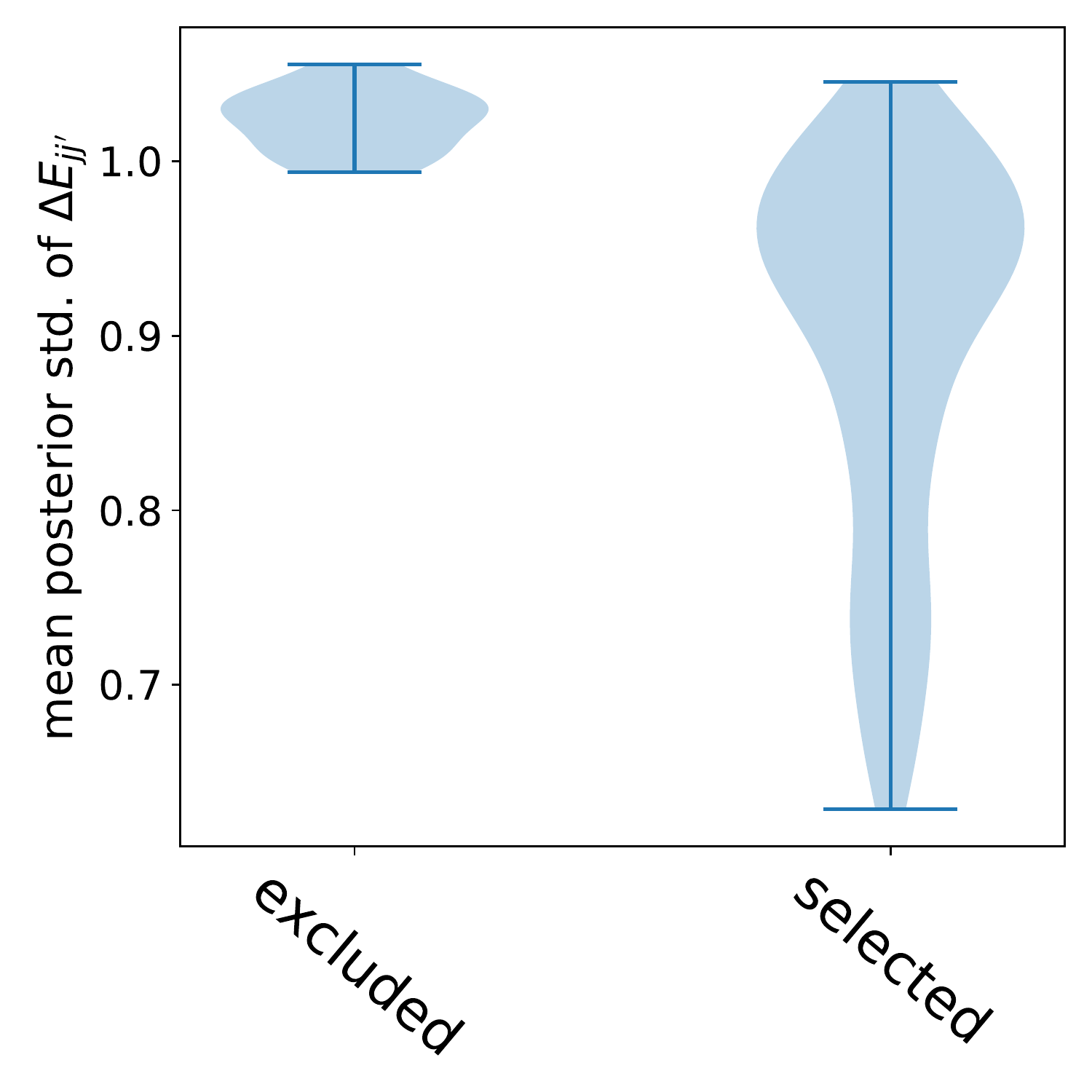}
    \caption{}
    \label{fig:interaction_uncertainty}
\end{subfigure}
\begin{subfigure}[t]{0.49\textwidth}
    \centering
    \includegraphics[height=3in]{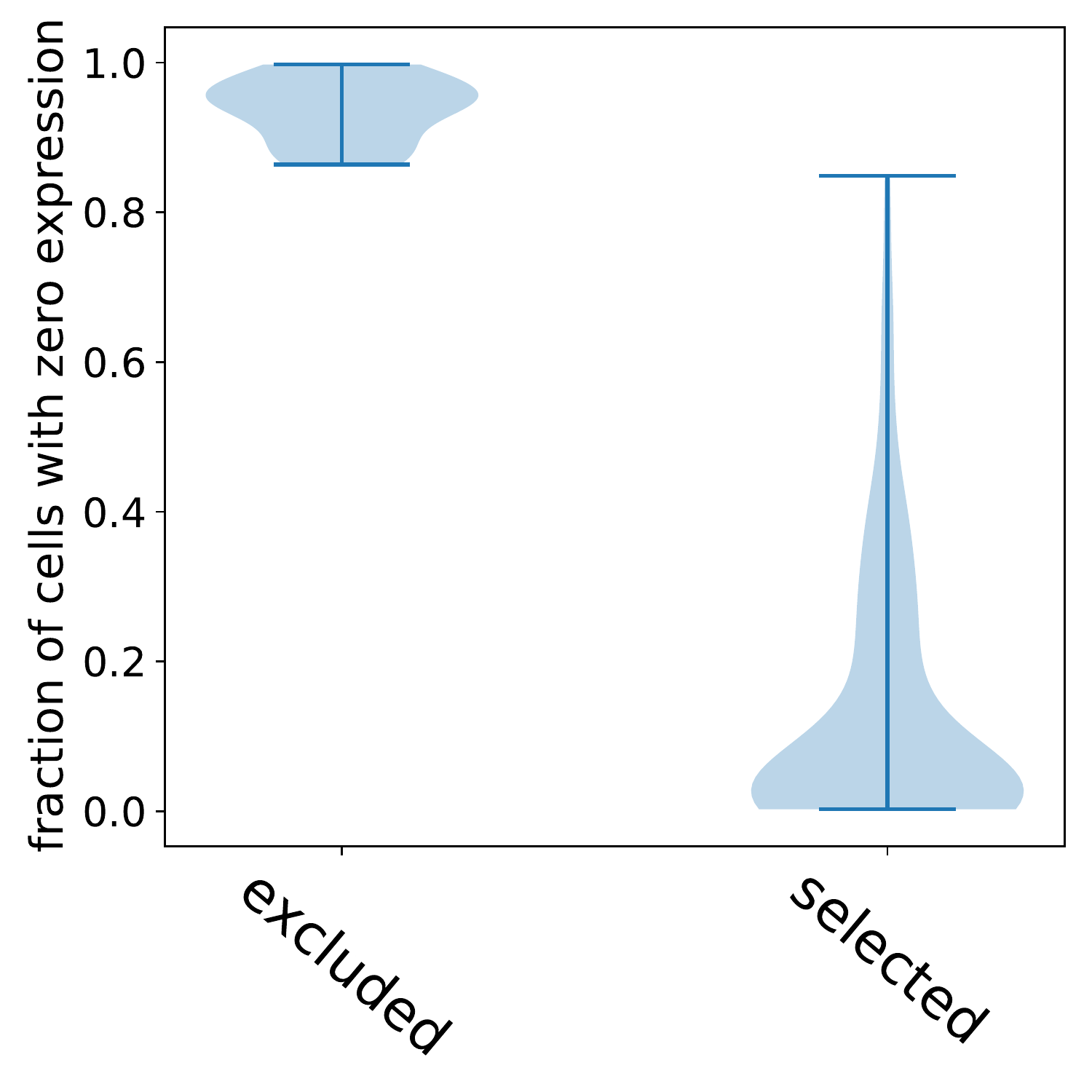}
    \caption{}
    \label{fig:fract_zero_selected}
\end{subfigure}
\caption{Comparison of the 187 selected genes and 13 excluded genes using data selection. (a) Violin plot of $\bar{\sigma}_j$ over all excluded and selected genes $j$, respectively, when applying the model to all 200 genes, where $\bar{\sigma}_j$ is the mean posterior standard deviation of the interaction energies $\Delta E_{j j'}$ for gene $j$, that is, $\bar{\sigma}_j := \frac{1}{d-1} \sum_{j' \neq j} \mathrm{std}(\Delta E_{j j'}\mid\mathrm{data})$.
(b) Violin plot of $f_j$ over all excluded and selected genes $j$, respectively, where $f_j$ is the fraction of cells with count equal to zero for gene $j$. The data selection procedure excluded all genes with more than 85\% zeros and selected all genes with fewer than 85\% zeros.}
\end{figure}

\newpage
\bibliographystyle{abbrvnat}
\small
\bibliography{references}

\begin{thebibliography}{83}
\providecommand{\natexlab}[1]{#1}
\providecommand{\url}[1]{\texttt{#1}}
\expandafter\ifx\csname urlstyle\endcsname\relax
  \providecommand{\doi}[1]{doi: #1}\else
  \providecommand{\doi}{doi: \begingroup \urlstyle{rm}\Url}\fi

\bibitem[Alon(2019)]{Alon2019-ni}
U.~Alon.
\newblock \emph{An Introduction to Systems Biology: Design Principles of
  Biological Circuits}.
\newblock CRC Press, July 2019.

\bibitem[Anastasiou et~al.(2021)Anastasiou, Barp, Briol, Ebner, Gaunt,
  Ghaderinezhad, Gorham, Gretton, Ley, Liu, Mackey, Oates, Reinert, and
  Swan]{Anastasiou2021-kf}
A.~Anastasiou, A.~Barp, F.-X. Briol, B.~Ebner, R.~E. Gaunt, F.~Ghaderinezhad,
  J.~Gorham, A.~Gretton, C.~Ley, Q.~Liu, L.~Mackey, C.~J. Oates, G.~Reinert,
  and Y.~Swan.
\newblock Stein's method meets statistics: {A} review of some recent
  developments.
\newblock May 2021.

\bibitem[Banerjee et~al.(2008)Banerjee, Ghaoui, and
  d'Aspremont]{Banerjee2008-nh}
O.~Banerjee, L.~E. Ghaoui, and A.~d'Aspremont.
\newblock Model selection through sparse maximum likelihood estimation for
  multivariate {Gaussian} or binary data.
\newblock \emph{Journal of Machine Learning Research}, 9\penalty0
  (Mar):\penalty0 485--516, 2008.

\bibitem[Barp et~al.(2019)Barp, Briol, Duncan, Girolami, and
  Mackey]{Barp2019-ut}
A.~Barp, F.-X. Briol, A.~B. Duncan, M.~Girolami, and L.~Mackey.
\newblock Minimum {S}tein discrepancy estimators.
\newblock \emph{arXiv preprint arXiv:1906.08283}, June 2019.

\bibitem[Barron(1989)]{Barron1989-jg}
A.~R. Barron.
\newblock Uniformly powerful goodness of fit tests.
\newblock \emph{The Annals of Statistics}, 17\penalty0 (1):\penalty0 107--124,
  1989.

\bibitem[Baydin et~al.(2018)Baydin, Pearlmutter, Radul, and
  Siskind]{Baydin2018-qp}
A.~G. Baydin, B.~A. Pearlmutter, A.~A. Radul, and J.~M. Siskind.
\newblock Automatic differentiation in machine learning: {A} survey.
\newblock \emph{Journal of Machine Learning Research}, 18\penalty0 (153), 2018.

\bibitem[Berger and Guglielmi(2001)]{Berger2001-li}
J.~O. Berger and A.~Guglielmi.
\newblock Bayesian and conditional frequentist testing of a parametric model
  versus nonparametric alternatives.
\newblock \emph{Journal of the American Statistical Association}, 96\penalty0
  (453):\penalty0 174--184, 2001.

\bibitem[Bingham et~al.(2019)Bingham, Chen, Jankowiak, Obermeyer, Pradhan,
  Karaletsos, Singh, Szerlip, Horsfall, and Goodman]{Bingham2019-aa}
E.~Bingham, J.~P. Chen, M.~Jankowiak, F.~Obermeyer, N.~Pradhan, T.~Karaletsos,
  R.~Singh, P.~Szerlip, P.~Horsfall, and N.~D. Goodman.
\newblock Pyro: Deep universal probabilistic programming.
\newblock \emph{Journal of Machine Learning Research}, 20\penalty0
  (28):\penalty0 1--6, 2019.

\bibitem[Bissiri et~al.(2016)Bissiri, Holmes, and Walker]{Bissiri2016-gk}
P.~G. Bissiri, C.~C. Holmes, and S.~G. Walker.
\newblock A general framework for updating belief distributions.
\newblock \emph{Journal of the Royal Statistical Society. Series B, Statistical
  Methodology}, 78\penalty0 (5):\penalty0 1103--1130, 2016.

\bibitem[Blei(2014)]{Blei2014-vn}
D.~M. Blei.
\newblock Build, compute, critique, repeat: {D}ata analysis with latent
  variable models.
\newblock \emph{Annual Review of Statistics and Its Application}, 1\penalty0
  (1):\penalty0 203--232, 2014.

\bibitem[Caticha(2004)]{Caticha2003-gb}
A.~Caticha.
\newblock Relative entropy and inductive inference.
\newblock \emph{AIP Conference Proceedings}, 707\penalty0 (1):\penalty0 75--96,
  2004.

\bibitem[Caticha(2011)]{Caticha2010-fm}
A.~Caticha.
\newblock Entropic inference.
\newblock \emph{AIP Conference Proceedings}, 1305\penalty0 (1):\penalty0
  20--29, 2011.

\bibitem[Chen et~al.(2015)Chen, Guo, Mishra, Robson, Niranjan, and
  Zheng]{Chen2015-ig}
H.~Chen, J.~Guo, S.~K. Mishra, P.~Robson, M.~Niranjan, and J.~Zheng.
\newblock Single-cell transcriptional analysis to uncover regulatory circuits
  driving cell fate decisions in early mouse development.
\newblock \emph{Bioinformatics}, 31\penalty0 (7):\penalty0 1060--1066, Apr.
  2015.

\bibitem[Chwialkowski et~al.(2016)Chwialkowski, Strathmann, and
  Gretton]{Chwialkowski2016-tk}
K.~Chwialkowski, H.~Strathmann, and A.~Gretton.
\newblock A kernel test of goodness of fit.
\newblock In \emph{International Conference on Machine Learning}, pages
  2606--2615, 2016.

\bibitem[Dawid(2011)]{Dawid2011-kb}
A.~P. Dawid.
\newblock Posterior model probabilities.
\newblock In P.~S. Bandyopadhyay and M.~R. Forster, editors, \emph{Philosophy
  of Statistics}, volume~7, pages 607--630. North-Holland, Amsterdam, Jan.
  2011.

\bibitem[de~Winde et~al.(2016)de~Winde, Veenbergen, Young, Xu-Monette, Wang,
  Xia, Jabbar, van~den Brand, van~der Schaaf, Elfrink, van Houdt, Gijbels,
  van~de Loo, Bennink, Hebeda, Groenen, van Krieken, Figdor, and van
  Spriel]{De_Winde2016-pk}
C.~M. de~Winde, S.~Veenbergen, K.~H. Young, Z.~Y. Xu-Monette, X.-X. Wang,
  Y.~Xia, K.~J. Jabbar, M.~van~den Brand, A.~van~der Schaaf, S.~Elfrink, I.~S.
  van Houdt, M.~J. Gijbels, F.~A.~J. van~de Loo, M.~B. Bennink, K.~M. Hebeda,
  P.~J. T.~A. Groenen, J.~H. van Krieken, C.~G. Figdor, and A.~B. van Spriel.
\newblock Tetraspanin {CD37} protects against the development of {B} cell
  lymphoma.
\newblock \emph{The Journal of Clinical Investigation}, 126\penalty0
  (2):\penalty0 653--666, Feb. 2016.

\bibitem[Dimitriev and Zhou(2021)]{Dimitriev2021-to}
A.~Dimitriev and M.~Zhou.
\newblock {ARMS}: {Antithetic-REINFORCE-Multi-Sample} gradient for binary
  variables.
\newblock In \emph{Proceedings of the 38th International Conference on Machine
  Learning}, 2021.

\bibitem[Ding and Peng(2005)]{Ding2005-bw}
C.~Ding and H.~Peng.
\newblock Minimum redundancy feature selection from microarray gene expression
  data.
\newblock \emph{Journal of Bioinformatics and Computational Biology},
  3\penalty0 (2):\penalty0 185--205, Apr. 2005.

\bibitem[Doksum and Lo(1990)]{Doksum1990-pd}
K.~A. Doksum and A.~Y. Lo.
\newblock Consistent and robust {B}ayes procedures for location based on
  partial information.
\newblock \emph{The Annals of Statistics}, 18\penalty0 (1):\penalty0 443--453,
  1990.

\bibitem[Duvenaud et~al.(2008)Duvenaud, Eaton, Murphy, and
  Schmidt]{Duvenaud2008-ir}
D.~Duvenaud, D.~Eaton, K.~Murphy, and M.~Schmidt.
\newblock Causal learning without {DAGs}.
\newblock In \emph{{NIPS} workshop on causality}. jmlr.org, 2008.

\bibitem[Ferguson(1974)]{Ferguson1974-qn}
T.~S. Ferguson.
\newblock Prior distributions on spaces of probability measures.
\newblock \emph{The Annals of Statistics}, 2\penalty0 (4):\penalty0 615--629,
  July 1974.

\bibitem[Folland(1999)]{folland1999real}
G.~B. Folland.
\newblock \emph{Real Analysis: Modern Techniques and Their Applications}.
\newblock John Wiley \& Sons, 1999.

\bibitem[Friedman(2004)]{Friedman2004-vp}
N.~Friedman.
\newblock Inferring cellular networks using probabilistic graphical models.
\newblock \emph{Science}, 303\penalty0 (5659):\penalty0 799--805, Feb. 2004.

\bibitem[Friedman et~al.(2000)Friedman, Linial, Nachman, and
  Pe'er]{Friedman2000-fd}
N.~Friedman, M.~Linial, I.~Nachman, and D.~Pe'er.
\newblock Using {B}ayesian networks to analyze expression data.
\newblock \emph{J. Comput. Biol.}, 7\penalty0 (3-4):\penalty0 601--620, 2000.

\bibitem[Gelman et~al.(2013)Gelman, Carlin, Stern, Dunson, Vehtari, and
  Rubin]{Gelman2013-al}
A.~Gelman, J.~B. Carlin, H.~S. Stern, D.~B. Dunson, A.~Vehtari, and D.~B.
  Rubin.
\newblock \emph{{B}ayesian Data Analysis}.
\newblock Chapman and Hall/CRC, 2013.

\bibitem[Ghosh and Ramamoorthi(2003)]{Ghosh2003-yv}
J.~K. Ghosh and R.~V. Ramamoorthi.
\newblock \emph{{B}ayesian Nonparametrics}.
\newblock Series in Statistics. Springer, 2003.

\bibitem[Gigante et~al.(2020)Gigante, Burkhardt, Dager, Stanley, and
  Tong]{Gigante2020-sc}
S.~Gigante, D.~Burkhardt, D.~Dager, J.~Stanley, and A.~Tong.
\newblock scprep.
\newblock \url{https://github.com/KrishnaswamyLab/scprep}, 2020.

\bibitem[Giordano et~al.(2019)Giordano, Stephenson, Liu, Jordan, and
  Broderick]{Giordano2018-xh}
R.~Giordano, W.~Stephenson, R.~Liu, M.~Jordan, and T.~Broderick.
\newblock A {S}wiss {A}rmy infinitesimal jackknife.
\newblock In \emph{The 22nd International Conference on Artificial Intelligence
  and Statistics}, pages 1139--1147. PMLR, 2019.

\bibitem[Gorham and Mackey(2017)]{Gorham2017-sd}
J.~Gorham and L.~Mackey.
\newblock Measuring sample quality with kernels.
\newblock In \emph{Proceedings of the 34th International Conference on Machine
  Learning - Volume 70}, pages 1292--1301, Sydney, NSW, Australia, 2017.

\bibitem[Grathwohl et~al.(2020)Grathwohl, Wang, Jacobsen, Duvenaud, and
  Zemel]{Grathwohl2020-pw}
W.~Grathwohl, K.-C. Wang, J.-H. Jacobsen, D.~Duvenaud, and R.~Zemel.
\newblock Learning the {S}tein discrepancy for training and evaluating
  energy-based models without sampling.
\newblock In \emph{Proceedings of the 37th International Conference on Machine
  Learning}, 2020.

\bibitem[Gretton et~al.(2012)Gretton, Borgwardt, Rasch, Sch{\"o}lkopf, and
  Smola]{Gretton2012-do}
A.~Gretton, K.~M. Borgwardt, M.~J. Rasch, B.~Sch{\"o}lkopf, and A.~Smola.
\newblock A kernel two-sample test.
\newblock \emph{Journal of Machine Learning Research}, 13:\penalty0 723--773,
  2012.

\bibitem[Gy{\"o}rfi and Van Der~Meulen(1991)]{Gyorfi1991-ki}
L.~Gy{\"o}rfi and E.~C. Van Der~Meulen.
\newblock A consistent goodness of fit test based on the total variation
  distance.
\newblock In G.~Roussas, editor, \emph{Nonparametric Functional Estimation and
  Related Topics}, pages 631--645. Springer Netherlands, Dordrecht, 1991.

\bibitem[Hicks et~al.(2018)Hicks, Townes, Teng, and Irizarry]{Hicks2018-zk}
S.~C. Hicks, F.~W. Townes, M.~Teng, and R.~A. Irizarry.
\newblock Missing data and technical variability in single-cell
  {RNA-sequencing} experiments.
\newblock \emph{Biostatistics}, 19\penalty0 (4):\penalty0 562--578, Oct. 2018.

\bibitem[Hoffman et~al.(2013)Hoffman, Blei, Wang, and Paisley]{Hoffman2013-cq}
M.~D. Hoffman, D.~M. Blei, C.~Wang, and J.~Paisley.
\newblock Stochastic variational inference.
\newblock \emph{Journal of Machine Learning Research}, 14:\penalty0 1303--1347,
  2013.

\bibitem[Hong and Preston(2005)]{Hong2005-kf}
H.~Hong and B.~Preston.
\newblock Nonnested model selection criteria.
\newblock 2005.

\bibitem[Huber(1985)]{Huber1985-xl}
P.~J. Huber.
\newblock Projection pursuit.
\newblock \emph{The Annals of Statistics}, 13\penalty0 (2):\penalty0 435--475,
  1985.

\bibitem[Huggins and Mackey(2018)]{Huggins2018-lr}
J.~H. Huggins and L.~Mackey.
\newblock Random feature {S}tein discrepancies.
\newblock \emph{arXiv preprint arXiv:1806.07788}, 2018.

\bibitem[Huggins and Miller(2021)]{Huggins2020-bu}
J.~H. Huggins and J.~W. Miller.
\newblock Reproducible model selection using bagged posteriors.
\newblock \emph{arXiv preprint arXiv:2007.14845}, 2021.

\bibitem[Huynh-Thu et~al.(2010)Huynh-Thu, Irrthum, Wehenkel, and
  Geurts]{Huynh-Thu2010-ee}
V.~A. Huynh-Thu, A.~Irrthum, L.~Wehenkel, and P.~Geurts.
\newblock Inferring regulatory networks from expression data using tree-based
  methods.
\newblock \emph{PLoS One}, 5\penalty0 (9), Sept. 2010.

\bibitem[Jacob et~al.(2017)Jacob, Murray, Holmes, and Robert]{Jacob2017-hu}
P.~E. Jacob, L.~M. Murray, C.~C. Holmes, and C.~P. Robert.
\newblock Better together? {S}tatistical learning in models made of modules.
\newblock \emph{arXiv preprint arXiv:1708.08719}, 2017.

\bibitem[Jewson et~al.(2018)Jewson, Smith, and Holmes]{Jewson2018-mw}
J.~Jewson, J.~Q. Smith, and C.~Holmes.
\newblock Principles of {B}ayesian inference using general divergence criteria.
\newblock \emph{Entropy}, 20\penalty0 (6):\penalty0 442, 2018.

\bibitem[Jiang and Tanner(2008)]{Jiang2008-zw}
W.~Jiang and M.~A. Tanner.
\newblock Gibbs posterior for variable selection in high-dimensional
  classification and data mining.
\newblock \emph{The Annals of Statistics}, 36\penalty0 (5):\penalty0
  2207--2231, 2008.

\bibitem[Kingma and Ba(2015)]{Kingma2015-ej}
D.~P. Kingma and J.~Ba.
\newblock Adam: {A} method for stochastic optimization.
\newblock In \emph{{ICLR}}, 2015.

\bibitem[Kingma and Welling(2014)]{Kingma2014-sp}
D.~P. Kingma and M.~Welling.
\newblock {Auto-encoding} variational {B}ayes.
\newblock In \emph{International Conference on Learning Representations
  (ICLR)}, Apr. 2014.

\bibitem[Knoblauch et~al.(2019)Knoblauch, Jewson, and
  Damoulas]{Knoblauch2019-nc}
J.~Knoblauch, J.~Jewson, and T.~Damoulas.
\newblock Generalized variational inference: {T}hree arguments for deriving new
  posteriors.
\newblock Apr. 2019.

\bibitem[Kool et~al.(2019)Kool, van Hoof, and Welling]{Kool2019-hl}
W.~Kool, H.~van Hoof, and M.~Welling.
\newblock Buy 4 {REINFORCE} samples, get a baseline for free.
\newblock In \emph{{ICLR} Workshop: Deep Reinforcement Learning Meets
  Structured Prediction}, 2019.

\bibitem[Kucukelbir et~al.(2017)Kucukelbir, Tran, Ranganath, Gelman, and
  Blei]{Kucukelbir2017-ir}
A.~Kucukelbir, D.~Tran, R.~Ranganath, A.~Gelman, and D.~M. Blei.
\newblock Automatic differentiation variational inference.
\newblock \emph{Journal of Machine Learning Research}, 18:\penalty0 1--45, Jan.
  2017.

\bibitem[Lavine(1992)]{Lavine1992-fu}
M.~Lavine.
\newblock Some aspects of {P}olya tree distributions for statistical modelling.
\newblock \emph{The Annals of Statistics}, 20\penalty0 (3):\penalty0
  1222--1235, 1992.

\bibitem[Lewis et~al.(2021)Lewis, MacEachern, and Lee]{Lewis2018-ri}
J.~R. Lewis, S.~N. MacEachern, and Y.~Lee.
\newblock Bayesian restricted likelihood methods: {C}onditioning on
  insufficient statistics in {B}ayesian regression.
\newblock \emph{Bayesian Analysis}, 1\penalty0 (1):\penalty0 1--38, 2021.

\bibitem[Lindsay(1988)]{Lindsay1988-vi}
B.~G. Lindsay.
\newblock Composite likelihood methods.
\newblock \emph{Contemporary Mathematics}, 80\penalty0 (1):\penalty0 221--239,
  1988.

\bibitem[Liu et~al.(2009)Liu, Lafferty, and Wasserman]{Liu2009-kx}
H.~Liu, J.~Lafferty, and L.~Wasserman.
\newblock The nonparanormal: Semiparametric estimation of high dimensional
  undirected graphs.
\newblock \emph{Journal of Machine Learning Research}, 10\penalty0
  (Oct):\penalty0 2295--2328, 2009.

\bibitem[Liu et~al.(2016)Liu, Lee, and Jordan]{Liu2016-bp}
Q.~Liu, J.~D. Lee, and M.~Jordan.
\newblock A kernelized stein discrepancy for goodness-of-fit tests.
\newblock In \emph{International Conference on Machine Learning}, volume~33,
  pages 276--284, 2016.

\bibitem[Matsubara et~al.(2021)Matsubara, Knoblauch, Briol, and
  Oates]{Matsubara2021-kz}
T.~Matsubara, J.~Knoblauch, F.-X. Briol, and C.~J. Oates.
\newblock Robust generalised bayesian inference for intractable likelihoods.
\newblock Apr. 2021.

\bibitem[Matsumoto et~al.(2017)Matsumoto, Kiryu, Furusawa, Ko, Ko, Gouda,
  Hayashi, and Nikaido]{Matsumoto2017-xi}
H.~Matsumoto, H.~Kiryu, C.~Furusawa, M.~S.~H. Ko, S.~B.~H. Ko, N.~Gouda,
  T.~Hayashi, and I.~Nikaido.
\newblock {SCODE}: an efficient regulatory network inference algorithm from
  single-cell {RNA-Seq} during differentiation.
\newblock \emph{Bioinformatics}, 33\penalty0 (15):\penalty0 2314--2321, Aug.
  2017.

\bibitem[Mauldin et~al.(1992)Mauldin, Sudderth, and Williams]{Mauldin1992-pd}
R.~D. Mauldin, W.~D. Sudderth, and S.~C. Williams.
\newblock Polya trees and random distributions.
\newblock \emph{The Annals of Statistics}, 20\penalty0 (3):\penalty0
  1203--1221, 1992.

\bibitem[Miller(2021)]{Miller2019-ur}
J.~W. Miller.
\newblock Asymptotic normality, concentration, and coverage of generalized
  posteriors.
\newblock \emph{Journal of Machine Learning Research}, 22\penalty0
  (168):\penalty0 1--53, 2021.

\bibitem[Miller and Dunson(2019)]{Miller2019-zj}
J.~W. Miller and D.~B. Dunson.
\newblock Robust {B}ayesian inference via coarsening.
\newblock \emph{Journal of the American Statistical Association}, 114\penalty0
  (527):\penalty0 1113--1125, 2019.

\bibitem[Minka(2000)]{Minka2000-rx}
T.~Minka.
\newblock Old and new matrix algebra useful for statistics, 2000.

\bibitem[Minka(2001)]{Minka2000-uw}
T.~P. Minka.
\newblock Automatic choice of dimensionality for {PCA}.
\newblock In \emph{Advances in Neural Information Processing Systems}, pages
  598--604, 2001.

\bibitem[Moignard et~al.(2015)Moignard, Woodhouse, Haghverdi, Lilly, Tanaka,
  Wilkinson, Buettner, Macaulay, Jawaid, Diamanti, Nishikawa, Piterman,
  Kouskoff, Theis, Fisher, and G{\"o}ttgens]{Moignard2015-lc}
V.~Moignard, S.~Woodhouse, L.~Haghverdi, A.~J. Lilly, Y.~Tanaka, A.~C.
  Wilkinson, F.~Buettner, I.~C. Macaulay, W.~Jawaid, E.~Diamanti, S.-I.
  Nishikawa, N.~Piterman, V.~Kouskoff, F.~J. Theis, J.~Fisher, and
  B.~G{\"o}ttgens.
\newblock Decoding the regulatory network of early blood development from
  single-cell gene expression measurements.
\newblock \emph{Nature Biotechnology}, 33\penalty0 (3):\penalty0 269--276, Mar.
  2015.

\bibitem[Pierson and Yau(2015)]{Pierson2015-fv}
E.~Pierson and C.~Yau.
\newblock {ZIFA}: Dimensionality reduction for zero-inflated single-cell gene
  expression analysis.
\newblock \emph{Genome Biology}, 16:\penalty0 241, Nov. 2015.

\bibitem[Pitman(2002)]{Pitman2002-ma}
J.~Pitman.
\newblock Combinatorial stochastic processes.
\newblock Technical Report 621, Dept of Statistics, UC Berkeley, 2002.

\bibitem[Qiu et~al.(2017)Qiu, Mao, Tang, Wang, Chawla, Pliner, and
  Trapnell]{qiu2017reversed}
X.~Qiu, Q.~Mao, Y.~Tang, L.~Wang, R.~Chawla, H.~A. Pliner, and C.~Trapnell.
\newblock Reversed graph embedding resolves complex single-cell trajectories.
\newblock \emph{Nature Methods}, 14\penalty0 (10):\penalty0 979, 2017.

\bibitem[Rezende(2018)]{Rezende2018-rp}
D.~J. Rezende.
\newblock Short notes on divergence measures.
\newblock July 2018.

\bibitem[Rockafellar(1970)]{Rockafellar1970-al}
R.~T. Rockafellar.
\newblock \emph{Convex Analysis}.
\newblock Princeton University Press, 1970.

\bibitem[Serfling(2009)]{Serfling2009-dq}
R.~J. Serfling.
\newblock \emph{Approximation Theorems of Mathematical Statistics}.
\newblock John Wiley \& Sons, Sept. 2009.

\bibitem[Shalek et~al.(2013)Shalek, Satija, Adiconis, Gertner, Gaublomme,
  Raychowdhury, Schwartz, Yosef, Malboeuf, Lu, Trombetta, Gennert, Gnirke,
  Goren, Hacohen, Levin, Park, and Regev]{Shalek2013-gb}
A.~K. Shalek, R.~Satija, X.~Adiconis, R.~S. Gertner, J.~T. Gaublomme,
  R.~Raychowdhury, S.~Schwartz, N.~Yosef, C.~Malboeuf, D.~Lu, J.~J. Trombetta,
  D.~Gennert, A.~Gnirke, A.~Goren, N.~Hacohen, J.~Z. Levin, H.~Park, and
  A.~Regev.
\newblock Single-cell transcriptomics reveals bimodality in expression and
  splicing in immune cells.
\newblock \emph{Nature}, 498\penalty0 (7453):\penalty0 236--240, June 2013.

\bibitem[Shao et~al.(2018)Shao, Jacob, Ding, and Tarokh]{Shao2018-ef}
S.~Shao, P.~E. Jacob, J.~Ding, and V.~Tarokh.
\newblock {B}ayesian model comparison with the {H}yv{\"a}rinen score:
  Computation and consistency.
\newblock \emph{Journal of the American Statistical Association}, pages 1--24,
  Sept. 2018.

\bibitem[Singer et~al.(2014)Singer, Yong, Tischler, Hackett, Altinok, Surani,
  Cai, and Elowitz]{Singer2014-lg}
Z.~S. Singer, J.~Yong, J.~Tischler, J.~A. Hackett, A.~Altinok, M.~A. Surani,
  L.~Cai, and M.~B. Elowitz.
\newblock Dynamic heterogeneity and {DNA} methylation in embryonic stem cells.
\newblock \emph{Molecular Cell}, 55\penalty0 (2):\penalty0 319--331, July 2014.

\bibitem[Sriperumbudur et~al.(2010)Sriperumbudur, Gretton, Fukumizu,
  Sch{\"o}lkopf, and Lanckriet]{Sriperumbudur2010-ew}
B.~K. Sriperumbudur, A.~Gretton, K.~Fukumizu, B.~Sch{\"o}lkopf, and G.~R.~G.
  Lanckriet.
\newblock Hilbert space embeddings and metrics on probability measures.
\newblock \emph{Journal of Machine Learning Research}, 11:\penalty0 1517--1561,
  2010.

\bibitem[Sriperumbudur et~al.(2011)Sriperumbudur, Fukumizu, and
  Lanckriet]{Sriperumbudur2011-lv}
B.~K. Sriperumbudur, K.~Fukumizu, and G.~R.~G. Lanckriet.
\newblock Universality, characteristic kernels and {RKHS} embedding of
  measures.
\newblock \emph{Journal of Machine Learning Research}, 12:\penalty0 2389--2410,
  2011.

\bibitem[Steinwart and Christmann(2008)]{Steinwart2008-sz}
I.~Steinwart and A.~Christmann.
\newblock \emph{Support Vector Machines}.
\newblock Springer Science \& Business Media, Sept. 2008.

\bibitem[Stuart et~al.(2019)Stuart, Butler, Hoffman, Hafemeister, Papalexi,
  Mauck, Hao, Stoeckius, Smibert, and Satija]{Stuart2019-gc}
T.~Stuart, A.~Butler, P.~Hoffman, C.~Hafemeister, E.~Papalexi, W.~M. Mauck,
  3rd, Y.~Hao, M.~Stoeckius, P.~Smibert, and R.~Satija.
\newblock Comprehensive integration of single-cell data.
\newblock \emph{Cell}, 177\penalty0 (7):\penalty0 1888--1902.e21, June 2019.

\bibitem[Townsend et~al.(2016)Townsend, Koep, and Weichwald]{Townsend2016-qk}
J.~Townsend, N.~Koep, and S.~Weichwald.
\newblock Pymanopt: A {P}ython toolbox for optimization on manifolds using
  automatic differentiation.
\newblock \emph{Journal of Machine Learning Research}, 17\penalty0
  (1):\penalty0 4755--4759, 2016.

\bibitem[van Dijk et~al.(2018)van Dijk, Sharma, Nainys, Yim, Kathail, Carr,
  Burdziak, Moon, Chaffer, Pattabiraman, Bierie, Mazutis, Wolf, Krishnaswamy,
  and Pe'er]{Van_Dijk2018-jw}
D.~van Dijk, R.~Sharma, J.~Nainys, K.~Yim, P.~Kathail, A.~J. Carr, C.~Burdziak,
  K.~R. Moon, C.~L. Chaffer, D.~Pattabiraman, B.~Bierie, L.~Mazutis, G.~Wolf,
  S.~Krishnaswamy, and D.~Pe'er.
\newblock Recovering gene interactions from single-cell data using data
  diffusion.
\newblock \emph{Cell}, 174\penalty0 (3):\penalty0 716--729.e27, July 2018.

\bibitem[Varin et~al.(2011)Varin, Reid, and Firth]{Varin2011-sy}
C.~Varin, N.~Reid, and D.~Firth.
\newblock An overview of composite likelihood methods.
\newblock \emph{Statistica Sinica}, 21\penalty0 (1):\penalty0 5--42, Jan. 2011.

\bibitem[Verdinelli and Wasserman(1998)]{Verdinelli1998-dx}
I.~Verdinelli and L.~Wasserman.
\newblock {B}ayesian goodness-of-fit testing using infinite-dimensional
  exponential families.
\newblock \emph{The Annals of Statistics}, 26\penalty0 (4):\penalty0
  1215--1241, Aug. 1998.

\bibitem[Vuong(1989)]{Vuong1989-gj}
Q.~H. Vuong.
\newblock Likelihood ratio tests for model selection and non-nested hypotheses.
\newblock \emph{Econometrica: Journal of the Econometric Society}, 57\penalty0
  (2):\penalty0 307--333, 1989.

\bibitem[Wolf et~al.(2018)Wolf, Angerer, and Theis]{Wolf2018-xz}
F.~A. Wolf, P.~Angerer, and F.~J. Theis.
\newblock {SCANPY}: Large-scale single-cell gene expression data analysis.
\newblock \emph{Genome Biology}, 19\penalty0 (1):\penalty0 15, Feb. 2018.

\bibitem[Xu-Monette et~al.(2016)Xu-Monette, Li, Byrd, Jabbar, Manyam, Maria~de
  Winde, van~den Brand, Tzankov, Visco, Wang, Dybkaer, Chiu, Orazi, Zu, Bhagat,
  Richards, Hsi, Choi, Huh, Ponzoni, Ferreri, M{\o}ller, Parsons, Winter, Wang,
  Hagemeister, Piris, Han~van Krieken, Medeiros, Li, van Spriel, and
  Young]{Xu-Monette2016-so}
Z.~Y. Xu-Monette, L.~Li, J.~C. Byrd, K.~J. Jabbar, G.~C. Manyam, C.~Maria~de
  Winde, M.~van~den Brand, A.~Tzankov, C.~Visco, J.~Wang, K.~Dybkaer, A.~Chiu,
  A.~Orazi, Y.~Zu, G.~Bhagat, K.~L. Richards, E.~D. Hsi, W.~W.~L. Choi, J.~Huh,
  M.~Ponzoni, A.~J.~M. Ferreri, M.~B. M{\o}ller, B.~M. Parsons, J.~N. Winter,
  M.~Wang, F.~B. Hagemeister, M.~A. Piris, J.~Han~van Krieken, L.~J. Medeiros,
  Y.~Li, A.~B. van Spriel, and K.~H. Young.
\newblock Assessment of {CD37} {B}-cell antigen and cell of origin
  significantly improves risk prediction in diffuse large {B}-cell lymphoma.
\newblock \emph{Blood}, 128\penalty0 (26):\penalty0 3083--3100, Dec. 2016.

\bibitem[Yeo and Johnson(2001)]{Yeo2001-jb}
I.-K. Yeo and R.~A. Johnson.
\newblock A uniform strong law of large numbers for {U}-statistics with
  application to transforming to near symmetry.
\newblock \emph{Statistics \& Probability Letters}, 51\penalty0 (1):\penalty0
  63--69, 2001.

\bibitem[Zhang(2006{\natexlab{a}})]{Zhang2006-sq}
T.~Zhang.
\newblock Information-theoretic upper and lower bounds for statistical
  estimation.
\newblock \emph{IEEE Transactions on Information Theory}, 52\penalty0
  (4):\penalty0 1307--1321, 2006{\natexlab{a}}.

\bibitem[Zhang(2006{\natexlab{b}})]{Zhang2006-wo}
T.~Zhang.
\newblock From $\epsilon$-entropy to {KL}-entropy: {A}nalysis of minimum
  information complexity density estimation.
\newblock \emph{The Annals of Statistics}, 34\penalty0 (5):\penalty0
  2180--2210, 2006{\natexlab{b}}.

\end{thebibliography}


\begin{thebibliography}{31}
\providecommand{\natexlab}[1]{#1}
\providecommand{\url}[1]{\texttt{#1}}
\expandafter\ifx\csname urlstyle\endcsname\relax
  \providecommand{\doi}[1]{doi: #1}\else
  \providecommand{\doi}{doi: \begingroup \urlstyle{rm}\Url}\fi

\bibitem[Barp et~al.(2019)Barp, Briol, Duncan, Girolami, and
  Mackey]{Barp2019-ut}
A.~Barp, F.-X. Briol, A.~B. Duncan, M.~Girolami, and L.~Mackey.
\newblock Minimum {S}tein discrepancy estimators.
\newblock \emph{arXiv preprint arXiv:1906.08283}, June 2019.

\bibitem[Berger and Guglielmi(2001)]{Berger2001-li}
J.~O. Berger and A.~Guglielmi.
\newblock Bayesian and conditional frequentist testing of a parametric model
  versus nonparametric alternatives.
\newblock \emph{Journal of the American Statistical Association}, 96\penalty0
  (453):\penalty0 174--184, 2001.

\bibitem[Bingham et~al.(2019)Bingham, Chen, Jankowiak, Obermeyer, Pradhan,
  Karaletsos, Singh, Szerlip, Horsfall, and Goodman]{Bingham2019-aa}
E.~Bingham, J.~P. Chen, M.~Jankowiak, F.~Obermeyer, N.~Pradhan, T.~Karaletsos,
  R.~Singh, P.~Szerlip, P.~Horsfall, and N.~D. Goodman.
\newblock Pyro: Deep universal probabilistic programming.
\newblock \emph{Journal of Machine Learning Research}, 20\penalty0
  (28):\penalty0 1--6, 2019.

\bibitem[Caticha(2004)]{Caticha2003-gb}
A.~Caticha.
\newblock Relative entropy and inductive inference.
\newblock \emph{AIP Conference Proceedings}, 707\penalty0 (1):\penalty0 75--96,
  2004.

\bibitem[Caticha(2011)]{Caticha2010-fm}
A.~Caticha.
\newblock Entropic inference.
\newblock \emph{AIP Conference Proceedings}, 1305\penalty0 (1):\penalty0
  20--29, 2011.

\bibitem[Dawid(2011)]{Dawid2011-kb}
A.~P. Dawid.
\newblock Posterior model probabilities.
\newblock In P.~S. Bandyopadhyay and M.~R. Forster, editors, \emph{Philosophy
  of Statistics}, volume~7, pages 607--630. North-Holland, Amsterdam, Jan.
  2011.

\bibitem[Dimitriev and Zhou(2021)]{Dimitriev2021-to}
A.~Dimitriev and M.~Zhou.
\newblock {ARMS}: {Antithetic-REINFORCE-Multi-Sample} gradient for binary
  variables.
\newblock In \emph{Proceedings of the 38th International Conference on Machine
  Learning}, 2021.

\bibitem[Ferguson(1974)]{Ferguson1974-qn}
T.~S. Ferguson.
\newblock Prior distributions on spaces of probability measures.
\newblock \emph{The Annals of Statistics}, 2\penalty0 (4):\penalty0 615--629,
  July 1974.

\bibitem[Folland(1999)]{folland1999real}
G.~B. Folland.
\newblock \emph{Real Analysis: Modern Techniques and Their Applications}.
\newblock John Wiley \& Sons, 1999.

\bibitem[Ghosh and Ramamoorthi(2003)]{Ghosh2003-yv}
J.~K. Ghosh and R.~V. Ramamoorthi.
\newblock \emph{{B}ayesian Nonparametrics}.
\newblock Series in Statistics. Springer, 2003.

\bibitem[Gigante et~al.(2020)Gigante, Burkhardt, Dager, Stanley, and
  Tong]{Gigante2020-sc}
S.~Gigante, D.~Burkhardt, D.~Dager, J.~Stanley, and A.~Tong.
\newblock scprep.
\newblock \url{https://github.com/KrishnaswamyLab/scprep}, 2020.

\bibitem[Gorham and Mackey(2017)]{Gorham2017-sd}
J.~Gorham and L.~Mackey.
\newblock Measuring sample quality with kernels.
\newblock In \emph{Proceedings of the 34th International Conference on Machine
  Learning - Volume 70}, pages 1292--1301, Sydney, NSW, Australia, 2017.

\bibitem[Grathwohl et~al.(2020)Grathwohl, Wang, Jacobsen, Duvenaud, and
  Zemel]{Grathwohl2020-pw}
W.~Grathwohl, K.-C. Wang, J.-H. Jacobsen, D.~Duvenaud, and R.~Zemel.
\newblock Learning the {S}tein discrepancy for training and evaluating
  energy-based models without sampling.
\newblock In \emph{Proceedings of the 37th International Conference on Machine
  Learning}, 2020.

\bibitem[Hong and Preston(2005)]{Hong2005-kf}
H.~Hong and B.~Preston.
\newblock Nonnested model selection criteria.
\newblock 2005.

\bibitem[Huggins and Mackey(2018)]{Huggins2018-lr}
J.~H. Huggins and L.~Mackey.
\newblock Random feature {S}tein discrepancies.
\newblock \emph{arXiv preprint arXiv:1806.07788}, 2018.

\bibitem[Kingma and Ba(2015)]{Kingma2015-ej}
D.~P. Kingma and J.~Ba.
\newblock Adam: {A} method for stochastic optimization.
\newblock In \emph{{ICLR}}, 2015.

\bibitem[Lavine(1992)]{Lavine1992-fu}
M.~Lavine.
\newblock Some aspects of {P}olya tree distributions for statistical modelling.
\newblock \emph{The Annals of Statistics}, 20\penalty0 (3):\penalty0
  1222--1235, 1992.

\bibitem[Liu et~al.(2016)Liu, Lee, and Jordan]{Liu2016-bp}
Q.~Liu, J.~D. Lee, and M.~Jordan.
\newblock A kernelized stein discrepancy for goodness-of-fit tests.
\newblock In \emph{International Conference on Machine Learning}, volume~33,
  pages 276--284, 2016.

\bibitem[Mauldin et~al.(1992)Mauldin, Sudderth, and Williams]{Mauldin1992-pd}
R.~D. Mauldin, W.~D. Sudderth, and S.~C. Williams.
\newblock Polya trees and random distributions.
\newblock \emph{The Annals of Statistics}, 20\penalty0 (3):\penalty0
  1203--1221, 1992.

\bibitem[Miller(2021)]{Miller2019-ur}
J.~W. Miller.
\newblock Asymptotic normality, concentration, and coverage of generalized
  posteriors.
\newblock \emph{Journal of Machine Learning Research}, 22\penalty0
  (168):\penalty0 1--53, 2021.

\bibitem[Minka(2001)]{Minka2000-uw}
T.~P. Minka.
\newblock Automatic choice of dimensionality for {PCA}.
\newblock In \emph{Advances in Neural Information Processing Systems}, pages
  598--604, 2001.

\bibitem[Rezende(2018)]{Rezende2018-rp}
D.~J. Rezende.
\newblock Short notes on divergence measures.
\newblock July 2018.

\bibitem[Rockafellar(1970)]{Rockafellar1970-al}
R.~T. Rockafellar.
\newblock \emph{Convex Analysis}.
\newblock Princeton University Press, 1970.

\bibitem[Serfling(2009)]{Serfling2009-dq}
R.~J. Serfling.
\newblock \emph{Approximation Theorems of Mathematical Statistics}.
\newblock John Wiley \& Sons, Sept. 2009.

\bibitem[Sriperumbudur et~al.(2010)Sriperumbudur, Gretton, Fukumizu,
  Sch{\"o}lkopf, and Lanckriet]{Sriperumbudur2010-ew}
B.~K. Sriperumbudur, A.~Gretton, K.~Fukumizu, B.~Sch{\"o}lkopf, and G.~R.~G.
  Lanckriet.
\newblock Hilbert space embeddings and metrics on probability measures.
\newblock \emph{Journal of Machine Learning Research}, 11:\penalty0 1517--1561,
  2010.

\bibitem[Sriperumbudur et~al.(2011)Sriperumbudur, Fukumizu, and
  Lanckriet]{Sriperumbudur2011-lv}
B.~K. Sriperumbudur, K.~Fukumizu, and G.~R.~G. Lanckriet.
\newblock Universality, characteristic kernels and {RKHS} embedding of
  measures.
\newblock \emph{Journal of Machine Learning Research}, 12:\penalty0 2389--2410,
  2011.

\bibitem[Steinwart and Christmann(2008)]{Steinwart2008-sz}
I.~Steinwart and A.~Christmann.
\newblock \emph{Support Vector Machines}.
\newblock Springer Science \& Business Media, Sept. 2008.

\bibitem[Stuart et~al.(2019)Stuart, Butler, Hoffman, Hafemeister, Papalexi,
  Mauck, Hao, Stoeckius, Smibert, and Satija]{Stuart2019-gc}
T.~Stuart, A.~Butler, P.~Hoffman, C.~Hafemeister, E.~Papalexi, W.~M. Mauck,
  3rd, Y.~Hao, M.~Stoeckius, P.~Smibert, and R.~Satija.
\newblock Comprehensive integration of single-cell data.
\newblock \emph{Cell}, 177\penalty0 (7):\penalty0 1888--1902.e21, June 2019.

\bibitem[Townsend et~al.(2016)Townsend, Koep, and Weichwald]{Townsend2016-qk}
J.~Townsend, N.~Koep, and S.~Weichwald.
\newblock Pymanopt: A {P}ython toolbox for optimization on manifolds using
  automatic differentiation.
\newblock \emph{Journal of Machine Learning Research}, 17\penalty0
  (1):\penalty0 4755--4759, 2016.

\bibitem[Wolf et~al.(2018)Wolf, Angerer, and Theis]{Wolf2018-xz}
F.~A. Wolf, P.~Angerer, and F.~J. Theis.
\newblock {SCANPY}: Large-scale single-cell gene expression data analysis.
\newblock \emph{Genome Biology}, 19\penalty0 (1):\penalty0 15, Feb. 2018.

\bibitem[Yeo and Johnson(2001)]{Yeo2001-jb}
I.-K. Yeo and R.~A. Johnson.
\newblock A uniform strong law of large numbers for {U}-statistics with
  application to transforming to near symmetry.
\newblock \emph{Statistics \& Probability Letters}, 51\penalty0 (1):\penalty0
  63--69, 2001.

\end{thebibliography}
\normalsize

\end{document}